\documentclass[a4paper,10pt,numbers=noenddot,twoside=on,headinclude=true,footinclude=false,headsepline=true]{scrartcl}

\usepackage[utf8]{inputenc}
\usepackage[T1]{fontenc}
\usepackage{textcomp}
\usepackage[english]{babel}
\usepackage{color}
\usepackage{amsmath}
\usepackage{amsopn}
\usepackage{amsbsy}
\usepackage{enumerate}
\usepackage{graphics}
\usepackage{units}
\usepackage[cal=boondoxo,bb=boondox]{mathalfa}
\usepackage{dsfont}
\usepackage{csquotes}

\usepackage[scaled]{beramono}
\usepackage{helvet}
\usepackage[charter]{mathdesign} 
\SetMathAlphabet{\mathcal}{normal}{OMS}{cmsy}{m}{n} 
\SetMathAlphabet{\mathcal}{bold}{OMS}{cmsy}{m}{n} 
\usepackage{diagxy}
\usepackage{marvosym}
\usepackage{todonotes}

\numberwithin{equation}{section} 
\numberwithin{table}{section} 
\numberwithin{figure}{section} 

\usepackage[amsmath,thmmarks,thref,hyperref]{ntheorem}

\theoremstyle{plain}
\newtheorem{theorem}{Theorem}[subsection]
\newtheorem{definition}[theorem]{Definition}
\newtheorem{lemma}[theorem]{Lemma}
\newtheorem{corollary}[theorem]{Corollary}

\newtheorem{proposition}[theorem]{Proposition}
\newtheorem{assumption}[theorem]{Assumption}
\newtheorem{metatheorem}[theorem]{Meta Theorem}

\theorembodyfont{\upshape}
\newtheorem{remark}[theorem]{Remark}

\theoremstyle{nonumberplain}

\theoremsymbol{\ensuremath{_\Box}}
\newtheorem{proof}{Proof}

\usepackage[style=alphabetic,
	citestyle=alphabetic,
	giveninits=true,
	backend=biber,
	hyperref=auto,
	sorting=nyt,
	maxcitenames=3,
	maxbibnames=99,
	minnames=3,
	arxiv=pdf,
	url=false,
	isbn=false,
	dateabbrev=true,
	datezeros=true,
	bibencoding=utf8]{biblatex}
\newbibmacro{string+doi}[1]{%
	\iffieldundef{doi}{#1}{\href{http://dx.doi.org/\thefield{doi}}{#1}}}
\DeclareFieldFormat
	{title}{\usebibmacro{string+doi}{\mkbibemph{#1}}\isdot}
\DeclareFieldFormat
	[article,inbook,incollection,inproceedings,patent,thesis,unpublished]
	{title}{\usebibmacro{string+doi}{\mkbibemph{#1}}\isdot}
\DeclareFieldFormat
	[article,inbook,incollection,inproceedings,patent,thesis,unpublished]
	{pages}{#1}
\DeclareFieldFormat
	[article]
	{journaltitle}{#1\isdot}
\DeclareFieldFormat
	[article]
	{volume}{\textbf{#1}}
\DefineBibliographyStrings{english}{pages={}}
\DeclareBibliographyDriver{article}{%
	\usebibmacro{bibindex}%
	\usebibmacro{begentry}%
	\usebibmacro{author/translator+others}%
	\setunit{\labelnamepunct}\newblock
	\usebibmacro{title}%
	\usebibmacro{journal+issuetitle}%
	\setunit*{\addcomma\space}
	\printfield{pages}
	\setunit*{\addcomma\space}
	\printfield{year}
	\usebibmacro{finentry}%
}

\renewbibmacro*{journal+issuetitle}{%
	\usebibmacro{journal}%
	\setunit*{\addspace}%
	\iffieldundef{series}
	{}
	{\newunit
		\printfield{series}%
		\setunit{\addspace}}%
	\iffieldundef{volume}
		{}
		{\printfield{volume}%
		\setunit{\addspace}}%
	\setunit*{\addcolon\space}%
	\usebibmacro{issue}%
	\newunit}

\providecommand{\ie}{i.~e.~}
\providecommand{\eg}{e.~g.~}
\providecommand{\cf}{cf.~}

\providecommand{\R}{\mathbb{R}}

\providecommand{\C}{\mathbb{C}}
\renewcommand{\C}{\mathbb{C}}

\providecommand{\T}{\mathbb{T}}
\renewcommand{\T}{\mathbb{T}}
\providecommand{\N}{\mathbb{N}}
\providecommand{\Z}{\mathbb{Z}}
\providecommand{\ii}{\mathrm{i}}
\providecommand{\e}{\mathrm{e}}
\renewcommand{\Re}{\mathrm{Re} \,}

\providecommand{\Hil}{\mathcal{H}}

\providecommand{\eps}{\varepsilon}
\providecommand{\Cont}{\mathcal{C}}

\providecommand{\ran}{\mathrm{ran} \, }

\providecommand{\supp}{\mathrm{supp} \,}

\providecommand{\trace}{\mathrm{Tr} \,}

\providecommand{\dd}{\mathrm{d}}
\providecommand{\id}{\mathds{1}}

\providecommand{\order}{\mathcal{O}}
\providecommand{\Fourier}{\mathcal{F}}

\providecommand{\trace}{\mathrm{Tr}}

\providecommand{\abs}[1]{\left \lvert #1 \right \rvert}
\providecommand{\sabs}[1]{\lvert #1 \vert}
\providecommand{\babs}[1]{\bigl \lvert #1 \bigr \rvert}

\providecommand{\norm}[1]{\left \lVert #1 \right \rVert}
\providecommand{\snorm}[1]{\lVert #1 \rVert}
\providecommand{\bnorm}[1]{\bigl \lVert #1 \bigr \rVert}
\providecommand{\Bnorm}[1]{\Bigl \lVert #1 \Bigr \rVert}
\providecommand{\scpro}[2]{\left \langle #1 , #2 \right \rangle}
\providecommand{\sscpro}[2]{\langle #1 , #2 \rangle}
\providecommand{\bscpro}[2]{\bigl \langle #1 , #2 \bigr \rangle}
\providecommand{\Bscpro}[2]{\Bigl \langle #1 , #2 \Bigr \rangle}

\providecommand{\sket}[1]{\vert #1 \rangle}
\providecommand{\bket}[1]{\bigl \vert #1 \bigr \rangle}

\providecommand{\sbra}[1]{\langle #1 \vert}
\providecommand{\bbra}[1]{\bigl \langle #1 \bigr \vert}

\providecommand{\sopro}[2]{\vert #1 \rangle \langle #2 \vert}

\providecommand{\expval}[1]{\left \langle #1 \right \rangle}
\providecommand{\sexpval}[1]{\langle #1 \rangle}
\providecommand{\bexpval}[1]{\bigl \langle #1 \bigr \rangle}

\providecommand{\ad}{\mathrm{ad}}
\providecommand{\BZ}{\T^*}
\providecommand{\Op}{\mathrm{Op}}
\providecommand{\Hper}{H_{\mathrm{per}}}
\providecommand{\eq}{\mathrm{eq}}
\providecommand{\per}{\mathrm{per}}
\providecommand{\Opeq}{\Op_{\eq}}

\providecommand{\Weyl}{\sharp^B}
\providecommand{\Int}{\mathrm{Int}}
\providecommand{\Wigner}{\mathcal{W}}

\providecommand{\Schwartz}{\mathcal{S}}

\providecommand{\WeylSys}{W}

\providecommand{\Schwartz}{\mathcal{S}}

\providecommand{\Fs}{\mathcal{F}_{\sigma}}

\providecommand{\KA}{{\mathsf{K}^A}}
\providecommand{\Reps}{{\mathsf{R}}}
\providecommand{\Zak}{\mathcal{Z}}
\providecommand{\domain}{\mathcal{D}}

\providecommand{\eq}{\mathrm{eq}}
\providecommand{\Hoer}[1]{S^{#1}}
\providecommand{\Hoereq}[1]{S_{\eq}^{#1}}

\providecommand{\SemiHoer}[1]{\mathrm{A}\Hoer{#1}}

\providecommand{\SemiHoereq}[1]{\mathrm{A} \Hoereq{#1}}

\providecommand{\MoyalSpace}{\mathcal{M}}

\usepackage{hyperref}
\usepackage{tikz}
\usetikzlibrary{fadings}
\usetikzlibrary{patterns}
\usetikzlibrary{shadows.blur}
\usetikzlibrary{shapes}

\begin{document}

\title{A Magnetic Pseudodifferential Calculus \\ for Operator-Valued and \\ Equivariant Operator-Valued Symbols}
\author{Giuseppe De Nittis${}^1$, Max Lein${}^2$ \& Marcello Seri${}^3$}

\maketitle
\vspace{-9mm}
\begin{center}
	${}^1$ Facultad de Matemáticas \& Instituto de Física, 
	Pontificia Universidad Católica de Chile \linebreak
	Avenida Vicuña Mackenna 4860, 
	Santiago, 
	Chile \linebreak
	\medskip
	\\
	${}^2$ Advanced Institute of Materials Research,  
	Tohoku University \linebreak
	2-1-1 Katahira, Aoba-ku, 
	Sendai, 980-8577, 
	Japan \linebreak
	\\
	${}^3$ Bernoulli Institute for Mathematics, Computer Science and Artificial Intelligence, University of Groningen \linebreak
	P.O. Box 407, 9700 AK Groningen, 
	The Netherlands \linebreak
\end{center}
\begin{abstract}
	In this monograph we develop magnetic pseudodifferential theory for oper\-ator-valued and \emph{equivariant} operator-valued functions and distributions from first principles. These have found plentiful applications in mathematical physics, including in rigorous perturbation theory for slow-fast systems and perturbed periodic operators. Yet, a systematic treatise was hitherto missing. While many of the results can be found piecemeal in appendices and as sketches in other articles, this article does contain new results. For instance, we have established Beals-type commutator criteria for both cases, which then imply the existence of Moyal resolvents for (equivariant) selfadjoint-operator-valued, elliptic Hörmander symbols and allows one to construct functional calculi. What is more, we give criteria on the function under which a magnetic pseudodifferential operator is (locally) trace class. 
	
	Our aims for this article are three-fold: (1)~Create a single, solid work that colleagues can refer to. (2)~Be pedagogical and precise. And (3)~give a straightforward strategy for extending results from the operator-valued to the equivariant case, pointing out some caveats and pitfalls that need to be kept in mind.
\end{abstract}
%

\tableofcontents

\section{Introduction} 
\label{intro}
Various incarnations of pseudodifferential theory have proven to be an invaluable tool in mathematics and mathematical physics. Nowadays there exist many different flavors, including pseudodifferential operators on manifolds \cite{Derezinski_al:Pseudodifferential_pseudoriemannian:2020,McKeag_Safarov:Pseudodifferential_manifolds:2011,Shubin:pseudodifferential:2001}, vector bundles \cite{Lampert_Teufel:adiabatic_limit_Schroedinger_operators_fiber_bundles:2014} and magnetic pseudodifferential operators \cite{Mantoiu_Purice:magnetic_Weyl_calculus:2004,Iftimie_Mantoiu_Purice:magnetic_psido:2006,Iftimie_Mantoiu_Purice:commutator_criteria:2008,Lein:progress_magWQ:2010,DeNittis_Lein:Bloch_electron:2009}. Applications range from the study of adiabatic \cite{PST:sapt:2002,Stiepan_Teufel:semiclassics_op_valued_symbols:2012} and Born-Oppenheimer-type systems \cite{PST:Born-Oppenheimer:2007}, the non- and semirelativistic limits of the Dirac equation \cite{Cordes:pseudodifferential_FW_transform:1983,Cordes:pseudodifferential_FW_transform:2004,Lein:two_parameter_asymptotics:2008,Fuerst_Lein:scaling_limits_Dirac:2008} as well as the study of effective dynamics in perturbed periodic quantum systems \cite{PST:effective_dynamics_Bloch:2003,Panati_Sparber_Teufel:polarization:2006,Freund_Teufel:non_trivial_Bloch_sapt:2013}
and their analogs in classical waves \cite{DeNittis_Lein:adiabatic_periodic_Maxwell_PsiDO:2013,DeNittis_Lein:sapt_photonic_crystals:2013,DeNittis_Lein:ray_optics_photonic_crystals:2014}. One of the reasons why pseudodifferential theory has proved so powerful and popular in applications is that we can make rigorous perturbation expansions and derive effective dynamics in the form of effective operators or semiclassical equations of motion. 

The main technical tool in the aforementioned publications are calculi for (magnetic) pseudodifferential operators defined from \emph{operator-valued} and \emph{equivariant}, operator-valued symbols. Unfortunately, the theory only exists piecemeal in appendices (\eg \cite[Appendix~A]{PST:effective_dynamics_Bloch:2003}, \cite[Appendix~A]{PST:sapt:2002}), sketches (\eg \cite[Section~2.2]{DeNittis_Lein:Bloch_electron:2009} or \cite[Section~4.1]{DeNittis_Lein:adiabatic_periodic_Maxwell_PsiDO:2013}) or as side remarks in the literature. In fact, the present work started life as an appendix to \cite{DeNittis_Lein_Seri:semiclassics_Bloch_electron_Fermi_surface:2021}. Given the broad utility of equivariant pseudodifferential operators, we believe it is past time to dedicate a work focused on this topic. 

As far as we can tell, the main reason why such a work does not yet exist is that, succinctly put, all major results amount to “this statement can be extended from $\Psi$DOs associated to scalar-valued symbols to $\Psi$DOs for (equivariant) operator-valued symbols.” To illustrate this, let us consider the case of the magnetic Weyl product $\Weyl$, which emulates the operator product on the level of symbols. A standard result from pseudodifferential theory states that it defines a continuous, bilinear map 
\begin{align*}
	\Weyl : S^{m_1}_{\rho,\delta} \times S^{m_2}_{\rho,\delta} \longrightarrow S^{m_1 + m_2}_{\rho,\delta}
\end{align*}
between Hörmander spaces (\cf \eg \cite[Proposition~(2.25)]{Folland:harmonic_analysis_hase_space:1989}); we will give precise definitions of $\Weyl$ and these Fréchet spaces below. This has a natural extension to operator-valued functions, 
\begin{align*}
	\Weyl : S^{m_1}_{\rho,\delta} \bigl ( \mathcal{B}(\Hil,\Hil') \bigr ) \times S^{m_2}_{\rho,\delta} \bigl ( \mathcal{B}(\Hil',\Hil'') \bigr ) \longrightarrow S^{m_1 + m_2}_{\rho,\delta} \bigl ( \mathcal{B}(\Hil,\Hil'') \bigr )
	, 
\end{align*}
where $\Hil$, $\Hil'$ and $\Hil''$ are Hilbert spaces; the proofs require only fairly straight-forward modifications. However, this need not be the case for more advanced results, \eg when one wants to generalize Beals-type commutator criteria to $\Psi$DOs defined from (equivariant) operator-valued symbols. For example, the first proof for scalar-valued symbols \cite{Iftimie_Mantoiu_Purice:commutator_criteria:2008} was fairly sophisticated and laborious, and we can see no obvious path for a straightforward extension without going into the details. Fortunately, Cornean, Helffer and Purice recently found a very elegant and brief proof of the commutator criteria in \cite{Cornean_Helffer_Purice:simple_proof_Beals_criterion_magnetic_PsiDOs:2018}.

One of the mathematical reasons for why this is not as straightforward as it would appear has to do with the subtleties of the topologies of the relevant spaces. Perhaps one is tempted to systematically extend existing results for scalar-valued symbols by expressing operator-valued symbol classes 
\begin{align*}
	S^m_{\rho,\delta} \bigl ( \mathcal{B}(\Hil,\Hil') \bigr ) \overset{\mbox{\Lightning}}{\cong} S^m_{\rho,\delta} \otimes \mathcal{B}(\Hil,\Hil')
\end{align*}
as a tensor product of scalar-valued symbol classes and the Banach space of bounded operators mapping $\Hil$ to $\Hil'$, and then leverage existing results. Unfortunately, unless $\Hil$ and $\Hil'$ are finite-dimensional, this is generally not true, because neither symbol spaces \cite{Witt:weak_topology_symbol_spaces:1997} nor infinite-dimensional Banach spaces are nuclear (\cf Definition~50.1, Theorem~50.1 and Corollary~2 in Chapter~50 of \cite{Treves:topological_vector_spaces:1967}). Hence, there are at least two ways to complete the algebraic tensor product and it is not obvious which, if any, reproduces the Fréchet topology of the space on the left (see \cite{Witt:weak_topology_symbol_spaces:1997}, specifically Theorem~3.2, Remark~3.3 and Proposition~4.4). Furthermore, equivariance is another obstacle to identifying $S^m_{\rho,\delta} \bigl ( \mathcal{B}(\Hil,\Hil') \bigr )$ with a tensor product even when $\Hil$ and $\Hil'$ are both finite-dimensional. 

That is why we believe there is merit to this paper even if none of our results may sound surprising to experts. 
\medskip

\noindent
We start by introducing the setting in Section~\ref{setting}. Next, we give the details of the extensions to $\Psi$DOs for operator-valued symbols (Section~\ref{operator_valued_calculus}) and equivariant operator-valued symbols (Section~\ref{equivariant_calculus}); apart from the fundamentals of the calculus, we will cover more advanced results such as Beals-type commutator criteria, inversion and functional calculi. We will put our results in perspective in Section~\ref{outlook} and give an outlook on potential and planned applications. Lastly, we include an appendix where we have collected some of the more technical arguments and proofs. 

\section{Defining the precise setting} 
\label{setting}
This section has a dual purpose: on the one hand, we will motivate the setting. And on the other, we will introduce suitable notation for a more abstract construction. However, to make the presentation more focussed, we will postpone some of the mathematical definitions until Sections~\ref{operator_valued_calculus} and \ref{equivariant_calculus}.

\subsection{The motivating example for this work} 
\label{setting:motivating_example}
So far we still have not explained to the reader what equivariant pseudodifferential operators are and why they appear naturally in the analysis of perturbed periodic systems. We will only give a very rough outline here and postpone the technical details until later; the interested reader may also look at \eg  \cite{PST:effective_dynamics_Bloch:2003,DeNittis_Lein:Bloch_electron:2009,DeNittis_Lein:sapt_photonic_crystals:2013} for additional context and rigorous definitions. 

A crystalline solid subjected to an external electromagnetic field in the one-electron approximation is described by the magnetic Schrödinger operator 
\begin{align}
  H^{A,\phi}_{\eps} = \bigl ( - \ii \nabla - A(\eps x) \bigr )^2 + V_{\mathrm{per}}(x) + \phi(\eps x)
	. 
	\label{setting:eqn:perturbed_periodic_Schroedinger_operator}
\end{align}
Here, the $A$ is a vector potential for the magnetic field $B = \dd A$. The adiabatic parameter $\eps > 0$ quantifies the difference in spatial scales on which the periodic potential $V_{\mathrm{per}}$ and the potentials $A$ and $\phi$ vary. To avoid inessential technical complications, let us assume that $V_{\mathrm{per}}$ is infinitesimally bounded with respect to the Laplacian. Hence, all operators below are selfadjoint on their standard domains.

As the notation suggests, the potential $V_{\mathrm{per}}$ is periodic, that is, $V_{\mathrm{per}}(x + \gamma) = V_{\mathrm{per}}(x)$ holds whenever 
\begin{align*}
	\gamma \in \Gamma := \mathrm{span}_{\Z} \bigl \{ e_1 , \ldots , e_d \bigr \}
	= \Bigl \{ \gamma = \mbox{$\sum_{j = 1}^d$} n_j \, e_j \; \; \big \vert \; \; n_1 , \ldots , n_d \in \Z \Bigr \}
\end{align*}
belongs to the lattice $\Gamma \cong \Z^d$ spanned by the (non-unique) fundamental lattice vectors $\{ e_1 , \ldots , e_d \}$. 

What works like \cite{PST:effective_dynamics_Bloch:2003,DeNittis_Lein:Bloch_electron:2009,DeNittis_Lein:sapt_photonic_crystals:2013} rest on is that $H_{\eps}^A$ can be regarded as a perturbation of the periodic magnetic Schrödinger operator 
\begin{align}
	\Hper := (- \ii \nabla)^2 + V_{\mathrm{per}}(x) 
	\label{setting:eqn:definition_Hper}
\end{align}
in the following sense: exploiting periodicity, a variant of the discrete Fourier transform called Bloch-Floquet-Zak transform, namely 
\begin{align}
	(\Zak \Psi)(k,y) := \sum_{\gamma \in \Gamma} \e^{- \ii k \cdot (y + \gamma)} \, \Psi(y + \gamma)
	, 
	\label{setting:eqn:Zak_transform}
\end{align}
fibers the periodic operator~\eqref{setting:eqn:definition_Hper} in Bloch momentum $k$ over a fundamental cell $\BZ$ of the dual lattice (\cf equation~\eqref{setting:eqn:Brillouin_zone}), 
\begin{align*}
	\Zak \, \Hper \, \Zak^{-1} = \int_{\BZ}^{\oplus} \dd k \, \Hper(k) 
	. 
\end{align*}
Then the perturbed operator~\eqref{setting:eqn:perturbed_periodic_Schroedinger_operator} can be related to a pseudodifferential operator 
\begin{align}
	\Zak \, H_{\eps}^A \, \Zak^{-1} = \Opeq^A(h)
	\label{setting:eqn:perturbed_periodic_Schroedinger_operators_as_magnetic_PsiDO}
\end{align}
associated to the operator-valued function  
\begin{align}
	h(r,k) = \Hper(k) + \phi(r) 
	\label{setting:eqn:symbol_perturbed_periodic_operator}
\end{align}
that is given in terms of the scalar potential $\phi(r)$ and the fiber operator 
\begin{align*}
	\Hper(k) = (- \ii \nabla + k)^2 + V_{\mathrm{per}}(y) 
\end{align*}
acting on $L^2(\T^d)$. 

\emph{One of the main aims of this paper is to give rigorous meaning to the symbol $\Opeq^A(h)$ in \eqref{setting:eqn:perturbed_periodic_Schroedinger_operators_as_magnetic_PsiDO} and develop a calculus for pseudodifferential operators defined from equivariant, operator-valued functions in Section~\ref{equivariant_calculus}.}

\subsubsection{$h$ is an operator-valued symbol} 
\label{setting:motivating_example:h_as_operator_valued_symbol}
We can already read off the last part: indeed, $(r,k) \mapsto h(r,k) = h(r,k)^*$ takes values in the selfadjoint operators and defines a bounded operator from the second Sobolev space $H^2(\T^d)$ to the Hilbert space $L^2(\T^d)$ over the fundamental cell in real space. Once we assume that $\phi$ lies in class $\Cont^{\infty}_{\mathrm{b}}(\R^d)$, the space of bounded, smooth functions with bounded derivatives to any order, then $h(r,k)$ grows quadratically in $k$ and remains uniformly bounded in $r$. Consequently, we may think of it as a Hörmander symbol that takes values in $\mathcal{B} \bigl ( H^2(\T^d) , L^2(\T^d) \bigr )$ and grows quadratically, 
\begin{align*}
	h \in S^2_{0,0} \bigl ( \mathcal{B} \bigl ( H^2(\T^d) , L^2(\T^d) \bigr ) \bigr ) 
	. 
\end{align*}
%

\subsubsection{The Zak transform} 
\label{setting:motivating_example:Zak_transform}
To understand what we mean by equivariant, we need to have a second look at the Zak transform. It basically exploits the periodicity by splitting real space coordinates 
\begin{align*}
	x \in \R^d \cong \T^d \times \Gamma \ni (y,\gamma)
\end{align*}
into a coordinate $y$ inside the fundamental cell in real space and a lattice vector $\gamma$; by periodicity, we identify the fundamental cell with a torus. Likewise, we can split the conjugate variable momentum after we introduce the reciprocal lattice 
\begin{align}
	\Gamma^* := \mathrm{span}_{\Z} \bigl \{ e_1^* , \ldots , e_d^* \bigr \}
	\label{setting:eqn:Brillouin_zone}
\end{align}
as the integer span of the base vectors $\{ e_1^* , \ldots , e_d^* \}$ that are characterized by 
\begin{align*}
	e_j \cdot e_n^* = 2 \pi \, \delta_{jn} 
	. 
\end{align*}
Then any momentum vector can be decomposed 
\begin{align*}
	p \in \R^d \cong \Gamma^* \times \BZ \ni (\gamma^*,k) 
\end{align*}
into a reciprocal lattice vector $\gamma^*$ and the component lying inside the first Brillouin zone 
\begin{align*}
	\BZ := \Bigl \{ k = \mbox{$\sum_{j = 1}^d$} \mu_j \, e_j^* \in \R^d \; \; \big \vert \; \; \mu_1 , \ldots , \mu_d \in [-\nicefrac{1}{2} , +\nicefrac{1}{2}] \Bigr \} 
	\cong \T^d 
	. 
\end{align*}
As our notation already suggests, we can often think of $\BZ$ as a $d$-dimensional torus, but we chose to add a ${}^{\ast}$ to make it easier for us and the reader to distinguish it from the unit cell $\T^d$ in real space. Group theoretically, $(\gamma,k)$ and $(y,\gamma^*)$ are conjugate pairs of variables as one can see from the definition~\eqref{setting:eqn:Zak_transform} of the Zak transform. 

On the dense set $\Schwartz(\R^d) \subset L^2(\R^d)$ a direct computation confirms 
\begin{subequations}\label{setting:eqn:periodicity_Zak_transform}
	\begin{align}
		(\Zak \Psi)(k,y+\gamma) &= (\Zak \Psi)(k,y)
		&&
		\forall \gamma \in \Gamma 
		, 
		\label{setting:eqn:periodicity_Zak_transform:y}
		\\
		(\Zak \Psi)(k+\gamma^*,y) &= \e^{+ \ii \gamma^* \cdot y} \, (\Zak \Psi)(k,y)
		=: \bigl ( \tau(\gamma^*) \Zak \Psi \bigr )(k,y)
		&&
		\forall \gamma^* \in \Gamma^* 
		\label{setting:eqn:periodicity_Zak_transform:k}
		, 
	\end{align}
\end{subequations}
namely the Zak transform is $\Gamma$-periodic in $y$ and $\Gamma^*$-quasiperiodic in $k$. 

The range of the Zak transform is 
\begin{align*}
	L^2_{\eq} \bigl ( \R^d,L^2(\T^d) \bigr ) := \Bigl \{ \psi \in L^2_{\mathrm{loc}} \bigl ( \R^d , L^2(\T^d) \bigr ) \; \; \big \vert \; \; \psi(k - \gamma^*) = \tau(\gamma^*) \psi(k) \; \; \forall \gamma^* \in \Gamma^* \mbox{ a.~e.~in $k$} \Bigr \} 
	, 
\end{align*}
which is a Hilbert space once we endow it with the scalar product 
\begin{align*}
	\scpro{\varphi}{\psi}_{\tau} := \int_{\BZ} \dd k \, \scpro{\varphi(k)}{\psi(k)}_{L^2(\T^d)} 
	. 
\end{align*}
Although, whenever convenient, we may identify it with 
\begin{align*}
	L^2(\BZ) \otimes L^2(\T^d) \cong L^2_{\eq} \bigl ( \R^d,L^2(\T^d) \bigr ) 
\end{align*}
with the help of Lemma~\ref{appendix:equivariant_operators:lem:weighted_Hilbert_spaces_translated_Brillouin_zones}. 

The operator $\tau(\gamma^*)$ that appears in the quasiperiodicity condition~\eqref{setting:eqn:periodicity_Zak_transform:k} can be regarded as a group representation of $\Gamma^*$: 
\begin{align*}
	\tau : \, &\Gamma^* \longrightarrow \mathcal{U} \bigl ( L^2(\T^d) \bigr )
	, 
	\\
	&\gamma^* \mapsto \e^{+ \ii \gamma^* \cdot \hat{y}} 
\end{align*}
While $\tau$ clearly takes values in the \emph{unitary} operators on $L^2(\T^d)$, by restriction we may also define $\gamma^* \mapsto \e^{+ \ii \gamma^* \cdot \hat{y}}$ on dense subsets of $L^2(\T^d)$ such as $H^m(\T^d)$; to distinguish the unitary group action $\tau$ on $L^2(\T^d)$ from the bounded-invertible-operator-valued group action 
\begin{align*}
	\tau' : \, &\Gamma^* \longrightarrow \mathcal{GL} \bigl ( H^m(\T^d) \bigr )
	, 
	\\
	&\gamma^* \mapsto \e^{+ \ii \gamma^* \cdot \hat{y}} 
	, 
\end{align*}
on the Sobolev space $H^m(\T^d)$, we have added a prime. Hilbert spaces like $L^2_{\eq} \bigl ( \R^d , H^m(\T^d) \bigr )$ are defined analogously to $L^2_{\eq} \bigl ( \R^d , L^2(\T^d) \bigr )$. Importantly, while $\tau'$ is not unitary, it has tempered growth in the sense that the operator norm 
\begin{align}
	\bnorm{\tau'(\gamma^*)}_{\mathcal{B}(H^m(\T^d))} \leq C_m \, \sexpval{\gamma^*}^m 
	\label{setting:eqn:tempered_growth_group_action_H_m}
\end{align}
grows like $\sabs{\gamma^*}^m$ for large $\sabs{\gamma^*}$; here, $\sexpval{\gamma^*} := \sqrt{1 + \sabs{\gamma^*}^2}$ is the Japanese bracket and $C_m > 0$ a suitable constant depending only on $m$. 

\subsubsection{Equivariant symbols} 
\label{setting:motivating_example:equivariant_symbols}
The periodicity of $\Hper$ now manifests itself in the $\Gamma^*$-equivariance of 
\begin{align*}
	\Hper(k - \gamma^*) = \tau(\gamma^*) \, \Hper(k) \, \tau'(\gamma^*)^{-1} 
	. 
\end{align*}
This \emph{equivariance condition} extends to the level of symbols as 
\begin{align}
	h(r,k - \gamma^*) = \tau(\gamma^*) \, h(r,k) \, \tau'(\gamma^*)^{-1} 
	. 
	\label{setting:eqn:equivariance:symbols}
\end{align}
Consequently, the symbol of the magnetic pseudodifferential operator $\Zak \, H^{A,\phi}_{\eps} \, \Zak^{-1}$ is an element of the Hörmander class composed of equivariant symbols of order $2$ (\cf Definition~\ref{equivariant_calculus:defn:equivariance_distribution}),  
\begin{align*}
	h \in S^2_{\eq} \bigl ( \mathcal{B} \bigl ( H^2(\T^d) , L^2(\T^d) \bigr ) \bigr ) 
	. 
\end{align*}
At first glance, it may seem that the equivariance condition~\eqref{setting:eqn:equivariance:symbols} is incompatible with the quadratic growth of the symbol in the momentum variable. Fortunately, this apparent contradiction can be resolved with the help of the growth estimate \eqref{setting:eqn:tempered_growth_group_action_H_m} for $m = 2$ and $m = 0$. Then indeed, for all $\gamma^* \in \Gamma^*$ we can verify that the operator norm 
\begin{align*}
	\bnorm{h(r,k - \gamma^*)}_{\mathcal{B}(H^2(\T^d),L^2(\T^d))} &\leq \bnorm{\tau'(\gamma^*)}_{H^m(\T^d)} \, \bnorm{\tau(\gamma^*)}_{L^2(\T^d)} \, \sup_{(r,k) \in \R^d \times \BZ} \bnorm{h(r,k)}_{\mathcal{B}(H^2(\T^d),L^2(\T^d))}
	\\
	&\leq C_m \, \sexpval{\gamma^*}^2 \, \sup_{(r,k) \in \R^d \times \BZ} \bnorm{h(r,k)}_{\mathcal{B}(H^2(\T^d),L^2(\T^d))}
	= C(h) \; \sexpval{\gamma^*}^2
\end{align*}
grows quadratically as $\sabs{\gamma^*} \rightarrow \infty$ — as it should. 

In fact, we may repeat the above estimate verbatim for arbitrary derivatives $\partial_r^a \partial_k^{\alpha} h$ of our symbol, and conclude that also they grow \emph{quadratically} in $k$. But this seems to conflict with the simple observation that $h(r,k)$ is a quadratic polynomial in $k$ and second-order momentum derivatives such as $\partial_{k_1}^2 h(r,k) \propto \id$ are proportional to the identity. Nevertheless, our conclusion is completely correct: when computing operator norms it matters whether we view \eg $\partial_{k_1}^2 h(r,k)$ as taking values in $\mathcal{B} \bigl ( H^2(\T^d) , L^2(\T^d) \bigr )$ or in $\mathcal{B} \bigl ( L^2(\T^d) \bigr )$. Lemma~\ref{equivariant_calculus:magnetic_PsiDOs:lem:bound_Hoermander_order_by_tau_orders} will show that the order of the symbol class only depends on the orders of the group actions — which is $2$ in the first case and $0$ in the second (\cf Lemma~\ref{setting:lem:magnetic_Sobolev_norm_operator_tau}). 

Therefore, we may view the space of equivariant Hörmander symbols of order $2$ 
\begin{align*}
	S^2_{\eq} \bigl ( \mathcal{B} \bigl ( H^2(\T^d) , L^2(\T^d) \bigr ) \bigr ) \subseteq S^2_{0,0} \bigl ( \mathcal{B} \bigl ( H^2(\T^d) , L^2(\T^d) \bigr ) \bigr )
\end{align*}
as a subspace of the operator-valued Hörmander symbol class of order $2$ and type $(0,0)$. 

\subsubsection{Utility of an equivariant magnetic pseudodifferential calculus} 
\label{setting:motivating_example:utility}
Before we abstract the setting, we think it is worthwhile to revisit the reasons for developing the theory of equivariant pseudodifferential operators from first principles. Our motivation was another upcoming work \cite{DeNittis_Lein_Seri:semiclassics_Bloch_electron_Fermi_surface:2021} where we will need a robust functional calculi for equivariant magnetic $\Psi$DOs in our proofs. These were beyond what one can achieve within a short appendix. So instead, we decided to build a framework within which the extension of more advanced results for $\Psi$DOs defined from scalar-valued Hörmander symbols to equivariant operator-valued symbols would be straightforward. As part of that effort, we have obtained a number of smaller results like the above mentioned Lemma~\ref{equivariant_calculus:magnetic_PsiDOs:lem:bound_Hoermander_order_by_tau_orders}, which simplifies some of the proofs in applications and clarifies some peculiarities of the equivariant calculus. 

Coming back to the question raised in the section header: what are the advantages of using pseudodifferential theory to analyze perturbed periodic operators in the first place? Big picture, we see three major advantages: first, it allows us to infer properties of the operator $\Opeq^A(f)$ from properties of the function $f$. Secondly, products of pseudodifferential operators are again pseudodifferential operators, 
\begin{align*}
	\Opeq^A(f) \, \Opeq^A(g) =: \Opeq^A(f \Weyl g) 
	, 
\end{align*}
and the operator product can be pulled back to the level of functions or distributions. Once we have a notion of product, we can also consider inverses with respect to the (magnetic) Weyl product $\Weyl$, which leads to Moyal resolvents $(f - z)^{(-1)_{\Weyl}}$; these are just the preimages of the resolvent operators under $\Opeq^A$,  
\begin{align*}
	\bigl ( \Opeq^A(f) - z \bigr )^{-1} = \Opeq^A \bigl ( (f - z)^{(-1)_{\Weyl}} \bigr ) 
	. 
\end{align*}
Under certain conditions, we can prove that these Moyal resolvents are not just tempered distributions, but functions; their properties are directly and straightforwardly related to properties of $f$. Moyal resolvents are then the basis for functional calculi on the level of functions on phase space, and it is these functional cacluli that are central to many arguments in the literature. 

The third advantage is that we can make many constructions \emph{microlocally}, \ie local in phase space: when a small parameter is present, then the Weyl product $\Weyl$ and hence, the Moyal resolvent have expansions in the small parameter. For example, we could scale the magnetic field $\lambda \, B$ with $\lambda \ll 1$ and consider the small magnetic field limit \cite{Cornean_Iftimie_Purice:Peierls_substitution_magnetic_PsiDOs:2019}. More commonly, one has a semiclassical parameter $\eps \ll 1$ in the construction, and we can expand 
\begin{align*}
	f \Weyl g \asymp \sum_{n = 0}^{\infty} \eps^n \, (f \Weyl g)_{(n)} 
\end{align*}
asymptotically in powers of $\eps$. The terms $(f \Weyl g)_{(n)}(x,\xi)$ only depend on $f(x,\xi)$, $g(x,\xi)$, the coefficients of the magnetic field $B_{lj}(x)$ as well as all their derivatives evaluated at $(x,\xi)$. Hence, the semiclassical expansion of the Moyal resolvent makes the calculus amenable to microlocalization. That is, in essence, the idea to construct the almost-invariant projection in \cite{Nenciu:effective_dynamics_Bloch:1991,PST:sapt:2002,PST:effective_dynamics_Bloch:2003,DeNittis_Lein:Bloch_electron:2009,DeNittis_Lein:sapt_photonic_crystals:2013,DeNittis_Lein:ray_optics_photonic_crystals:2014}. 

Other operators like effective hamiltonians can likewise be computed order-by-order, resulting in \eg an asymptotic expansion of the semi- and non-relativistic dynamics from the Dirac equation \cite{Fuerst_Lein:scaling_limits_Dirac:2008}. 

\subsection{Abstract setting} 
\label{setting:abstract}
The purpose of Section~\ref{setting:motivating_example} was to motivate the abstract setting, which we will detail now. To define the abstract version of $L^2_{\eq} \bigl ( \R^d , L^2(\T^d) \bigr )$, we will replace $L^2(\T^d)$ with some separable Hilbert space $\Hil$ that is endowed with a group action 
\begin{align*}
	\tau : \Gamma^* \longrightarrow \mathcal{GL}(\Hil)
\end{align*}
For consistency's sake, we shall stick to $\Gamma^*$ to denote our group even though $\Gamma^*$ can be mathematically identified with $\Z^d$. Similarly, we will continue to write $\BZ \cong \T^d$ for the fundamental cell associated to the lattice $\Gamma^*$. Then $\R^d \cong \Gamma^* \times \BZ$ splits into a lattice vector and a vector inside the fundamental cell. 

The Hilbert space of equivariant $\Hil$-valued $L^2$ functions
\begin{align}
	L^2_{\eq}(\R^d,\Hil) := \Bigl \{ \psi \in L^2_{\mathrm{loc}}(\R^d,\Hil) \; \; \big \vert \; \; &\psi(k - \gamma^*) = \tau(\gamma^*) \, \psi(k) 
	\Bigr . \notag \\
	&\Bigl . 
	\mbox{ for almost all $k \in \R^d$} , \; 
	\forall \gamma^* \in \Gamma^* \Bigr \} 
	\label{setting:eqn:equivariant_L2_space}
\end{align}
comes equipped with the scalar product
\begin{align}
	\scpro{\varphi}{\psi}_{L^2_{\eq}} := \int_{\BZ} \dd k \, \scpro{\varphi(k)}{\psi(k)}_{\Hil}
	. 
	\label{setting:eqn:equivariant_scalar_product}
\end{align}
Throughout our work the group action is assumed to have tempered growth of order $q$ in the following precise sense: 
\begin{definition}[Order of $\tau$]\label{setting:defn:order_tau}
	We say that the group representation $\tau : \Gamma^* \longrightarrow \mathcal{GL}(\Hil)$ is of order $q \in \R$ if there exists $C > 0$ so that
	\begin{align*}
		\bnorm{\tau(\gamma^*)}_{\mathcal{B}(\Hil)} \leq C \, \sexpval{\gamma^*}^q
	\end{align*}
	holds for all $\gamma^* \in \Gamma^*$.
\end{definition}
While we do not need this tempered growth assumption to define the Hilbert space, it is \emph{essential} if we want to view 
\begin{align}
	L^2_{\eq}(\R^d,\Hil) \subset \Schwartz^*_{\eq}(\R^d,\Hil) 
	\label{setting:eqn:embedding_L2_eq_Sprime_eq}
\end{align}
as a subset of \emph{equivariant} tempered distributions, which we will define below. Involving distributions becomes necessary since the standard construction of pseudodifferential operators, which is very elegantly explained in \cite[Section~II]{Mantoiu_Purice:magnetic_Weyl_calculus:2004}, first defines $\Op^A$ on Schwartz functions and then extends them by duality to $\Schwartz^*$ (the reason why we deviate from the standard notation and use $\Schwartz^*$, will become clear in the Section~\ref{operator_valued_calculus:extension_by_duality}).

To give rigorous meaning to $\Schwartz^*_{\eq}(\R^d,\Hil)$, we first extend translations 
\begin{align*}
	(T_{\gamma^*} \varphi)(k) := \varphi(k - \gamma^*) 
\end{align*}
from $\Schwartz(\R^d,\Hil)$ to their dual $\Schwartz^*(\R^d,\Hil)$ via 
\begin{align*}
	\bscpro{T_{\gamma^*} F \, }{ \, \varphi}_{\Schwartz(\R^d,\Hil)} := \bscpro{F \, }{ \, T_{-\gamma^*} \varphi}_{\Schwartz(\R^d,\Hil)}
	&&
	\forall F \in \Schwartz^*(\R^d,\Hil)
	, \; 
	\varphi \in \Schwartz(\R^d,\Hil)
	. 
\end{align*}
Flipping the sign is necessary if we want an extension that is consistent with
\begin{align*}
	\bscpro{T_{\gamma^*} \psi \, }{ \, \varphi}_{\Schwartz(\R^d,\Hil)} &= \int_{\BZ} \dd k \, \bscpro{\psi(k - \gamma^*)}{\varphi(k)}_{\Hil} 
	= \int_{\BZ} \dd k \, \bscpro{\psi(k)}{\varphi(k + \gamma^*)}_{\Hil} 
	\\
	&= \bscpro{\psi \, }{ \, T_{-\gamma^*} \varphi}_{\Schwartz(\R^d,\Hil)}
\end{align*}
when also $\psi \in \Schwartz(\R^d,\Hil)$ in the first argument is a test function. Fortunately, this is entirely consistent with translations on $L^2(\R^d,\Hil)$, where the Hilbert space adjoint $T_{\gamma^*}^* = T_{\gamma^*}^{-1} = T_{- \gamma^*}$ also equals the inverse translation. 

Similarly, but keeping into account the complex conjugation,
we define
\begin{align*}
	\bscpro{\tau(\gamma^*) \, F \, }{ \, \varphi}_{\Schwartz(\R^d,\Hil)} := \bscpro{F \, }{ \, \tau(\gamma^*)^* \, \varphi}_{\Schwartz(\R^d,\Hil)}
	, 
\end{align*}
where $\tau(\gamma^*)^*$ is the $\Hil$-adjoint of $\tau(\gamma^*)$. 
\begin{definition}[$\Schwartz^*_{\eq}(\R^d,\Hil)$]
	Suppose $\tau$ has tempered growth. Then the tempered distribution $F$ is equivariant if and only if 
	\begin{align*}
		T_{\gamma^*} F = \tau(\gamma^*) \, F 
	\end{align*}
	holds true for all $\gamma^* \in \Gamma^*$. The Fréchet space of equivariant tempered distributions is denoted with 
	\begin{align*}
		\Schwartz^*_{\eq}(\R^d,\Hil) := \bigl \{ F \in \Schwartz^*(\R^d,\Hil) \; \; \vert \; \; T_{\gamma^*} F = \tau(\gamma^*) \, F \quad \forall \gamma^* \in \Gamma^* \bigr \} 
		. 
	\end{align*}
\end{definition}
The tempered growth assumption is necessary to ensure that we can bound the action of $\tau(\gamma^*)$ on $\Schwartz(\R^d,\Hil)$ uniformly in $\gamma^*$, \ie the group action is \emph{uniformly} continuous. Hence, the Fréchet topology of $\Schwartz^*_{\eq}(\R^d,\Hil)$ is the restriction of the Fréchet topology of $\Schwartz^*(\R^d,\Hil)$ to the subspace of equivariant distributions. 

\subsection{Example of the generalized formalism: a condensed matter system subjected to strong magnetic fields} 
\label{setting:strong_internal_B}
The abstract setting from Section~\ref{setting:abstract} immediately applies to a generalization of our motivating example from Section~\ref{setting:motivating_example}. Here, the perturbed hamiltonian 
\begin{align}
  H^{A,\phi}_{\eps} = \bigl ( - \ii \nabla - A_0(x) - \lambda A(\eps x) \bigr )^2 + V_{\mathrm{per}}(x) + \phi(\eps x)
	. 
	\label{setting:eqn:perturbed_periodic_magnetic_Schroedinger_operator}
\end{align}
contains an additional vector potential $A_0$ for a $\Gamma$-periodic magnetic field $B_0 = \dd A_0$ with \emph{rational flux} through a unit cell. As opposed to the external magnetic field $\eps \lambda \, B(\eps x) = \dd \bigl ( \lambda A(\eps x) \bigr )$ the field $B_0(x)$ does not scale with $\eps$; it is in this sense that the magnetic field $B_0(x)$ is stronger than the external magnetic field $\eps \lambda \, B(\eps x)$. The most common example is that of a constant magnetic field $B_0 = \mathrm{const.}$, but restricting ourselves to this particular case is mathematically not necessary. 

The rational flux assumption on $B_0$ allows us to consider 
\begin{align*}
	\Hper = \bigl ( - \ii \nabla - A_0(x) \bigr )^2 + V_{\mathrm{per}}(x)
\end{align*}
as a periodic operator on a suitably chosen superlattice 
\begin{align*}
	\Gamma_{B_0} \subseteq \Gamma 
	. 
\end{align*}
Just like in Section~\ref{setting:motivating_example} the periodic operator 
\begin{align*}
	\Zak_{A_0} \, \Hper \, \Zak_{A_0}^{-1} = \int_{\BZ_{B_0}}^{\oplus} \dd k \, \Hper(k) 
	, 
\end{align*}
admits a fiber decomposition; the indices $A_0$ and $B_0$ emphasize the fact that the Zak transform $\Zak_{A_0}$ and the Brillouin torus $\BZ_{B_0}$ are those defined with respect to the super lattice $\Gamma_{B_0}$ and magnetic translations. The domain of $\Hper(k)$ is the \emph{magnetic} Sobolev space $H^2_{A_0}(\T^d)$ over the torus. However, the presence of the magnetic vector potential does not alter the tempered growth of the group action $\tau : \Gamma_{B_0} \longrightarrow \mathcal{GL} \bigl ( H^m_{A_0}(\T^d) \bigr )$. 
\begin{lemma}\label{setting:lem:magnetic_Sobolev_norm_operator_tau}
	Suppose the vector potential $A_0$ is polynomially bounded. Then the group action $\gamma^* \mapsto \e^{+ \ii \gamma^* \cdot \hat{y}}$ on the magnetic Sobolev space $H^m_{A_0}(\T^d)$ is of order $m$. 
\end{lemma}
The proof, which can be found in its entirety in Appendix~\ref{appendix:tempered_growth_magnetic_Sobolev_spaces_torus}, rests on the observation that the first term 
\begin{align*}
	\sexpval{- \ii \nabla_y + A_0(\hat{y})}^m \, \e^{+ \ii \gamma^* \cdot \hat{y}} \, \psi &= \e^{+ \ii \gamma^* \cdot \hat{y}} \, \sexpval{\gamma^* + A_0(\hat{y})}^m \, \psi + \e^{+ \ii \gamma^* \cdot \hat{y}} \, \sexpval{- \ii \nabla_y + A_0(\hat{y})}^m \, \psi 
\end{align*}
grows like $\sabs{\gamma^*}^m$ for large $\sabs{\gamma^*}$ in $\mathcal{B} \bigl ( L^2(\T^d) \bigr )$ norm. This example illustrates why the tempered growth condition is quite natural in many applications, but a necessity if we want to go from the concrete to the abstract. 

All that being said, we are once again in the abstract setting described in Section~\ref{setting:abstract}, and Section~\ref{equivariant_calculus} will explain how to make sense of 
\begin{align*}
  H^{A,\phi}_{\eps} = \Opeq^A \bigl ( \Hper + \phi \bigr ) 
\end{align*}
as a magnetic pseudodifferential operator. 

\section{Magnetic pseudodifferential calculus for operator-valued symbols} 
\label{operator_valued_calculus}
First, we will develop magnetic pseudodifferential theory for operator-valued symbols. Even if we were only interested in the equivariant case, this is a necessary first step. And given the plethora of applications of operator-valued pseudodifferential theory to a wide range of problems such as Born-Oppenheimer-type systems \cite{PST:Born-Oppenheimer:2007}, the Dirac equation \cite{Cordes:pseudodifferential_FW_transform:1983,Cordes:pseudodifferential_FW_transform:2004,Fuerst_Lein:scaling_limits_Dirac:2008} and many others, we believe the results of this section will be interesting in and of themselves. 

To be more general, we allow for the presence of a magnetic field $B$, although readers not familiar with the \emph{magnetic} variant of pseudodifferential theory \cite{Mueller:product_rule_gauge_invariant_Weyl_symbols:1999,Mantoiu_Purice:magnetic_Weyl_calculus:2004} can ignore it in much of what follows. In fact, we recommend \cite{Mantoiu_Purice:magnetic_Weyl_calculus:2004} even to readers for whom $B = 0$, since Măntoiu and Purice very elegantly explained the basic construction of pseudodifferential theory. 

For the first part of this section, where we construct the pseudodifferential calculus on $\Schwartz$ and $\Schwartz'$, polynomial growth of the magnetic field can be accommodated. \textbf{Therefore, unless specified otherwise, we will tacitly assume throughout this section that the following assumptions are verified:} 
\begin{assumption}[Polynomially bounded magnetic fields]\label{operator_valued_calculus:assumption:polynomially_bounded_B}
	\begin{enumerate}[(a)]
		\item The components of the magnetic field $B_{jl} \in \Cont^{\infty}_{\mathrm{pol}}(\R^d)$ are polynomially bounded, smooth and have polynomially bounded derivatives to any order. All vector potentials $A \in \Cont^{\infty}_{\mathrm{pol}}(\R^d,\R^d)$ for such magnetic fields $B = \dd A$ are of the same class. 
		\item All Hilbert spaces such as $\Hil$, $\Hil'$, etc.\ are separable. 
	\end{enumerate}
\end{assumption}

\subsection{Building block operators and the magnetic Weyl system} 
\label{operator_valued_calculus:building_block_operators}
A pseudodifferential calculus is built around two families of “building block” operators, namely position and momentum. In the so-called adiabatic representation they are given by 
\begin{subequations}\label{operator_valued_calculus:eqn:building_block_operators}
	\begin{align}
		Q &:= \eps \hat{x} 
		,
		\\
		P^A &:= - \ii \nabla - \lambda A(\eps \hat{x}) 
		, 
	\end{align}
\end{subequations}
where $\eps$ is the semiclassical parameter and $\lambda$ the coupling constant to the magnetic field. We regard these as selfadjoint operators on $L^2(\R^d)$ equipped with the usual domains. 

The operators $(Q,P^A)$ are unitarily equivalent to position and momentum operators in the perhaps more common scaling 
\begin{align*}
	Q' &:= \hat{x} 
	,
	\\
	P^{\prime \, A} &:= - \ii \eps \nabla - \lambda A(\hat{x})
	,
\end{align*}
where the semiclassical parameter appears in front of the derivative in the momentum operator. 

Nevertheless, both sets of operators will lead to the same pseudodifferential calculus as they implement the same commutation relations, which formally read 
\begin{align}
	\ii \, [Q_j , Q_l] = 0 
	, 
	&&
	\ii \, [P^A_j , P^A_l] = \eps \lambda \, B_{jl}(Q)
	, 
	&&
	\ii \, [P^A_j , Q_l] = \eps \delta_{jl} 
	. 
	\label{operator_valued_calculus:eqn:commutation_relations}
\end{align}
In fact, we can conjugate $Q$ and $P^A$ with any unitary operator, the calculus we are developing will be the same. Other relevant unitary transformations are the Fourier transform or the gauge transformation $\e^{+ \ii \lambda \vartheta(Q)}$, which implements the change of gauge $A \mapsto A + \eps \dd \vartheta$. 

The building block operators enter into the definition of the magnetic Weyl system 
\begin{align}
	W^A(X) := \e^{- \ii \sigma(X,(Q,P^A))} \otimes \id_{\Hil} \equiv \e^{- \ii \sigma(X,(Q,P^A))} 
	,
	\label{operator_valued_calculus:eqn:magnetic_Weyl_system}
\end{align}
where $X = (x,\xi) \in T^* \R^d$ and $Y = (y,\eta) \in T^* \R^d$ are phase space variables, $\sigma(X,Y) := \xi \cdot y - x \cdot \eta$ is the non-magnetic symplectic form and the exponential is defined in terms of functional calculus for the selfadjoint operator $\sigma(X,(Q,P^A)) = \xi \cdot Q - x \cdot P^A$. 

Later on, we will also commonly use the phase space variable $Z = (z,\zeta) \in T^* \R^d$. As a matter of convention, the small roman letters $x$, $y$ and $z$ denote the position components and the greek letters $\xi$, $\eta$ and $\zeta$ the momentum components. 

The Weyl system acts on $\Psi \in L^2(\R^d,\Hil)$ as 
\begin{align*}
	\bigl ( W^A(Y) \Psi \bigr )(x) = \e^{- \ii \frac{\lambda}{\eps} \int_{[\eps x , \eps x + \eps y]} A} \, \e^{- \ii \eps \eta \cdot (x + \frac{y}{2})} \, \Psi(x + y)
	. 
\end{align*}
Compared to the theory for scalar-valued symbols developed in \eg \cite{Mueller:product_rule_gauge_invariant_Weyl_symbols:1999,Mantoiu_Purice:magnetic_Weyl_calculus:2004,Iftimie_Mantoiu_Purice:magnetic_psido:2006,Iftimie_Mantoiu_Purice:commutator_criteria:2008,Iftimie_Purice:magnetic_Fourier_integral_operators:2011,Lein:two_parameter_asymptotics:2008,Lein:progress_magWQ:2010}, we have merely tensored on $\id_{\Hil}$. However, to unclutter the notation, we will only make “$\otimes \id_{\Hil}$” explicit when emphasis is needed and write $W^A(X)$ to mean both, $\e^{- \ii \sigma(X,(Q,P^A))} \otimes \id_{\Hil}$ and $\e^{- \ii \sigma(X,(Q,P^A))} \otimes \id_{\Hil'}$ even when $\Hil \neq \Hil'$. 

The magnetic Weyl system rigorously implements the commutation relations~\eqref{operator_valued_calculus:eqn:commutation_relations} since in general the Weyl system evaluated at different points $X \neq Y$ does not commute with itself, 
\begin{align}
	W^A(X) \, W^A(Y) = \e^{+ \ii \frac{\eps}{2} \sigma(X,Y)} \, \e^{+ \ii \frac{\lambda}{\eps} \Gamma^B(Q , Q + \eps x , Q + \eps x + \eps y)} \, W^A(X + Y) 
	. 
\end{align}
Indeed, the first phase factor is a consequence of $\ii \, [P^A_j , Q_l] = \eps \delta_{jl}$. The magnetic phase factor is the exponential of the magnetic flux through the triangle with corners $Q$, $Q + \eps x$ and $Q + \eps x + \eps y$, 
\begin{align*}
	\Gamma^B(q, q + \eps x , q + \eps x + \eps y) := \int_{\sexpval{q,q + \eps x,q + \eps x + \eps y}} B 
	, 
\end{align*}
and stems from $\ii \, [P^A_j , P^A_l] = \eps \lambda \, B_{jl}(Q)$. Note that the area of the triangle is $\order(\eps^2)$ so that the magnetic phase is again of $\order(\eps)$, just like the non-magnetic phase factor. 

The covariance of the Weyl system under unitary transformations is a direct consequence of functional calculus: if we fix a unitary $U$ and replace $(Q,P^A)$ by $Q_U := U \, Q \, U^{-1}$ and $P^A_U := U \, P^A \, U^{-1}$, then the associated Weyl system 
\begin{align}
	W_U^A(X) := \e^{- \ii \sigma(X,(Q_U,P_U^A))} 
	&= U \, \e^{- \ii \sigma(X,(Q,P^A))} \, U^{-1} 
	\notag \\
	&= U \, W^A(X) \, U^{-1} 
	\label{operator_valued_calculus:eqn:covariance_Weyl_system}
\end{align}
is related to the original Weyl system $W^A(X)$ by adjoining the unitary $U$. 

The gauge-covariance of the calculus emerges from a special case of equation~\eqref{operator_valued_calculus:eqn:covariance_Weyl_system}: for changes of gauge $U = \e^{+ \ii \lambda \vartheta(Q)}$ the above reduces to 
\begin{align*}
	\e^{+ \ii \lambda \vartheta(Q)} \, W^A(X) \, e^{- \ii \lambda \vartheta(Q)} &= W^{A + \eps \dd \vartheta}(X)
	. 
\end{align*}
%

\subsection{Rigorous construction on Schwartz spaces} 
\label{operator_valued_calculus:construction_on_Schwartz}
For much of the construction, we can repeat the arguments from \cite[Sections~III–V]{Mantoiu_Purice:magnetic_Weyl_calculus:2004} verbatim, the presence of the extra factor $\otimes \id_{\Hil}$ in the Weyl system does not introduce any extra complications.

\subsubsection{Exploiting the tensor product structure} 
\label{operator_valued_calculus:construction_on_Schwartz:tensor_product_structure}
At first glance, one might suppose that the majority of this section is unnecessary, especially that all it seems we are doing is re-tracing the steps laid out in the literature such as \cite{Mantoiu_Purice:magnetic_Weyl_calculus:2004,Iftimie_Mantoiu_Purice:magnetic_psido:2006,Lein:two_parameter_asymptotics:2008,Iftimie_Mantoiu_Purice:commutator_criteria:2008,Cornean_Helffer_Purice:simple_proof_Beals_criterion_magnetic_PsiDOs:2018}: it is tempting to first prove the relevant facts for operator-valued functions that are tensor products $f_j(r,k) = g_j(r,k) \, G_j$ where $g_j : \R^d \times \R^d \longrightarrow \C$ is an ordinary, scalar-valued Hörmander symbol and $G_j \in \mathcal{B}(\Hil,\Hil')$ some operator. Then we could bootstrap these arguments by approximating arbitrary operator-valued symbols 
\begin{align*}
	f \approx \sum_{j = 1}^n g_j \otimes G_j
\end{align*}
by finite linear combinations of tensor products, thereby extending all sorts of results almost immediately to operator-valued functions. Of particular interest are Hörmander symbol classes, where the $g_j \in S^m_{\rho,\delta}$ (\cf Definition~\ref{operator_valued_calculus:defn:Hoermander_symbols} below) are standard, scalar-valued Hörmander symbols and $f$ is an operator-valued symbol. However, this argument falsely assumes that finite linear combinations lie dense. And unless we are dealing with matrix-valued functions, that is, $\dim \Hil , \dim \Hil' < \infty$, this is unfortunately not correct. In fact, in general there exist at least two topologies with respect to which to complete the algebraic tensor product. 

Specifically, there may be more than one tensor product $\mathcal{X} = \mathcal{X}_1 \otimes \mathcal{X}_2$ of two locally convex topological vector spaces (\cf \cite[Chapter~43–48]{Treves:topological_vector_spaces:1967}). Only when $\mathcal{X}_1$ or $\mathcal{X}_2$ are nuclear (\cf \cite[Definition~50.1]{Treves:topological_vector_spaces:1967}) is the tensor product necessarily unique (\cf \cite[Theorem~50.1]{Treves:topological_vector_spaces:1967}). Infinite-dimensional Banach spaces such as $L^p(\T^d)$ are not nuclear \cite[Chapter~50, Corollary~2]{Treves:topological_vector_spaces:1967} and neither are Hörmander class symbol spaces (see \cite{Witt:weak_topology_symbol_spaces:1997}, specifically Theorem~3.2, Remark~3.3 and Proposition~4.4). 

Fortunately, though, we \emph{can} exploit the tensor product structure for the initial part of the construction: the topological vector space of Schwartz functions $\Schwartz(\R^d)$ \emph{is} nuclear as is its dual $\Schwartz^*(\R^d)$, the space of tempered distributions (\cf \cite[Chapter~51, p.~530, Corollary]{Treves:topological_vector_spaces:1967}); we will go into our reasons for labeling tempered distribution spaces with ${}^{\ast}$ rather than ${}^{\prime}$ in Section~\ref{operator_valued_calculus:extension_by_duality:relevant_spaces} below. Therefore, the completion of the algebraic tensor product is unique and we do have the identifications 
\begin{align*}
	\Schwartz^{(\ast)} \bigl ( T^* \R^d , \mathcal{B}(\Hil,\Hil') \bigr ) \cong \Schwartz^{(\ast)}(T^* \R^d) \otimes \mathcal{B}(\Hil,\Hil') 
	, 
\end{align*}
$\Schwartz^{(\ast)}(\R^d,\Hil) \cong \Schwartz^{(\ast)}(\R^d) \otimes \Hil$ and $\Schwartz^{(\ast)}(\R^d \times \R^d) \cong \Schwartz^{(\ast)}(\R^d) \otimes \Schwartz^{(\ast)}(\R^d)$ in the sense of the relevant Fréchet topologies (\cf \cite[Theorem~51.6 and its Corollary]{Treves:topological_vector_spaces:1967}). The brackets around ${}^{(\ast)}$ are meant to indicate that the above identifications hold for Schwartz spaces as well as tempered distribution spaces. 

\subsubsection{Magnetic Weyl quantization} 
\label{operator_valued_calculus:construction_on_Schwartz:magnetic_weyl_quantization}
Magnetic Weyl quantization can be viewed as a functional calculus associated to a set of non-commuting operators, \ie for suitable operator-valued functions $f : T^* \R^d \longrightarrow \mathcal{B}(\Hil,\Hil')$ we wish to make sense of “$f(Q,P^A)$”. We do this by Fourier transforming back and forth with the symplectic Fourier transform 
\begin{align}
	(\Fourier_{\sigma} f)(X) := \frac{1}{(2\pi)^d} \int_{T^* \R^d} \dd X' \, \e^{+ \ii \sigma(X,X')} \, f(X') 
	, 
	\label{operator_valued_calculus:eqn:symplectic_Fourier_transform}
\end{align}
which gives us the magnetic Weyl quantization of $f$, 
\begin{align}
	\Op^A(f) &= \frac{1}{(2\pi)^d} \int_{T^* \R^d} \dd X \, (\Fourier_{\sigma} f)(X) \, W^A(X)
	\label{operator_valued_calculus:eqn:definition_Op_A}
	\\
	&= \frac{1}{(2\pi)^d} \int_{T^* \R^d} \dd X \, \e^{- \ii \sigma(X,(Q,P^A))} \otimes (\Fourier_{\sigma} f)(X)
	. 
	\notag 
\end{align}
This integral is absolutely convergent for operator-valued Schwartz functions 
\begin{align*}
	f \in \Schwartz \bigl ( T^* \R^d , \mathcal{B}(\Hil,\Hil') \bigr ) 
	, 
\end{align*}
so at this stage it does not matter whether we define the Bochner integral (see \eg \cite[Section~5.5]{Teschl:real_analysis:2019}, \cite[Chapter V.5]{Yosida:Functional_analysis:1995}) in the norm topology, strong sense or weak sense. However, for consistency's sake, we shall interpret \eqref{operator_valued_calculus:eqn:definition_Op_A} in the weak sense, \ie look at “matrix elements” $\scpro{\varphi'}{\Op^A(f) \psi}_{L^2(\R^d,\Hil')}$ for test functions $\psi \in \Schwartz(\R^d,\Hil)$ and $\varphi' \in \Schwartz(\R^d,\Hil')$. 
\begin{definition}[Magnetic Weyl quantization]\label{operator_valued_calculus:defn:magnetic_Weyl_quantization}
	The magnetic pseudodifferential operator associated with $f \in \Schwartz \bigl ( T^* \R^d , \mathcal{B}(\Hil,\Hil') \bigr )$ is given by \eqref{operator_valued_calculus:eqn:definition_Op_A}, where the integral is defined in the weak sense as a map 
	\begin{align*}
		\Op^A(f) : \Schwartz(\R^d,\Hil) \longrightarrow \Schwartz(\R^d,\Hil') 
		. 
	\end{align*}
\end{definition}
We could have worked with the regular Fourier transform in the above definition, but the symplectic Fourier transform~\eqref{operator_valued_calculus:eqn:symplectic_Fourier_transform} has the added benefit of being its own inverse, $\Fourier_{\sigma}^2 = \id$. Hence, \emph{formally} equation~\eqref{operator_valued_calculus:eqn:definition_Op_A} reduces to 
\begin{align*}
	\Op^A(f) = (\Fourier_{\sigma}^2 f)(Q,P^A) 
	= f(Q,P^A) 
	. 
\end{align*}
The fact that Schwartz functions possess an integrable symplectic Fourier transform immediately yields the following 
\begin{lemma}\label{operator_valued_calculus:lem:Weyl_quantization_Schwartz_class}
	Consider any $f \in \Schwartz \bigl ( T^* \R^d , \mathcal{B}(\Hil,\Hil') \bigr )$. Then the following holds true: 
	\begin{enumerate}[(1)]
		\item The associated magnetic pseudodifferential operator 
		\begin{align*}
			\Op^A(f) : \Schwartz(\R^d,\Hil) \longrightarrow \Schwartz(\R^d,\Hil') 
		\end{align*}
		is a continuous linear operator between Fréchet spaces. 
		\item $\Op^A(f)$ extends to a continuous, \ie bounded operator 
		\begin{align*}
			\Op^A(f) : L^2(\R^d,\Hil) \longrightarrow L^2(\R^d,\Hil')
		\end{align*}
		between $L^2$ spaces that we denote with the same symbol. 
	\end{enumerate}
\end{lemma}
Magnetic Weyl quantization inherits the covariance~\eqref{operator_valued_calculus:eqn:covariance_Weyl_system} of the Weyl system: should we use the operators $(Q_U,P^A_U) = \bigl ( U \, Q \, U^{-1} \, , \, U \, P^A \, U^{-1} \bigr )$ that have been conjugated with a unitary $U$ instead, then the associated pseudodifferential operator 
\begin{align*}
	\Op_U^A(f) &= U \, \Op^A(f) \, U^{-1} 
\end{align*}
is unitarily equivalent to $\Op^A(f)$. For example, if $U = \e^{+ \ii \lambda \vartheta(Q)}$ is a gauge transformation, then the above reads 
\begin{align}
	\Op^{A + \eps \dd \vartheta}(f) &= \e^{+ \ii \lambda \vartheta(Q)} \, \Op^A(f) \, \e^{- \ii \lambda \vartheta(Q)} 
	. 
	\label{operator_valued_calculus:eqn:covariance_Op_A}
\end{align}
\begin{lemma}\label{operator_valued_calculus:lem:covariance_Weyl_quantization_Schwartz_class}
	Magnetic Weyl quantization $\Op^A$ on $\Schwartz \bigl ( T^* \R^d , \mathcal{B}(\Hil,\Hil') \bigr )$ is gauge-covariant in the sense of equation~\eqref{operator_valued_calculus:eqn:covariance_Op_A}.
\end{lemma}
Note that the other common recipe to treat magnetic system, combining ordinary Weyl calculus for the operators $Q$ and non-magnetic momentum $P = - \ii \nabla$ with minimal substitution, is \emph{not} gauge-invariant; we refer to \eg \cite[Section~6.2]{Iftimie_Mantoiu_Purice:magnetic_psido:2006} and \cite[Chapter~2.2.1]{Lein:progress_magWQ:2010} for details. The advantage of our covariant prescription is that all essential properties of the pseudodifferential operator only depend on the magnetic field rather than on our choice of vector potential. 

\subsubsection{The integral kernel map} 
\label{operator_valued_calculus:construction_on_Schwartz:integral_kernel_map}
While not strictly necessary for the definition, the magnetic pseudodifferential operator 
\begin{align*}
	\Op^A(f) := \Int(K^A_f) 
\end{align*}
can be expressed in terms of its operator kernel 
\begin{align}
	K^A_f(x,y) := \e^{- \ii \frac{\lambda}{\eps} \int_{[\eps x , \eps y]} A} \frac{1}{(2\pi)^d} \int_{\R^d} \dd \eta \, \e^{- \ii (y - x) \cdot \eta} \, f \bigl ( \tfrac{\eps}{2} (x + y) , \eta \bigr ) 
	\in \mathcal{B}(\Hil,\Hil')
	\label{operator_valued_calculus:eqn:operator_kernel}
\end{align}
and the integral map 
\begin{align*}
	\bigl ( \Int(K) \Psi \bigr )(x) := \int_{\R^d} \dd y \, K(x,y) \, \Psi(y) 
	\in \Hil' 
	. 
\end{align*}
This will be insightful conceptually and useful practically when extending $\Op^A$ to tempered distributions. 
\begin{lemma}\label{operator_valued_calculus:lem:kernel_map_Schwartz_class}
	The kernel map $f \mapsto K_f^A$ defined through \eqref{operator_valued_calculus:eqn:operator_kernel} is a topological vector space isomorphism 
	\begin{align*}
		\Schwartz \bigl ( T^* \R^d , \mathcal{B}(\Hil,\Hil') \bigr ) \longrightarrow \Schwartz \bigl ( \R^d \times \R^d , \mathcal{B}(\Hil,\Hil') \bigr )
	\end{align*}
	between Schwartz spaces. 
\end{lemma}
\begin{proof}
	We can repeat the arguments in the beginning of \cite[Section~IV.B]{Mantoiu_Purice:magnetic_Weyl_calculus:2004} verbatim, but working with Bochner integrals this time to take the operator-valuedness of $f$ into account. The kernel map is a composition of multiplication by the $\Cont^{\infty}_{\mathrm{pol}}$ phase function $\e^{- \ii \frac{\lambda}{\eps} \int_{[\eps x , \eps y]} A}$, a linear, invertible change of variables and a partial Fourier transform. All of these are continuous maps between Schwartz spaces, and hence, as a composition of continuous maps between Schwartz functions, the kernel map $f \mapsto K_f^A$ is continuous. 
	
	Further, all of these are evidently invertible with continuous inverses, which proves that $f \mapsto K_f^A$ is indeed a topological vector space isomorphism. 
\end{proof}
%

\subsubsection{The adjoint} 
\label{operator_valued_calculus:construction_on_Schwartz:adjoint}
When $\Hil \subseteq \Hil'$ is a dense subspace of the target space, we may interpret $f(X) \in \mathcal{B}(\Hil,\Hil')$ as taking values in the unbounded operators on $\Hil'$ with $X$-independent domain. 

For such operators we can define the adjoint $\Op^A(f)^*$ via the pointwise adjoint 
\begin{align*}
	\bscpro{\Op^A(f)^* \Phi'}{\Psi}_{L^2(\R^d,\Hil')} := \bscpro{\Phi'}{\Op^A(f) \Psi}_{L^2(\R^d,\Hil')}
	, 
\end{align*}
where we have used the same scalar product on both sides of the equation since $\Psi \in \Schwartz(\R^d,\Hil)$ can also be regarded as an element of $\Schwartz(\R^d,\Hil') \subset L^2(\R^d,\Hil')$. A quick computation confirms that the adjoint of the quantization is the quantization of the adjoint, 
\begin{align*}
	\Op^A(f)^* = \Op^A(f^*) 
	, 
	&&
	f \in \Schwartz \bigl ( T^* \R^d , \mathcal{B}(\Hil,\Hil') \bigr ) 
	. 
\end{align*}
%

\subsubsection{The magnetic Wigner transform} 
\label{operator_valued_calculus:construction_on_Schwartz:magnetic_wigner_transform}
Conceptually speaking, the Wigner transform 
\begin{align}
	(\Wigner^A K)(X) = \frac{1}{(2\pi)^{\nicefrac{d}{2}}}\int_{\R^d} \dd y \, \e^{- \ii y \cdot \xi} \, \e^{- \ii \frac{\lambda}{\eps} \int_{[x - \frac{\eps}{2} y , x + \frac{\eps}{2} y]} A} \, K \bigl ( \tfrac{x}{\eps} + \tfrac{y}{2} , \tfrac{x}{\eps} - \tfrac{y}{2} \bigr )
	\in \mathcal{B}(\Hil,\Hil')
	\label{operator_valued_calculus:eqn:Wigner_transform}
\end{align}
fulfills two roles: first of all, in applications it is used to relate quantum states with quasi-classical states. And secondly, it implements the inverse to $\Op^A$, that is, if $F = \Int(K_F)$ is an operator with Schwartz class kernel $K_F \in \Schwartz \bigl ( \R^d \times \R^d , \mathcal{B}(\Hil,\Hil') \bigr )$, then we have 
\begin{align*}
	\bigl ( \Op^A \bigr )^{-1}(F) = \Wigner^A K_F 
	. 
\end{align*}
The quickest way to see this is to compare the kernel map~\eqref{operator_valued_calculus:eqn:operator_kernel} with the Wigner transform~\eqref{operator_valued_calculus:eqn:Wigner_transform}: all the operations are inverses to one another and thus, $\Wigner^A$ is the inverse to $f \mapsto K^A_f$. 

Mathematically, the Wigner transform is implicitly defined through the relation 
\begin{align*}
	\int_{T^* \R^d} \dd X \, \mathrm{Tr}_{\Hil'} \Bigl ( f(X) \, \Wigner^A \bigl ( K_{\sopro{\Phi'}{\Psi}} \bigr )(X) \Bigr ) := \scpro{\Phi'}{\Op^A(f) \Psi}_{L^2(\R^d,\Hil')}
	, 
\end{align*}
where $f \in \Schwartz \bigl ( T^* \R^d , \mathcal{B}(\Hil,\Hil') \bigr )$, $\Psi \in \Schwartz(\R^d,\Hil)$ and $\Phi' \in \Schwartz(\R^d,\Hil')$ are all Schwartz class. 
\begin{lemma}\label{operator_valued_calculus:lem:Wigner_transform_Schwartz_class}
	The Wigner transform~\eqref{operator_valued_calculus:eqn:Wigner_transform} defines a topological vector space isomorphism 
	\begin{align*}
		\Wigner^A : \Schwartz \bigl ( \R^d \times \R^d , \mathcal{B}(\Hil,\Hil') \bigr ) \longrightarrow \Schwartz \bigl ( T^* \R^d , \mathcal{B}(\Hil,\Hil') \bigr )
		. 
	\end{align*}
	Its inverse is the kernel map~\eqref{operator_valued_calculus:eqn:operator_kernel}. 
\end{lemma}
\begin{proof}
	An explicit computation verifies that the Wigner transform is the inverse of the kernel map. And since the latter is a topological vector space isomorphism by Lemma~\ref{operator_valued_calculus:lem:kernel_map_Schwartz_class}, so is the Wigner transform $\Wigner^A$. 
\end{proof}
%

\subsubsection{The magnetic Weyl product} 
\label{operator_valued_calculus:construction_on_Schwartz:magnetic_weyl_product}
The magnetic Weyl product $\Weyl$ pulls back the operator product to the level of functions or distributions on phase space, 
\begin{align*}
	\Op^A(f \Weyl g) := \Op^A(f) \, \Op^A(g) 
	. 
\end{align*}
Of course, for the above equation to make sense, the ranges of $f$ and $g$ need to be composable, that is, $f$ has to take values in $\mathcal{B}(\Hil,\Hil')$ and $g$ must map onto $\mathcal{B}(\Hil',\Hil'')$. 

As the notation suggests, $\Weyl$ only depends on the magnetic field, something that can be inferred from gauge-covariance~\eqref{operator_valued_calculus:eqn:covariance_Op_A} and confirmed by deriving an explicit integral formula for it, 
\begin{align}
	(f \Weyl g)(X) = \frac{1}{(2\pi)^{2d}} \int_{T^* \R^d} \dd Y \int_{T^* \R^d} \dd Z \, &\e^{+ \ii \sigma(X,Y+Z)} \, \e^{+ \ii\frac{\eps}{2} \sigma(Y,Z)} \, \e^{- \ii \lambda \gamma^B(x,y,z)} 
	\, \cdot \notag \\
	&\cdot 
	(\Fourier_{\sigma} f)(Y) \, (\Fourier_{\sigma} g)(Z) 
	, 
	\label{operator_valued_calculus:eqn:Weyl_product}
\end{align}
where we have introduced the abbreviation 
\begin{align*}
	\gamma^B(x,y,z) := \tfrac{1}{\eps} \, \Gamma^B \bigl ( x - \tfrac{\eps}{2}(y+z) \, , x + \tfrac{\eps}{2}(y-z) \, , x + \tfrac{\eps}{2}(y+z) \bigr ) 
	. 
\end{align*}
There are other equivalent integral expressions for the product, \eg \cite[equation~(36)]{Mantoiu_Purice:magnetic_Weyl_calculus:2004}. 

Retracing the arguments for the scalar-valued case, replacing the absolute value with operator norms where appropriate, yields the following 
\begin{proposition}
	The magnetic Weyl product
	\begin{align*}
		\Weyl : \Schwartz \bigl ( T^* \R^d , \mathcal{B}(\Hil,\Hil') \bigr ) \times \Schwartz \bigl ( T^* \R^d , \mathcal{B}(\Hil',\Hil'') \bigr ) \longrightarrow \Schwartz \bigl ( T^* \R^d , \mathcal{B}(\Hil,\Hil'') \bigr )
	\end{align*}
	defines a bilinear, continuous map between Schwartz spaces, and its explicit expression as an absolutely convergent integral is given by equation~\eqref{operator_valued_calculus:eqn:Weyl_product}. 
\end{proposition}
\begin{proof}
	The proof is a straightforward modification of \cite[Proposition~13]{Mantoiu_Purice:magnetic_Weyl_calculus:2004} or \cite[Theorem~2.2.17]{Lein:progress_magWQ:2010}: we just need to replace absolute values by operator products. Estimates on the magnetic phase $\e^{+ \ii \lambda \gamma^B(x,y,z)}$ carry over verbatim. 
\end{proof}
While we will postpone a more in-depth explanation of the asymptotic expansion of 
\begin{align*}
	f \Weyl g = f \, g - \eps \, \tfrac{\ii}{2} \{ f , g \}_{\lambda B} + \order(\eps^2)
	, 
\end{align*}
let us take the opportunity to introduce the magnetic Poisson bracket 
\begin{align*}
	 \{ f , g \}_{\lambda B} := \nabla_{\xi} f \cdot \nabla_x g - \nabla_x f \cdot \nabla_{\xi} g - \lambda \, \sum_{j , k = 1}^d B_{jk} \, \partial_{\xi_j} f \, \partial_{\xi_k} g 
\end{align*}
and mention that even when $\Hil = \Hil' = \Hil''$ the lowest order term $f(X) \, g(X) \neq g(X) \, f(X)$ is no longer commutative. 

\subsection{Extension to tempered distributions by duality} 
\label{operator_valued_calculus:extension_by_duality}
In spirit the extension by duality proceeds just as in the scalar-valued case. However, there are some additional technical complications we need to be cognizant of when adapting the definitions via the duality bracket. The first step is to enumerate and characterize the spaces of Schwartz functions and tempered distributions that will be involved.

\subsubsection{Relevant spaces of Schwartz functions and distrributions} 
\label{operator_valued_calculus:extension_by_duality:relevant_spaces}
For this subsubsection and this subsubsection alone, for any Fréchet space $\mathcal{X}$ let $\mathcal{X}'$ denote its strong dual; in case $\mathcal{X}$ is a Banach space, this is nothing but the Banach space dual. To avoid confusion, within this subsubsection we will number Hilbert spacesi $\Hil_1,\ \Hil_2,\ \ldots$ rather than add primes. 

We need to introduce distributions to complete our rigged Hilbert space or Gel'fand triple \cite{berezanskii1968expansions}
\begin{align}
	\Schwartz(\R^d,\Hil) \subset L^2(\R^d,\Hil) \subset \Schwartz^*(\R^d,\Hil) 
	. 
	\label{operator_valued_calculus:eqn:rigged_Hilbert_space_S_Sstar}
\end{align}
Evidently, we can also consider magnetic Sobolev spaces $H^m_A(\R^d,\Hil)$ as rigged Hilbert spaces in the above sense. Here, we prefer to use ${}^*$ rather than ${}'$ to indicate that instead of the standard bilinear duality bracket, we use the “scalar product”-like, sesquilinear duality bracket 
\begin{align*}
	\scpro{\varphi}{\psi}_{\Schwartz(\R^d,\Hil)} := \int_{\R^d} \dd x \, \bscpro{\varphi(x)}{\psi(x)}_{\Hil} 
	.  
\end{align*}
Clearly, any construction such as extending derivatives and other linear continuous operators on $\Schwartz(\R^d,\Hil)$ via the ordinary, bilinear duality bracket can be directly translated to the sesquilinear one. 

The advantage of using $\scpro{\, \cdot \,}{\, \cdot \,}_{\Schwartz(\R^d,\Hil)}$ is that this leads to the canonical identification of 
\begin{align*}
	\Schwartz(\R^d,\Hil)^* = \Schwartz^*(\R^d,\Hil)
	. 
\end{align*}
Had we used the standard duality bracket instead, then  we would have had to distinguish between $\Hil$ and continuous linear functionals on $\Hil$. Specifically, according to \cite[p.~534]{Treves:topological_vector_spaces:1967} the strong dual can be thought of in two ways: 
\begin{align*}
	\Schwartz(\R^d,\Hil)' = \Schwartz'(\R^d,\Hil')
\end{align*}
After those definitions, we can consider linear continuous operators acting between these spaces. Given two Fréchet spaces $\mathcal{X}$ and $\mathcal{Y}$, we denote the space of linear continuous operators between them with $\mathcal{L}(\mathcal{X},\mathcal{Y})$. Then the rigged Hilbert space~\eqref{operator_valued_calculus:eqn:rigged_Hilbert_space_S_Sstar} immediately leads to two continuous linear inclusions between spaces of continuous linear operators, 
\begin{align*}
	\mathcal{L} \bigl ( \Schwartz^*(\R^d,\Hil_1) \, , \, \Schwartz(\R^d,\Hil_2) \bigr ) \subset \mathcal{B} \bigl ( L^2(\R^d,\Hil_1) \, , \, L^2(\R^d,\Hil_2) \bigr ) 
	\subset \mathcal{L} \bigl ( \Schwartz(\R^d,\Hil_1) \, , \, \Schwartz^*(\R^d,\Hil_2) \bigr )
	. 
\end{align*}
These inclusions are the basis for the (extension of) the 
\begin{lemma}[Schwartz Kernel Theorem]\label{operator_valued_calculus:lem:Schwartz_kernel_theorem}
	Any element $F$ of $\mathcal{B} \bigl ( L^2(\R^d,\Hil_1) \, , \, L^2(\R^d,\Hil_2) \bigr )$ has a distributional operator kernel $K_F \in \Schwartz^* \bigl ( \R^d \times \R^d \, , \, \mathcal{B}(\Hil_1,\Hil_2) \bigr )$. 
\end{lemma}
The second class of tempered distributions enters when we wish to make sense of $\Op^A(F)$ when 
\begin{align*}
	F \in \Schwartz^* \bigl ( T^* \R^d , \mathcal{B}(\Hil_1,\Hil_2) \bigr ) 
	. 
\end{align*}
We are in the unusual position that we have determined what dual space we wish to work with and need to find the pre-dual. Importantly, the spaces of Schwartz functions and tempered distributions are nuclear (\cf \cite[p.~530, Corollary]{Treves:topological_vector_spaces:1967}), which means we can think of the Schwartz space of Banach space $\mathcal{X}$-valued functions 
\begin{align*}
	\Schwartz(T^* \R^d,\mathcal{X}) \cong \Schwartz(T^* \R^d) \otimes \mathcal{X}
\end{align*}
as \emph{the} tensor product of the (scalar-valued) Schwartz space $\Schwartz(T^* \R^d)$ and $\mathcal{X}$. In general there are several topologies with respect to which we can complete the algebraic tensor product (composed of finite linear combinations), but for nuclear spaces these all necessarily coincide (\cf \cite[Theorem~50.1]{Treves:topological_vector_spaces:1967}). 

The tensor product decomposition tells us that all we need to focus on is the Banach space $\mathcal{X}$ and its dual. Generally, we can think of the strong dual of the $\mathcal{X}$-valued Schwartz functions as 
\begin{align*}
  \Schwartz(T^* \R^d,\mathcal{X})' = \Schwartz'(\R^{2d},\mathcal{X}') 
	. 
\end{align*}
Importantly, $\Schwartz'(T^* \R^d,\mathcal{X}) \neq \Schwartz'(\R^{2d},\mathcal{X}')$ need not coincide. Put another way, we have to look for a Banach space $\mathcal{X}$ so that its Banach space dual $\mathcal{X}' = \mathcal{B}(\Hil_1,\Hil_2)$ gives us the bounded operators. 

Fortunately, we know the answer (\cf \cite[Theorem~48.5']{Treves:topological_vector_spaces:1967}), it is the space of nuclear operators (\cf \cite[Definition~47.2]{Treves:topological_vector_spaces:1967}); although we shall refer to it as the space of \emph{trace class operators} $\mathcal{L}^1(\Hil_2,\Hil_1)$ since $F \in \mathcal{L}^1(\Hil_2,\Hil_1)$ is equivalent to imposing 
\begin{align*}
  \sqrt{F^* F} \in \mathcal{L}^1(\Hil_2)
  \quad\mbox{or, equivalently,}\quad 
  \sqrt{F F^*} \in \mathcal{L}^1(\Hil_1)
\end{align*}
are trace class in the usual sense.
For the benefit of the reader, we have included all pertinent definitions and specific references in Appendix~\ref{appendix:identification_operator_valued_Schwartz_spaces_distributions}. 

Clearly, we can repeat the construction from Section~\ref{operator_valued_calculus:construction_on_Schwartz} for $\Schwartz \bigl ( T^* \R^d , \mathcal{L}^1(\Hil_2,\Hil_1) \bigr )$, we just have to use the trace norm instead of the operator norm when defining Bochner integrals. The duality bracket now involves the Hilbert space trace $\trace_{\Hil_1}$, for two Schwartz functions $f \in \Schwartz \bigl ( T^* \R^d , \mathcal{L}^1(\Hil_2,\Hil_1) \bigr )$ and $g \in \Schwartz \bigl ( T^* \R^d , \mathcal{L}^1(\Hil_1,\Hil_2) \bigr )$ we set 
\begin{align*}
	(f , g)_{\Schwartz(T^* \R^d , \mathcal{L}^1(\Hil_2,\Hil_1))} := \int_{T^* \R^d} \dd X \, \trace_{\Hil_2} \bigl ( f(X) \, g(X) \bigr ) 
	. 
\end{align*}
An important fact we will exploit later is that the cyclicity of the trace allows us to swap $f$ and $g$, 
\begin{align*}
	(f , g)_{\Schwartz(T^* \R^d , \mathcal{L}^1(\Hil_2,\Hil_1))} = (g , f)_{\Schwartz(T^* \R^d , \mathcal{L}^1(\Hil_1,\Hil_2))} 
	. 
\end{align*}
One annoying point of basing our arguments on the Banach space and strong duals is that in equations like the one above we need to sometimes exchange the order of the Hilbert spaces. That also prevents us from exploiting the inclusion 
\begin{align*}
	\Schwartz \bigl ( T^* \R^d , \mathcal{L}^1(\Hil_1,\Hil_2) \bigr ) \subset \Schwartz' \bigl ( T^* \R^d , \mathcal{L}^1(\Hil_1,\Hil_2) \bigr )
\end{align*}
in a natural way. So also here we will opt for a sesquilinear duality bracket 
\begin{align*}
	\scpro{F}{g}_{\Schwartz(T^* \R^d , \mathcal{L}^1(\Hil_1,\Hil_2))} := \int_{T^* \R^d} \dd X \, \trace_{\Hil_1} \bigl ( F^*(X) \, g(X) \bigr ) 
\end{align*}
and denote the associated space of distributions with 
\begin{align*}
	\Schwartz^* \bigl ( T^* \R^d , \mathcal{B}(\Hil_1,\Hil_2) \bigr ) = \Schwartz \bigl ( T^* \R^d , \mathcal{L}^1(\Hil_1,\Hil_2) \bigr )^* 
	. 
\end{align*}
Thanks to the inclusion $\mathcal{L}^1(\Hil_1,\Hil_2) \subseteq \mathcal{B}(\Hil_1,\Hil_2)$ of the trace class operators into the bounded operators, we can now view the space of operator-valued tempered distributions \emph{as an extension} of the class of test functions it is based on. 

\subsubsection{The magnetic quantization rule} 
\label{operator_valued_calculus:extension_by_duality:Weyl_quantization}
With the definitions of the relevant spaces out of the way, the extension is very easy. We begin with the kernel map~\eqref{operator_valued_calculus:eqn:operator_kernel}. 
\begin{lemma}\label{operator_valued_calculus:lem:extension_kernel_map_distributions}
	The kernel map $f \mapsto K^A_f$ from equation~\eqref{operator_valued_calculus:eqn:operator_kernel} extends to a topological vector space isomorphism on the level of distributions, 
	\begin{align*}
		f \mapsto K^A_f : \Schwartz^* \bigl ( T^* \R^d , \mathcal{B}(\Hil,\Hil') \bigr ) \longrightarrow \Schwartz^* \bigl ( \R^d \times \R^d , \mathcal{B}(\Hil,\Hil') \bigr )
		. 
	\end{align*}
\end{lemma}
\begin{proof}
	The kernel map is a composition of a partial Fourier transform, a linear, invertible change of coordinates on $T^* \R^d$ and multiplication by the $\Cont^{\infty}_{\mathrm{pol}}$ phase $\e^{- \ii \frac{\lambda}{\eps} \int_{[\eps x , \eps y]} A}$. All of them define continuous maps on operator-valued Schwartz functions that extend continuously to operator-valued tempered distributions. 
\end{proof}
We can now extend the map $\Op^A$ in two ways: we can either stick to $\Op^A(f)$ for Schwartz functions and extend the space these operators act on. Or we can define $\Op^A(F)$ for tempered distributions.  
\begin{proposition}\label{operator_valued_calculus:prop:extension_OpA_tempered_distributions}
	\begin{enumerate}[(1)]
		\item $\Op^A$ extends to the topological vector space isomorphism 
		\begin{align*}
			\Op^A : \Schwartz \bigl ( T^* \R^d , \mathcal{L}^1(\Hil,\Hil') \bigr ) \longrightarrow \mathcal{L} \bigl ( \Schwartz^*(\R^d,\Hil) \, , \, \Schwartz(\R^d,\Hil') \bigr )
			. 
		\end{align*}
		\item $\Op^A$ extends to the topological vector space isomorphism 
		\begin{align*}
			\Op^A : \Schwartz^* \bigl ( T^* \R^d , \mathcal{B}(\Hil,\Hil') \bigr ) \longrightarrow \mathcal{L} \bigl ( \Schwartz(\R^d,\Hil) \, , \, \Schwartz^*(\R^d,\Hil') \bigr ) 
			. 
		\end{align*}
		\item The two extensions from (1) and (2) are still gauge-covariant, \ie for any $\vartheta \in \Cont^{\infty}_{\mathrm{pol}}(\R^d,\R)$ we have 
		\begin{align*}
			\Op^{A + \eps \dd \vartheta}(f) &= \e^{+ \ii \lambda \vartheta(Q)} \, \Op^A(f) \, \e^{- \ii \lambda \vartheta(Q)}
		\end{align*}
		where $f$ is either an operator-valued Schwartz function or distribution. 
	\end{enumerate}
\end{proposition}
In a nutshell, the proof of \cite[Proposition~5]{Mantoiu_Purice:magnetic_Weyl_calculus:2004} — and our generalization — rely on tensor product decompositions of the various spaces. The rather technical proofs are relegated to Appendix~\ref{appendix:identification_operator_valued_Schwartz_spaces_distributions}. 

\subsubsection{The magnetic Wigner transform} 
\label{operator_valued_calculus:extension_by_duality:magnetic_wigner_transform}
Likewise, the Wigner transform has an extension to tempered distributions. Fortunately, seeing as the Wigner transform is the inverse of the kernel map $f \mapsto K_f^A$, all of the work has already been completed when proving Lemma~\ref{operator_valued_calculus:lem:extension_kernel_map_distributions}. So the following corollary is a combination of Lemmas~\ref{operator_valued_calculus:lem:kernel_map_Schwartz_class}, \ref{operator_valued_calculus:lem:Wigner_transform_Schwartz_class} and \ref{operator_valued_calculus:lem:extension_kernel_map_distributions}: 
\begin{corollary}
	The Wigner transform~\eqref{operator_valued_calculus:eqn:Wigner_transform} extends to a continuous map 
	\begin{align*}
		\Wigner^A : \Schwartz^* \bigl ( \R^d \times \R^d \, , \, \mathcal{B}(\Hil,\Hil') \bigr ) \longrightarrow \Schwartz^* \bigl ( T^* \R^d \, , \, \mathcal{B}(\Hil,\Hil') \bigr ) 
	\end{align*}
	between operator-valued tempered distributions. 
\end{corollary}
%

\subsubsection{The magnetic Weyl product} 
\label{operator_valued_calculus:extension_by_duality:the_magnetic_weyl_product}
To extend the magnetic Weyl product from Schwartz functions to tempered distributions in spirit we are going to follow the arguments in \cite[Section~V.B]{Mantoiu_Purice:magnetic_Weyl_calculus:2004} for the scalar-valued (as opposed to operator-valued) case: the extension rests on 
\begin{align*}
	\int_{T^* \R^d} \dd X \, (f \Weyl g)(X) &= \int_{T^* \R^d} \dd X \, f(X) \, g(X) = \bscpro{\overline{f}}{g} 
	= \bscpro{\overline{g}}{f} 
	&&
	\forall f , g \in \Schwartz(T^* \R^d) 
	, 
\end{align*}
which allows us to shift the Weyl product from one argument in the duality bracket to the other. The analog of the above formula for the operator-valued case is as follows: 
\begin{lemma}\label{operator_valued_calculus:lem:int_Weyl_product_equals_int_product}
	Let $f , g \in \Schwartz \bigl ( T^* \R^d , \mathcal{L}^1(\Hil,\Hil') \bigr )$ be two Schwartz functions. Then the following equality holds: 
	\begin{align}
		\int_{T^* \R^d} \dd X \, \trace_{\Hil} \bigl ( f^* \Weyl g(X) \bigr ) &= \int_{T^* \R^d} \dd X \, \trace_{\Hil} \bigl ( f^*(X) \, g(X) \bigr ) 
		= \scpro{f}{g}_{\Schwartz(T^* \R^d , \mathcal{L}^1(\Hil,\Hil'))} 
		\label{operator_valued_calculus:eqn:int_Weyl_product_equals_int_product}
	\end{align}
\end{lemma}
\begin{proof}
	We will only sketch the proof, since it is a straightforward extension of the arguments in \eg \cite[Lemma~14]{Mantoiu_Purice:magnetic_Weyl_calculus:2004} or \cite[Lemma~2.3.6]{Lein:progress_magWQ:2010}. 
	
	First of all, the three integral expressions in equation~\eqref{operator_valued_calculus:eqn:int_Weyl_product_equals_int_product} are well-defined, because any trace class operator is bounded and thus, its adjoint is also bounded. The product of a bounded and a trace class operator is again trace class. (In fact, Lemma~\ref{operator_valued_calculus:lem:facts_needed_proof_extension_product}~(1) below tells us that $f^*$ takes values in the trace class operators again.)
	
	To make the following arguments rigorous, we need to regularize the integral in $X$ on the left hand side of \eqref{operator_valued_calculus:eqn:int_Weyl_product_equals_int_product}. Starting with equation~\eqref{operator_valued_calculus:eqn:Weyl_product}, we first integrate over $\xi$, which gives a $(2\pi)^d \, \delta(y + z)$. Then we can eliminate the integral in, say, $z$. After setting $z = -y$, we see that the magnetic flux triangle collapses on one side and the exponential of the magnetic flux is $0$; this eliminates the $x$-dependence of the integrand besides the exponential $\e^{- \ii x \cdot (\eta + \zeta)}$. 
	
	Hence, we are now free to integrate over $x$, which produces $(2\pi)^d \, \delta(\eta + \zeta)$. Eliminating the $\zeta$ integral and using the unitarity of the symplectic Fourier transform now yields the claim. Note that all these phase factors are scalar and can be pulled out of the trace. Moreover, the potentially infinite sum implicit in the trace can be controlled using Dominated Convergence for sums. 
\end{proof}
Before we proceed with the modified statement, we will collect a few elementary facts needed in the subsequent proofs:
\begin{lemma}\label{operator_valued_calculus:lem:facts_needed_proof_extension_product}
	\begin{enumerate}[(1)]
		\item $f \in \Schwartz \bigl ( T^* \R^d , \mathcal{L}^1(\Hil,\Hil') \bigr )$ implies the pointwise adjoint 
		\begin{align*}
			f^* \in \Schwartz \bigl ( T^* \R^d , \mathcal{L}^1(\Hil',\Hil) \bigr ) 
		\end{align*}
		also takes values in the trace class operators. 
		\item For any $f \in \Schwartz \bigl ( T^* \R^d , \mathcal{L}^1(\Hil',\Hil'') \bigr )$ and $g \in \Schwartz \bigl ( T^* \R^d , \mathcal{L}^1(\Hil,\Hil') \bigr )$ we have 
		\begin{align*}
			(f \Weyl g)^* = g^* \Weyl f^*
			. 
		\end{align*}
	\end{enumerate}
\end{lemma}
\begin{proof}
	\begin{enumerate}[(1)]
		\item It suffices to prove this pointwise. One equivalent way to define $\mathcal{L}^1(\Hil,\Hil') \ni A$ is to impose the condition 
		\begin{align*}
			\sqrt{A^* \, A} \in \mathcal{L}^1(\Hil) 
			\quad \mbox{or, equivalently,} \quad 
			\sqrt{A \, A^*} \in \mathcal{L}^1(\Hil') 
		\end{align*}
		on its elements (\cf Definition~\ref{appendix:extension_by_duality_operator_valued_calculus:defn:trace_class_operators} and the Corollary on \cite[p.~494]{Treves:topological_vector_spaces:1967}). Hence, $f(X) \in \mathcal{L}^1(\Hil,\Hil')$ holds exactly when the adjoint $f^*(X) = f(X)^* \in \mathcal{L}^1(\Hil',\Hil)$ is trace class as well. 
		\item We can prove this using the intertwining between the operator and the Weyl products as well as the two adjoints via $\Op^A$: 
		\begin{align*}
			\Op^A \bigl ( (f \Weyl g)^* \bigr ) &= \bigl ( \Op^A(f \Weyl g) \bigr )^* 
			= \bigl ( \Op^A(f) \, \Op^A(g) \bigr )^* 
			\\
			&= \Op^A(g)^* \, \Op^A(f)^* 
			= \Op^A(g^*) \, \Op^A(f^*) 
			\\
			&= \Op^A(g^* \Weyl f^*)
		\end{align*}
	\end{enumerate}
\end{proof}
As a way to motivate the correct extension to tempered distributions, we will prove the following identities for Schwartz functions first. 
\begin{proposition}\label{operator_valued_calculus:prop:relation_Weyl_product_Schwartz_functions_duality_bracket}
	For all $f \in \Schwartz \bigl ( T^* \R^d , \mathcal{L}^1(\Hil',\Hil'') \bigr )$, $g \in \Schwartz \bigl ( T^* \R^d , \mathcal{L}^1(\Hil,\Hil') \bigr )$ and $h \in \Schwartz \bigl ( T^* \R^d , \mathcal{L}^1(\Hil,\Hil'') \bigr )$ the following identities hold true: 
	\begin{enumerate}[(1)]
		\item $\displaystyle \bscpro{f \Weyl g}{h}_{\Schwartz(T^* \R^d , \mathcal{L}^1(\Hil,\Hil''))} = \bscpro{f}{h \Weyl g^*}_{\Schwartz(T^* \R^d , \mathcal{L}^1(\Hil',\Hil''))} = \bscpro{f}{(g \Weyl h^*)^*}_{\Schwartz(T^* \R^d , \mathcal{L}^1(\Hil',\Hil''))}$
		\item $\displaystyle \bscpro{f \Weyl g}{h}_{\Schwartz(T^* \R^d , \mathcal{L}^1(\Hil,\Hil''))} = \bscpro{g}{f^* \Weyl h}_{\Schwartz(T^* \R^d , \mathcal{L}^1(\Hil,\Hil'))} = \bscpro{g}{(h^* \Weyl f)^*}_{\Schwartz(T^* \R^d , \mathcal{L}^1(\Hil,\Hil'))}$
	\end{enumerate}
\end{proposition}
\begin{proof}
	\begin{enumerate}[(1)]
		\item Three auxiliary lemmas will enter the computation below: apart from equation~\eqref{operator_valued_calculus:eqn:int_Weyl_product_equals_int_product} from Lemma~\ref{operator_valued_calculus:lem:int_Weyl_product_equals_int_product} (marked with $\star$), we will need the cyclicity of the trace (Lemma~\ref{appendix:extension_by_duality_operator_valued_calculus:prop:properties_generalized_L1_spaces}~(5), marked with $\natural$) as well as Lemma~\ref{operator_valued_calculus:lem:facts_needed_proof_extension_product}~(2) (marked with $\flat$), 
		\begin{align*}
			&\bscpro{f \Weyl g}{h}_{\Schwartz(T^* \R^d , \mathcal{L}^1(\Hil,\Hil''))} 
			= \\
			&\qquad 
			= \int_{T^* \R^d} \dd X \, \trace_{\Hil} \bigl ( (f \Weyl g)^*(X) \, h(X) \bigr ) 
			\overset{\natural,\flat}{=} \int_{T^* \R^d} \dd X \, \trace_{\Hil''} \bigl ( h(X) \, (g^* \Weyl f^*)(X) \bigr ) 
			\\
			&\qquad 
			\overset{\star}{=} \int_{T^* \R^d} \dd X \, \trace_{\Hil''} \bigl ( h \Weyl g^* \Weyl f^*(X) \bigr ) 
			\overset{\star}{=} \int_{T^* \R^d} \dd X \, \trace_{\Hil''} \bigl ( (h \Weyl g^*)(X) \, f^*(X) \bigr ) 
			\\
			&\qquad 
			\overset{\natural}{=} \int_{T^* \R^d} \dd X \, \trace_{\Hil''} \bigl ( f^*(X) \, (h \Weyl g^*)(X) \bigr ) 
			\\
			&\qquad 
			= \bscpro{f}{h \Weyl g^*}_{\Schwartz(T^* \R^d , \mathcal{L}^1(\Hil',\Hil''))} 
			\overset{\flat}{=} \bscpro{f}{(g \Weyl h^*)^*}_{\Schwartz(T^* \R^d , \mathcal{L}^1(\Hil',\Hil''))}
			. 
		\end{align*}
		\item Reusing the same abbreviations as in the proof of item~(1), we make a computation similar to the one above: 
		\begin{align*}
			&\bscpro{f \Weyl g}{h}_{\Schwartz(T^* \R^d , \mathcal{L}^1(\Hil,\Hil''))} 
			= \\
			&\qquad
			= \int_{T^* \R^d} \dd X \, \trace_{\Hil} \bigl ( (f \Weyl g)^*(X) \, h(X) \bigr ) 
			\overset{\natural,\flat}{=} \int_{T^* \R^d} \dd X \, \trace_{\Hil} \bigl ( (g^* \Weyl f^*)(X) \, h(X) \bigr ) 
			\\
			&\qquad
			\overset{\star}{=} \int_{T^* \R^d} \dd X \, \trace_{\Hil} \bigl ( g^* \Weyl f^* \Weyl h(X) \bigr ) 
			\overset{\star}{=} \int_{T^* \R^d} \dd X \, \trace_{\Hil} \bigl ( g^*(X) \, (f^* \Weyl h)(X) \bigr ) 
			\\
			&\qquad
			= \bscpro{g}{f^* \Weyl h}_{\Schwartz(T^* \R^d , \mathcal{L}^1(\Hil,\Hil'))} 
			\overset{\flat}{=} \bscpro{g}{(h^* \Weyl f)^*}_{\Schwartz(T^* \R^d , \mathcal{L}^1(\Hil,\Hil'))} 
		\end{align*}
	\end{enumerate}
\end{proof}
The above formulas provide a way to extend the Weyl product from $\mathcal{L}^1$-valued Schwartz functions to bounded-operator-valued tempered distributions. 
\begin{definition}[Extension of the Weyl product to distributions]\label{operator_valued_calculus:defn:extension_Weyl_product_distributions}
	The magnetic Weyl product can be extended as follows: 
	\begin{enumerate}[(1)]
		\item For $F \in \Schwartz^* \bigl ( T^* \R^d , \mathcal{B}(\Hil',\Hil'') \bigr )$ and $g \in \Schwartz \bigl ( T^* \R^d , \mathcal{L}^1(\Hil,\Hil') \bigr )$ we set 
		\begin{align*}
			\bscpro{F \Weyl g}{h}_{\Schwartz(T^* \R^d , \mathcal{L}^1(\Hil,\Hil''))} := \bscpro{F}{h \Weyl g^*}_{\Schwartz(T^* \R^d , \mathcal{L}^1(\Hil',\Hil''))}
			&&
			\forall h \in \Schwartz \bigl ( T^* \R^d , \mathcal{L}^1(\Hil,\Hil'') \bigr )
			. 
		\end{align*}
		\item For $f \in \Schwartz \bigl ( T^* \R^d , \mathcal{L}^1(\Hil',\Hil'') \bigr )$ and $G \in \Schwartz^* \bigl ( T^* \R^d , \mathcal{B}(\Hil,\Hil') \bigr )$ we set 
		\begin{align*}
			\bscpro{f \Weyl G}{h}_{\Schwartz(T^* \R^d , \mathcal{L}^1(\Hil,\Hil''))} := \bscpro{G}{f^* \Weyl h}_{\Schwartz(T^* \R^d , \mathcal{L}^1(\Hil,\Hil'))} 
			&&
			\forall h \in \Schwartz \bigl ( T^* \R^d , \mathcal{L}^1(\Hil,\Hil'') \bigr )
			. 
		\end{align*}
	\end{enumerate}
\end{definition}
Let us dot our i's and cross our t's, and check that this extension is well-defined.
\begin{lemma}\label{operator_valued_calculus:lem:extension_Weyl_product_to_tempered_distributions_associative}
	The extensions of the magnetic Weyl product in Definition~\ref{operator_valued_calculus:defn:extension_Weyl_product_distributions} are well-defined and associative. 
\end{lemma}
\begin{proof}
	First of all, the trace norm estimate~\eqref{appendix:extension_by_duality_operator_valued_calculus:eqn:trace_norm_estimate_product} implies that both, 
	\begin{align*}
		h \Weyl g^* &\in \Schwartz \bigl ( T^* \R^d , \mathcal{L}^1(\Hil',\Hil'') \bigr ) 
		, 
		\\
		f^* \Weyl h &\in \Schwartz \bigl ( T^* \R^d , \mathcal{L}^1(\Hil,\Hil') \bigr )
		, 
	\end{align*}
	are \emph{trace class}-valued Schwartz functions. Hence, the two extensions are well-defined. 
	
	When \eg $F \in \Schwartz \bigl ( T^* \R^d , \mathcal{L}^1(\Hil',\Hil'') \bigr ) \subset \Schwartz^* \bigl ( T^* \R^d , \mathcal{B}(\Hil',\Hil'') \bigr )$, then Definition~\ref{operator_valued_calculus:defn:extension_Weyl_product_distributions}~(1) is consistent with Proposition~\ref{operator_valued_calculus:prop:relation_Weyl_product_Schwartz_functions_duality_bracket}~(1). Similarly, the extension of multiplying with a distribution from the right is consistent with Proposition~\ref{operator_valued_calculus:prop:relation_Weyl_product_Schwartz_functions_duality_bracket}~(2). 
	
	Associativity of the Weyl product follows from the associativity of $\Weyl$ on Schwartz functions. Indeed, for any combination of Schwartz functions $g_R \in \Schwartz \bigl ( T^* \R^d , \mathcal{L}^1(\Hil'',\Hil''') \bigr )$, $g_L \in \Schwartz \bigl ( T^* \R^d , \mathcal{L}^1(\Hil,\Hil') \bigr )$ and $h \in \Schwartz(T^* \R^d , \mathcal{L}^1(\Hil,\Hil'''))$, and tempered distribution $F \in \Schwartz^* \bigl ( T^* \R^d , \mathcal{B}(\Hil',\Hil'') \bigr )$, we can move the brackets after pushing the Weyl product to the right-hand side of the duality bracket,
	\begin{align*}
		\bscpro{(g_L \Weyl F) \Weyl g_R}{h}_{\Schwartz(T^* \R^d , \mathcal{L}^1(\Hil,\Hil'''))} &= \bscpro{g_L \Weyl F}{h \Weyl g_R^*}_{\Schwartz(T^* \R^d , \mathcal{L}^1(\Hil',\Hil'''))}
		\\
		&= \bscpro{F}{g_L^* \Weyl (h \Weyl g_R^*)}_{\Schwartz(T^* \R^d , \mathcal{L}^1(\Hil',\Hil''))}
		\\
		&= \bscpro{F}{(g_L^* \Weyl h) \Weyl g_R^*}_{\Schwartz(T^* \R^d , \mathcal{L}^1(\Hil',\Hil''))} 
		\\
		&= \bscpro{F \Weyl g_R}{g_L^* \Weyl h}_{\Schwartz(T^* \R^d , \mathcal{L}^1(\Hil,\Hil''))} 
		\\
		&= \bscpro{g_L \Weyl (F \Weyl g_R)}{h}_{\Schwartz(T^* \R^d , \mathcal{L}^1(\Hil,\Hil'''))} 
		. 
	\end{align*}
	This finishes the proof. 
\end{proof}
The next step is to introduce the three magnetic Moyal spaces, which are the analogs of the Moyal algebras introduced in \cite[Section~V.C]{Mantoiu_Purice:magnetic_Weyl_calculus:2004}. 
\begin{definition}[Magnetic Moyal spaces]\label{operator_valued_calculus:defn:Moyal_spaces}
	We define the \emph{left Moyal space} 
	\begin{align*}
		\mathcal{M}^B_{\mathrm{L}} \bigl ( \mathcal{B}(\Hil',\Hil'') \bigr ) := \Bigl \{ F \in \Schwartz^* \bigl ( T^* \R^d , \mathcal{B}(\Hil',\Hil'') \bigr ) \; \; \big \vert \; \; &F \Weyl g \in \Schwartz \bigl ( T^* \R^d , \mathcal{L}^1(\Hil,\Hil'') \bigr ) 
		\Bigr . \\ 
		&\Bigl . 
		\forall g \in \Schwartz \bigl ( T^* \R^d , \mathcal{L}^1(\Hil,\Hil') \bigr ) \Bigr \} 
		, 
	\end{align*}
	the \emph{right Moyal space}
	\begin{align*}
		\mathcal{M}^B_{\mathrm{R}} \bigl ( \mathcal{B}(\Hil,\Hil') \bigr ) := \Bigl \{ G \in \Schwartz^* \bigl ( T^* \R^d , \mathcal{B}(\Hil,\Hil') \bigr ) \; \; \big \vert \; \; &f \Weyl G \in \Schwartz \bigl ( T^* \R^d , \mathcal{L}^1(\Hil,\Hil'') \bigr ) 
		\Bigr . \\ 
		&\Bigl . 
		\forall f \in \Schwartz \bigl ( T^* \R^d , \mathcal{L}^1(\Hil',\Hil'') \bigr ) \Bigr \} 
		, 
	\end{align*}
	and the \emph{Moyal space} 
	\begin{align*}
		\mathcal{M}^B \bigl ( \mathcal{B}(\Hil,\Hil') \bigr ) := \mathcal{M}^B_{\mathrm{L}} \bigl ( \mathcal{B}(\Hil,\Hil') \bigr ) \cap \mathcal{M}^B_{\mathrm{R}} \bigl ( \mathcal{B}(\Hil,\Hil') \bigr ) 
		. 
	\end{align*}
	We say two (left or right) Moyal spaces such as $\mathcal{M}^B \bigl ( \mathcal{B}(\Hil,\Hil') \bigr )$ and $\mathcal{M}^B \bigl ( \mathcal{B}(\Hil',\Hil'') \bigr )$ are \emph{composable} when initial space $\Hil'$ of one is the target space of the other. 
\end{definition}
Unless $\Hil = \Hil'$, we may not multiply elements of Moyal spaces with one another; hence, in general they only form (vector) \emph{spaces} rather than algebras. 

As we will soon see, these definitions are related to and consistent with the literature.
\begin{proposition}\label{operator_valued_calculus:prop:tensor_product_decomposition_Moyal_spaces}
	The three Moyal spaces admit a tensor product decomposition, 
	\begin{align*}
		\mathcal{M}^B_{\mathrm{L},\mathrm{R}} \bigl ( \mathcal{B}(\Hil,\Hil') \bigr ) &= \mathcal{M}^B_{\mathrm{L},\mathrm{R}}(\C) \otimes \mathcal{B}(\Hil,\Hil') 
		, 
		\\
		\mathcal{M}^B \bigl ( \mathcal{B}(\Hil,\Hil') \bigr ) &= \mathcal{M}^B(\C) \otimes \mathcal{B}(\Hil,\Hil') 
		, 
	\end{align*}
	where $\otimes$ can stand for either the $\eps$- or $\pi$-tensor product (\cf \cite[Chapter~43]{Treves:topological_vector_spaces:1967}), and $\mathcal{M}^B_{\mathrm{L},\mathrm{R}}(\C)$ and $\mathcal{M}^B(\C)$ are the Moyal spaces from \cite[Section~V.C]{Mantoiu_Purice:magnetic_Weyl_calculus:2004}. 
\end{proposition}
\begin{proof}
	As subspaces of a nuclear space are nuclear (\cf \cite[Proposition~50.1, (50.3)]{Treves:topological_vector_spaces:1967}), the three Moyal algebras for scalar-valued functions 
	\begin{align*}
		\mathcal{M}^B_{\mathrm{L},\mathrm{R}}(\C) , \mathcal{M}^B(\C) \subset \Schwartz^*(T^* \R^d) 
	\end{align*}
	are nuclear as well. Consequently, the tensor product of them with $\mathcal{B}(\Hil,\Hil')$ is unique and coincides with the Moyal spaces from Definition~\ref{operator_valued_calculus:defn:Moyal_spaces}. 
\end{proof}
Consequently, the mathematical properties of the Moyal \emph{spaces} defined here are essentially the same as the Moyal algebras from the literature (\cf \cite[Section~V.C]{Mantoiu_Purice:magnetic_Weyl_calculus:2004}). One example is the following 
\begin{corollary}\label{operator_valued_calculus:cor:left_right_Moyal_spaces_related_via_adjoint}
	Left and right Moyal spaces are related through the adjoint,  
	\begin{align*}
		\mathcal{M}^B_{\mathrm{L},\mathrm{R}} \bigl ( \mathcal{B}(\Hil,\Hil') \bigr )^* = \mathcal{M}^B_{\mathrm{R},\mathrm{L}} \bigl ( \mathcal{B}(\Hil',\Hil) \bigr ) 
		. 
	\end{align*}
\end{corollary}
However, in general, they are not algebras because when $\Hil \neq \Hil'$ we cannot multiply elements in these sets with one another. Still, elements of \emph{composable} Moyal spaces can be multiplied with one another just as before. 
\begin{definition}[Extension of Weyl product to Moyal spaces]\label{operator_valued_calculus:defn:extension_Weyl_product_Moyal_spaces}
	\begin{enumerate}[(1)]
		\item For all $F \in \mathcal{M}^B_{\mathrm{L}} \bigl ( \mathcal{B}(\Hil',\Hil'') \bigr )$ and $G \in \mathcal{M}^B_{\mathrm{L}} \bigl ( \mathcal{B}(\Hil,\Hil') \bigr )$ we define their Moyal product as 
		\begin{align*}
			\bscpro{F \Weyl G}{h}_{\Schwartz(T^* \R^d , \mathcal{L}^1(\Hil,\Hil''))} := \bscpro{F}{(G \Weyl h^*)^*}_{\Schwartz(T^* \R^d , \mathcal{L}^1(\Hil',\Hil''))}
			&&
			\forall h \in \Schwartz \bigl ( T^* \R^d , \mathcal{L}^1(\Hil,\Hil'') \bigr )
			. 
		\end{align*}
		\item For all $F \in \mathcal{M}^B_{\mathrm{R}} \bigl ( \mathcal{B}(\Hil',\Hil'') \bigr )$ and $G \in \mathcal{M}^B_{\mathrm{R}} \bigl ( \mathcal{B}(\Hil,\Hil') \bigr )$ we define their Moyal product as 
		\begin{align*}
			\bscpro{F \Weyl G}{h}_{\Schwartz(T^* \R^d , \mathcal{L}^1(\Hil,\Hil''))} := \bscpro{G}{(h^* \Weyl F)^*}_{\Schwartz(T^* \R^d , \mathcal{L}^1(\Hil,\Hil'))}
			&&
			\forall h \in \Schwartz \bigl ( T^* \R^d , \mathcal{L}^1(\Hil,\Hil'') \bigr )
			. 
		\end{align*}
		\item When $F \in \mathcal{M}^B \bigl ( \mathcal{B}(\Hil',\Hil'') \bigr )$ and $G \in \mathcal{M}^B \bigl ( \mathcal{B}(\Hil,\Hil') \bigr )$ lie in their corresponding composable Moyal spaces, we can define the product by (1) or (2). 
	\end{enumerate}
\end{definition}
Apart from making the technical arguments in this section, we will almost exclusively be concerned with Moyal spaces $\mathcal{M}^B \bigl ( \mathcal{B}(\Hil,\Hil') \bigr )$ rather than left or right Moyal spaces. One of the reasons is that we want $\Weyl$ to be associative. 
\begin{lemma}\label{operator_valued_calculus:lem:extension_Weyl_product_Moyal_spaces}
	\begin{enumerate}[(1)]
		\item The extension of the Weyl product to composable Moyal spaces is well-defined and associative. 
		\item The Weyl product defined on left and right Moyal spaces fails to be associative. 
		\item When $\Hil = \Hil'$, left and right Moyal spaces $\mathcal{M}^B_{\mathrm{L},\mathrm{R}} \bigl ( \mathcal{B}(\Hil) \bigr )$ form algebras. And the Moyal space $\mathcal{M}^B \bigl ( \mathcal{B}(\Hil) \bigr )$ endowed with the involution ${}^*$ forms a $\ast$-algebra. In that case, we will refer to them as (left/right) \emph{Moyal algebras}. 
	\end{enumerate}
\end{lemma}
\begin{proof}
	\begin{enumerate}[(1)]
		\item The strategy is to regularize the double product first, prove it for the regularized expression and then take the limit again. 
		
		To make our arguments more explicit, for the purpose of the proof we will distinguish between the left extension $\Weyl_{\mathrm{L}}$ (Definition~\ref{operator_valued_calculus:defn:extension_Weyl_product_Moyal_spaces}~(1)) and the right extension $\Weyl_{\mathrm{R}}$ (Definition~\ref{operator_valued_calculus:defn:extension_Weyl_product_Moyal_spaces}~(2)). Lastly, we will also need the two extensions from Definition~\ref{operator_valued_calculus:defn:extension_Weyl_product_distributions}, which we will denote with $\Weyl_{\Schwartz^*}$; importantly, $\Weyl_{\Schwartz^*}$ is associative when we multiply one distribution with finitely many Schwartz functions (Lemma~\ref{operator_valued_calculus:lem:extension_Weyl_product_to_tempered_distributions_associative}). 
		
		That being said, now we proceed with the actual computations. Let us pick two distributions $F \in \mathcal{M}^B \bigl ( \mathcal{B}(\Hil',\Hil'') \bigr )$ and $G \in \mathcal{M}^B \bigl ( \mathcal{B}(\Hil,\Hil') \bigr )$ from composable Moyal spaces as well as two regularizing Schwartz functions $h_{\mathrm{L}} \in \Schwartz \bigl ( T^* \R^d , \mathcal{L}^1(\Hil'') \bigr )$ and $h_{\mathrm{R}} \in \Schwartz \bigl ( T^* \R^d , \mathcal{L}^1(\Hil) \bigr )$. Later on, we will let $h_{\mathrm{L},\mathrm{R}}$ tend to $\id_{\Hil'',\Hil}$. Then for all test functions $j \in \Schwartz \bigl ( T^* \R^d , \mathcal{L}^1(\Hil,\Hil'') \bigr )$ we can rewrite the left product in the following fashion: 
		\begin{align*}
			\Bscpro{h_{\mathrm{L}} \Weyl_{\Schwartz^*} \bigl ( F \Weyl_{\mathrm{L}} G \bigr ) \Weyl_{\Schwartz^*} h_{\mathrm{R}} \, }{ \, j}_{\Schwartz(T^* \R^d,\mathcal{L}^1(\Hil,\Hil''))} &= \Bscpro{F \Weyl_{\mathrm{L}} G \, }{ \, h_{\mathrm{L}}^* \Weyl j \Weyl h_{\mathrm{R}}^*}_{\Schwartz(T^* \R^d,\mathcal{L}^1(\Hil,\Hil''))}
			\\
			&= \Bscpro{F \, }{ \, \Bigl ( G \Weyl_{\Schwartz^*} \bigl ( h_{\mathrm{L}}^* \Weyl_{\Schwartz^*} j \Weyl_{\Schwartz^*} h_{\mathrm{R}}^* \bigr )^* \Bigr )^*}_{\Schwartz(T^* \R^d,\mathcal{L}^1(\Hil',\Hil''))}
			\\
			&= \Bscpro{F \, }{ \, h_{\mathrm{L}}^* \Weyl_{\Schwartz^*} \bigl ( G \Weyl_{\Schwartz^*} h_{\mathrm{R}} \Weyl_{\Schwartz^*} j^* \bigr )^*}_{\Schwartz(T^* \R^d,\mathcal{L}^1(\Hil',\Hil''))}
			\\
			&= \Bscpro{h_{\mathrm{L}} \Weyl_{\Schwartz^*} F \, }{ \, \Bigl ( \bigl ( h_{\mathrm{R}}^* \Weyl_{\Schwartz^*} G^* \bigr )^* \Weyl_{\Schwartz^*} j^* \Bigr )^* }_{\Schwartz(T^* \R^d,\mathcal{L}^1(\Hil',\Hil''))}
			\\
			&= \Bscpro{\bigl ( h_{\mathrm{L}} \Weyl_{\Schwartz^*} F \bigr ) \Weyl_{\Schwartz^*} \bigl ( G \Weyl_{\Schwartz^*} h_{\mathrm{R}} \bigr ) \, }{ \, j}_{\Schwartz(T^* \R^d,\mathcal{L}^1(\Hil,\Hil''))}
		\end{align*}
		Note that we were able to use the associativity of $\Weyl_{\Schwartz^*}$, exploiting that \eg $h_{\mathrm{L}} \Weyl_{\Schwartz^*} F \in \Schwartz \bigl ( T^* \R^d , \mathcal{L}^1(\Hil',\Hil'') \bigr )$ is again a Schwartz function by assumption. 
		
		We can compare the above expression with the result of a similar computation, this time starting with the right product, 
		\begin{align*}
			\Bscpro{h_{\mathrm{L}} \Weyl_{\Schwartz^*} \bigl ( F \Weyl_{\mathrm{R}} G \bigr ) \Weyl_{\Schwartz^*} h_{\mathrm{R}} \, }{ \, j}_{\Schwartz(T^* \R^d,\mathcal{L}^1(\Hil,\Hil''))} &= \Bscpro{F \Weyl_{\mathrm{R}} G \, }{ \, h_{\mathrm{L}}^* \Weyl j \Weyl h_{\mathrm{R}}^*}_{\Schwartz(T^* \R^d,\mathcal{L}^1(\Hil,\Hil''))}
			\\
			&= \Bscpro{G \, }{ \, \Bigl ( \bigl ( h_{\mathrm{L}}^* \Weyl_{\Schwartz^*} j \Weyl_{\Schwartz^*} h_{\mathrm{R}}^* \bigr )^* \Weyl_{\Schwartz^*} F \Bigr )^*}_{\Schwartz(T^* \R^d,\mathcal{L}^1(\Hil,\Hil'))}
			\\
			&= \Bscpro{G \, }{ \, F^* \Weyl_{\Schwartz^*} h_{\mathrm{L}}^* \Weyl_{\Schwartz^*} j \Weyl_{\Schwartz^*} h_{\mathrm{R}}^*}_{\Schwartz(T^* \R^d,\mathcal{L}^1(\Hil,\Hil'))}
			\\
			&= \Bscpro{G \Weyl_{\Schwartz^*} h_{\mathrm{R}} \, }{ \, \bigl ( h_{\mathrm{L}} \Weyl_{\Schwartz^*} F \bigr )^* \Weyl_{\Schwartz^*} j}_{\Schwartz(T^* \R^d,\mathcal{L}^1(\Hil,\Hil'))}
			\\
			&= \Bscpro{\bigl ( h_{\mathrm{L}} \Weyl_{\Schwartz^*} F \bigr ) \Weyl_{\Schwartz^*} \bigl ( G \Weyl_{\Schwartz^*} h_{\mathrm{R}} \bigr ) \, }{ \, j}_{\Schwartz(T^* \R^d,\mathcal{L}^1(\Hil,\Hil''))} 
			. 
		\end{align*}
		Indeed, the two expressions agree, 
		\begin{align*}
			\Bscpro{h_{\mathrm{L}} \Weyl_{\Schwartz^*} \bigl ( F \Weyl_{\mathrm{L}} G \bigr ) \Weyl_{\Schwartz^*} h_{\mathrm{R}} \, }{ \, j}_{\Schwartz(T^* \R^d,\mathcal{L}^1(\Hil,\Hil''))} &= \Bscpro{h_{\mathrm{L}} \Weyl_{\Schwartz^*} \bigl ( F \Weyl_{\mathrm{R}} G \bigr ) \Weyl_{\Schwartz^*} h_{\mathrm{R}} \, }{ \, j}_{\Schwartz(T^* \R^d,\mathcal{L}^1(\Hil,\Hil''))} 
			, 
		\end{align*}
		and taking the limit $h_{\mathrm{L}} \rightarrow \id_{\Hil''}$ and $h_{\mathrm{R}} \rightarrow \id_{\Hil}$ gives $F \Weyl_{\mathrm{L}} G = F \Weyl_{\mathrm{R}} G$. 
		
		Furthermore, this also implies associativity of the extension of $\Weyl$ to composable Moyal spaces: adding two regularizing Schwartz functions $j_{\mathrm{L},\mathrm{R}}$, we compare 
		\begin{align*}
			&\Bscpro{j_{\mathrm{L}} \Weyl \bigl ( (F \Weyl G) \Weyl H \bigr ) \Weyl j_{\mathrm{R}}}{k}_{\Schwartz(T^* \R^d,\mathcal{L}^1(\Hil,\Hil'''))} 
			= \\
			&\qquad 
			= \Bscpro{(F \Weyl_{\mathrm{R}} G) \Weyl_{\mathrm{L}} H}{j_{\mathrm{L}}^* \Weyl_{\Schwartz^*} k  \Weyl_{\Schwartz^*} j_{\mathrm{R}}^*}_{\Schwartz(T^* \R^d,\mathcal{L}^1(\Hil,\Hil'''))}
			\\
			&\qquad 
			= \Bscpro{G}{\Bigl ( \Bigl ( \bigl ( H \Weyl_{\Schwartz^*} j_{\mathrm{R}} \bigr ) \Weyl_{\Schwartz^*} k^* \Bigr ) \Weyl_{\Schwartz^*} \bigl ( j_{\mathrm{L}} \Weyl_{\Schwartz^*} F \bigr ) \Bigr )^*}_{\Schwartz(T^* \R^d,\mathcal{L}^1(\Hil',\Hil''))}
			\\
			&\qquad 
			= \Bscpro{G}{\Bigl ( \bigl ( H \Weyl_{\Schwartz^*} j_{\mathrm{R}} \bigr ) \Weyl_{\Schwartz^*} k^* \Weyl_{\Schwartz^*} \bigl ( j_{\mathrm{L}} \Weyl_{\Schwartz^*} F \bigr ) \Bigr )^*}_{\Schwartz(T^* \R^d,\mathcal{L}^1(\Hil',\Hil''))}
		\end{align*}
		with 
		\begin{align*}
			&\Bscpro{j_{\mathrm{L}} \Weyl \bigl ( F \Weyl (G \Weyl H) \bigr ) \Weyl j_{\mathrm{R}}}{k}_{\Schwartz(T^* \R^d,\mathcal{L}^1(\Hil,\Hil'''))} 
			= \\
			&\qquad 
			= \Bscpro{F \Weyl_{\mathrm{L}} (G \Weyl_{\mathrm{R}} H) \, }{ \, j_{\mathrm{L}}^* \Weyl_{\Schwartz^*} k  \Weyl_{\Schwartz^*} j_{\mathrm{R}}^*}_{\Schwartz(T^* \R^d,\mathcal{L}^1(\Hil,\Hil'''))}
			\\
			&\qquad 
			= \Bscpro{G}{\Bigl ( \bigl ( H \Weyl_{\Schwartz^*} j_{\mathrm{R}} \bigr ) \Weyl_{\Schwartz^*} \Bigl ( k^* \Weyl_{\Schwartz^*} \bigl ( j_{\mathrm{L}} \Weyl_{\Schwartz^*} F \bigr ) \Bigr ) \Bigr )^*}_{\Schwartz(T^* \R^d,\mathcal{L}^1(\Hil',\Hil''))}
			\\
			&\qquad 
			= \Bscpro{G}{\Bigl ( \bigl ( H \Weyl_{\Schwartz^*} j_{\mathrm{R}} \bigr ) \Weyl_{\Schwartz^*} k^* \Weyl_{\Schwartz^*} \bigl ( j_{\mathrm{L}} \Weyl_{\Schwartz^*} F \bigr ) \Bigr )^*}_{\Schwartz(T^* \R^d,\mathcal{L}^1(\Hil',\Hil''))}
		\end{align*}
		and see that they agree; again, associativity of $\Weyl_{\Schwartz^*}$ was the essential ingredient. Taking once more the limit $j_{\mathrm{L}} \rightarrow \id_{\Hil'''}$ and $j_{\mathrm{R}} \rightarrow \id_{\Hil}$ and remarking that the limits converge in the correct Moyal spaces yields the claim. 
		\item This follows directly from the definition. More specifically, if we follow Definition~\ref{operator_valued_calculus:defn:extension_Weyl_product_Moyal_spaces} to give meaning to 
		\begin{align*}
			F \Weyl_{\mathrm{L}} G \Weyl_{\mathrm{L}} H &= \bigl ( F \Weyl_{\mathrm{L}} G \bigr ) \Weyl_{\mathrm{L}} H 
			, 
			\\
			F \Weyl_{\mathrm{R}} G \Weyl_{\mathrm{R}} H &= F \Weyl_{\mathrm{R}} \bigl ( G  \Weyl_{\mathrm{R}} H \bigr ) 
			, 
		\end{align*}
		we see that this implicitly chooses an order. 
		\item Definitions~\ref{operator_valued_calculus:defn:extension_Weyl_product_Moyal_spaces}~(1) and (2) declare products on left and right Moyal spaces $\mathcal{M}^B_{\mathrm{L},\mathrm{R}} \bigl ( \mathcal{B}(\Hil) \bigr )$, respectively. Thus, they indeed form algebras. 
		
		To show that $\mathcal{M}^B \bigl ( \mathcal{B}(\Hil) \bigr )$ defines a $\ast$-algebra, we first note that we can use either extension of the product in Definition~\ref{operator_valued_calculus:defn:extension_Weyl_product_Moyal_spaces}, according to item~(2) they are equivalent. Consequently, also $\mathcal{M}^B \bigl ( \mathcal{B}(\Hil) \bigr )$ defines an algebra. 
		
		Proving that it is left invariant by the adjoint, we note the relationship between left and right Moyal algebras 
		\begin{align*}
			\mathcal{M}^B_{\mathrm{L},\mathrm{R}} \bigl ( \mathcal{B}(\Hil) \bigr )^* = \mathcal{M}^B_{\mathrm{R},\mathrm{L}} \bigl ( \mathcal{B}(\Hil) \bigr ) 
		\end{align*}
		proven in Corollary~\ref{operator_valued_calculus:cor:left_right_Moyal_spaces_related_via_adjoint}. As the Moyal algebra is the intersection of left and right Moyal algebra, we deduce $\mathcal{M}^B \bigl ( \mathcal{B}(\Hil) \bigr )$ is invariant under ${}^*$, \ie a $\ast$-algebra. 
	\end{enumerate}
\end{proof}
Let us summarize some important properties of Moyal spaces. To formulate the statements, we need to introduce the space of linear continuous maps 
\begin{align*}
	\mathcal{L} \bigl ( \Schwartz(\R^d,\Hil) \, , \, \Schwartz(\R^d,\Hil') \bigr ) := \Bigl \{ \hat{F} \in \mathcal{L} \bigl ( \Schwartz(\R^d,\Hil) , \Schwartz^*(\R^d,\Hil') \bigr ) \; \; \big \vert \; \; \ran \hat{F} \subseteq \Schwartz(\R^d,\Hil') \Bigr \} 
\end{align*}
between Schwartz functions, defined as those elements which map into the Schwartz functions again. 

Similarly, we define the space of linear continuous maps 
\begin{align*}
	\mathcal{L} \bigl ( \Schwartz^*(\R^d,\Hil) \, , \, \Schwartz^*(\R^d,\Hil') \bigr ) := \Bigl \{ &\hat{F} \in \mathcal{L} \bigl ( \Schwartz(\R^d,\Hil) , \Schwartz^*(\R^d,\Hil') \bigr ) \; \; \big \vert \; \; 
	\Bigr . \\
	\Bigl . 
	&\mbox{$\exists$ continuous extension } \hat{G} : \Schwartz^*(\R^d,\Hil) \longrightarrow \Schwartz^*(\R^d,\Hil') \Bigr \} 
\end{align*}
between tempered distributions. 
\begin{proposition}[Properties of Moyal spaces]
	\begin{enumerate}[(1)]
		\item The image of left and right Moyal spaces under the magnetic quantization map are 
		\begin{align*}
			\Op^A \Bigl ( \mathcal{M}^B_{\mathrm{L}} \bigl ( \mathcal{B}(\Hil,\Hil') \bigr ) \Bigr ) &= \mathcal{L} \bigl ( \Schwartz(\R^d,\Hil) \, , \, \Schwartz(\R^d,\Hil') \bigr )
			, 
			\\
			\Op^A \Bigl ( \mathcal{M}^B_{\mathrm{R}} \bigl ( \mathcal{B}(\Hil,\Hil') \bigr ) \Bigr ) &= \mathcal{L} \bigl ( \Schwartz^*(\R^d,\Hil) \, , \, \Schwartz^*(\R^d,\Hil') \bigr )
			. 
		\end{align*}
		\item $\Op^A$ is a topological vector space isomorphism between $\mathcal{M}^B \bigl ( \mathcal{B}(\Hil,\Hil') \bigr )$ and 
		\begin{align*}
			\mathcal{L} \bigl ( \Schwartz(\R^d,\Hil) , \Schwartz(\R^d,\Hil') \bigr ) \cap \mathcal{L} \bigl ( \Schwartz^*(\R^d,\Hil) , \Schwartz^*(\R^d,\Hil') \bigr ) 
			. 
		\end{align*}
		\item Left and right Moyal spaces have the following characterization: 
		\begin{align*}
			\Schwartz^* \bigl ( T^* \R^d , \mathcal{B}(\Hil,\Hil') \bigr ) \Weyl \Schwartz \bigl ( T^* \R^d , \mathcal{L}^1(\Hil',\Hil'') \bigr ) &\subseteq \mathcal{M}^B_{\mathrm{R}} \bigl ( \mathcal{B}(\Hil,\Hil'') \bigr ) 
			, 
			\\
			\Schwartz \bigl ( T^* \R^d , \mathcal{L}^1(\Hil,\Hil') \bigr ) \Weyl \Schwartz^* \bigl ( T^* \R^d , \mathcal{B}(\Hil',\Hil'') \bigr ) &\subseteq \mathcal{M}^B_{\mathrm{L}} \bigl ( \mathcal{B}(\Hil,\Hil'') \bigr ) 
			. 
		\end{align*}
		\item The Moyal space $\mathcal{M}^B \bigl ( \mathcal{B}(\Hil,\Hil') \bigr )$ contains $\Schwartz(T^* \R^d,\C)$ as a selfadjoint two-sided ideal. 
		\item $\mathcal{M}^B \bigl ( \mathcal{B}(\Hil) \bigr )$ is a unital $\ast$-algebra containing $\Schwartz(T^* \R^d,\C)$ as a selfadjoint two-sided ideal. 
	\end{enumerate}
\end{proposition}
\begin{proof}
	\begin{enumerate}[(1)]
		\item This follows from a straightforward adaptation of the arguments in the proof of \cite[Proposition~21]{Mantoiu_Purice:magnetic_Weyl_calculus:2004} with the help of the tensor product decomposition from Proposition~\ref{operator_valued_calculus:prop:tensor_product_decomposition_Moyal_spaces}: we just need to make sure to multiply only operators that are composable in the sense of Definition~\ref{operator_valued_calculus:defn:Moyal_spaces}. 
		\item This follows directly from the definition of 
		\begin{align*}
			\mathcal{M}^B \bigl ( \mathcal{B}(\Hil,\Hil') \bigr ) = \mathcal{M}^B_{\mathrm{L}} \bigl ( \mathcal{B}(\Hil,\Hil') \bigr ) \cap \mathcal{M}^B_{\mathrm{R}} \bigl ( \mathcal{B}(\Hil,\Hil') \bigr ) 
		\end{align*}
		as the intersection of left and right Moyal spaces, and combining that with item~(2). 
		\item We leave it to the reader to adapt the arguments in the proof of  \cite[Proposition~22]{Mantoiu_Purice:magnetic_Weyl_calculus:2004}. 
		\item Here, we again need to make use of the tensor product decomposition of the Moyal space  (Proposition~\ref{operator_valued_calculus:prop:tensor_product_decomposition_Moyal_spaces}) and adapt the arguments from the proof of \cite[Proposition~22]{Mantoiu_Purice:magnetic_Weyl_calculus:2004} accordingly. 
		\item Closedness under the product can be inferred directly from the definition of the Moyal algebra $\mathcal{M}^B \bigl ( \mathcal{B}(\Hil) \bigr )$, Lemma~\ref{operator_valued_calculus:lem:extension_Weyl_product_Moyal_spaces}~(3) and item~(4). 
	\end{enumerate}
\end{proof}
%

\subsection{A pseudodifferential calculus of operator-valued Hörmander symbols} 
\label{operator_valued_calculus:Hoermander_symbols}
Everything we have done up until now was for the purpose of giving a rigorous description of pseuododifferential operators for operator-valued Hörmander symbols. The idea is to view functions like Hörmander symbols and elements of $L^2(\R^d,\Hil)$ as tempered distributions, use the extension of the calculus from Section~\ref{operator_valued_calculus:extension_by_duality} and then study the range once we restrict 
\begin{align*}
	\Op^A(f) \big \vert_{L^2(\R^d,\Hil)} : L^2(\R^d,\Hil) \subseteq \Schwartz^*(\R^d,\Hil) \longrightarrow \Schwartz^*(\R^d,\Hil')
\end{align*}
to a Hilbert space like $L^2(\R^d,\Hil)$ or some magnetic Sobolev space. 

We first introduce the notion of operator-valued Hörmander symbols, which is the relevant class of functions we wish to consider. 
\begin{definition}[Operator-valued Hörmander symbols]\label{operator_valued_calculus:defn:Hoermander_symbols}
	\begin{enumerate}[(1)]
		\item For $m \in \R$, $0 \leq \delta \leq \rho \leq 1$ the operator-valued Hörmander classes are the families of operator valued functions
		\begin{align*}
			S^m_{\rho,\delta} \bigl ( \mathcal{B}(\Hil',\Hil) \bigr ) := \Bigl \{ f \in \Cont^{\infty} \bigl ( T^* \R^d , \mathcal{B}(\Hil',\Hil) \bigr ) \; \; \big \vert \; \; \snorm{f}_{m,a \alpha} < \infty \; \forall a , \alpha \in \N_0^d \Bigr \} 
			, 
		\end{align*}
		where the seminorms are defined akin to the scalar-valued case, 
		\begin{align*}
			\snorm{f}_{m,a \alpha} := \sup_{(x,\xi) \in T^* \R^d} \Bigl ( \sexpval{\xi}^{- m + \rho \sabs{\alpha} - \delta \sabs{a}} \, \bnorm{\partial_x^a \partial_{\xi}^{\alpha} f(x,\xi)}_{\mathcal{B}(\Hil',\Hil)} \Bigr )
			, 
		\end{align*}
		with the absolute values replaced by operator norms. 
		\item For $0 \leq \delta \leq \rho \leq 1$ we define $S_{\rho,\delta}^{\infty} \bigl ( \mathcal{B}(\Hil',\Hil) \bigr ) := \bigcup_{m \in \R} S^m_{\rho,\delta} \bigl ( \mathcal{B}(\Hil',\Hil) \bigr )$ and endow it with the inductive limit topology. 
		\item We define $S^{-\infty} \bigl ( \mathcal{B}(\Hil',\Hil) \bigr ) := \bigcap_{m \in \R} S^m_{\rho,\delta} \bigl ( \mathcal{B}(\Hil',\Hil) \bigr )$ and endow it with the projective limit topology. This space is independent of the values of $0 \leq \delta \leq \rho \leq 1$. 
	\end{enumerate}
\end{definition}
Since Hörmander symbols are bounded in $x$, we need to impose stricter assumptions on the magnetic fields, namely that they are bounded in the following sense: 
\begin{assumption}[Bounded magnetic fields]\label{operator_valued_calculus:assumption:bounded_magnetic_fields}
	The components of the magnetic field $B_{jl} \in \Cont^{\infty}_{\mathrm{b}}(\R^d)$ are bounded, smooth and have smooth derivatives to any order. All vector potentials $A \in \Cont^{\infty}_{\mathrm{pol}}(\R^d,\R^d)$ for $B = \dd A$ are smooth, polynomially bounded functions. 
\end{assumption}
As we have already explained in Section~\ref{operator_valued_calculus:construction_on_Schwartz:tensor_product_structure}, we may \emph{not} think of 
\begin{align*}
	S^m_{\rho,\delta} \bigl ( \mathcal{B}(\Hil',\Hil) \bigr ) \neq S^m_{\rho,\delta}(\C) \otimes \mathcal{B}(\Hil',\Hil) 
\end{align*}
as the tensor product of the usual Hörmander class with the appropriate Banach space of bounded operators, \emph{unless} $\Hil$ and $\Hil'$ are finite-dimensional. When one of the Hilbert spaces is infinite-dimensional, neither $S^m_{\rho,\delta}(\C)$ nor $\mathcal{B}(\Hil',\Hil)$ are nuclear and there are several ways to complete the algebraic tensor product. 

Nevertheless, we are still able to exploit the tensor product structures of 
\begin{align*}
	\Schwartz^{(\ast)}(\R^d,\Hil) &\cong \Schwartz^{(\ast)}(\R^d) \otimes \Hil 
\end{align*}
as these are the spaces our operators are initially defined on. The magnetic Sobolev spaces 
\begin{align*}
	H^m_A(\R^d,\Hil) :& \negmedspace = \bigl \{ \Psi \in L^2(\R^d,\Hil) \; \; \vert \; \; \sexpval{P^A}^m \Psi \in L^2(\R^d,\Hil) \bigr \} 
	\\
	&\cong H^m_A(\R^d) \otimes \Hil 
\end{align*}
for $m \geq 0$ are just the usual Hilbert space tensor product of the magnetic Sobolev spaces with $\Hil$.
 The operator 
\begin{align*}
	\sexpval{P^A}^m = \left(\sqrt{1 + (P^A)^2}\right)^{m}
\end{align*}
that appears in the definition and as a weight in the scalar product 
\begin{align}
	\scpro{\Phi}{\Psi}_{H^m_A(\R^d,\Hil)} &:= \scpro{\sexpval{P^A}^m \Phi}{\sexpval{P^A}^m \Psi}_{L^2(\R^d,\Hil)} 
	\label{operator_valued_calculus:eqn:scalar_product_Sobolev_space}
\end{align}
is defined through functional calculus for the magnetic Laplacian. 

We can also define Hörmander spaces of negative order as the anti-dual of the corresponding positive-order Hörmander space, we just need to modify \cite[Definition~4.8]{Iftimie_Mantoiu_Purice:magnetic_psido:2006} to include the $\Hil$-valuedness.

\subsubsection{Properties of magnetic pseudodifferential operators associated to Hörmander symbols} 
\label{operator_valued_calculus:Hoermander_symbols:PsiDOs}
First of all, Hörmander class symbols are all uniformly polynomially bounded functions and therefore lie inside the magnetic Moyal spaces (which can be proven by modifying the arguments in \cite[Section~V.D]{Mantoiu_Purice:magnetic_Weyl_calculus:2004}),  
\begin{align*}
	S^m_{\rho,\delta} \bigl ( \mathcal{B}(\Hil,\Hil') \bigr ) \subset \Cont^{\infty}_{\mathrm{u,pol}} \bigl ( T^* \R^d , \mathcal{B}(\Hil,\Hil') \bigr ) 
	\subset \MoyalSpace \bigl ( \mathcal{B}(\Hil,\Hil') \bigr ) 
	, 
\end{align*}
and therefore they define continuous linear operators 
\begin{align}
	\Op^A(f) : \Schwartz^*(\R^d,\Hil) \longrightarrow \Schwartz^*(\R^d,\Hil')
	, 
	\label{operator_valued_calculus:eqn:Hoermander_symbols_continuous_map_Sprime_Sprime}
\end{align}
which restrict to continuous linear operators $\Op^A(f) : \Schwartz(\R^d,\Hil) \longrightarrow \Schwartz(\R^d,\Hil')$ between the corresponding Schwartz spaces. 

However, ultimately we are usually not interested in operators between Schwartz spaces or tempered distributions, they are just a stepping stone to define operators on magnetic Sobolev spaces and other relevant Hilbert spaces. The basis for defining magnetic pseudodifferential operators on Hilbert spaces is the rigged Hilbert space
\begin{align*}
	\Schwartz(\R^d,\Hil) \hookrightarrow H^m_A(\R^d,\Hil) \hookrightarrow \Schwartz^*(\R^d,\Hil)
	. 
\end{align*}
Vector-valued test functions $\Schwartz(\R^d,\Hil)$ are dense in the magnetic Sobolev spaces. To convince ourselves of that, we just express these spaces as Hilbert space tensor products and invoke the density of $\Schwartz(\R^d) \subseteq H^m_A(\R^d)$, $m \geq 0$.

Clearly, these nested inclusions allow us to define 
\begin{align*}
	\Op^A(f) : H^m_A(\R^d,\Hil) \longrightarrow \Schwartz^*(\R^d,\Hil')
\end{align*}
as the restriction of \eqref{operator_valued_calculus:eqn:Hoermander_symbols_continuous_map_Sprime_Sprime} to the magnetic Sobolev space. \emph{A priori} we do \emph{not} know whether the range is some $H^{m'}_A(\R^d,\Hil')$, $m' \in \R$, though. 

The first result in this direction is a combination of the magnetic Calderón-Vaillancourt Theorem due to Iftimie, Măntoiu and Purice, Theorem~3.1 in \cite{Iftimie_Mantoiu_Purice:magnetic_psido:2006}, and the operator-valued, non-magnetic version proven by Panati, Spohn and Teufel \cite[Proposition~A.4]{PST:sapt:2002}: 
\begin{theorem}[Calderón-Vaillancourt Theorem]\label{operator_valued_calculus:thm:Calderon_Vaillancourt}
	Suppose the magnetic field $B$ is bounded in the sense of Assumption~\ref{operator_valued_calculus:assumption:bounded_magnetic_fields}. Then a Hörmander symbol $f \in S^0_{\rho,\delta} \bigl ( \mathcal{B}(\Hil,\Hil') \bigr )$ where either $0 \leq \delta < \rho \leq 1$ or $\rho = \delta \in [0,1)$, defines a bounded magnetic pseudodifferential operator
	\begin{align*}
		\Op^A(f) : L^2(\R^d,\Hil) \longrightarrow L^2(\R^d,\Hil') 
	\end{align*}
	between Hilbert spaces. More precisely, there exist constants $C(d)$ and $n(d)$ that only depend on the dimension of the ambient space that allow us to bound the operator norm 
	\begin{align}
		\bnorm{\Op^A(f)}_{\mathcal{B}(L^2(\R^d,\Hil) , L^2(\R^d,\Hil') )} \leq C(d) \, \max_{\sabs{a} , \sabs{\alpha} \leq n(d)} \snorm{f}_{0,a \alpha}
		\label{operator_valued_calculus:eqn:Calderon_Vaillancourt_estimate}
	\end{align}
	in terms of finitely many Hörmander seminorms. 
\end{theorem}
\begin{remark}
	Since the natural inclusions 
	\begin{align*}
		S^m_{\rho,\delta} \bigl ( \mathcal{B}(\Hil,\Hil') \bigr ) \subseteq S^{m'}_{\rho',\delta'} \bigl ( \mathcal{B}(\Hil,\Hil') \bigr )
		, 
		&&
		m \leq m' 
		, \; 
		\rho \geq \rho' 
		, \; 
		\delta \leq \delta' 
		, 
	\end{align*}
	extend to operator-valued Hörmander symbols, it suffices to prove the $L^2$-boundedness for the case $m = 0$ and $\rho = \delta \in [0,1)$. 
\end{remark}
\begin{proof}[Sketch]
	The proof is essentially identical to that of \cite[Theorem~3.1]{Iftimie_Mantoiu_Purice:magnetic_psido:2006}, we merely replace absolute value with operator norms where necessary. Therefore, we will only sketch the argument. 
	
	The first step is to replace $f \in S^0_{\rho,\rho} \bigl ( \mathcal{B}(\Hil,\Hil') \bigr )$ with a Schwartz function and employ \cite[Remark~3.2]{Iftimie_Mantoiu_Purice:magnetic_psido:2006}: multiplying $f$ with a smoothened bump function $\chi(\epsilon X)$ with compact support and scaling the plateaux with $\epsilon \in (0,1]$ gives us a sequence of Schwartz functions for which the desired estimate~\eqref{operator_valued_calculus:eqn:Calderon_Vaillancourt_estimate} holds uniformly in $\epsilon$.
	
	The second step proceeds under the assumption $f \in \Schwartz \bigl ( T^* \R^d , \mathcal{B}(\Hil,\Hil') \bigr )$ and relies crucially on the Stein-Cotlar-Knapp Lemma; the latter states that under certain circumstances sums of “almost orthogonal”, bounded operators are bounded. Since the “Hörmander $L$ operators” used in the proof, $L_{\xi}$ and $M$, act trivially on the $\mathcal{B}(\Hil,\Hil')$ part, these arguments still go through. To account for the operator- or Hilbert space-valuedness, we \eg need to pick $u \in \Schwartz(\R^d,\Hil)$ in Steps~1–2 and $v \in \Schwartz(\R^d,\Hil')$ in Step~3 as well as replace the complex conjugate $\overline{g(z,y;\eta)} \in \C$ by the Hilbert space adjoint $g(z,y;\eta)^* \in \mathcal{B}(\Hil',\Hil)$. 
\end{proof}
The commutator criteria we give in Section~\ref{operator_valued_calculus:Hoermander_symbols:commutator_criteria} below do the opposite of the Calderón-Vaillancourt Theorem, it tells us which operators are magnetic pseudodifferential operators with a Hörmander symbol.

A direct consequence of the Calderón-Vaillancourt Theorem~\ref{operator_valued_calculus:thm:Calderon_Vaillancourt} is the following 
\begin{corollary}\label{operator_valued_calculus:cor:boundedness_magnetic_Sobolev_spaces}
	Suppose the magnetic field $B$ is bounded in the sense of Assumption~\ref{operator_valued_calculus:assumption:bounded_magnetic_fields} and either $0 \leq \delta < \rho \leq 1$ or $\rho = \delta \in [0,1)$. Then a Hörmander symbol $f \in S^m_{\rho,\delta} \bigl ( \mathcal{B}(\Hil,\Hil') \bigr )$ defines a bounded magnetic pseudodifferential operator
	\begin{align*}
		\Op^A(f) : H^s_A(\R^d,\Hil) \longrightarrow H^{s-m}_A(\R^d,\Hil') 
	\end{align*}
	between magnetic Sobolev spaces for any $s \in \R$. 
\end{corollary}
\begin{remark}
	Compared to \cite[Proposition~4.3]{Iftimie_Mantoiu_Purice:magnetic_psido:2006} we are also able to deal with the case $\rho = \delta \in [0,1)$. The main ingredient in the proof by Iftimie et.\ al needed was a parametrix construction for $\sexpval{P^A}^m$; the restriction $\rho > \delta$ is necessary to ensure the remainder is a smoothing operator. When we replace $\sexpval{P^A}^m$ by $\sexpval{P^A}^m + \lambda(m)$ for some suitably chosen non-negative constant $\lambda(m) \geq 0$, later results cited below guarantee the existence of an exact inverse \emph{without} remainder. 
\end{remark}
The proof also relies on a fact we will only prove below, Theorem~\ref{operator_valued_calculus:thm:Weyl_product_Hoermander_symbols}. It tells us that the product of two Hörmander symbols is a Hörmander symbol whose order is the sum of the orders of the two factors. 
\begin{proof}
	Let us at first assume $m \geq 0$ is non-negative. The first ingredient is that the weight 
	\begin{align*}
		\sexpval{P^A}^m = \Op^A \bigl ( \sexpval{\xi}^m \bigr ) 
		,
	\end{align*}
	which enters the definition of $H^m_A(\R^d,\Hil)$, can also be viewed as a pseudodifferential operator associated with $\sexpval{\xi}^m \in S^m_{1,0}(\C)$. Moreover, it is evidently elliptic in the usual sense (\cf also Definition~\ref{operator_valued_calculus:defn:elliptic_symbols} which generalizes the concept of ellipticity to operator-valued symbols). 
	
	In fact, we may replace $\sexpval{\xi}^m$ by $p_m(\xi) := \sexpval{\xi}^m + \lambda(m)$ where $\lambda(m) \geq 0$ is a suitable non-negative constant. Following the arguments of \cite[§3.3]{Lein_Mantoiu_Richard:anisotropic_mag_pseudo:2009} which hinge on \cite[Theorem~1.8]{Mantoiu_Purice_Richard:spectral_propagation_results_magnetic_Schroedinger_operators:2007}, we may choose the value of $\lambda(m)$ large enough so that the Moyal inverse $p^{(-1)_{\Weyl}} \in S^{-m}_{1,0}(\C)$ exists as a Hörmander symbol, \ie the symbol that satisfies 
	\begin{align*}
		p_m \Weyl p_m^{(-1)_{\Weyl}} = 1 = p_m^{(-1)_{\Weyl}} \Weyl p_m 
		. 
	\end{align*}
	At this stage we can drop the assumption $m \geq 0$: we can now define weights $w_m$ for $m \in \R$ by setting 
	\begin{align}
		w_m := 
		\begin{cases}
			p_m & m \geq 0 \\
			p_{\sabs{m}}^{(-1)_{\Weyl}} & m < 0 \\
		\end{cases}
		\in S^m_{1,0}(\C) 
		.
		\label{operator_valued_calculus:eqn:definition_alternate_weights_magnetic_Sobolev_spaces}
	\end{align}
	Importantly, we may replace $\Op^A(\sexpval{\xi}^m)$ by $\Op^A(w_m)$ when defining the \emph{Banach} spaces $H^m_A(\R^d,\Hil)$, for the norms are evidently equivalent. 
	
	Once we show that the composition of two operator-valued Hörmander symbols is a Hörmander symbol whose order is the sum of the orders, a debt we will pay with Theorem~\ref{operator_valued_calculus:thm:Weyl_product_Hoermander_symbols} below, we can finish the proof: using $\Op^A(w_m)$ as a weight for the magnetic Sobolev norms, we can insert the weights and factor the product as 
	\begin{align}
		\norm{\Op^A(f) \Psi}_{H^{s-m}_A(\R^d,\Hil')} &= \norm{\Op^A(w_{s-m}) \, \Op^A(f) \, \Op^A(w_{-s}) \, \Op^A(w_s) \Psi}_{L^2(\R^d,\Hil')} 
		\notag \\
		&= \norm{\Op^A \bigl ( w_{s-m} \Weyl f \Weyl w_{-s} \bigr ) \, \Op^A(w_s) \Psi}_{L^2(\R^d,\Hil')} 
		. 
		\label{operator_valued_calculus:eqn:proof_corollary_Calderon_Vaillancourt_Sobolev_norm_L2_norm}
	\end{align}
	Since $\Psi \in H^s_A(\R^d,\Hil)$ belongs to the magnetic Sobolev space of order $s$, the vector $\Op^A(w_s) \Psi \in L^2(\R^d,\Hil)$ lies in $L^2$. And thanks to Theorem~\ref{operator_valued_calculus:thm:Weyl_product_Hoermander_symbols} we know 
	\begin{align*}
		w_{s-m} \Weyl f \Weyl w_{-s} \in S^{s - m + m - s}_{\rho,\delta} \bigl ( \mathcal{B}(\Hil,\Hil') \bigr ) 
		= S^0_{\rho,\delta} \bigl ( \mathcal{B}(\Hil,\Hil') \bigr ) 
	\end{align*}
	is a Hörmander symbol of order $0$; consequently, the magnetic Calderón-Vaillancourt Theorem~\ref{operator_valued_calculus:thm:Calderon_Vaillancourt} applies and we get a bounded operator on $L^2(\R^d,\Hil) \longrightarrow L^2(\R^d,\Hil')$.
	
	Finally, we can estimate the right-hand side of \eqref{operator_valued_calculus:eqn:proof_corollary_Calderon_Vaillancourt_Sobolev_norm_L2_norm} by
	\begin{align*}
		\ldots &\leq C(d) \, \Bigl ( \max_{\sabs{a} , \sabs{\alpha} \leq n(d)} \snorm{f}_{m,a \alpha} \Bigr ) \, 
		\snorm{\Psi}_{H^s_A(\R^d,\Hil)}
		, 
	\end{align*}
	and therefore, $\Op^A(f) : H^s_A(\R^d,\Hil) \longrightarrow H^{s-m}_A(\R^d,\Hil')$ defines a bounded operator. 
\end{proof}
The previous statement is not optimal in the sense that $f \in S^m_{\rho,\delta} \bigl ( \mathcal{B}(\Hil,\Hil') \bigr )$ implies $f \in S^{m + m'}_{\rho,\delta} \bigl ( \mathcal{B}(\Hil,\Hil') \bigr )$ for any $m' > 0$. Thus, $m$ does not necessarily characterize the growth of $f$ as $\sabs{\xi} \rightarrow \infty$, it just gives an upper bound. The degree of growth is captured by the notion of ellipticity: 
\begin{definition}[Elliptic operator-valued symbol]\label{operator_valued_calculus:defn:elliptic_symbols}
	An operator-valued Hörmander symbol $f \in S^m_{\rho,\delta} \bigl ( \mathcal{B}(\Hil,\Hil') \bigr )$ is called \emph{elliptic} if and only if there exist two constants $C > 0$ and $R > 0$ such that 
	\begin{align*}
		\bnorm{f(x,\xi)}_{\mathcal{B}(\Hil,\Hil')} \geq C \, \sexpval{\xi}^m 
	\end{align*}
	holds for all $\sabs{\xi} \geq R$. 
\end{definition}
When $f$ is scalar-valued, then \emph{real}-valued elliptic symbols define selfadjoint operators. In our setting, the real-valuedness assumption needs to be replaced by \emph{selfadjoint}-operator-valuedness. 
\begin{theorem}\label{operator_valued_calculus:thm:selfadjointness_elliptic_symbols}
	Assume the magnetic field is bounded in the sense of Assumption~\ref{operator_valued_calculus:assumption:bounded_magnetic_fields}, and that the initial Hilbert space $\Hil \subseteq \Hil'$ can be regarded as a dense subspace of the target Hilbert space $\Hil'$. Suppose $h \in S^m_{\rho,\delta} \bigl ( \mathcal{B}(\Hil,\Hil') \bigr )$ is a Hörmander symbol of order $m \in \R$ and type $0 \leq \delta < \rho \leq 1$ that takes values in the selfadjoint operators on $\Hil'$, 
	\begin{align*}
		h(X)^* = h(X) 
		&&
		\forall X \in T^* \R^d 
		. 
	\end{align*}
	In case $m > 0$ we assume in addition that $h$ is elliptic of order $m$ (\cf Definition~\ref{operator_valued_calculus:defn:elliptic_symbols}). 
	
	Then $\Op^A(h) = \Op^A(h)^*$ defines a selfadjoint operator on $L^2(\R^d,\Hil)$. Its domain is either $L^2(\R^d,\Hil)$ ($m \leq 0$) or $H^m_A(\R^d,\Hil)$ ($m > 0$); the space of Schwartz functions $\Schwartz(\R^d,\Hil) \subseteq L^2(\R^d,\Hil)$ is a core. 
\end{theorem}
In the above statement it is important that the selfadjointness is with respect to the Hilbert space $L^2(\R^d,\Hil')$ (\ie we need to consider $L^2$-functions taking values in the \emph{primed} Hilbert space). That is because $\Hil \subseteq \Hil'$ is the \emph{“pointwise”} domain of selfadjointness of the symbol $h(x,\xi) = h(x,\xi)^*$. 
\begin{remark}[When $m > 0$ then $\rho > \delta$ is a necessary assumption]
	While for many applications the assumption $\rho > \delta$ is natural, for equivariant magnetic pseudodifferential operators we invariably have $\rho = 0 = \delta$. Unfortunately, the case $\rho = \delta \in [0,1)$ is not covered by our extension of \cite[Theorem~5.1]{Iftimie_Mantoiu_Purice:magnetic_psido:2006}: a central piece in the argument was to show that the graph norm is equivalent to the magnetic Sobolev norm. The upper bound follows from Corollary~\ref{operator_valued_calculus:cor:boundedness_magnetic_Sobolev_spaces}. For the lower bound Iftimie et al.\ constructed a parametrix for the operator $\Op^A(h)$, \ie an operator $\Op^A(g)$ which “inverts $\Op^A(h)$ up to a compact operator”, 
	\begin{align*}
		\Op^A(g) \, \Op^A(h) - \id_{\Hil} =: \Op^A(r)
		. 
	\end{align*}
	Such a non-unique operator $\Op^A(g)$ can be explicitly constructed with the help of pseudodifferential theory (\cf Section~\ref{operator_valued_calculus:Hoermander_symbols:Weyl_product} and \cite[Theorem~2.8]{Iftimie_Mantoiu_Purice:magnetic_psido:2006}). The assumption $\rho > \delta$ is necessary to ensure that $r \in S^{-\infty} \bigl ( \mathcal{B}(\Hil) \bigr ) \cap S^{-\infty} \bigl ( \mathcal{B}(\Hil') \bigr )$ is a smoothing symbol. If $\rho = \delta$, then the symbol $r \in S^m_{\rho,\rho} \bigl ( \mathcal{B}(\Hil) \bigr )$ of the remainder would be of the same order $m$ as the operator $h$. Therefore, the remainder $\Op^A(r)$ would not even define a bounded operator between $L^2$-spaces. 
\end{remark}
\begin{proof}
	First of all, our assumption that $h(x,\xi)^* = h(x,\xi)$ takes values in the selfadjoint operators means that the domain of selfadjointness $\domain \bigl ( h(x,\xi) \bigr ) = \Hil$ is just the Hilbert space $\Hil \subseteq \Hil'$. And to allow for $h(x,\xi)$ taking values in the \emph{unbounded} selfadjoint operators on $\Hil'$, the subspace $\Hil$ only needs to be dense. Further, the selfadjointness of $h(x,\xi)^* = h(x,\xi)$ implies that the pointwise adjoint $h(x,\xi)^*$ again defines a bounded operator from $\Hil$ to $\Hil'$ rather than the other way around. 
	
	We can check by hand that $\Op^A(h)$ is symmetric on $\Schwartz(\R^d,\Hil) \subseteq \Schwartz(\R^d,\Hil')$: for all $\phi , \psi \in \Schwartz(\R^d,\Hil)$ we plug in the selfadjointness of $h(x,\xi) = h(x,\xi)^*$ and compute 
	\begin{align*}
		&\scpro{\phi}{\Op^A(h) \psi}_{L^2(\R^d,\Hil')} = 
		\\
		&\quad 
		= \frac{1}{(2\pi)^d} \int_{\R^d} \dd x \int_{\R^d} \dd y \int_{\R^d} \dd \eta \, \e^{- \ii (y - x) \cdot \eta} \, \e^{- \ii \frac{\lambda}{\eps} \int_{[\eps x , \eps y]} A} \, \scpro{\phi(x) \, }{ \, h \bigl ( \tfrac{\eps}{2} (x + y) , \eta \bigr ) \psi(y)}_{\Hil'}
		\\
		&\quad 
		= \frac{1}{(2\pi)^d} \int_{\R^d} \dd x \int_{\R^d} \dd y \int_{\R^d} \dd \eta \, \scpro{\e^{+ \ii (y - x) \cdot \eta} \, \e^{+ \ii \frac{\lambda}{\eps} \int_{[\eps x , \eps y]} A} \, h \bigl ( \tfrac{\eps}{2} (x + y) , \eta \bigr ) \phi(x) \, }{ \, \psi(y)}_{\Hil'}
		\\
		&\quad 
		= \scpro{\Op^A(h) \phi}{\psi}_{L^2(\R^d,\Hil')}
		. 
	\end{align*}
	To get to the last equality, we remark that $x$ and $y$ now play opposite roles of one another and $[\eps x,\eps y] = - [\eps y , \eps x]$ reverses the orientation of the line integral. Consequently, $\Op^A(f) \vert_{\Schwartz(\R^d,\Hil)}$ is a symmetric operator. 
	
	In fact, Corollary~\ref{operator_valued_calculus:cor:boundedness_magnetic_Sobolev_spaces} tells us that 
	\begin{align*}
		H^A := \Op^A(h) : H^m_A(\R^d,\Hil) \longrightarrow L^2(\R^d,\Hil') 
	\end{align*}
	defines a bounded extension of $\Op^A(h)$, and the above computation applies verbatim to when 
	\begin{align*}
		\phi , \psi \in H^m_A(\R^d,\Hil) \subseteq L^2(\R^d,\Hil) \subseteq L^2(\R^d,\Hil') 
	\end{align*}
	are taken from the larger subspace of $L^2(\R^d,\Hil')$. 
	
	The ellipticity of $h$ implies that the graph norm of $H^A$ is equivalent to the Sobolev norm: we just write $H^m_A(\R^d,\Hil) = H^m_A(\R^d) \otimes \Hil$ as a (Hilbert space) tensor product, and apply \cite[Lemma~4.3~(3)]{Iftimie_Mantoiu_Purice:magnetic_psido:2006}. 
	
	It remains to check that the domain 
	\begin{align*}
		\mathcal{D} \bigl ( {H^A}^* \bigr ) = \mathcal{D}(H^A) = H^m_A(\R^d,\Hil) \subseteq L^2(\R^d,\Hil')
	\end{align*}
	of the adjoint operator coincides with the domain of $H^A$. Writing out the definition of the adjoint operator, we see that for any $\Phi \in \domain \bigl ( {H^A}^* \bigr )$ there exists $\Psi \in L^2(\R^d,\Hil)$ so that 
	\begin{align*}
		\scpro{H^A \varphi}{\Phi}_{L^2(\R^d,\Hil')} 
		= \scpro{\varphi}{\Psi}_{L^2(\R^d,\Hil')} 
	\end{align*}
	holds for all $\varphi \in \Schwartz(\R^d,\Hil)$. This proves ${H^A}^* \Phi = \Psi \in \Schwartz^*(\R^d,\Hil')$ in the weak sense. 
	
	Invoking Corollary~\ref{operator_valued_calculus:cor:boundedness_magnetic_Sobolev_spaces} once more tells us that $\Phi \in H^m_A(\R^d,\Hil)$ must hold, which proves the first (and only non-trivial) inclusion in 
	\begin{align*}
		\domain \bigl ( {H^A}^* \bigr ) \subseteq \domain(H^A) \subseteq \domain \bigl ( {H^A}^* \bigr ) 
		. 
	\end{align*}
	Thus, $H^A = {H^A}^*$ is selfadjoint and $\Schwartz(\R^d,\Hil)$ a core. 
\end{proof}
%

\subsubsection{The magnetic Weyl product and its asymptotic expansions} 
\label{operator_valued_calculus:Hoermander_symbols:Weyl_product}
The extension of the Weyl product to operator-valued Hörmander symbols is straightforward: we just need to replace the absolute value by operator norms in the estimates. While in principle, it is not clear that derivatives of an operator-valued function $f : T^* \R^d \longrightarrow \mathcal{B}(\Hil,\Hil')$ need to take values in $\mathcal{B}(\Hil,\Hil')$ as well, this assumption is baked into the definition of the operator-valued Hörmander classes. 
\begin{theorem}[Composition of operator-valued Hörmander symbols]\label{operator_valued_calculus:thm:Weyl_product_Hoermander_symbols}
	Suppose the magnetic field is bounded in the sense of Assumption~\ref{operator_valued_calculus:assumption:bounded_magnetic_fields} and $(\rho,\delta)$ which determine the order of the symbol class satisfy either $0 \leq \delta < \rho \leq 1$ or $0 \leq \rho \leq 1$ and $\delta = 0$. 
	Then Weyl product defines a continuous, bilinear map between operator-valued Hörmander spaces
	\begin{align*}
		\Weyl : S^{m_1}_{\rho,\delta} \bigl ( \mathcal{B}(\Hil,\Hil') \bigr ) \times S^{m_2}_{\rho,\delta} \bigl ( \mathcal{B}(\Hil',\Hil'') \bigr ) \longrightarrow S^{m_1 + m_2}_{\rho,\delta} \bigl ( \mathcal{B}(\Hil,\Hil'') \bigr ) 
		. 
	\end{align*}
\end{theorem}
Since the extension is really just a matter of replacing absolute values by suitable operator norms as well as elementary inequalities like 
\begin{align*}
	\bnorm{f(x,\xi) \, g(y,\eta)}_{\mathcal{B}(\Hil,\Hil'')} \leq \bnorm{f(x,\xi)}_{\mathcal{B}(\Hil,\Hil')} \, \bnorm{g(y,\eta)}_{\mathcal{B}(\Hil',\Hil'')} 
	, 
\end{align*}
we will not bother repeating one of the existing proofs; we refer the interested reader to \cite[Theorem~2.6]{Iftimie_Mantoiu_Purice:magnetic_psido:2006} for when $0 \leq \delta < \rho \leq 1$ or \cite[Appendix~D]{Lein:two_parameter_asymptotics:2008} for the case $\delta = 0$, $0 \leq \rho \leq 1$. 
\begin{remark}
	It also makes sense to consider products of operator-valued and \emph{scalar}-valued symbols, we just need to identify $f \equiv f \otimes \id_{\Hil} \in S^{m_1}_{\rho,\delta}(\C) \subseteq S^{m_1}_{\rho,\delta} \bigl ( \mathcal{B}(\Hil) \bigr )$. Therefore, the above Theorem implies the continuity of the maps 
	\begin{align*}
		\Weyl &: S^{m_1}_{\rho,\delta}(\C) \times S^{m_2}_{\rho,\delta} \bigl ( \mathcal{B}(\Hil,\Hil') \bigr ) \longrightarrow S^{m_1 + m_2}_{\rho,\delta} \bigl ( \mathcal{B}(\Hil,\Hil') \bigr ) 
		, 
		\\
		\Weyl &: S^{m_1}_{\rho,\delta} \bigl ( \mathcal{B}(\Hil,\Hil') \bigr ) \times S^{m_2}_{\rho,\delta}(\C) \longrightarrow S^{m_1 + m_2}_{\rho,\delta} \bigl ( \mathcal{B}(\Hil,\Hil') \bigr ) 
		. 
	\end{align*}
\end{remark}
The Weyl product admits asymptotic expansions in small parameters. For example, if we scale the magnetic field $B \rightarrow \lambda B$ by a coupling parameter $\lambda$, the leading-order term in the $\lambda$-expansion of the Weyl product 
\begin{align*}
	f \Weyl g = f \sharp^{B = 0} g + \order(\lambda) 
\end{align*}
is the non-magnetic Weyl product $\sharp^{B = 0}$. Similarly, we can obtain expansions in the semiclassical parameter $\eps$ and simultaneously in $\eps$ \emph{and} $\lambda$. Just like with the product itself, the proofs from \cite{Lein:two_parameter_asymptotics:2008} carry over after minimal modifications. 

Still, we need to clarify what we mean by asymptotic expansion: many constructions in pseudodifferential theory yield symbols that are defined order-by-order in a small parameter $\epsilon$ (such as $\eps$ or $\lambda$), and we need to make sense of \emph{formal} (or asymptotic) sums 
\begin{align*}
	f_{\epsilon} \asymp \sum_{n = 0}^{\infty} \epsilon^n \, f_n 
	. 
\end{align*}
Here, each of the terms $f_n \in S^{m - n \mu (\rho - \delta)}_{\rho,\delta} \bigl ( \mathcal{B}(\Hil,\Hil') \bigr )$ belongs to a particular symbol class. The parameter $\mu \geq 0$ determines whether and how quickly the order of the symbol classes decreases as $n \rightarrow \infty$; typically, one assumes $\mu > 0$, but to ensure all statements cover \emph{equivariant} symbol classes from Section~\ref{equivariant_calculus}, we will allow $\mu = 0$ as well. 

At least when $\mu > 0$ each formal sum has many resummations (\cf \cite[Proposition~2.26]{Folland:harmonic_analysis_hase_space:1989} and \cite[Proposition~2.5.33]{Ruzhansky_Turunen:pseudodifferential_operators_symmetries:2010}), but the difference between two resummations is $\order(\epsilon^{\infty})$ small. What this means in mathematical terms is covered in the next definition: 
\begin{definition}[Asymptotic Hörmander symbol class $\SemiHoer{m}_{\rho,\delta} \bigl ( \mathcal{B}(\Hil,\Hil') \bigr )$]\label{operator_valued_calculus:defn:asymptotic_Hoermander_class}
	Let $\epsilon \in (0,\epsilon_0)$, $\epsilon_0 \ll 1$, be a small parameter, $\mu \geq 0$, $m \in \R$ and $0 \leq \delta \leq \rho \leq 1$. A map 
	\begin{align*}
		\epsilon \mapsto f_{\epsilon} \in S^m_{\rho,\delta} \bigl ( \mathcal{B}(\Hil,\Hil') \bigr )
	\end{align*}
	is called an \emph{asymptotic Hörmander symbol} of order $m$ and type $(\rho,\delta)$ if and only if there exists a sequence $\{ f_n \}_{n \in \N_0}$ of Hörmander symbols $f_n \in S^{m - n \mu (\rho - \delta)}_{\rho,\delta} \bigl ( \mathcal{B}(\Hil,\Hil') \bigr )$ such that 
	\begin{align*}
		f_{\epsilon} - \sum_{n = 0}^N \epsilon^n \, f_n = \order(\epsilon^{N+1}) \in S^{m - (N+1) \mu (\rho - \delta)}_{\rho,\delta} \bigl ( \mathcal{B}(\Hil,\Hil') \bigr ) 
	\end{align*}
	holds uniformly in the following sense: for all $N \in \N_0$ and all multi indices $a , \alpha \in \N_0^d$ there exists a constant $C_{N,a \alpha} > 0$ that is independent of $\epsilon$ such that for all $\epsilon \in (0,\epsilon_0)$ we have 
	\begin{align*}
		\norm{f_{\epsilon} - \sum_{n = 0}^N \epsilon^n \, f_n}_{m - (N+1) \mu (\rho - \delta) ,  a \alpha} \leq C_{N,a \alpha} \, \epsilon^{N+1} 
		. 
	\end{align*}
	We denote the space of all asymptotic Hörmander symbols of order $m$ and type $(\rho,\delta)$ with $\SemiHoer{m}_{\rho,\delta} \bigl ( \mathcal{B}(\Hil,\Hil') \bigr )$; it is endowed with the inductive limit topology. 
	
	Moreover, we denote the space of formal (or asymptotic) sums with $\Sigma \Hoer{m}_{\rho,\delta} \bigl ( \mathcal{B}(\Hil,\Hil') \bigr )$; it is endowed with the inductive limit topology. 
\end{definition}
One of the authors extended this definition to the case where one has two small parameters; we refer the reader to \cite[Definition~2.3]{Lein:two_parameter_asymptotics:2008}. 
\begin{remark}[Explicit description of the topology of $\SemiHoer{m}_{\rho,\delta} \bigl ( \mathcal{B}(\Hil,\Hil') \bigr )$]
	We were unable to find an explicit description of the topology of $\SemiHoer{m}_{\rho,\delta} \bigl ( \mathcal{B}(\Hil,\Hil') \bigr )$ in the literature. Suppose, we are given some $N \in \N_0$ and consider the Hörmander symbols that admit an expansion up to $\order(\epsilon^{N+1})$. Specifically, we are interested in maps $\epsilon \mapsto f_{\epsilon}$ that can be written as the \emph{finite} sum 
	\begin{align*}
		f_{\epsilon} = \sum_{n = 0}^N \epsilon^N \, f_n + \epsilon^{N+1} \, R_N 
		, 
	\end{align*}
	where the $f_n \in S^{m - n \mu (\rho - \delta)}_{\rho,\delta} \bigl ( \mathcal{B}(\Hil,\Hil') \bigr )$, $n = 0 , \ldots , N$, as well as the remainder 
	\begin{align*}
		R_N := \epsilon^{-(N+1)} \left ( f_{\epsilon} - \sum_{n = 0}^N \epsilon^n \, f_n \right ) 
		\in S^{m - (N+1) \mu (\rho - \delta)}_{\rho,\delta} \bigl ( \mathcal{B}(\Hil,\Hil') \bigr )
	\end{align*}
	need to lie in the appropriate Hörmander classes by definition. Asymptotic symbols with an expansion up to order $N$ can therefore be identified with 
	\begin{align*}
		\mathrm{A}^N_{\mu} S^m_{\rho,\delta} \bigl ( \mathcal{B}(\Hil,\Hil') \bigr ) := \bigoplus_{n = 0}^{N+1} S^{m - n \mu (\rho - \delta)}_{\rho,\delta} \bigl ( \mathcal{B}(\Hil,\Hil') \bigr ) 
		. 
	\end{align*}
	Clearly, these finite-order spaces carry a Fréchet topology and nest into one another, 
	\begin{align*}
		\mathrm{A}^N_{\mu} S^m_{\rho,\delta} \bigl ( \mathcal{B}(\Hil,\Hil') \bigr ) \subseteq \mathrm{A}^{N+1}_{\mu} S^m_{\rho,\delta} \bigl ( \mathcal{B}(\Hil,\Hil') \bigr )
		. 
	\end{align*}
	Hence, we may view the space of asymptotic Hörmander symbols 
	\begin{align*}
		\SemiHoer{m}_{\rho,\delta} \bigl ( \mathcal{B}(\Hil,\Hil') \bigr ) := \bigcap_{N \in \N_0} \mathrm{A}^N_{\mu} S^m_{\rho,\delta} \bigl ( \mathcal{B}(\Hil,\Hil') \bigr )
	\end{align*}
	as the inductive limit of the finite-order spaces; to unburden the notation, we do not make the dependence of $\SemiHoer{m}_{\rho,\delta} \bigl ( \mathcal{B}(\Hil,\Hil') \bigr )$ on $\mu$ explicit. 
	
	Note, however, that the $\SemiHoer{m}_{\rho,\delta} \bigl ( \mathcal{B}(\Hil,\Hil') \bigr )$ are \emph{not} Fréchet spaces as they are no longer Hausdorff: each formal sum $\sum_{n = 0}^{\infty} \eps^n f_n$ has several resummations. To see this, we topologize the space of formal sums. As before, we begin with the space of finite sums 
	\begin{align*}
		\Sigma^N_{\mu} S^m_{\rho,\delta} \bigl ( \mathcal{B}(\Hil,\Hil') \bigr ) := \bigoplus_{n = 0}^N S^{m - n \mu (\rho - \delta)}_{\rho,\delta} \bigl ( \mathcal{B}(\Hil,\Hil') \bigr ) 
		\subsetneq \mathrm{A}^N_{\mu} S^m_{\rho,\delta} \bigl ( \mathcal{B}(\Hil,\Hil') \bigr ) 
		, 
	\end{align*}
	which differs from $\mathrm{A}^N_{\mu} S^m_{\rho,\delta} \bigl ( \mathcal{B}(\Hil,\Hil') \bigr )$ only in that we discard the last term in the direct sum for the remainder. As before, these spaces nest into one another and the intersection 
	\begin{align*}
		\Sigma S^m_{\rho,\delta} \bigl ( \mathcal{B}(\Hil,\Hil') \bigr ) := \bigcap_{N \in \N_0} \Sigma^N_{\mu} S^m_{\rho,\delta} \bigl ( \mathcal{B}(\Hil,\Hil') \bigr ) 
	\end{align*}
	of all finite sum spaces is therefore naturally endowed with the inductive limit topology. 
	
	When $\mu > 0$ and $\rho > \delta$, the order $m - n \mu (\rho - \delta) \rightarrow -\infty$ of the Hörmander spaces tend to $-\infty$. Consequently, the difference 
	\begin{align*}
		f_{\epsilon} - \sum_{n = 0}^{\infty} \epsilon^n \, f_n = \order(\epsilon^{\infty}) 
		\in S^{-\infty} \bigl ( \mathcal{B}(\Hil,\Hil') \bigr ) 
	\end{align*}
	is a smoothing symbol. Conversely, results such as \cite[Proposition~2.26]{Folland:harmonic_analysis_hase_space:1989} or \cite[Proposition~2.5.33]{Ruzhansky_Turunen:pseudodifferential_operators_symmetries:2010} explain how to construct \emph{a} resummation from a given formal sum; the constructive proofs make clear that resummations are not unique, and different resummations differ by $S^{-\infty} \bigl ( \mathcal{B}(\Hil,\Hil') \bigr )$. Put succinctly, for this choice of parameters the space of asymptotic Hörmander symbols is related to the space of asymptotic sums by 
	\begin{align*}
		\SemiHoer{m}_{\rho,\delta} \bigl ( \mathcal{B}(\Hil,\Hil') \bigr ) \Big / S^{-\infty}\bigl ( \mathcal{B}(\Hil,\Hil') \bigr ) \cong \Sigma \Hoer{m}_{\rho,\delta} \bigl ( \mathcal{B}(\Hil,\Hil') \bigr )
		. 
	\end{align*}
	Hence, the relationship between the two mimics that of $\mathcal{L}^p(\R^d)$ and $L^p(\R^d) = \mathcal{L}^p(\R^d) \big / \sim$. 
\end{remark}
In perturbative expansions one constructs a formal sum $\sum_{n = 0}^{\infty} \epsilon^n \, f_n$ order-by-order from some recursive procedure and not an asymptotic symbol $f_{\epsilon}$. Fortunately, when $\rho > \delta$ and $\mu > 0$ we usually need not distinguish between formal sums and asymptotic symbols. 
\begin{lemma}\label{operator_valued_calculus:lem:existence_resummation_formal_sum}
	Suppose $\rho > \delta$ and $\mu > 0$. Then there is a one-to-one correspondence between asymptotic symbols and formal sums up to $S^{-\infty} \bigl ( \mathcal{B}(\Hil,\Hil') \bigr )$ and $\order(\epsilon^{\infty})$ in the following sense: any asymptotic symbol $f_{\epsilon} \in \SemiHoer{m}_{\rho,\delta} \bigl ( \mathcal{B}(\Hil,\Hil') \bigr )$ determines the coefficients in the formal sum $\sum_{n = 0}^{\infty} \epsilon^n \, f_n \in \Sigma S^m_{\rho,\delta} \bigl ( \mathcal{B}(\Hil,\Hil') \bigr )$. 
	
	Conversely, any formal sum $\sum_{n = 0}^{\infty} \epsilon^n \, f_n \in \Sigma S^m_{\rho,\delta} \bigl ( \mathcal{B}(\Hil,\Hil') \bigr )$ admits a resummation $f_{\epsilon} \in \SemiHoer{m}_{\rho,\delta} \bigl ( \mathcal{B}(\Hil,\Hil') \bigr )$ so that 
	\begin{align*}
		f_{\epsilon} - \sum_{n = 0}^{\infty} \epsilon^n \, f_n = \order(\epsilon^{\infty})
		\in S^{-\infty} \bigl ( \mathcal{B}(\Hil,\Hil') \bigr ) 
		. 
	\end{align*}
	The resummation is unique up to $S^{-\infty} \bigl ( \mathcal{B}(\Hil,\Hil') \bigr )$ and $\order(\epsilon^{\infty})$, \ie the difference of two resummations $f_{\epsilon} - \tilde{f}_{\epsilon} = \order(\epsilon^{\infty}) \in S^{-\infty} \bigl ( \mathcal{B}(\Hil,\Hil') \bigr )$ is a smoothing symbol. 
\end{lemma}
\begin{proof}
	The first implication, determining the formal sum from an asymptotic symbol, follows directly from the definition of $\SemiHoer{m}_{\rho,\delta} \bigl ( \mathcal{B}(\Hil,\Hil') \bigr )$. 
	
	Only for the other direction do we need the assumption $\rho > \delta$ and $\mu > 0$: suppose we are given a formal sum $\sum_{n = 0}^{\infty} \epsilon^n \, f_n \in \Sigma S^m_{\rho,\delta} \bigl ( \mathcal{B}(\Hil,\Hil') \bigr )$. Then we can use the strategy in the proofs of \cite[Proposition~2.26]{Folland:harmonic_analysis_hase_space:1989} or \cite[Proposition~2.5.33]{Ruzhansky_Turunen:pseudodifferential_operators_symmetries:2010} to construct a resummation $f_{\epsilon}$ with the help of a cutoff function; the assumptions $\rho > \delta$ and $\mu > 0$ ensure that the order 
	\begin{align*}
		m - n \mu (\rho - \delta) \xrightarrow{n \rightarrow \infty} -\infty
	\end{align*}
	of the terms goes to $-\infty$ independently of $m \in \R$. Clearly, the resummation depends on our choice of cutoff function, but ultimately, the difference amounts to a smoothing symbol $S^{-\infty} \bigl ( \mathcal{B}(\Hil,\Hil') \bigr )$ of order $\epsilon^{\infty}$. The operator-valuedness only enters insofar that we need to replace the ordinary Hörmander seminorms by the operator-valued ones from Definition~\ref{operator_valued_calculus:defn:Hoermander_symbols}. 
\end{proof}
\begin{remark}[No results for resummations when $\mu = 0$ or $\rho = \delta$]\label{operator_valued_calculus:rem:lack_resummation_results_rho_equal_delta}
	Despite a thorough search of the literature, we were unable to find results that proved the existence of resummations for when $\rho = \delta$ or $\mu = 0$. The proofs for the case $\rho > \delta$ such as \cite[Proposition~2.26]{Folland:harmonic_analysis_hase_space:1989} or \cite[Proposition~2.5.33]{Ruzhansky_Turunen:pseudodifferential_operators_symmetries:2010} are constructive, and they fail when $\rho = \delta$ or $\mu = 0$. 
\end{remark}
Asymptotic expansions of the Weyl product are naturally formulated on asymptotic Hörmander symbol classes: 
\begin{theorem}[Asymptotic expansions of $\Weyl$]
	Suppose the magnetic field is bounded in the sense of Assumption~\ref{operator_valued_calculus:assumption:bounded_magnetic_fields} and $\rho \in [0,1]$. 
	Then the Weyl product 
	\begin{align*}
		\Weyl : S^{m_1}_{\rho,0} \bigl ( \mathcal{B}(\Hil,\Hil') \bigr ) \times S^{m_2}_{\rho,0} \bigl ( \mathcal{B}(\Hil',\Hil'') \bigr ) \longrightarrow S^{m_1 + m_2}_{\rho,0} \bigl ( \mathcal{B}(\Hil,\Hil'') \bigr )
	\end{align*}
	has asymptotic expansions in $\eps$, $\lambda$ as well as $\eps$ \emph{and} $\lambda$, all of which define bilinear continuous maps 
	\begin{align*}
		&\Weyl : \SemiHoer{m_1}_{\rho,0} \bigl ( \mathcal{B}(\Hil,\Hil') \bigr ) \times \SemiHoer{m_2}_{\rho,0} \bigl ( \mathcal{B}(\Hil',\Hil'') \bigr ) \longrightarrow \SemiHoer{m_1 + m_2}_{\rho,0} \bigl ( \mathcal{B}(\Hil,\Hil'') \bigr ) 
		, 
		\\
		&\Weyl : \Sigma \Hoer{m_1}_{\rho,0} \bigl ( \mathcal{B}(\Hil,\Hil') \bigr ) \times \Sigma \Hoer{m_2}_{\rho,0} \bigl ( \mathcal{B}(\Hil',\Hil'') \bigr ) \longrightarrow \Sigma \Hoer{m_1 + m_2}_{\rho,0} \bigl ( \mathcal{B}(\Hil,\Hil'') \bigr ) 
		. 
	\end{align*}
	The terms of the expansion can be found in \cite{Lein:two_parameter_asymptotics:2008}, specifically Theorems~1.1 and 2.12. 
\end{theorem}
Note that asymptotic expansions other than the expansion in $\eps$ and $\lambda$ have been considered in the literature. In \cite{Fuerst_Lein:scaling_limits_Dirac:2008} one of the authors showed how to obtain the semirelativistic limit of the Dirac dynamics; the scaling $\eps = \nicefrac{v}{c}$ and $\lambda = \nicefrac{v^2}{c^2}$ emerges then, where $v$ is the typical speed the particle travels at and $c$ is the speed of light. 
\begin{remark}
	We reckon we can extend the above theorem to symbol classes $S^m_{\rho,\delta}$ where $0 \leq \delta < \rho \leq 1$, all one needs to do is extend the existence results of the oscillatory integral in \cite[Appendix~D]{Lein:two_parameter_asymptotics:2008}. 
\end{remark}
One important application is the existence of a parametrix for elliptic magnetic pseudodifferential operators. When a scalar-valued pseudodifferential operator $\Op^A(f)$ possesses a bounded inverse, then its parametrix is the asymptotic expansion of 
\begin{align*}
	\Op^A(f)^{-1} \asymp \sum_{n = 0}^{\infty} \Op^A(g_n) 
\end{align*}
One of the most common applications is the case $f = h - z$, $z\in \C\setminus\sigma(\Op^A(h))$, where $\Op^A(h)$ defines a selfadjoint operator. Then the parametrix yields an asymptotic expansion of the operator resolvent. However, there are situations where $\Op^A(f)$ is not invertible, but the parametrix nevertheless exists. While this may seem like a downside, in many applications it is actually a feature and not a bug. For example, one can still formulate a holomorphic functional calculus analogous to the one described in Section~\ref{operator_valued_calculus:Hoermander_symbols:inversion_and_functional_calculus} below and define \eg projections perturbatively order-by-order \cite{Panati_Spohn_Teufel:sapt_PRL:2002,DeNittis_Lein:Bloch_electron:2009,Fuerst_Lein:scaling_limits_Dirac:2008,DeNittis_Lein:sapt_photonic_crystals:2013,Freund_Teufel:non_trivial_Bloch_sapt:2013}. Moreover, they enter in the proofs of many technical results such as that of  Corollary~\ref{operator_valued_calculus:cor:boundedness_magnetic_Sobolev_spaces}. 

Given that there are many asymptotic expansions of the magnetic Weyl product (\cf \eg \cite{Lein:two_parameter_asymptotics:2008,Fuerst_Lein:scaling_limits_Dirac:2008}), we will keep the discussion a bit more abstract. 
\begin{proposition}[Existence of a left parametrix]\label{operator_valued_calculus:prop:existence_left_parametrix}
	Suppose he magnetic field $B$ is bounded in the sense of Assumption~\ref{operator_valued_calculus:assumption:bounded_magnetic_fields}. Furthermore, we make the following assumptions: 
	\begin{enumerate}[(a)]
		\item We are given a symbol $f \in S^m_{\rho,\delta} \bigl ( \mathcal{B}(\Hil,\Hil') \bigr )$, $m \in \R$ and $\rho > \delta \geq 0$. 
		\item We are given an asymptotic expansion of the magnetic Weyl product 
		\begin{align*}
			h \Weyl k \asymp \sum_{n = 0}^{\infty} \epsilon^n \, (h \Weyl k)_{(n)} 
			\in \SemiHoer{m_1 + m_2}_{\rho,\delta} \bigl ( \mathcal{B}(\Hil,\Hil'') \bigr ) 
		\end{align*}
		of $h \in \SemiHoer{m_1}_{\rho,\delta} \bigl ( \mathcal{B}(\Hil',\Hil'') \bigr )$ and $k \in \SemiHoer{m_2}_{\rho,\delta} \bigl ( \mathcal{B}(\Hil,\Hil') \bigr )$ in a small parameter $\epsilon \ll 1$ in the sense of Definition~\ref{operator_valued_calculus:defn:asymptotic_Hoermander_class} with $\mu > 0$. 
		\item There exists a symbol $g_0 \in S^{-m}_{\rho,\delta} \bigl ( \mathcal{B}(\Hil',\Hil) \bigr )$ that satisfies 
		\begin{align*}
			g_0 \Weyl f - \id_{\Hil} &= \order(\epsilon) \in S^{-\mu (\rho - \delta)}_{\rho,\delta} \bigl ( \mathcal{B}(\Hil) \bigr )
			. 
		\end{align*}
	\end{enumerate}
	Then there exists a symbol 
	\begin{align*}
		f^{(-1)_{\epsilon}} = \sum_{n = 0}^{\infty} \epsilon^n \, f^{(-1)_{\epsilon}}_n
		\in S^{-m}_{\rho,\delta} \bigl ( \mathcal{B}(\Hil',\Hil) \bigr ) 
	\end{align*}
	called the \emph{left parametrix}, which is the inverse with respect to $\Weyl$ up to $\order(\epsilon^{\infty})$, 
	\begin{align*}
		f^{(-1)_{\epsilon}} \Weyl f - \id_{\Hil} &= \order(\epsilon^{\infty}) 
		\in S^{-\infty} \bigl ( \mathcal{B}(\Hil) \bigr )
		. 
	\end{align*}
	The terms $f^{(-1)_{\epsilon}}_n \in S^{-m - n \mu(\rho - \delta)}_{\rho,\delta} \bigl ( \mathcal{B}(\Hil',\Hil) \bigr )$ can be computed explicitly order-by-order, and the $0$th-order term $f^{(-1)_{\epsilon}}_0 = g_0$ is the symbol from assumption~(c). The symbol $f^{(-1)_{\epsilon}}$ is unique up to $\order(\epsilon^{\infty}) \in S^{-\infty} \bigl ( \mathcal{B}(\Hil',\Hil) \bigr )$. 
\end{proposition}
\begin{proof}
	The strategy of the proof is quite standard and can be found in \eg \cite[Theorem~2.8]{Iftimie_Mantoiu_Purice:magnetic_psido:2006}. At the core the idea is to derive the asymptotic expansion of the inverse $f^{(-1)_{\epsilon}}$ from the \emph{inversion defect}
	\begin{align*}
		r_0 := \epsilon^{-1} \bigl ( g_0 \Weyl f - \id_{\Hil} \bigr ) 
		= \order(1) 
		\in S^{- \mu (\rho - \delta)}_{\rho,\delta} \bigl ( \mathcal{B}(\Hil) \bigr )
	\end{align*}
	that belongs to the associated symbol class by assumption~(c). Since we have added a factor of $\epsilon^{-1}$, the inversion defect is $r_0 = \order(1)$. Now we define another symbol by asymptotically summing the formal geometric series 
	\begin{align*}
		r := \id_{\Hil} + \sum_{n = 1}^{\infty} (-1)^n \, \epsilon^n \, r_0^{\Weyl \, n}
	\end{align*}
	where $r_0^{\Weyl \, n} := r_0 \Weyl \cdots \Weyl r_0$ is the $n$-fold Weyl product of $r_0$ with itself. \emph{Formally,} $r = \bigl ( \id_{\Hil} - \eps \, r_0 \bigr )^{(-1)_{\epsilon}}$ is just the inverse of $\id_{\Hil} - \eps \, r_0$ with respect to the magnetic Weyl product, which explains why 
	\begin{align*}
		f^{(-1)_{\epsilon}} := r \Weyl g_0 
	\end{align*}
	is a left-inverse of $f$, 
	\begin{align*}
		f^{(-1)_{\epsilon}} \Weyl f - \id_{\Hil} &= r \Weyl \bigl ( \epsilon \, r_0 - \id_{\Hil} \bigr )
		= \order(\epsilon^{\infty})
		\in S^{-\infty} \bigl ( \mathcal{B}(\Hil) \bigr ) 
		. 
	\end{align*}
	Doing this rigorously just amounts to checking that order-by-order in $\epsilon$ the terms cancel. 
\end{proof}
An analogous statement holds for the right parametrix. Since the proof is a straightforward modification of the one above, we will omit it for brevity. 
\begin{corollary}[Existence of a right parametrix]\label{operator_valued_calculus:cor:existence_right_parametrix}
	Suppose we are in the setting of Proposition~\ref{operator_valued_calculus:prop:existence_left_parametrix}, but we replace assumption~(c) with 
	\begin{enumerate}[(a')]
		\setcounter{enumi}{2}
		\item There exists a symbol $g_0 \in S^{-m}_{\rho,\delta} \bigl ( \mathcal{B}(\Hil',\Hil) \bigr )$ that satisfies 
		\begin{align*}
			f \Weyl g_0 - \id_{\Hil'} &= \order(\epsilon) \in S^{-\mu (\rho - \delta)}_{\rho,\delta} \bigl ( \mathcal{B}(\Hil') \bigr )
			. 
		\end{align*}
	\end{enumerate}
	Then there exists a symbol 
	\begin{align*}
		f^{(-1)_{\epsilon}} = \sum_{n = 0}^{\infty} \epsilon^n \, f^{(-1)_n}
		\in S^{-m}_{\rho,\delta} \bigl ( \mathcal{B}(\Hil',\Hil) \bigr ) 
	\end{align*}
	called the \emph{right parametrix}, which is the inverse with respect to $\Weyl$ up to $\order(\epsilon^{\infty})$, 
	\begin{align*}
		f \Weyl f^{(-1)_{\epsilon}} - \id_{\Hil'} &= \order(\epsilon^{\infty}) 
		\in S^{-\infty} \bigl ( \mathcal{B}(\Hil') \bigr )
		. 
	\end{align*}
	The terms $f^{(-1)_n} \in S^{-m - n \mu(\rho - \delta)}_{\rho,\delta} \bigl ( \mathcal{B}(\Hil',\Hil) \bigr )$ can be computed explicitly order-by-order, and the $0$th-order term $f^{(-1)_0} = g_0$ is the symbol from assumption~(c'). The symbol $f^{(-1)_{\epsilon}}$ is unique up to $\order(\epsilon^{\infty}) \in S^{-\infty} \bigl ( \mathcal{B}(\Hil',\Hil) \bigr )$. 
\end{corollary}
One can combine Proposition~\ref{operator_valued_calculus:prop:existence_left_parametrix} and Corollary~\ref{operator_valued_calculus:cor:existence_right_parametrix} to deduce the existence of a (left \emph{and} right) parametrix. 
\begin{corollary}[Existence of a parametrix]\label{operator_valued_calculus:cor:existence_parametrix}
	Suppose the conditions listed in Proposition~\ref{operator_valued_calculus:prop:existence_left_parametrix} \emph{and} Corollary~\ref{operator_valued_calculus:cor:existence_right_parametrix} are satisfied. Moreover, we make the following additional assumptions: 
	\begin{enumerate}[(a'')]
		\item The Hilbert space $\Hil \hookrightarrow \Hil'$ can be continuously and densely embedded into $\Hil'$. 
		\item The symbol $g_0 \in S^{-m}_{\rho,\delta} \bigl ( \mathcal{B}(\Hil',\Hil) \bigr )$ from assumptions~(c) and (c') is such that the expresssion
		\begin{align*}
			g_0 \Weyl f - \id_{\Hil'} &= \order(\epsilon) \in S^{-\mu (\rho - \delta)}_{\rho,\delta} \bigl ( \mathcal{B}(\Hil') \bigr )
		\end{align*}
		continuously extends to operator-valued symbols on $\Hil' \supseteq \Hil$ and the expression 
		\begin{align*}
			f \Weyl g_0 - \id_{\Hil} &= \order(\epsilon) \in S^{-\mu (\rho - \delta)}_{\rho,\delta} \bigl ( \mathcal{B}(\Hil) \bigr )
		\end{align*}
		restricts to a symbol taking values in $\mathcal{B}(\Hil)$. 
	\end{enumerate}
	Then the left and right parametrices $f^{(-1)_{\epsilon}}_{\mathrm{L}}$ and $f^{(-1)_{\epsilon}}_{\mathrm{R}}$ constructed in Proposition~\ref{operator_valued_calculus:prop:existence_left_parametrix} and Corollary~\ref{operator_valued_calculus:cor:existence_right_parametrix} agree asymptotically in the sense that 
	\begin{align*}
		f^{(-1)_{\epsilon}}_{\mathrm{L}} - f^{(-1)_{\epsilon}}_{\mathrm{R}} &= \order(\epsilon^{\infty}) 
		\in S^{-\infty} \bigl ( \mathcal{B}(\Hil) \bigr ) \cap S^{-\infty} \bigl ( \mathcal{B}(\Hil') \bigr ) 
		. 
	\end{align*}
\end{corollary}
Assumption~(a'') is quite natural when \eg $\Hil = H^m_A(\R^d,\C^n)$ and $\Hil' = H^{m'}_A(\R^d,\C^n)$ are magnetic Sobolev spaces of different orders. When $m \leq m'$ we can naturally include $H^m_A(\R^d,\C^n) \subseteq  H^{m'}_A(\R^d,\C^n)$ the Sobolev space of lower order into the one of higher order. 
\begin{proof}
	Let us denote left and right parametrices with $f^{(-1)_{\epsilon}}_{\mathrm{L}}$ and $f^{(-1)_{\epsilon}}_{\mathrm{R}}$. Note that the symbol $g_0$ that enters assumptions~(c) and (c') is one and the same. 
	
	Assumption~(a'') is necessary to make sure \eg we may view $\id_{\Hil} = \id_{\Hil'} \big \vert_{\Hil}$ as restriction of the identity on $\Hil'$.
	
	Assumption~(b'') tells us we can view the symbols 
	\begin{align*}
		r_{\mathrm{L},0} := \epsilon^{-1} \bigl ( g_0 \Weyl f - \id_{\Hil} \bigr ) 
		\in S^{-\mu (\rho - \delta)}_{\rho,\delta} \bigl ( \mathcal{B}(\Hil) \bigr ) \cap S^{-\mu (\rho - \delta)}_{\rho,\delta} \bigl ( \mathcal{B}(\Hil') \bigr )
		, 
		\\
		r_{\mathrm{R},0} := \epsilon^{-1} \bigl ( f \Weyl g_0 - \id_{\Hil'} \bigr )
		\in S^{-\mu (\rho - \delta)}_{\rho,\delta} \bigl ( \mathcal{B}(\Hil) \bigr ) \cap S^{-\mu (\rho - \delta)}_{\rho,\delta} \bigl ( \mathcal{B}(\Hil') \bigr )
		, 
	\end{align*}
	as elements of symbol spaces that take values either in $\mathcal{B}(\Hil)$ or $\mathcal{B}(\Hil')$. Thus, the symbols $f^{(-1)_{\epsilon}}_{\mathrm{L}}$ and $f^{(-1)_{\epsilon}}_{\mathrm{R}}$ constructed in Proposition~\ref{operator_valued_calculus:prop:existence_left_parametrix} and Corollary~\ref{operator_valued_calculus:cor:existence_right_parametrix}, respectively, can be viewed as elements of 
	\begin{align*}
		f^{(-1)_{\epsilon}}_{\mathrm{L}} , f^{(-1)_{\epsilon}}_{\mathrm{R}} \in S^{-m}_{\rho,\delta} \bigl ( \mathcal{B}(\Hil) \bigr ) \cap S^{-m}_{\rho,\delta} \bigl ( \mathcal{B}(\Hil') \bigr )
		. 
	\end{align*}
	Consequently, it makes sense to consider the difference 
	\begin{align*}
		f^{(-1)_{\epsilon}}_{\mathrm{L}} - f^{(-1)_{\epsilon}}_{\mathrm{R}} &= f^{(-1)_{\epsilon}}_{\mathrm{L}} \Weyl \bigl ( \id - f \Weyl f^{(-1)_{\epsilon}}_{\mathrm{R}} \bigr ) - \bigl ( \id - f^{(-1)_{\epsilon}}_{\mathrm{L}} \Weyl f \bigr ) \Weyl f^{(-1)_{\epsilon}}_{\mathrm{R}} 
		\\
		&= \order(\epsilon^{\infty}) 
		\in S^{-\infty} \bigl ( \mathcal{B}(\Hil) \bigr ) \cap S^{-\infty} \bigl ( \mathcal{B}(\Hil') \bigr )
		, 
	\end{align*}
	which, after adding and subtracting the mixed term $f^{(-1)_{\epsilon}}_{\mathrm{L}} \Weyl f \Weyl f^{(-1)_{\epsilon}}_{\mathrm{R}}$, can be shown to be small thanks to Theorem~\ref{operator_valued_calculus:thm:Weyl_product_Hoermander_symbols}. 
\end{proof}
\begin{remark}[Existence of (left/right) parametrices when $\rho = 0$ and $\mu = 0$]\label{operator_valued_calculus:rem:existence_parametrices}~\\
	The conditions $\rho > 0$ and $\mu > 0$ can be lifted. For instance, when we construct parametrices for \emph{equivariant} magnetic pseudodifferential operators, the equivariance condition~\eqref{equivariant_calculus:eqn:equivariance_condition_symbols} forces $\rho = 0 = \delta$ (\cf Lemma~\ref{equivariant_calculus:magnetic_PsiDOs:lem:bound_Hoermander_order_by_tau_orders}). 
	
	Fortunately, the statements and proofs of Proposition~\ref{operator_valued_calculus:prop:existence_left_parametrix} as well as Corollaries~\ref{operator_valued_calculus:cor:existence_right_parametrix} and \ref{operator_valued_calculus:cor:existence_parametrix} can be modified to cover the case $\rho = 0$ and $\mu = 0$. Here, all the terms and the remainders of the asymptotic expansion of the magnetic Weyl product $\Weyl$ lie in the same symbol class, namely $\Sigma S^{m_1 + m_2}_{0,0} \bigl ( \mathcal{B}(\Hil,\Hil'') \bigr )$. 
	
	The only price we have to pay is that the parametrix exists only as a formal sum, and not as an asymptotic symbol. That is because we are not aware of a general existence result for a resummation when $\rho = \delta$ or $\mu = 0$ (\cf Remark~\ref{operator_valued_calculus:rem:lack_resummation_results_rho_equal_delta}). The lack of a resummation may or may not matter: for perturbation expansions, it usually does not matter. However, there are cases such as Theorem~\ref{operator_valued_calculus:thm:selfadjointness_elliptic_symbols} where the existence of a resummation of the parametrix is crucial for the proof. We will return to this point in Section~\ref{equivariant_calculus:calculus:boundedness_results} in the discussion of Theorem~\ref{equivariant_calculus:thm:selfadjointness_elliptic_symbols_weak_result}. 
\end{remark}
%

\subsubsection{Beals' Commutator Criterion to identify magnetic $\Psi$DOs with Hörmander symbols} 
\label{operator_valued_calculus:Hoermander_symbols:commutator_criteria}
The Calderón-Vaillancourt Theorem~\ref{operator_valued_calculus:thm:Calderon_Vaillancourt} tells us that magnetic $\Psi$DOs defined from Hörmander symbols of order $m$ define bounded operators $H^m_A(\R^d,\Hil) \longrightarrow L^2(\R^d,\Hil')$. We could turn this question on its head and ask: under what conditions is a given operator a magnetic $\Psi$DO defined from a Hörmander symbol? The standard tool to answer this conclusively is a commutator criterion, the most common ones being Beals' and Bony's Commutator Criteria \cite{Beals:characterization_psido:1977,Bony:characterization_psido:1996}. 

The first proof of Beals' and Bony's Commutator Criteria for \emph{magnetic} pseudodifferential operators is due to Iftimie, Măntoiu and Purice \cite{Iftimie_Mantoiu_Purice:commutator_criteria:2008}; however, rather than following the original work, we will adapt a more modern and elegant proof of Beals' Criterion given by Cornean, Helffer and Purice \cite{Cornean_Helffer_Purice:simple_proof_Beals_criterion_magnetic_PsiDOs:2018}. We will show that it extends in a straightforward fashion to operator-valued symbols. Small parameters do not play a role in the proof. Nevertheless, we will include the semiclassical parameter $\eps$ and the parameter $\lambda$ for the magnetic field strength in the formulæ below. 

The idea can be most easily explained for the symbol class 
\begin{align*}
	S^0_{0,0} \bigl ( \mathcal{B}(\Hil,\Hil') \bigr ) = \Cont^{\infty}_{\mathrm{b}} \bigl ( T^* \R^d , \mathcal{B}(\Hil,\Hil') \bigr ) 
	, 
\end{align*}
which consists of bounded functions whose derivatives in $x$ and $\xi$ are all bounded. Any $f \in S^0_{0,0} \bigl ( \mathcal{B}(\Hil,\Hil') \bigr )$ defines a \emph{bounded} magnetic pseudodifferential operator
\begin{align*}
	\Op^A(f) : L^2(\R^d,\Hil) \longrightarrow L^2(\R^d,\Hil') 
	,
\end{align*}
precisely because we can control sufficiently many derivatives; this is the content of the Calderón-Vaillancourt Theorem~\ref{operator_valued_calculus:thm:Calderon_Vaillancourt}. 

If we reverse the premise, suppose we are given a bounded operator 
\begin{align*}
	F : L^2(\R^d,\Hil) \longrightarrow L^2(\R^d,\Hil') 
	, 
\end{align*}
and we would like to know whether $F = \Op^A(f)$ for some $f \in S^0_{0,0} \bigl ( \mathcal{B}(\Hil,\Hil') \bigr )$. \emph{A priori} all we know from the Schwartz Kernel Theorem~\ref{operator_valued_calculus:lem:Schwartz_kernel_theorem} is that there exists a \emph{tempered distribution} $f \in \Schwartz^* \bigl ( T^* \R^d , \mathcal{B}(\Hil,\Hil') \bigr )$ with $F = \Op^A(f)$. 

They key is to control “derivatives” of the operator $F$. On the level of operators we can define a notion of “derivative” from the building block operators $Q_j$ and $P_j^A$ through commutators. More properly, though, we should refer to%
\begin{subequations}\label{operator_valued_calculus:eqn:derivations_for_commutator_criteria}
	\begin{align}
		\ad_{Q_j}(F) :& \negmedspace = [Q_j , F] := Q_j \, F - F \, Q_j 
		, 
		\\
		\ad_{P^A_j}(F) :& \negmedspace = [P^A_j , F] := P^A_j \, F - F \, P^A_j 
		, 
	\end{align}
\end{subequations}
for $j = 1 , \ldots , d$ as \emph{derivations}, because they satisfy the product rule (Leibniz's law), \eg 
\begin{align*}
	\ad_{Q_j}(F \, G) = \ad_{Q_j}(F) \, G + F \, \ad_{Q_j}(G) 
\end{align*}
holds for all suitable composable operators $F$ and $G$. Of course, in general the derivations~\eqref{operator_valued_calculus:eqn:derivations_for_commutator_criteria} need not exist. Indeed, their existence as bounded operators enters as an assumption in Theorem~\ref{operator_valued_calculus:thm:Beals_commutator_criterion} below. 

Note that the commutation relations~\eqref{operator_valued_calculus:eqn:commutation_relations} for position and momentum imply the derivations do not all commute with one another, and in principle, their order is important. For example, in the presence of magnetic fields two momentum derivations 
\begin{align*}
	\ad_{P^A_j} \circ \ad_{P^A_k}(F) &= \ad_{P^A_k} \circ \ad_{P^A_j}(F) + \bigl [ [P^A_j , P^A_k] \, , F \bigr ] 
	\\
	&= \ad_{P^A_k} \circ \ad_{P^A_j}(F) - \ii \eps \lambda \, \bigr [ B_{jk}(Q) \, , F \bigr ]
	\\
	&\neq \ad_{P^A_k} \circ \ad_{P^A_j}(F) 
\end{align*}
need not commute. This is where the Boundedness Assumption~\ref{operator_valued_calculus:assumption:bounded_magnetic_fields} on the magnetic field $B$ enters into our reasoning. It ensures that the extra term is the commutator of the \emph{bounded} operators $B_{jk}(Q)$ and $F$, and as a result its presence does not impose additional regularity conditions. This reasoning extends to all other and higher-order derivations: the extra terms that appear when changing the order of derivations involve $B$, derivatives of $B$ and lower-order derivations of $F$. 

Consequently, the order in which we take the derivations is immaterial for characterizing “smoothness”, and we may opt for lexicographical order, 
\begin{align}
	\partial^{(a,\alpha)}_{(Q,P^A)} F := \ad_{P^A_1}^{a_1} \circ \cdots \circ \ad_{P^A_d}^{a_d} \circ \ad_{Q_1}^{\alpha_1} \circ \cdots \circ \ad_{Q_d}^{\alpha_d}(F)
	, 
	&&
	a , \alpha \in \N_0^d 
	. 
	\label{operator_valued_calculus:eqn:commutator_criterion_derivations_operators}
\end{align}
This way we can express the $\Cont^{\infty}_{\mathrm{b}}$ condition on the operator more compactly as 
\begin{align}
	\partial^{(a,\alpha)}_{(Q,P^A)} F \in \mathcal{B} \bigl ( L^2(\R^d,\Hil) , L^2(\R^d,\Hil') \bigr ) 
	&&
	\forall a , \alpha \in \N_0^d 
	. 
	\label{operator_valued_calculus:eqn:boundedness_derivations_operators}
\end{align}
For practical purposes it is more expedient to directly work with distributions rather than operators. We pull back the Banach space of bounded operators and define 
\begin{align*}
	\mathfrak{A}^B \bigl ( \mathcal{B}(\Hil,\Hil') \bigr ) := \bigl ( \Op^A \bigr )^{-1} \Bigl ( \mathcal{B} \bigl ( L^2(\R^d,\Hil) \, , L^2(\R^d,\Hil') \bigr ) \Bigr ) 
	\subseteq \Schwartz^* \bigl ( T^* \R^d , \mathcal{B}(\Hil,\Hil') \bigr ) 
	.
\end{align*}
Because we can continuously embed the bounded operators into 
\begin{align*}
	\mathcal{B} \bigl ( L^2(\R^d,\Hil) , L^2(\R^d,\Hil') \bigr ) \subset \mathcal{L} \bigl ( \Schwartz(\R^d,\Hil) , \Schwartz^*(\R^d,\Hil') \bigr ) 
	= \Op^A \Bigl ( \Schwartz^* \bigl ( T^* \R^d , \mathcal{B}(\Hil,\Hil') \bigr ) \Bigr ) 
	, 
\end{align*}
and the latter is the range of $\Op^A$ under the operator-valued tempered distributions (\cf Proposition~\ref{operator_valued_calculus:prop:extension_OpA_tempered_distributions}~(2)), the vector space $\mathfrak{A}^B \bigl ( \mathcal{B}(\Hil,\Hil') \bigr )$ indeed consists of tempered distributions. Endowed with the pulled back norm 
\begin{align*}
	\norm{f}_{\mathfrak{A}^B(\mathcal{B}(\Hil,\Hil'))} := \bnorm{\Op^A(f)}_{\mathcal{B}(L^2(\R^d,\Hil),L^2(\R^d,\Hil'))} 
	, 
\end{align*}
it again forms a Banach space. One of the advantages is that due to gauge covariance, Proposition~\ref{operator_valued_calculus:prop:extension_OpA_tempered_distributions}~(3), the Banach space $\mathfrak{A}^B \bigl ( \mathcal{B}(\Hil,\Hil') \bigr )$ only depends on the magnetic \emph{field} $B$ rather than the vector potential $A$. With this notation in hand, we can reformulate the Calderon-Vallaincourt Theorem~\ref{operator_valued_calculus:thm:Calderon_Vaillancourt} as 
\begin{corollary}\label{operator_valued_calculus:cor:Calderon_Vaillancourt}
	For any $0 \leq \delta \leq \rho \leq 1$ with $\delta \neq 1$ the Hörmander symbol class
	\begin{align*}
		S^0_{\rho,\delta} \bigl ( \mathcal{B}(\Hil,\Hil') \bigr ) \hookrightarrow \mathfrak{A}^B \bigl ( \mathcal{B}(\Hil,\Hil') \bigr )
	\end{align*}
	of order $0$ and type $(\rho,\delta)$ can be continuously injected into $\mathfrak{A}^B \bigl ( \mathcal{B}(\Hil,\Hil') \bigr )$. 
\end{corollary}
When two such Banach spaces $\mathfrak{A}^B \bigl ( \mathcal{B}(\Hil',\Hil'') \bigr ) \ni f$ and $\mathfrak{A}^B \bigl ( \mathcal{B}(\Hil,\Hil') \bigr ) \ni g$ are composable, we can pull back the operator product to the level of distributions, 
\begin{align}
	f \Weyl g := \bigl ( \Op^A \bigr )^{-1} \Bigl ( \Op^A(f) \, \Op^A(g) \Bigr ) 
	\in \mathfrak{A}^B \bigl ( \mathcal{B}(\Hil,\Hil'') \bigr )
	. 
	\label{operator_valued_calculus:eqn:Weyl_product_Banach_space_distributions_preimage_bounded_operators}
\end{align}
Evidently, the notation is not at all accidental, for suitable elements, say Schwartz functions, the product $f \Weyl g$ defined through \eqref{operator_valued_calculus:eqn:Weyl_product_Banach_space_distributions_preimage_bounded_operators} coincides with the magnetic Weyl product~\eqref{operator_valued_calculus:eqn:Weyl_product}. Likewise, we can introduce a notion of adjoint ${}^*$ that assigns $f \mapsto f^* \in \mathfrak{A}^B \bigl ( \mathcal{B}(\Hil',\Hil) \bigr )$; note that the order of $\Hil$ and $\Hil'$ is reversed. 

Hence, the derivations~\eqref{operator_valued_calculus:eqn:derivations_for_commutator_criteria} pull back to Moyal commutators%
\begin{subequations}\label{operator_valued_calculus:eqn:Moyal_derivations_for_commutator_criteria}
	\begin{align}
		\ad_{x_j}(f) :& \negmedspace = [x_j , f]_{\Weyl} := x_j \Weyl f - f \Weyl x_j 
		, 
		\\
		\ad_{\xi_j}(f) :& \negmedspace = [\xi_j , f]_{\Weyl} := \xi_j \Weyl f - f \Weyl \xi_j
		, 
	\end{align}
\end{subequations}
which suggests why the boundedness of $\partial^{(a,\alpha)}_{(Q,P^A)}(F)$ implies $\Cont^{\infty}_{\mathrm{b}}$-regularity of the prequantization $f$: formally, with the help of the asymptotic expansion of the Weyl product $\Weyl$ in $\eps$ (\cf \cite[Theorem~1.1]{Lein:two_parameter_asymptotics:2008}) the Moyal commutators simplify to 
\begin{subequations}\label{operator_valued_calculus:eqn:action_ad_x_ad_xi_functions}
	\begin{align}
		\ad_{x_j}(f) &= + \eps \, \ii \partial_{\xi_j} f 
		, 
		\label{operator_valued_calculus:eqn:action_ad_x_ad_xi_functions:ad_x}
		\\
		\ad_{\xi_j}(f) &= - \eps \, \ii \partial_{x_j} f - \eps \, \lambda \sum_{k = 1}^d B_{jk} \, \partial_{\xi_k} f + \order(\eps^3)
		\label{operator_valued_calculus:eqn:action_ad_x_ad_xi_functions:ad_xi}
		. 
	\end{align}
\end{subequations}
Up to a prefactor, the position \emph{derivation} equals the momentum \emph{derivative} (the order is reversed!); the momentum \emph{derivation} involves both, $x$- and $\xi$-derivatives of $f$ as well as the magnetic field $B_{jk}$. We again see the Boundedness Assumption~\ref{operator_valued_calculus:assumption:bounded_magnetic_fields} on $B_{jk} \in \Cont^{\infty}_{\mathrm{b}}(\R^d)$ implies that we can control the derivations $\ad_{x_j}(f)$ and $\ad_{\xi_j}(f)$ in terms of \emph{derivatives} of $f$ — and vice versa. 

Just as above, we will pick lexicographical order for our higher-order derivations 
\begin{align}
	\partial^{(a,\alpha)}_{(x,\xi)} f := \ad_{\xi_1}^{a_1} \circ \cdots \circ \ad_{\xi_d}^{a_d} \circ \ad_{x_1}^{\alpha_1} \circ \cdots \circ \ad_{x_d}^{\alpha_d}(f) 
	, 
	&&
	a , \alpha \in \N_0^d
	, 
	\label{operator_valued_calculus:eqn:Moyal_commutator_derivations_distributions_phase_space}
\end{align}
on the level of tempered distributions. Again, due to our assumptions on $B$, the order in which we take these derivations is not important for Theorem~\ref{operator_valued_calculus:thm:Beals_commutator_criterion} below. 

With these notions in hand, we can find an equivalent form for \eqref{operator_valued_calculus:eqn:boundedness_derivations_operators}, which places regularity conditions on the distribution $f$ rather than the operator $F = \Op^A(f)$:
\begin{align*}
	\partial^{(a,\alpha)}_{(x,\xi)} f &\in \mathfrak{A}^B \bigl ( \mathcal{B}(\Hil,\Hil') \bigr ) 
	\quad \forall a , \alpha \in \N_0^d 
	\\
	&\qquad \qquad \big\Updownarrow 
	\\
	\partial^{(a,\alpha)}_{(Q,P^A)} F  &\in \mathcal{B} \bigl ( L^2(\R^d,\Hil) , L^2(\R^d,\Hil') \bigr ) 
	\quad \forall a , \alpha \in \N_0^d 
\end{align*}
As the notation suggests, in the end the Fréchet topology defined by the \emph{derivations}~\eqref{operator_valued_calculus:eqn:Moyal_commutator_derivations_distributions_phase_space} is equivalent to the usual Fréchet topology of $\Cont^{\infty}_{\mathrm{b}} \bigl ( T^* \R^d , \mathcal{B}(\Hil,\Hil') \bigr )$ generated by the family of sup norms of all \emph{derivatives} $\partial_x^a \partial_{\xi}^{\alpha} f$. Indeed, this is the point of view taken in \cite{Iftimie_Mantoiu_Purice:commutator_criteria:2008} to prove Beals' Commutator Criterion: 
\begin{theorem}[Beals' Commutator Criterion]\label{operator_valued_calculus:thm:Beals_commutator_criterion}
	Suppose the magnetic field $B$ is bounded in the sense of Assumption~\ref{operator_valued_calculus:assumption:bounded_magnetic_fields}. 
	Then a bounded operator 
	\begin{align*}
		F  = \Op^A(f) \in \mathcal{B} \bigl ( L^2(\R^d,\Hil) \, , \, L^2(\R^d,\Hil') \bigr )
		. 
	\end{align*}
	is a magnetic pseudodifferential operator associated to a Hörmander symbol $f \in S^0_{0,0} \bigl ( \mathcal{B}(\Hil,\Hil') \bigr )$ if and only if  
	\begin{align*}
		\partial^{(a,\alpha)}_{(x,\xi)} f \in \mathfrak{A}^B \bigl ( \mathcal{B}(\Hil,\Hil') \bigr ) 
		&&
		\forall a , \alpha \in \N_0^d 
		. 
	\end{align*}
\end{theorem}
Before we furnish the proof, let us give an important generalization that identifies Hörmander symbols $S^m_{\rho,0} \bigl ( \mathcal{B}(\Hil,\Hil') \bigr )$ of arbitrary order and more general type. However, to be able to apply Beal's Criterion~\ref{operator_valued_calculus:thm:Beals_commutator_criterion}, we wish to relate $\Op^A(f)$ to a map between $L^2$-spaces. The trick here is the same as in the proof of Corollary~\ref{operator_valued_calculus:cor:boundedness_magnetic_Sobolev_spaces}, which states that a magnetic $\Psi$DO defined from $f \in S^m_{\rho,0} \bigl ( \mathcal{B}(\Hil,\Hil') \bigr )$ gives rise to a continuous operator 
\begin{align*}
	\Op^A(f) : H^s_A(\R^d,\Hil) \longrightarrow H^{s-m}_A(\R^d,\Hil')
\end{align*}
between the magnetic Sobolev space of order $s$ and $s-m$. 

One way to do this is to multiply $\Op^A(f)$ with another operator that maps $H^{s-m}_A(\R^d,\Hil')$ to $H^s_A(\R^d,\Hil)$. Our choice is the magnetic $\Psi$DO associated to the symbol 
\begin{align*}
	w_m \in S^m_{1,0}(\C) \subseteq S^m_{\rho,0} \bigl ( \mathcal{B}(\Hil) \bigr )
\end{align*}
defined in \eqref{operator_valued_calculus:eqn:definition_alternate_weights_magnetic_Sobolev_spaces}. Importantly, $w_m^{(-1)_{\Weyl}} = w_{-m} \in S^{-m}_{1,0}(\C)$ has already been demonstrated in the literature (\eg as a consequence of \cite[Proposition~6.31]{Iftimie_Mantoiu_Purice:commutator_criteria:2008} or \cite[Theorem~1.8]{Mantoiu_Purice_Richard:spectral_propagation_results_magnetic_Schroedinger_operators:2007}). 

For non-negative $m \geq 0$ it is elliptic and equals 
\begin{align*}
	w_m(x,\xi) = \sexpval{\xi}^m + \lambda(m)
	, 
\end{align*}
where the constant $\lambda(m) \geq 0$ is chosen large enough to ensure that the Moyal inverse $w_m^{(-1)_{\Weyl}} =: w_{-m}$ exists. Such a constant $\lambda(m)$ always exists: the Gårding Inequality for magnetic $\Psi$DOs \cite[Corollary~6.4]{Iftimie_Mantoiu_Purice:magnetic_psido:2006} guarantees that $\Op^A \bigl ( \sexpval{\xi}^m \bigr )$ is bounded from below. 

Naturally, $w_m$ can be used to define an equivalent norm on magnetic Sobolev spaces by using $\Op^A(w_m)$ as a weight in the scalar product~\eqref{operator_valued_calculus:eqn:scalar_product_Sobolev_space}. Therefore, invoking Corollary~\ref{operator_valued_calculus:cor:boundedness_magnetic_Sobolev_spaces} once more, we may view 
\begin{align*}
	\Op^A(w_m) : H^{-m}_A(\R^d,\Hil') \longrightarrow L^2(\R^d,\Hil')
\end{align*}
as a continuous operator that maps back into $L^2(\R^d,\Hil)$. Thus, their composition 
\begin{align*}
	\Op^A(w_{-m}) \, \Op^A(f) = \Op^A \bigl ( w_{-m} \Weyl f \bigr ) : L^2(\R^d,\Hil) \longrightarrow L^2(\R^d,\Hil') 
\end{align*}
indeed maps between $L^2$-spaces, which proves 
\begin{align*}
	w_{-m} \Weyl f \in \mathfrak{A}^B \bigl ( \mathcal{B}(\Hil,\Hil') \bigr ) 
	. 
\end{align*}
To see where this condition comes from, pretend we already know $f \in S^m_{\rho,0} \bigl ( \mathcal{B}(\Hil,\Hil') \bigr )$ is a Hörmander symbol. Then by virtue of Theorem~\ref{operator_valued_calculus:thm:Weyl_product_Hoermander_symbols} their product is mapped into 
\begin{align*}
	w_{-m} \Weyl f \in S^{-m}_{\rho,0}(\C) \Weyl S^m_{\rho,0} \bigl ( \mathcal{B}(\Hil,\Hil') \bigr ) 
	\subseteq S^0_{0,0} \bigl ( \mathcal{B}(\Hil,\Hil') \bigr ) 
	. 
\end{align*}
Derivatives of $f$ in $\xi$ lower the order of the Hörmander class by $\rho$ whereas derivatives in $x$ do not change symbol class, 
\begin{align}
	\partial_x^a \partial_{\xi}^{\alpha} f \in S^{m - \sabs{\alpha} \rho}_{\rho,0} \bigl ( \mathcal{B}(\Hil,\Hil') \bigr ) 
	&&
	\forall a , \alpha \in \N_0^d 
	. 
\end{align}
We may equivalently express this as 
\begin{align}
	w_{-m + \sabs{\alpha} \rho} \Weyl \partial_x^a \partial_{\xi}^{\alpha} f \in S^0_{0,0} \bigl ( \mathcal{B}(\Hil,\Hil') \bigr ) \subset \mathfrak{A}^B \bigl ( \mathcal{B}(\Hil,\Hil') \bigr ) 
	. 
	\label{operator_valued_calculus:eqn:Beals_commutator_criterion_condition_w_minus_m_Weyl_f}
\end{align}
These arguments not only explain how to modify the condition in Theorem~\ref{operator_valued_calculus:thm:Beals_commutator_criterion}, it proves one of the two directions of the following if-and-only-if statement: 
%
%
\begin{corollary}\label{operator_valued_calculus:cor:Beals_commutator_criterion}
	Suppose the magnetic field $B$ is bounded in the sense of Assumption~\ref{operator_valued_calculus:assumption:bounded_magnetic_fields}. 
	Then a bounded operator 
	\begin{align*}
		F  = \Op^A(f) \in \mathcal{B} \bigl ( H^m_A(\R^d,\Hil) \, , \, L^2(\R^d,\Hil') \bigr )
		. 
	\end{align*}
	is a magnetic pseudodifferential operator associated to a Hörmander symbol $f \in S^m_{\rho,0} \bigl ( \mathcal{B}(\Hil,\Hil') \bigr )$ if and only if  
	\begin{align}
		w_{- m + \sabs{\alpha} \rho} \Weyl \partial^{(a,\alpha)}_{(x,\xi)} f \in \mathfrak{A}^B \bigl ( \mathcal{B}(\Hil,\Hil') \bigr ) 
		&&
		\forall a , \alpha \in \N_0^d 
		. 
		\label{operator_valued_calculus:eqn:Beals_commutator_criterion_condition_S_m_rho}
	\end{align}
\end{corollary}
Alternatively, we may put the weights to the right: 
%
%
\begin{corollary}
	In Corollary~\ref{operator_valued_calculus:cor:Beals_commutator_criterion} we may equivalently impose 
	\begin{align*}
		\partial^{(a,\alpha)}_{(x,\xi)} f \Weyl w_{- m + \sabs{\alpha} \rho} \in \mathfrak{A}^B \bigl ( \mathcal{B}(\Hil,\Hil') \bigr ) 
		&&
		\forall a , \alpha \in \N_0^d 
		. 
	\end{align*}
\end{corollary}
This corollary follows immediately from Corollary~\ref{operator_valued_calculus:cor:Beals_commutator_criterion}, we merely have to apply the adjoint ${}^*$, which flips the order of $w_{- m + \sabs{\alpha} \rho}$ and $\partial^{(a,\alpha)}_{(x,\xi)} f$. 

To tidy up the presentation of the proof of these two corollaries, we will factor out some auxiliary statements and prove them first. 
\begin{lemma}\label{operator_valued_calculus:lem:characterization_Hoermander_symbols_via_derivations}
	\begin{enumerate}[(1)]
		\item For Hörmander symbols $f \in S^m_{\rho,0} \bigl ( \mathcal{B}(\Hil,\Hil') \bigr )$ the derivations with respect to $x_j$ and $\xi_j$ simplify to equations~\eqref{operator_valued_calculus:eqn:action_ad_x_ad_xi_functions}. 
		\item We have the following characterization of operator-valued Hörmander symbols: 
		\begin{align*}
			f \in S^m_{\rho,0} \bigl ( \mathcal{B}(\Hil,\Hil') \bigr ) 
			\qquad \Longleftrightarrow \qquad 
			\partial_{(x,\xi)}^{(a,\alpha)} f \in S^{m - \sabs{\alpha} \rho}_{0,0} \bigl ( \mathcal{B}(\Hil,\Hil') \bigr ) 
			&&
			\forall a , \alpha \in \N_0^d 
		\end{align*}
		\item We have the following characterization of operator-valued Hörmander symbols: 
		\begin{align*}
			f \in S^m_{\rho,0} \bigl ( \mathcal{B}(\Hil,\Hil') \bigr ) 
			\qquad \Longleftrightarrow \qquad 
			\partial_{(x,\xi)}^{(0,\alpha)} f \in S^{m - \sabs{\alpha} \rho}_{0,0} \bigl ( \mathcal{B}(\Hil,\Hil') \bigr ) 
			&&
			\forall \alpha \in \N_0^d 
		\end{align*}
		\item The first-order derivations~\eqref{operator_valued_calculus:eqn:action_ad_x_ad_xi_functions} define continuous mappings 
		\begin{subequations}\label{operator_valued_calculus:eqn:continuity_derivations_Hoermander_symbol_classes}
			\begin{align}
				\ad_{x_j} &: S^m_{\rho,0} \bigl ( \mathcal{B}(\Hil,\Hil') \bigr ) \longrightarrow S^{m - \rho}_{\rho,0} \bigl ( \mathcal{B}(\Hil,\Hil') \bigr ) 
				, 
				\\
				\ad_{\xi_j} &: S^m_{\rho,0} \bigl ( \mathcal{B}(\Hil,\Hil') \bigr ) \longrightarrow S^m_{\rho,0} \bigl ( \mathcal{B}(\Hil,\Hil') \bigr ) 
				, 
			\end{align}
		\end{subequations}
		Thus, the higher-order derivations in lexicographical order, 
		\begin{align*}
			\partial_{(x,\xi)}^{(a,\alpha)} : S^m_{\rho,0} \bigl ( \mathcal{B}(\Hil,\Hil') \bigr ) \longrightarrow S^{m - \sabs{\alpha} \rho}_{\rho,0} \bigl ( \mathcal{B}(\Hil,\Hil') \bigr ) 
			, 
			&&
			a , \alpha \in \N_0^d 
			, 
		\end{align*}
		define continuous maps between Hörmander symbol classes. 
	\end{enumerate}
\end{lemma}
\begin{proof}
	\begin{enumerate}[(1)]
		\item Given that the momentum derivation is defined in terms of $\xi_j \in S^1_{1,0}(\C)$, Theorem~\ref{operator_valued_calculus:thm:Weyl_product_Hoermander_symbols} guarantees that $\ad_{\xi_j}(f)$ exists in $S^{m+1}_{\rho,0} \bigl ( \mathcal{B}(\Hil,\Hil') \bigr )$. Moreover, we can insert the asymptotic expansion of $\Weyl$ in $\eps$ into $\ad_{\xi_j}(f)$ and exploit that $(x,\xi) \mapsto \xi_j$ is a scalar-valued function so that all pointwise commutators vanish. Consequently, the asymptotic expansion starts at $\order(\eps)$ and $\ad_{\xi_j}(f) \in S^m_{\rho,0} \bigl ( \mathcal{B}(\Hil,\Hil') \bigr )$ is at most of the same order as $f$, namely $m$. This gives the first expression, equation~\eqref{operator_valued_calculus:eqn:action_ad_x_ad_xi_functions:ad_xi}. 
		
		To handle equation~\eqref{operator_valued_calculus:eqn:action_ad_x_ad_xi_functions:ad_x}, we need to make some straightforward modifications to standard existence results of the relevant oscillatory integrals such as \cite[Lemma~D.3]{Lein:two_parameter_asymptotics:2008}; this shows that $\ad_{x_j}(f)$ exists. 
		
		What is more, because the \emph{magnetic} Weyl product of a function that depends only on $x$ with another function reduces to the \emph{non-magnetic} Weyl product, we deduce 
		\begin{align*}
			\ad_{x_j}(f) &= x_j \Weyl f - f \Weyl x_j 
			= x_j \sharp^{B = 0} f - f \sharp^{B = 0} x_j 
			. 
		\end{align*}
		Exploiting the scalar-valuedness of $(x,\xi) \mapsto x_j$ yields that the zeroth-order term in the $\eps$ expansion, the pointwise commutator of $x_j$ and $f$ vanishes. Moreover, all magnetic terms vanish, and hence, the asymptotic expansion of $\ad_{x_j}(f)$ terminates after the $\order(\eps)$ term. Consequently, equation~\eqref{operator_valued_calculus:eqn:action_ad_x_ad_xi_functions:ad_x} is exact and we deduce $\ad_{x_j}(f) \in S^{m - \rho}_{\rho,0} \bigl ( \mathcal{B}(\Hil,\Hil') \bigr )$. 
		\item The $\Longrightarrow$ direction is obvious, it follows directly from Definition~\ref{operator_valued_calculus:defn:Hoermander_symbols}. 
	
		To show the converse implication, suppose $\partial_{(x,\xi)}^{(a,\alpha)} f \in S^{m - \sabs{\alpha} \rho}_{0,0} \bigl ( \mathcal{B}(\Hil,\Hil') \bigr )$ holds true for all $a , \alpha \in \N_0^d$. When we choose $a = 0 = \alpha$, the assumption implies $f \in S^m_{0,0} \bigl ( \mathcal{B}(\Hil,\Hil') \bigr )$ is a Hörmander symbol of order $m$. 
	
		Moreover, lexicographical ordering of the derivations yields $\partial_{(x,\xi)}^{(a,\alpha)} f = \partial_{(x,\xi)}^{(a,0)} \bigl ( \partial_{(x,\xi)}^{(0,\alpha)} f \bigr )$, and we may consider position and momentum derivations separately. From the continuity of $\ad_{x_j}$ on the level of Hörmander symbols and equation~\eqref{operator_valued_calculus:eqn:action_ad_x_ad_xi_functions:ad_x}, the momentum derivations 
		\begin{align*}
			\partial_{(x,\xi)}^{(0,\alpha)} f = (\ii \eps)^{\sabs{\alpha}} \, \partial_{\xi}^{\alpha} f \overset{\star}{\in} S^{m - \sabs{\alpha} \rho}_{0,0} \bigl ( \mathcal{B}(\Hil,\Hil') \bigr ) 
		\end{align*}
		reduce to the usual derivatives, and those lie in the correct Hörmander class ($\star$). 
	
		While we do not have such a neat equality for the position derivations, we may proceed by induction to deduce that $\partial_{(x,\xi)}^{(a,0)} f$ is a linear combination of derivatives of $f$ in $x$ as well as terms that depend on (derivatives of) $B_{jk}$ and $\partial_{\xi_k} f$, $k = 1 , \ldots , d$. 
		\item This follows directly from (2). 
		\item The proof of continuity for all first-order derivations was part of the proof of items~(1) and (2). That extends directly to higher-order derivations by induction. 
	\end{enumerate}
\end{proof}
With this out of the way, we can proceed with a proof. While we could in principle ask the reader to generalize Theorems~5.20 and 5.21 in \cite{Iftimie_Mantoiu_Purice:commutator_criteria:2008}, we can actually leverage their results and simplify our proofs. One substantial simplification is that we already know $w_{-m} = w_m^{(-1)_{\Weyl}} \in S^{-m}_{1,0}(\C)$ holds for all values $m \in \R$. This either follows from \cite[Proposition~6.31]{Iftimie_Mantoiu_Purice:commutator_criteria:2008} and the G{\r a}rding Inequality (\cite[Corollary~5.4]{Iftimie_Mantoiu_Purice:magnetic_psido:2006}); alternatively, we may refer to \cite[Theorem~1.8]{Mantoiu_Purice_Richard:spectral_propagation_results_magnetic_Schroedinger_operators:2007} for a direct proof. 
\begin{proof}[Corollary~\ref{operator_valued_calculus:cor:Beals_commutator_criterion}]
	$\Longrightarrow$: Suppose $f \in S^m_{\rho,0} \bigl ( \mathcal{B}(\Hil,\Hil') \bigr )$ is a Hörmander class symbol. Then according to Corollary~\ref{operator_valued_calculus:cor:boundedness_magnetic_Sobolev_spaces} to the Calderón-Vaillancourt Theorem~\ref{operator_valued_calculus:thm:Calderon_Vaillancourt} the magnetic pseudodifferential operator 
	\begin{align*}
		\Op^A(f) : H^m_A(\R^d,\Hil) \longrightarrow L^2(\R^d,\Hil')
	\end{align*}
	is bounded. What is more, in view of Theorem~\ref{operator_valued_calculus:thm:Weyl_product_Hoermander_symbols} the composition 
	\begin{align*}
		w_{-m + \sabs{\alpha} \rho} \Weyl \partial_x^a \partial_{\xi}^{\alpha} f \in S^0_{\rho,0} \bigl ( \mathcal{B}(\Hil,\Hil') \bigr ) 
	\end{align*}
	is a Hörmander symbol of order $0$, no matter what multi indices $a , \alpha  \in \N_0^d$ we pick. Consequently, it also defines an element of $\mathfrak{A}^B \bigl ( \mathcal{B}(\Hil,\Hil') \bigr )$ (Corollary~\ref{operator_valued_calculus:cor:Calderon_Vaillancourt}), and we deduce that condition~\eqref{operator_valued_calculus:eqn:Beals_commutator_criterion_condition_S_m_rho} is satisfied. 
	\medskip
	
	\noindent
	$\Longleftarrow$: Suppose the operator $\Op^A(f)$ is such that equation~\eqref{operator_valued_calculus:eqn:Beals_commutator_criterion_condition_S_m_rho} holds true. 

	When $m = 0$ and $\rho = 0$, this is just the content of Beals' Commutator Criterion~\ref{operator_valued_calculus:cor:Beals_commutator_criterion}. 
	
	So assume $m \neq 0$ or $\rho = 0$; we will use a bootstrap argument to infer the case $\rho \in (0,1]$ from the case $\rho = 0$ afterwards. Lemma~\ref{operator_valued_calculus:lem:characterization_Hoermander_symbols_via_derivations}~(1) gives us an equivalent characterization of elements $f \in S^m_{0,0} \bigl ( \mathcal{B}(\Hil,\Hil') \bigr )$ in terms of derivations, 
	\begin{align*}
		w_{-m} \Weyl \partial_{(x,\xi)}^{(a,\alpha)} f \in S^0_{0,0} \bigl ( \mathcal{B}(\Hil,\Hil') \bigr )
		&&
		\forall a , \alpha \in \N_0^d 
		. 
	\end{align*}
	Since we want to invoke Beals' Commutator Criterion~\ref{operator_valued_calculus:cor:Beals_commutator_criterion}, we need to prove that $f$ verifies the condition 
	\begin{align*}
		\partial_{(x,\xi)}^{(a,\alpha)} \bigl ( w_{-m} \Weyl f \bigr ) \in \mathfrak{A}^B \bigl ( \mathcal{B}(\Hil,\Hil') \bigr )
		&&
		\forall a , \alpha \in \N_0^d 
		. 
	\end{align*}
	Fortunately, this condition is equivalent to \eqref{operator_valued_calculus:eqn:Beals_commutator_criterion_condition_S_m_rho}: for any pair of multi indices $a , \alpha \in \N_0^d$ we can use Leibniz's rule to express 
	\begin{align*}
		\partial_{(x,\xi)}^{(a,\alpha)} \bigl ( w_{-m} \Weyl f \bigr ) &= \sum_{\substack{b + c = a \\ \beta + \gamma = \alpha}} \bigl ( \partial_{(x,\xi)}^{(b,\beta)} w_{-m} \bigr ) \Weyl \partial_{(x,\xi)}^{(c,\gamma)} f 
		\\
		&= \sum_{\substack{b + c = a \\ \beta + \gamma = \alpha}} \underbrace{\bigl ( \partial_{(x,\xi)}^{(b,\beta)} w_{-m} \bigr ) \Weyl w_m}_{\in S^0_{0,0}(\C) \subseteq \mathfrak{A}^B(\mathcal{B}(\Hil'))} \Weyl \underbrace{w_{-m} \Weyl \partial_{(x,\xi)}^{(c,\gamma)} f}_{\in \mathfrak{A}^B(\mathcal{B}(\Hil,\Hil'))} 
		\in \mathfrak{A}^B \bigl ( \mathcal{B}(\Hil,\Hil') \bigr ) 
	\end{align*}
	as a sum of derivations of $w_{-m} = w_m^{(-1)_{\Weyl}}$ and $f$. 
	
	The first term in the product of the sum is an element of $S^0_{1,0}(\C) \subseteq S^0_{\rho,0} \bigl ( \mathcal{B}(\Hil') \bigr )$: from the characterization of Hörmander classes through Lemma~\ref{operator_valued_calculus:lem:characterization_Hoermander_symbols_via_derivations}~(2) we have learnt 
	\begin{align*}
		\partial_{(x,\xi)}^{(b,\beta)} w_{-m} \in S^{- m - \sabs{\beta}}_{1,0}(\C)
		\subseteq S^{-m}_{0,0}(\C)
		. 
	\end{align*}
	Combined with $\beta + \gamma = \alpha$ this confirms that the first term lies in $S^0_{0,0}(\C) \subset \mathfrak{A}^B \bigl ( \mathcal{B}(\Hil') \bigr )$ (Corollary~\ref{operator_valued_calculus:cor:Calderon_Vaillancourt}). The second term is an element of $\mathfrak{A}^B \bigl ( \mathcal{B}(\Hil,\Hil') \bigr )$ by assumption. Consequently, also their product is an element of $\mathfrak{A}^B \bigl ( \mathcal{B}(\Hil,\Hil') \bigr )$. 
	
	From Beals' Commutator Criterion~\ref{operator_valued_calculus:thm:Beals_commutator_criterion} we therefore deduce 
	\begin{align*}
		w_{-m} \Weyl f \in S^0_{0,0} \bigl ( \mathcal{B}(\Hil,\Hil') \bigr ) 
	\end{align*}
	is a Hörmander symbol, and we have shown the case $\rho = 0$ even when $m \neq 0$. 
	
	Lastly, we will consider the case $m \in \R$ and $\rho \in (0,1]$. Our arguments will hinge on two facts: the first is the characterization of elements of $S^m_{\rho,0} \bigl ( \mathcal{B}(\Hil,\Hil') \bigr )$ in terms of position derivations only  (Lemma~\ref{operator_valued_calculus:lem:characterization_Hoermander_symbols_via_derivations}~(3)). The other ingredient is equation~\eqref{operator_valued_calculus:eqn:action_ad_x_ad_xi_functions:ad_x}, \ie the fact that for Hörmander symbols position \emph{derivations} equal momentum \emph{derivatives} up to a factor $(\ii \eps)^{\pm \sabs{\alpha}}$ (Lemma~\ref{operator_valued_calculus:lem:characterization_Hoermander_symbols_via_derivations}~(1)). 
	
	Consequently, with the help of the commutator criterion for $\rho = 0$ that we have just proven we can give the following equivalent characterization for any $\beta \in \N_0^d$: 
	\begin{align*}
		\partial_{(x,\xi)}^{(0,\beta)} f \in S^{m - \sabs{\beta} \rho}_{0,0} \bigl ( \mathcal{B}(\Hil,\Hil') \bigr )
		\quad \Longleftrightarrow \quad 
		w_{- m + \sabs{\beta} \rho} \Weyl \partial_{(x,\xi)}^{(a,\alpha + \beta)} f \in \mathfrak{A}^B \bigl ( \mathcal{B}(\Hil,\Hil') \bigr ) 
		&&
		\forall a , \alpha \in \N_0^d 
		. 
	\end{align*}
	But all we have to do here is insert $1 = w_{m - (\sabs{\alpha} + \sabs{\beta}) \rho} \Weyl w_{m - (\sabs{\alpha} + \sabs{\beta}) \rho}^{(-1)_{\Weyl}}$ into the middle of the equation on the right, 
	\begin{align*}
		w_{- m + \sabs{\beta} \rho} \Weyl \partial_{(x,\xi)}^{(a,\alpha)} \partial_{\xi}^{\beta} f &= \Bigl ( w_{-m + \sabs{\beta} \rho} \Weyl w_{m - (\sabs{\alpha} + \sabs{\beta}) \rho} \Bigr ) \Weyl \Bigl ( w_{-m + (\sabs{\alpha} + \sabs{\beta}) \rho} \Weyl \partial_{(x,\xi)}^{(a,\alpha+\beta)} f \Bigr )
		, 
	\end{align*}
	and consider the two factors separately: the first defines an element of $\mathfrak{A}^B \bigl ( \mathcal{B}(\Hil') \bigr )$ since Theorem~\ref{operator_valued_calculus:thm:Weyl_product_Hoermander_symbols} and Corollary~\ref{operator_valued_calculus:cor:Calderon_Vaillancourt} imply 
	\begin{align*}
		w_{- m + \sabs{\beta} \rho} \Weyl w_{m - (\sabs{\alpha} + \sabs{\beta}) \rho} \in S^{- m + \sabs{\beta} \rho}_{1,0}(\C) \Weyl S^{m - (\sabs{\alpha} + \sabs{\beta}) \rho}_{1,0}(\C) 
		\subseteq S^{- \sabs{\alpha} \rho}_{1,0}(\C) 
		\subseteq \mathfrak{A}^B \bigl ( \mathcal{B}(\Hil') \bigr ) 
		. 
	\end{align*}
	The second factor is an element of $\mathfrak{A}^B \bigl ( \mathcal{B}(\Hil') \bigr )$ by assumption. Hence, no matter the value $\beta \in \N_0^d$ the momentum derivative 
	\begin{align*}
		\partial_{\xi}^{\beta} f = (\ii \eps)^{-\sabs{\beta}} \, \partial_{(x,\xi)}^{(0,\beta)} f \in S^{m - \sabs{\beta} \rho}_{0,0} \bigl ( \mathcal{B}(\Hil,\Hil') \bigr ) 
	\end{align*}
	lies in the Hörmander space of the correct order, which is equivalent to $f \in S^m_{\rho,0} \bigl ( \mathcal{B}(\Hil,\Hil') \bigr )$. This not only finishes the proof of the case $\rho \in (0,1]$, but the proof of the corollary as a whole. 
\end{proof}
Now is our turn to furnish a proof of Beal's Commutator Criterion~\ref{operator_valued_calculus:thm:Beals_commutator_criterion}. Clearly, one of the implications, 
\begin{align*}
	f \in S^0_{0,0} \bigl ( \mathcal{B}(\Hil,\Hil') \bigr )
	\quad \Longrightarrow \quad
	\partial^{(a,\alpha)}_{(x,\xi)} f \in \mathfrak{A}^B \bigl ( \mathcal{B}(\Hil,\Hil') \bigr ) 
	&&
	\forall a , \alpha \in \N_0^d 
	.
\end{align*}
is a direct consequence of the Calderón-Vaillancourt Theorem~\ref{operator_valued_calculus:thm:Calderon_Vaillancourt}. 

Thus, we just owe the reader a proof of the converse implication of Theorem~\ref{operator_valued_calculus:thm:Beals_commutator_criterion}. 
\begin{proposition}\label{operator_valued_calculus:prop:other_direction_Beals_commutator_criterion}
	Suppose the magnetic field $B$ is bounded in the sense of Assumption~\ref{operator_valued_calculus:assumption:bounded_magnetic_fields}. Then the implication $\Longleftarrow$ in Theorem~\ref{operator_valued_calculus:thm:Beals_commutator_criterion} holds true:
	\begin{align*}
		f \in S^0_{0,0} \bigl ( \mathcal{B}(\Hil,\Hil') \bigr )
		\quad \Longleftarrow \quad
		\partial^{(a,\alpha)}_{(x,\xi)} f \in \mathfrak{A}^B \bigl ( \mathcal{B}(\Hil,\Hil') \bigr ) 
		&&
		\forall a , \alpha \in \N_0^d 
		.
	\end{align*}
	More specifically, Lemmas~3.1, 3.2 and 3.3 as well as their magnetic analogs from \cite{Cornean_Helffer_Purice:simple_proof_Beals_criterion_magnetic_PsiDOs:2018} extend to the case where $F$ and its commutators with $Q$ and $P^A$ lie in $\mathcal{B} \bigl ( L^2(\R^d,\Hil) , L^2(\R^d,\Hil') \bigr )$. That is, there exists an operator-valued symbol $f \in S^0_{0,0} \bigl ( \mathcal{B}(\Hil,\Hil') \bigr ) = \Cont^{\infty}_{\mathrm{b}} \bigl ( T^* \R^d , \mathcal{B}(\Hil,\Hil') \bigr )$ such that for all $\Psi' \in \Schwartz(\R^d,\Hil')$ and $\Psi \in \Schwartz(\R^d,\Hil)$ the inner product is given by the phase space integral
	\begin{align}
		&\bscpro{\Psi'}{F \Psi}_{L^2(\R^d,\Hil')}
		=
		\label{operator_valued_calculus:eqn:relation:matrix_element_operator_OpA_f}
		\\
		&\qquad
		= \frac{1}{(2\pi)^d} \int_{\R^d} \dd x \int_{\R^d} \dd y \int_{\R^d} \dd \eta \, \e^{- \ii \frac{\lambda}{\eps} \int_{[\eps x , \eps y]} A} \, \e^{- \ii \eta \cdot (y - x)} \; \bscpro{\Psi'(x) \, }{ \, f \bigl ( \tfrac{\eps}{2}(x+y) , \eta \bigr ) \, \Psi(y)}_{\Hil'}
		.
		\notag
	\end{align}
	Equivalently, $F = \Op^A(f)$ is the magnetic Weyl quantization of $f$. 
\end{proposition}
A reader looking at our straightforward modifications to the proof of Cornean et al.\ \cite{Cornean_Helffer_Purice:simple_proof_Beals_criterion_magnetic_PsiDOs:2018} may think our efforts are unnecessary. While morally, they are correct, we feel compelled to give the details for several reasons. First of all, it is complex enough for the uninitiated and merely pointing to it would place an undue burden on the reader. 

The second reason is that — unfortunately — to our knowledge there is no abstract result we can cite that allows us to leverage \cite[Theorem~1.2]{Cornean_Helffer_Purice:simple_proof_Beals_criterion_magnetic_PsiDOs:2018} directly: intuitively, we expect that because $x_j$ and $\xi_j$ are scalar-valued functions, the conditions derived from the Moyal commutators are independent of $f$ being operator-valued. Making this idea mathematically precise will involve tensor products like 
\begin{align*}
	\mathfrak{A}^B(\C) \otimes \mathcal{B}(\Hil,\Hil') \neq \mathfrak{A}^B \bigl ( \mathcal{B}(\Hil,\Hil') \bigr ) 
\end{align*}
and 
\begin{align*}
	S^0_{0,0}(\C) \otimes \mathcal{B}(\Hil,\Hil') \neq S^0_{0,0} \bigl ( \mathcal{B}(\Hil,\Hil') \bigr ) 
\end{align*}
for the relevant Hörmander symbol class. Unfortunately, the infinite-dimensional Banach space $\mathfrak{A}^B(\C) = \bigl ( \Op^A \bigr )^{-1} \bigl ( \mathcal{B} \bigl ( L^2(\R^d) \bigr ) \bigr )$ and the scalar-valued Hörmander class $S^0_{0,0}(\C)$ are not nuclear (\cf \cite[Chapter~50, Corollary~2]{Treves:topological_vector_spaces:1967}, and Theorem~3.2, Remark~3.3 and Proposition~4.4 in \cite{Witt:weak_topology_symbol_spaces:1997}). Therefore, left- and right-hand side are not even different: the former is \emph{not even well-defined} as there are several distinct topologies with respect to which we can complete the algebraic tensor product. 

Reasoning about topologies may seem like a minor technical detail, but it is central to the proof: one way to understand the hitherto unproven direction of Theorem~\ref{operator_valued_calculus:thm:Beals_commutator_criterion} is to prove an equivalence of Fréchet topologies; indeed, this is the point of view of Iftimie et al.\ (\cf \cite[Theorem~2.5]{Iftimie_Mantoiu_Purice:commutator_criteria:2008}). 
\medskip

\noindent
The arguments in the proof of Cornean, Helffer and Purice \cite{Cornean_Helffer_Purice:simple_proof_Beals_criterion_magnetic_PsiDOs:2018} rest on Gabor frames (\cf \eg \cite[Chapter~11]{Christensen:frames_Riesz_bases:2016}), which is an overcomplete set of vectors that in many respects behave like orthonormal bases. Then they relate properties of the operator $\Op^A(f)$ to properties of the “matrix elements” with respect to elements of the frame they have chosen: 
\begin{definition}
	Suppose $\chi \in \Cont^{\infty}_{\mathrm{c}}(\R^d,[0,1])$ is a cutoff function that possesses the following properties: 
	\begin{enumerate}[(a)]
		\item $\supp \chi \subseteq [-1,+1]^d$
		\item $\sum_{\gamma \in \Z^d} \chi_{\gamma}(x)^2 = 1$ holds for all $x \in \R^d$ where $\chi_{\gamma}(x) := \chi(x - \gamma)$
	\end{enumerate}
	Then for $\gamma \in \Z^d$ and $k \in \Z^d$ we define 
	\begin{align*}
		G_{\gamma,k}^A(x) :& \negmedspace= (2\pi)^{\nicefrac{d}{2}} \, \e^{- \ii \frac{\lambda}{\eps} \int_{[\eps x , \eps \gamma]} A} \, \chi(x - \gamma) \, \e^{+ \ii k \cdot (x - \gamma)} 
		\\
		&= (2\pi)^{\nicefrac{d}{2}} \, \e^{- \ii \frac{\lambda}{\eps} \int_{[\eps x , \eps \gamma]} A} \, \chi_{\gamma}(x) \, \e^{+ \ii k \cdot (x - \gamma)} 
		. 
	\end{align*}
\end{definition}
The Gabor frame $\bigl \{ G_{\gamma,k}^A \bigr \}_{\gamma \in \Gamma , \, k \in \Z^d}$ extend straightforwardly to Hilbert space-valued $L^2$-spaces. 
\begin{lemma}
	Suppose $\{ \varphi_n \}_{n \in \mathcal{I}}$ is a tight, normalized frame of $\Hil$. 
	\begin{enumerate}[(1)]
		\item The set $\bigl \{ G_{\gamma,k}^A \otimes \varphi_n \bigr \}_{\substack{\gamma \in \Z^d , \, k \in \Z^d \\ n \in \mathcal{I}}}$ is a tight, normalized frame of $L^2(\R^d,\Hil)$. 
		\item For any $\Psi \in \Schwartz(\R^d,\Hil)$ left- and right-hand side of 
		\begin{align}
			\Psi(x) &= \sum_{\substack{\gamma \in \Z^d , \, k \in \Z^d \\ n \in \mathcal{I}}} \bscpro{G_{\gamma,k}^A \otimes \varphi_n}{\Psi}_{L^2(\R^d,\Hil)} \, G_{\gamma,k}^A(x) \otimes \varphi_n 
			\label{appendix:commutator_criterion:eqn:normalized_tight_tensor_product_frame}
		\end{align}
		agree for all $x \in \R^d$ and the series on the right-hand side is absolutely convergent in the $\ell^1 \bigl ( \Z^d \times \Z^d , \Hil \bigr )$ sense. Moreover, the coefficients decay rapidly in $\gamma$ and $k$, \ie for any $N \in \N$ there exits $C_{\Psi,N} > 0$ such that 
		\begin{align}
			\sum_{n \in \mathcal{I}} \; \babs{\bscpro{G_{\gamma,k}^A \otimes \varphi_n}{\Psi}_{L^2(\R^d,\Hil)}}^2 \leq C_{f,N} \, \expval{\gamma}^{-N} \, \expval{k}^{-N}
			\label{appendix:commutator_criterion:eqn:basis_coefficients_estimate}
		\end{align}
		holds for all $\gamma \in \Z^d$ and $k \in \Z^d$. 
	\end{enumerate}
\end{lemma}
\begin{remark}
	Evidently, if $\Hil$ is finite-dimensional, the convergence of the sum in $n$ becomes a trivial matter and the extension is immediate. Only if the index set $\mathcal{I} \cong \N$ of the orthonormal basis is countably infinite, do we need to make sure that the sum over $n$ converges in a suitable sense. 
\end{remark}
\begin{proof}
	\begin{enumerate}[(1)]
		\item According to \cite[Lemma~2.1]{Cornean_Helffer_Purice:simple_proof_Beals_criterion_magnetic_PsiDOs:2018} the set $\bigl \{ G_{\gamma,k}^A \bigr \}_{_{\substack{\gamma \in \Z^d , \, k \in \Z^d}}}$ is a normalized, tight frame. Then also the tensor product of the two normalized, tight frames $\bigl \{ G_{\gamma,k}^A \otimes \varphi_n \bigr \}_{\substack{\gamma \in \Z^d , \, k \in \Z^d \\ n \in \mathcal{I}}}$ is a normalized, tight frame of the (Hilbert space) tensor product space $L^2(\R^d) \otimes \Hil \cong L^2(\R^d,\Hil)$ \cite[Theorem~2.3]{Khosravi_Asgari:tensor_product_normalized_tight_frame:2003}. 
		\item The normalized tightness proven in (1) shows that the sum on the right-hand side of \eqref{appendix:commutator_criterion:eqn:normalized_tight_tensor_product_frame} converges weakly in $L^2(\R^d,\Hil)$. 
		
		But we can do better than that: once we obtain the estimate~\eqref{appendix:commutator_criterion:eqn:basis_coefficients_estimate}, then this will imply $\ell^1 \bigl ( \Z^d \times \Z^d , \Hil \bigr )$-convergence of the sum in equation~\eqref{appendix:commutator_criterion:eqn:normalized_tight_tensor_product_frame}. 
		
		What we need to study is how the scalar product 
		\begin{align}
			\bscpro{G_{\gamma,k}^A \otimes \varphi_n}{\Psi}_{L^2(\R^d,\Hil)} &= \int_{\R^d} \dd x \, (2\pi)^{-\nicefrac{d}{2}} \; \chi(x - \gamma) \, \e^{+ \ii k \cdot (x - \gamma)} \, \e^{- \ii \frac{\lambda}{\eps} \int_{[\eps x , \eps \gamma]} A} \, \scpro{\varphi_n \, }{ \, \Psi(x)}_{\Hil} 
			\notag \\
			&= (2\pi)^{-\nicefrac{d}{2}} \int_{\R^d} \dd x \, \chi(x) \, \e^{+ \ii k \cdot x} \, \e^{- \ii \frac{\lambda}{\eps} \int_{[\eps (x + \gamma) , \eps \gamma]} A} \, \bscpro{\varphi_n \, }{ \, \Psi(x + \gamma)}_{\Hil} 
			\notag \\
			&= \Bigl ( \Fourier \bigl ( \chi \, \tau_{-\gamma} \bigl ( \e^{+ \ii \frac{\lambda}{\eps} \int_{[\eps \, \cdot , \eps \gamma]} A} \, \scpro{\varphi_n}{\Psi(\, \cdot \,)}_{\Hil} \bigr ) \bigr ) \Bigr )(k) 
			=: \bigl ( \Fourier \Psi^A_{\chi,\gamma,n} \bigr )(k)
			\label{appendix:commutator_criterion:eqn:frame_coefficient_Fourier_transform_nice_function_sans_sum_over_n}
		\end{align}
		behaves as a function of $n$, $k$ and $\gamma$. Here, we have used $\tau_{\gamma}$ to denote the translation of functions by $\gamma$, that is $(\tau_{\gamma} \Psi)(x) := \Psi(x - \gamma)$. 
		
		First of all, the above shows we can write this coefficient as the Fourier transform of the $\Cont^{\infty}_{\mathrm{c}}(\R^d,\Hil) \subset \Schwartz(\R^d,\Hil)$ function $\Psi_{\chi,\gamma,n}^A$. Given that Fourier transforms of Schwartz functions decay rapidly, this proves that the coefficients decay as $\sexpval{k}^{-N}$ for any $N \in \N_0$. 
		
		On to the decay in $\gamma$: let us pretend for the moment that $\gamma \in \R^d$ is a continuum variable. Then the rapid decay of $\Psi \in \Schwartz(\R^d,\Hil)$ and the fact that the derivatives of the vector potential $A \in \Cont^{\infty}_{\mathrm{pol}}(\R^d,\R^d)$ are polynomially bounded imply that also 
		\begin{align*}
			(\gamma,x) \mapsto \Psi^A_{\chi,\gamma,n}(x) 
			= \chi(x) \, \tau_{-\gamma} \bigl ( \e^{+ \ii \frac{\lambda}{\eps} \int_{[\eps \, \cdot , \eps \gamma]} A} \, \Psi \bigr )(x)
			\in \Schwartz \bigl ( \R^d \times \R^d \bigr )
		\end{align*}
		is a Schwartz function in both variables, $\Psi^A_{\chi,\gamma,n}(x)$ and all its derivatives in $\gamma$ and $x$ decay rapidly. 
		
		Because the Fourier transform in $x$ is continuous with respect to the Fréchet topology of $\Schwartz \bigl ( \R^d \times \R^d \bigr )$, we can estimate each seminorm of $(\gamma,k) \mapsto \bigl ( \Fourier \Psi^A_{\chi,\gamma,n} \bigr )(k)$ by a finite number of seminorms of $(\gamma,x) \mapsto \Psi^A_{\chi,\gamma,n}(x)$. Hence, the coefficient~\eqref{appendix:commutator_criterion:eqn:frame_coefficient_Fourier_transform_nice_function_sans_sum_over_n} vanishes faster than any inverse polynomial as $\abs{\gamma} \rightarrow \infty$. 
		
		The last step is to consider the sum over $n \in \mathcal{I}$. As mentioned before, this is only non-trivial if $\dim \Hil = \infty$. We again view all expressions below as functions of $(\gamma,k)$ or $(\gamma,x)$. First of all, as $\{ \varphi_n \}_{n \in \mathcal{I}}$ is a tight, normalized frame, the sum 
		\begin{align*}
			\sum_{n \in \mathcal{I}} \bscpro{G_{\gamma,k}^A \otimes \varphi_n}{\Psi}_{L^2(\R^d,\Hil)} \; \varphi_n 
		\end{align*}
		converges in $\Hil$. Exploiting the continuity of the Fourier transform on $\Schwartz \bigl ( \R^d \times \R^d , \Hil \bigr )$, we may interchange the sum over $n \in \mathcal{I}$ with the Fourier transform, 
		\begin{align*}
			\sum_{n \in \mathcal{I}} \bscpro{G_{\gamma,k}^A \otimes \varphi_n}{\Psi}_{L^2(\R^d,\Hil)} \; \varphi_n 
			&= \sum_{n \in \mathcal{I}} \Bigl ( \Fourier \bigl ( \chi \, \tau_{-\gamma} \bigl ( \e^{+ \ii \frac{\lambda}{\eps} \int_{[\eps \, \cdot , \eps \gamma]}} \, \scpro{\varphi_n}{\Psi(\, \cdot \,)}_{\Hil} \, \varphi_n \bigr ) \bigr ) \Bigr )(k)
			\\
			&= \Bigl ( \Fourier \Bigl ( \chi \, \tau_{-\gamma} \Bigl ( \e^{+ \ii \frac{\lambda}{\eps} \int_{[\eps \, \cdot , \eps \gamma]}} \, \mbox{$\sum_{n \in \mathcal{I}}$} \scpro{\varphi_n}{\Psi(\, \cdot \,)}_{\Hil} \, \varphi_n \Bigr ) \Bigr ) \Bigr )(k)
			\\
			&= \Bigl ( \Fourier \bigl ( \chi \, \tau_{-\gamma} \bigl ( \e^{+ \ii \frac{\lambda}{\eps} \int_{[\eps \, \cdot , \eps \gamma]}} \, \Psi(\, \cdot \,) \bigr ) \bigr ) \Bigr )(k) 
			=: (\Fourier \Psi_{\chi,\gamma}^A)(k) 
			, 
		\end{align*}
		where the convergence is understood to be in the Hilbert space $\Hil$. 
		
		Evidently, all the arguments we have made about $(\gamma,x) \mapsto \Psi^A_{\chi,\gamma,n}(x) \in \Schwartz \bigl (\R^d \times \R^d \bigr )$ and its Fourier transform regarding its behavior in $\gamma$ and $k$ apply verbatim to 
		\begin{align*}
			(\gamma,x) \mapsto \Psi^A_{\chi,\gamma}(x) \in \Schwartz \bigl ( \R^d \times \R^d , \Hil \bigr ) 
			. 
		\end{align*}
		This shows not only the bound \eqref{appendix:commutator_criterion:eqn:basis_coefficients_estimate}, but also the convergence of the sum on the right-hand side of equation~\eqref{appendix:commutator_criterion:eqn:normalized_tight_tensor_product_frame} in the $\ell^1 \bigl ( \Z^d \times \Z^d , \Hil \bigr )$ sense. 
	\end{enumerate}
\end{proof}
For the purpose of the proof, we assume we are given orthonormal bases $\{ \varphi_n \}_{n \in \mathcal{I}}$ and $\{ \varphi'_{n'} \}_{n' \in \mathcal{I}'}$ of $\Hil$ and $\Hil'$; these are of course also normalized, tights frames. The crucial quantity are the coefficients 
\begin{align}
	F^A_{\gamma',\gamma;k',k;n',n} := \Bscpro{G^A_{\gamma',k'} \otimes \varphi'_{n'} \, }{ \, F \; G^A_{\gamma,k} \otimes \varphi_n}_{L^2(\R^d,\Hil')}
	, 
	\label{operator_valued_calculus:eqn:coefficient_F_tight_frames}
\end{align}
and we need to study its behavior as $\gamma$, $\gamma'$, $k$, $k'$, $n$ and $n'$ vary. \emph{Formally,} the operator kernel of $F$ is given by the infinite sum 
\begin{align*}
	K_F(x',x) &= \sum_{\substack{\gamma , \gamma' \in \Z^d \\ k , k' \in \Z^d}} \sum_{\substack{n \in \mathcal{I} \\ n' \in \mathcal{I}'}} F^A_{\gamma',\gamma;k',k;n',n} \; G^A_{\gamma',k'}(x') \; \overline{G^A_{\gamma,k}(x)} \; \bket{\varphi'_{n'}}_{\Hil'} \bbra{\varphi_n}_{\Hil} 
	. 
\end{align*}
Cornean, Helffer and Purice \cite{Cornean_Helffer_Purice:simple_proof_Beals_criterion_magnetic_PsiDOs:2018} show how to deal with the sum over $\gamma$, $\gamma'$, $k$ and $k'$, and our job is to justify the sum over $n$ and $n'$. 

Fortunately, this is very easy. We will try to re-use the notation in \cite{Cornean_Helffer_Purice:simple_proof_Beals_criterion_magnetic_PsiDOs:2018} as much as possible. First, we show how to define 
\begin{align}
	F^A_{\gamma',\gamma;k',k} := \sum_{\substack{n \in \mathcal{I} \\ n' \in \mathcal{I}'}} F^A_{\gamma',\gamma;k',k;n',n} \; \bket{\varphi_n}_{\Hil} \bbra{\varphi'_{n'}}_{\Hil'} 
	\label{operator_valued_calculus:eqn:operator_valued_coefficient_F_tight_frames}
\end{align}
properly. 
\begin{lemma}
	Let $F \in \mathcal{B} \bigl ( L^2(\R^d,\Hil) , L^2(\R^d,\Hil') \bigr )$ be a bounded operator, and $\{ \varphi_n \}_{n \in \mathcal{I}}$ and $\{ \varphi'_{n'} \}_{n' \in \mathcal{I}'}$ be two normalized, tight frames of $\Hil$ and $\Hil'$. Then the operator $F^A_{\gamma',\gamma;k',k}$ defined by the right-hand side of \eqref{operator_valued_calculus:eqn:operator_valued_coefficient_F_tight_frames} exists as an element of $\mathcal{B}(\Hil,\Hil')$ whose norm can be uniformly estimated in $\gamma$, $\gamma'$, $k$ and $k'$, 
	\begin{align}
		\bnorm{F^A_{\gamma',\gamma;k',k}}_{\mathcal{B}(\Hil,\Hil')} \leq (2\pi)^{-d} \, \snorm{\chi}_{L^2(\R^d)}^2 \; \snorm{F}_{\mathcal{B}( L^2(\R^d,\Hil) , L^2(\R^d,\Hil') )}
		. 
		\label{operator_valued_calculus:eqn:norm_bound_F_gamma_gammaprime_k_kprime}
	\end{align}
\end{lemma}
\begin{proof}
	Pick some arbitrary $\psi \in \Hil$ and $\psi' \in \Hil'$ and set 
	\begin{align*}
		\Psi &:= G^A_{\gamma,k} \otimes \psi 
		,
		\\
		\Psi' &:= G^A_{\gamma',k'} \otimes \psi' 
		, 
	\end{align*}
	for some $\gamma , \gamma' \in \Z^d$ and $k , k' \in \Z^d$. The (potentially infinite) sums 
	\begin{align*}
		\Psi &= \sum_{n \in \mathcal{I}} \scpro{\varphi_n}{\psi}_{\Hil} \; G_{\gamma,k}^A \otimes \varphi_n 
		, 
		\\
		\Psi' &= \sum_{n' \in \mathcal{I}'} \sscpro{\varphi'_{n'}}{\psi'}_{\Hil'} \; G_{\gamma,k}^A \otimes \varphi'_{n'} 
		, 
	\end{align*}
	converge in $L^2(\R^d,\Hil)$ and $L^2(\R^d,\Hil')$, respectively. Since $F$ and the scalar product are (jointly) continuous, we can plug in the sums for $\Psi$ and $\Psi'$ into 
	\begin{align}
		\bscpro{\Psi'}{F \Psi}_{L^2(\R^d,\Hil')} &= \sum_{\substack{n \in \mathcal{I} \\ n' \in \mathcal{I}'}} \scpro{\varphi_n}{\psi}_{\Hil} \; \overline{\sscpro{\varphi'_{n'}}{\psi'}_{\Hil'}} \; \bscpro{G^A_{\gamma',k'} \otimes \varphi'_{n'} \, }{ \, F \, G^A_{\gamma,k} \otimes \varphi_n}_{L^2(\R^d,\Hil')} 
		\label{operator_valued_calculus:eqn:sum_coefficients_over_n_nprime}
		\\
		&= \sum_{\substack{n \in \mathcal{I} \\ n' \in \mathcal{I}'}} F^A_{\gamma',\gamma;k',k;n',n} \; \scpro{\varphi_n}{\psi}_{\Hil} \; \overline{\sscpro{\varphi'_{n'}}{\psi'}_{\Hil'}} 
		\notag \\
		&=: \bscpro{\psi'}{F^A_{\gamma',\gamma;k',k} \psi}_{\Hil'} 
		\notag 
	\end{align}
	Since linear operators between separable Hilbert spaces are uniquely defined by their matrix elements, the above gives a linear operator 
	\begin{align*}
		F^A_{\gamma',\gamma;k',k} : \Hil \longrightarrow \Hil' 
		. 
	\end{align*}
	Its matrix elements with respect to $\varphi'_{n'} \in \Hil'$ and $\varphi_n \in \Hil$ 
	\begin{align*}
		\bscpro{\psi'}{F^A_{\gamma',\gamma;k',k} \psi}_{\Hil'} &= \sum_{\substack{n \in \mathcal{I} \\ n' \in \mathcal{I}'}} F^A_{\gamma',\gamma;k',k;n',n} \; \scpro{\varphi_n}{\psi}_{\Hil} \; \overline{\sscpro{\varphi'_{n'}}{\psi'}_{\Hil'}} 
	\end{align*}
	reproduce the coefficients~\eqref{operator_valued_calculus:eqn:coefficient_F_tight_frames}. 
	
	To show that $\bnorm{F^A_{\gamma',\gamma;k',k}}_{\mathcal{B}(\Hil,\Hil')}$ is bounded, we apply the Cauchy-Schwarz inequality to the left-hand side of \eqref{operator_valued_calculus:eqn:sum_coefficients_over_n_nprime} and plug in $\bnorm{G^A_{\gamma,k}}_{L^2(\R^d)} = (2\pi)^{-\nicefrac{d}{2}} \, \snorm{\chi}_{L^2(\R^d)}$ to obtain 
	\begin{align*}
		\babs{\bscpro{\psi'}{F^A_{\gamma',\gamma;k',k} \psi}_{\Hil'}} &\leq \snorm{F}_{\mathcal{B}( L^2(\R^d,\Hil) , L^2(\R^d,\Hil') )} \, \bnorm{\Psi}_{L^2(\R^d,\Hil)} \, \bnorm{\Psi'}_{L^2(\R^d,\Hil')} 
		\\
		&= \snorm{F}_{\mathcal{B}( L^2(\R^d,\Hil) , L^2(\R^d,\Hil') )} \, \bnorm{G^A_{\gamma,k}}_{L^2(\R^d)} \, \bnorm{G^A_{\gamma',k'}}_{L^2(\R^d)} \, \snorm{\psi}_{\Hil} \, \snorm{\psi'}_{\Hil'} 
		\\
		&= (2\pi)^{-d} \, \snorm{\chi}_{L^2(\R^d)}^2 \, \snorm{F}_{\mathcal{B}( L^2(\R^d,\Hil) , L^2(\R^d,\Hil') )} \, \snorm{\psi}_{\Hil} \, \snorm{\psi'}_{\Hil'} 
		. 
	\end{align*}
	The operator norm of $F^A_{\gamma',\gamma;k',k}$ can be computed by taking the $\sup$ of the left-hand side over $\psi \in \Hil$ and $\psi' \in \Hil'$ with $\norm{\psi}_{\Hil} = 1 = \snorm{\psi'}_{\Hil'}$. Therefore, the above estimate not only shows $F^A_{\gamma',\gamma;k',k} \in \mathcal{B}(\Hil,\Hil')$, it in fact proves the desired norm bound~\eqref{operator_valued_calculus:eqn:norm_bound_F_gamma_gammaprime_k_kprime} that is uniform in $\gamma , \gamma' \in \Z^d$ and $k , k' \in \Z^d$. 
\end{proof}
Consequently, we can transliterate the arguments of \cite{Cornean_Helffer_Purice:simple_proof_Beals_criterion_magnetic_PsiDOs:2018} to the present context by viewing $F^A_{\gamma',\gamma;k',k}$ not as a complex number, but an element of $\mathcal{B}(\Hil,\Hil')$. E.~g.~each finite sum on the right-hand side of \cite[equation~(3.2)]{Cornean_Helffer_Purice:simple_proof_Beals_criterion_magnetic_PsiDOs:2018} belongs to $\Cont^{\infty}_{\mathrm{b}} \bigl ( \R^d \times \R^d , \mathcal{B}(\Hil,\Hil') \bigr )$ instead of the scalar valued $\Cont^{\infty}_{\mathrm{b}}$ functions. Similarly, the sum \cite[equation~(3.3)]{Cornean_Helffer_Purice:simple_proof_Beals_criterion_magnetic_PsiDOs:2018} should be understood in $\mathcal{B}(\Hil,\Hil')$. Thus, generally, absolute values just need to be replaced by $\norm{\cdot}_{\mathcal{B}(\Hil,\Hil')}$. 

As position and momentum operators act trivially on $\Hil$ and $\Hil'$, we can directly extend Lemmas~3.1, 3.2 and 3.3 as well as their magnetic analogs from Section~4 of \cite{Cornean_Helffer_Purice:simple_proof_Beals_criterion_magnetic_PsiDOs:2018} to operator-valued symbols. 
\begin{proof}[Proposition~\ref{operator_valued_calculus:prop:other_direction_Beals_commutator_criterion}]
	With slight abuse of notation, we will view $\Hil'$ as a subset of $\Hil$ and exploit the density when computing the operator norm 
	\begin{align*}
		\bnorm{F^A_{\gamma',\gamma;k',k}}_{\mathcal{B}(\Hil,\Hil')} &= \sup_{\substack{\psi , \psi' \in \Hil' \\ \snorm{\psi}_{\Hil} = 1 = \snorm{\psi'}_{\Hil'}}} \babs{\bscpro{\psi'}{F^A_{\gamma',\gamma;k',k} \psi}_{\Hil'}} 
	\end{align*}
	by taking the supremum only over $\Hil'$. Moreover, we will keep the magnetic field in our notation, even though Lemmas~3.1–3.3 address the non-magnetic case. 
	
	To obtain the estimates in \cite[Lemma~3.1]{Cornean_Helffer_Purice:simple_proof_Beals_criterion_magnetic_PsiDOs:2018}, we merely replace $F^A_{\gamma',\gamma;k',k}$ with the scalar product $\bscpro{\psi'}{F^A_{\gamma',\gamma;k',k} \psi}_{\Hil'}$ and repeat the computations for 
	\begin{align*}
		(\gamma_1 - \gamma_1') \; \bscpro{\psi' \, }{ \, F^A_{\gamma',\gamma;k',k} \psi}_{\Hil'} &= - \, \Bscpro{\bigl ( \eps^{-1} \, Q_1 - \gamma_1 \bigr ) \, G^A_{\gamma',k'} \otimes \psi' \, }{ \, F \, G^A_{\gamma,k} \otimes \psi}_{L^2(\R^d,\Hil')} 
		+ \\
		&\qquad 
		+ \Bscpro{G^A_{\gamma',k'} \otimes \psi' \, }{ \, \bigl [ \eps^{-1} \, Q_1 , F \bigr ] \, G^A_{\gamma,k} \otimes \psi}_{L^2(\R^d,\Hil')} 
		+ \\
		&\qquad 
		+ \Bscpro{G^A_{\gamma',k'} \otimes \psi' \, }{ \, F \, \bigl ( \eps^{-1} \, Q_1 - \gamma_1' \bigr ) \, G^A_{\gamma,k} \otimes \psi}_{L^2(\R^d,\Hil')} 
	\end{align*}
	and conclude that each of the three terms is bounded: The boundedness of the first and the last term stems from $G^A_{\gamma,k} \in \Schwartz(\R^d)$. And the middle commutator is bounded by assumption on $F$. 
	
	Taking the supremum as described above yields the desired operator norm bound. The operator-valuedness plays no role, because $Q$ and $P^A$ act on $\Hil' \subseteq \Hil$ trivially by assumption. 
	
	The estimates for the other, including higher-order commutators follow analogously, which gives us Lemma~3.1 and its magnetic sibling. 
	
	The arguments in the proofs of Lemma~3.2 and 3.3 as well as their magnetic versions can be transcribed in a similar fashion. To account for the presence of the small parameters, we \eg have to replace the magnetic Wigner transform for $\eps = 1 = \lambda$ in \cite[equation~(4.4)]{Cornean_Helffer_Purice:simple_proof_Beals_criterion_magnetic_PsiDOs:2018} by the appropriately scaled magnetic Wigner transform~\eqref{operator_valued_calculus:eqn:Wigner_transform}. 
\end{proof}
%

\subsubsection{Results on inversion and functional calculus} 
\label{operator_valued_calculus:Hoermander_symbols:inversion_and_functional_calculus}
Probably the most important application of commutator criteria is to prove that resolvent operators 
\begin{align*}
	\bigl ( \Op^A(h) - z \bigr )^{-1} = \Op^A \bigl ( (h - z)^{(-1)_{\Weyl}} \bigr ) 
\end{align*}
are themselves pseudodifferential operators. Once this has been established, we can pull back a functional calculus to the level of Hörmander symbols: for suitable functions such as $\varphi \in \Cont^{\infty}_{\mathrm{c}}(\R,\C)$ we can choose some quasi-analytic extension $\widetilde{\varphi}$ (see \eg \cite[p.~169, equations~(2)–(3)]{Davies:functional_calculus:1995} or \cite[Chapter~8]{Dimassi_Sjoestrand:spectral_asymptotics:1999}) and use the Helffer-Sjöstrand formula 
\begin{align}
	\varphi^B(h) := \frac{1}{\pi} \int_{\C} \dd z \, \partial_{\bar{z}} \widetilde{\varphi}(z) \, (h - z)^{(-1)_{\Weyl}} 
	\label{operator_valued_calculus:eqn:Helffer_Sjoestrand_functional_calculus}
\end{align}
to assign another symbol to each $\varphi$ and $h$. Alternatively, we may instead consider a holomorphic functional calculus 
\begin{align}
	\varphi^B(h) := \frac{\ii}{2 \pi} \int_{\Gamma} \dd z \, \varphi(z) \, (h - z)^{(-1)_{\Weyl}} 
	\label{operator_valued_calculus:eqn:holomorphic_functional_calculus}
\end{align}
for functions $\varphi$ that are holomorphic in some neighborhood of the spectrum $\sigma \bigl ( \Op^A(h) \bigr )$ and a contour $\Gamma$ that encloses $\supp \varphi \cap \sigma \bigl ( \Op^A(h) \bigr )$ in the complex plane. 
\begin{theorem}[Invertibility]\label{operator_valued_calculus:thm:invertiblity}
	Suppose the magnetic field is bounded in the sense of Assumption~\ref{operator_valued_calculus:assumption:bounded_magnetic_fields}, $\rho \in [0,1]$ and that the Hilbert spaces $\Hil$ and $\Hil'$ are isomorphic. 
	Suppose $f \in S^m_{\rho,0} \bigl ( \mathcal{B}(\Hil,\Hil') \bigr )$ is a Hörmander symbol of non-negative order $m \geq 0$ and $\rho \in [0,1]$, which is invertible as an element of either $\mathfrak{A}^B \bigl ( \mathcal{B}(\Hil,\Hil') \bigr )$ ($m = 0$) or $\mathcal{M}^B \bigl ( \mathcal{B}(\Hil,\Hil') \bigr )$ ($m \geq 0$). 
	Then its inverse 
	\begin{align*}
		f^{(-1)_{\Weyl}} \in S^{-m}_{\rho,0} \bigl ( \mathcal{B}(\Hil',\Hil) \bigr ) 
	\end{align*}
	with respect to the magnetic Weyl product $\Weyl$ is a Hörmander symbol of order $-m \leq 0$. 
\end{theorem}
Iftimie et al.\ resort to \emph{Bony's} Commutator Criterion in \cite{Iftimie_Mantoiu_Purice:commutator_criteria:2008} instead of Beals' Criterion, presumably because it allows for a more elegant proof: rather than repeat the three-step strategy, they were able to skip the bootstrap argument from $\rho = 0$ to $\rho \neq 0$. 

While we could establish that also Bony's Commutator Criterion extends to operator-valued symbols, we shall not do so here and simply repeat the above three-step strategy from the proof of Corollary~\ref{operator_valued_calculus:cor:Beals_commutator_criterion}: we (1) establish the result for $m = 0$, $\rho = 0$, (2) extend it to $m \neq 0$, $\rho = 0$, and (3) then prove the statement for $m \in \R$, $\rho \in (0,1]$. 
\begin{proof}
	Suppose $m = 0$ and $\rho = 0$. Then Beals' Commutator Criterion~\ref{operator_valued_calculus:thm:Beals_commutator_criterion} tells us to consider $\partial_{(x,\xi)}^{(a,\alpha)} f^{(-1)_{\Weyl}}$ and prove that it belongs to $\mathfrak{A}^B \bigl ( \mathcal{B}(\Hil',\Hil) \bigr )$. 
	
	The idea is to apply Leibniz's rule to 
	\begin{align*}
		f^{(-1)_{\Weyl}} \Weyl f = \id_{\Hil} \in \mathfrak{A}^B \bigl ( \mathcal{B}(\Hil) \bigr ) 
	\end{align*}
	and deduce 
	\begin{align}
		\ad_{x_j} \bigl ( f^{(-1)_{\Weyl}} \bigr ) &= - f^{(-1)_{\Weyl}} \Weyl \ad_{x_j}(f) \Weyl f^{(-1)_{\Weyl}} 
		\label{operator_valued_calculus:eqn:proof_invertibility_ad_Moyal_inverse_equals_Moyal_inverse_ad_f_Moyal_inverse}
	\end{align}
	as well as an analogous expression for the momentum derivation $\ad_{\xi_j} \bigl ( f^{(-1)_{\Weyl}} \bigr )$. As a product of elments of $\mathfrak{A}^B \bigl ( \mathcal{B}(\Hil',\Hil) \bigr )$ and $\mathfrak{A}^B \bigl ( \mathcal{B}(\Hil,\Hil') \bigr )$, we deduce that all first-order derivations are again elements of $\mathfrak{A}^B \bigl ( \mathcal{B}(\Hil',\Hil) \bigr )$. 
	
	This argument extends by induction to higher-order derivations. Consequently, Beals' Commutator Criterion ensures the inverse $f^{(-1)_{\Weyl}} \in S^0_{0,0} \bigl ( \mathcal{B}(\Hil',\Hil) \bigr )$ is a Hörmander symbol, and we have established the claim for $m = 0$ and $\rho = 0$. 
	\medskip
	
	\noindent
	Next we consider the case $m > 0$ and $\rho = 0$. In order to apply Corollary~\ref{operator_valued_calculus:cor:Beals_commutator_criterion} to Beals' Commutator Criterion~\ref{operator_valued_calculus:thm:Beals_commutator_criterion} we need to verify whether $w_m \Weyl \partial_{(x,\xi)}^{(a,\alpha)} f^{(-1)_{\Weyl}}$ belongs to $\mathfrak{A}^B \bigl ( \mathcal{B}(\Hil',\Hil) \bigr )$ or not. 
	
	When $a = 0 = \alpha$, we exploit that $f \in S^m_{0,0} \bigl ( \mathcal{B}(\Hil,\Hil') \bigr )$ is an invertible element of $\mathcal{M}^B \bigl ( \mathcal{B}(\Hil,\Hil) \bigr )$. Fortunately, we have already established that also $w_m = w_{-m}^{(-1)_{\Weyl}}$ is invertible in $\mathcal{M}^B(\C)$. Moreover, the relevant Hörmander classes lie in the corresponding Moyal spaces. Hence, the product 
	\begin{align*}
		f \Weyl w_{-m} \in S^m_{0,0} \bigl ( \mathcal{B}(\Hil,\Hil') \bigr ) \Weyl S^{-m}_{0,0}(\C) \subseteq S^0_{0,0} \bigl ( \mathcal{B}(\Hil,\Hil') \bigr ) 
	\end{align*}
	belongs to the Hörmander class of order $0$ and type $(0,0)$. Moreover, it defines an invertible element, because of the following computation in $\mathcal{M}^B \bigl ( \mathcal{B}(\Hil',\Hil) \bigr )$: 
	\begin{align*}
		\bigl ( f \Weyl w_{-m} \bigr )^{(-1)_{\Weyl}} &= w_m \Weyl f^{(-1)_{\Weyl}} \in \mathcal{M}^B \bigl ( \mathcal{B}(\Hil',\Hil) \bigr ) 
	\end{align*}
	In fact, the inverse of the product $f \Weyl w_{-m}$ exists in $\mathfrak{A}^B \bigl ( \mathcal{B}(\Hil',\Hil) \bigr )$. The first step of the proof ($m = 0$, $\rho = 0$) now applies to $f \Weyl w_{-m}$ and we deduce $w_m \Weyl f^{(-1)_{\Weyl}} \in S^0_{0,0} \bigl ( \mathcal{B}(\Hil',\Hil) \bigr )$. Theorem~\ref{operator_valued_calculus:thm:Weyl_product_Hoermander_symbols} on the composition of Hörmander symbols then directly yields that the inverse 
	\begin{align*}
		f^{(-1)_{\Weyl}} &= w_{-m} \Weyl \bigl ( w_m \Weyl f^{(-1)_{\Weyl}} \bigr ) \in S^{-m}_{0,0}(\C) \Weyl S^0_{0,0} \bigl ( \mathcal{B}(\Hil',\Hil) \bigr ) 
		\subseteq S^{-m}_{0,0} \bigl ( \mathcal{B}(\Hil',\Hil) \bigr ) 
	\end{align*}
	belongs to the Hörmander class of order $-m$. 
	\medskip
	
	\noindent
	That leaves us with the case $m \geq 0$ and $\rho \in (0,1]$. As before, we will bootstrap the results for $m \geq 0$ and $\rho = 0$. Since we may choose $w_0 = 1$, we will immediately write the proof for the more complicated case where $m > 0$ might be positive. 
	
	Our aim is to show 
	\begin{align*}
		\partial_{\xi}^{\beta} f^{(-1)_{\Weyl}} \in S^{-m - \sabs{\beta} \rho}_{0,0} \bigl ( \mathcal{B}(\Hil',\Hil) \bigr ) 
		&&
		\forall \beta \in \N_0^d
	\end{align*}
	which is another way of saying that $f^{(-1)_{\Weyl}} \in S^{-m}_{\rho,0} \bigl ( \mathcal{B}(\Hil',\Hil) \bigr )$.
	
	First of all, as the Hörmander classes nest into each other, 
	\begin{align*}
		S^m_{\rho,0} \big ( \mathcal{B}(\Hil,\Hil') \bigr ) \subseteq S^m_{\rho',0} \big ( \mathcal{B}(\Hil,\Hil') \bigr ) 
		&&
		\forall \rho \leq \rho' 
		, 
	\end{align*}
	we may apply either the result from the first ($m = 0$, $\rho = 0$) or the second step ($m > 0$, $\rho = 0$) and deduce $f \in S^{-m}_{0,0} \bigl ( \mathcal{B}(\Hil',\Hil) \bigr )$. 
	
	That means that momentum derivatives 
	\begin{align*}
		\partial_{\xi}^{\beta} f^{(-1)_{\Weyl}} = \partial_{(x,\xi)}^{(0,\beta)} f^{(-1)_{\Weyl}} \in S^{-m}_{0,0} \bigl ( \mathcal{B}(\Hil',\Hil) \bigr ) 
	\end{align*}
	are all well-defined and give elements of some Hörmander class whose order is at most $-m$. Moreover, they coincide with position derivations by Lemma~\ref{operator_valued_calculus:lem:characterization_Hoermander_symbols_via_derivations}. 
	
	Equation~\eqref{operator_valued_calculus:eqn:proof_invertibility_ad_Moyal_inverse_equals_Moyal_inverse_ad_f_Moyal_inverse} and the analogous expression for derivations with respect to $\xi_j$ allow us to iteratively convert derivations of $f^{(-1)_{\Weyl}}$ into derivations of $f$ times $f^{(-1)_{\Weyl}}$. Thus, we can express 
	\begin{align*}
		&w_{m + \sabs{\beta} \rho} \Weyl \partial_{(x,\xi)}^{(a,\alpha)} \partial_{\xi}^{\beta} f^{(-1)_{\Weyl}} = w_{-m - \sabs{\beta}\rho} \Weyl \partial_{(x,\xi)}^{(a,\alpha + \beta)} f^{(-1)_{\Weyl}} 
		\\
		&\qquad
		= \sum_{\substack{\sum_{j = 1}^q a_j = a \\ \sum_{j = 1}^q \alpha_j = \alpha + \beta}} C_{\{ a_j \} , \{ \alpha_j \}} 
		\cdot \\
		&\qquad \qquad \qquad \cdot 
		w_{- m - \sabs{\beta} \rho} \Weyl f^{(-1)_{\Weyl}} \Weyl \partial_{(x,\xi)}^{(a_1,\alpha_1)} f \Weyl f^{(-1)_{\Weyl}} \Weyl \cdots \Weyl f^{(-1)_{\Weyl}} \Weyl \partial_{(x,\xi)}^{(a_q,\alpha_q)} f \Weyl f^{(-1)_{\Weyl}}
		\\
		&\qquad 
		\in S^{m + \sabs{\beta} \rho}_{0,0}(\C) \Weyl S^{- m - (\sabs{\alpha} + \sabs{\beta}) \rho}_{0,0} \bigl ( \mathcal{B}(\Hil',\Hil) \bigr ) 
		\subseteq S^{- \sabs{\alpha} \rho}_{0,0} \bigl ( \mathcal{B}(\Hil',\Hil) \bigr ) 
	\end{align*}
	on the level of Hörmander symbols as a sum that involves $f^{(-1)_{\Weyl}}$ and derivations of only $f$ (but importantly \emph{not} $f^{(-1)_{\Weyl}}$). Here, $q := \sabs{a} + \sabs{\alpha} + \sabs{\beta}$ is the total number of derivations and the constants $C_{\{ a_j \} , \{ \alpha_j \}} \in \{ 0 , \pm 1 \}$ are either just a sign or $0$ in case the combination of multi indices $(a_1 , \ldots , a_q , \alpha_1 , \ldots \alpha_q)$ does not contribute. When the summand involves higher-order derivations of $f$, some of the $(a_j,\alpha_j) = 0$ are $0$; those terms $f^{(-1)_{\Weyl}} \Weyl \partial_{(x,\xi)}^{(a_j,\alpha_j)} f = \id_{\Hil}$ in the product reduce to the identity. 
	%
	
	In the end, each term in the sum belongs to symbol class
	\begin{align*}
		f^{(-1)_{\Weyl}} \Weyl \partial_{(x,\xi)}^{(a_j,\alpha_j)} f &\in S^{-m}_{0,0} \bigl ( \mathcal{B}(\Hil',\Hil) \bigr ) \Weyl S^{m - \sabs{\alpha_j} \rho}_{0,0} \bigl ( \mathcal{B}(\Hil,\Hil') \bigr ) 
		\\
		&\subseteq S^{- \sabs{\alpha_j} \rho}_{0,0} \bigl ( \mathcal{B}(\Hil) \bigr ) 
		, 
	\end{align*}
	by assumption on $f \in S^m_{\rho,0} \bigl ( \mathcal{B}(\Hil,\Hil') \bigr )$. Being careful to not forget the trailing $f^{(-1)_{\Weyl}} \in S^{-m}_{0,0} \bigl ( \mathcal{B}(\Hil',\Hil) \bigr )$ factor, we deduce that 
	\begin{align*}
		w_{m + \sabs{\beta} \rho} \Weyl \partial_{(x,\xi)}^{(a,\alpha)} \partial_{\xi}^{\beta} f^{(-1)_{\Weyl}} \in S^{- \sabs{\alpha} \rho}_{0,0} \bigl ( \mathcal{B}(\Hil',\Hil) \bigr ) \subset \mathfrak{A}^B \bigl ( \mathcal{B}(\Hil',\Hil) \bigr ) 
	\end{align*}
	indeed quantizes to a bounded operator on $L^2$ by virtue of Corollary~\ref{operator_valued_calculus:cor:Calderon_Vaillancourt}. 
	
	Because these arguments hold for all $\beta \in \N_0^d$, we infer from Corollary~\ref{operator_valued_calculus:cor:Beals_commutator_criterion} that 
	$f^{(-1)_{\Weyl}} \in S^{-m}_{\rho,0} \bigl ( \mathcal{B}(\Hil',\Hil) \bigr )$ belongs to the symbol class of order $m \geq 0$ and type $(\rho,0)$. This finishes the proof. 
\end{proof}
\begin{remark}
	The restriction to $m \geq 0$ stems from us requiring that $f^{(-1)_{\Weyl}}$ and its derivations define bounded operators. For example, if $f \in S^{-m}_{\rho,0} \bigl ( \mathcal{B}(\Hil,\Hil') \bigr )$ were of negative order and suppose we know that $f^{(-1)_{\Weyl}} \in S^m_{0,0} \bigl ( \mathcal{B}(\Hil',\Hil) \bigr )$ defines a symbol of positive order, then 
	\begin{align*}
		\ad_{x_j} \bigl ( f^{(-1)_{\Weyl}} \bigr ) &= - f^{(-1)_{\Weyl}} \Weyl \ad_{x_j}(f) \Weyl f^{(-1)_{\Weyl}} 
		\\
		&\in S^m_{0,0} \bigl ( \mathcal{B}(\Hil',\Hil) \bigr ) \Weyl S^{-m - \rho}_{0,0} \bigl ( \mathcal{B}(\Hil,\Hil') \bigr ) \Weyl S^m_{0,0} \bigl ( \mathcal{B}(\Hil',\Hil) \bigr )
		\\
		&\subseteq S^m_{0,0} \bigl ( \mathcal{B}(\Hil',\Hil) \bigr )
	\end{align*}
	is of positive order and in general will not belong to $\mathfrak{A}^B \bigl ( \mathcal{B}(\Hil',\Hil) \bigr )$. Consequently, we may not invoke Beals' Commutator Criterion~\ref{operator_valued_calculus:thm:Beals_commutator_criterion} or its Corollary~\ref{operator_valued_calculus:cor:Beals_commutator_criterion}
\end{remark}
Existence of inverses on the level of Hörmander symbols is commonly applied to resolvents of pseudodifferential operators, \ie the operator that satisfies 
\begin{align*}
	\bigl ( \Op^A(h) - z \bigr )^{-1} \, \bigl ( \Op^A(h) - z \bigr ) &= \id_{L^2(\R^d,\Hil)} 
	, 
	\\
	\bigl ( \Op^A(h) - z \bigr ) \, \bigl ( \Op^A(h) - z \bigr )^{-1} &= \id_{L^2(\R^d,\Hil')} 
	. 
\end{align*}
Note that the identities in these two equations are not the same, but are defined over different $L^2$-spaces. Connected to that, depending on whether we are reading the upper or lower line, we need to make sense of $z = z \, \id$ as an operator on $L^2(\R^d,\Hil)$ and $L^2(\R^d,\Hil')$, respectively. Consequently, we need to make one additional assumption: $\Hil \hookrightarrow \Hil'$ has to inject densely. Hence, we may view $L^2(\R^d,\Hil) \subseteq L^2(\R^d,\Hil')$ as a dense subspace of the larger $L^2$-space. Then $\id_{L^2(\R^d,\Hil')}$ is the unique extension of $\id_{L^2(\R^d,\Hil)}$; likewise, we may view the latter as a restriction of the former. 

Consequently, Theorem~\ref{operator_valued_calculus:thm:invertiblity} immediately implies 
\begin{theorem}[Existence of Moyal resolvent]\label{operator_valued_calculus:thm:existence_Moyal_resolvent}
	Suppose the magnetic field is bounded in the sense of Assumption~\ref{operator_valued_calculus:assumption:bounded_magnetic_fields}, $\rho \in [0,1]$ and that $\Hil \hookrightarrow \Hil'$ can be continuously and densely injected into $\Hil'$. 
	Assume we are given a symbol $f \in S^m_{\rho,0} \bigl ( \mathcal{B}(\Hil,\Hil') \bigr )$ of non-negative order $m \geq 0$ and $f - z$ is an invertible element of $\mathcal{M}^B \bigl ( \mathcal{B}(\Hil,\Hil') \bigr )$ for some $z \not\in \sigma \bigl ( \Op^A(f) \bigr )$. Then the Moyal resolvent 
	\begin{align*}
		(f - z)^{(-1)_{\Weyl}} \in S^{-m}_{\rho,0} \bigl ( \mathcal{B}(\Hil',\Hil) \bigr )
	\end{align*}
	exists as a Hörmander symbol of order $-m \leq 0$. 
\end{theorem}
We may ask for conditions on $f$ when the Moyal resolvent exists. A common one concerns elliptic symbols that quantize to selfadjoint operators: 
\begin{corollary}
	Suppose the magnetic field is bounded in the sense of Assumption~\ref{operator_valued_calculus:assumption:bounded_magnetic_fields}, $\rho \in [0,1]$ and that $\Hil \hookrightarrow \Hil'$ can be continuously and densely injected into $\Hil'$. 
	
	Assume we are given a symbol $h \in S^m_{\rho,0} \bigl ( \mathcal{B}(\Hil,\Hil') \bigr )$ of order $m \geq 0$ that is elliptic in the sense of Definition~\ref{operator_valued_calculus:defn:elliptic_symbols} when $m > 0$ and takes values in the selfadjoint operators. Then for any $z \not\in \sigma \bigl ( \Op^A(h) \bigr )$ the Moyal resolvent 
	\begin{align*}
		(h - z)^{(-1)_{\Weyl}} \in S^{-m}_{\rho,0} \bigl ( \mathcal{B}(\Hil,\Hil') \bigr ) 
	\end{align*}
	exists as a Hörmander symbol.
\end{corollary}
The proof consists of adapting the arguments before \cite[Proposition~6.31]{Iftimie_Mantoiu_Purice:commutator_criteria:2008}. The essential idea is that elliptic symbols can be used to define an equivalent norm on magnetic Sobolev spaces and thus, 
\begin{align*}
	\Op^A(h - z) \, \Op^A(w_{-m}) = \Op^A \bigl ( (h - z) \Weyl w_{-m} \bigr ) : L^2(\R^d,\Hil) \longrightarrow L^2(\R^d,\Hil') 
\end{align*}
defines a continuous bijection with continuous inverse. Importantly, this is equivalent to saying that the domain of selfadjointness is $\mathcal{D} \bigl ( \Op^A(h) \bigr ) = H^m_A(\R^d,\Hil) = \ran \Op^A(w_{-m})$. 
\medskip

\noindent
We take this opportunity to compare the Moyal inverse constructed here to the parametrices from Section~\ref{operator_valued_calculus:Hoermander_symbols:Weyl_product}. At first glance, it seems that parametrices are — up to $\order(\epsilon^{\infty})$ — identical to Moyal resolvents. But this is not quite right. Instead, both approaches represent trade-offs: on the one hand, the (exact) Moyal resolvent uniquely locates the spectrum of $\Op^A(h)$ in the complex plane. The advantage of using a parametrix is that in some circumstances we \emph{are} able to construct an approximate resolvent \emph{even when} $z \in \sigma \bigl ( \Op^A(h) \bigr )$ lies in the spectrum of the operator. That is not contradictory: the parametrix exists only as an asymptotic expansion, \ie we have no control over whether and in what sense the asymptotic series converges. In fact, when $z \in \sigma \bigl ( \Op^A(h) \bigr )$ lies in the spectrum the formal series for the Moyal resolvent \emph{cannot} converge to a symbol $(h - z)^{(-1)_{\epsilon}}$ for which 
\begin{align*}
	(h - z)^{(-1)_{\epsilon}} \Weyl (h - z) &= \id_{\Hil} 
	\quad \mbox{and} \quad 
	(h - z) \Weyl (h - z)^{(-1)_{\epsilon}} = \id_{\Hil'} 
\end{align*}
hold \emph{without} a remainder. These trade-offs will become even more apparent when we insert the resolvents into equations~\eqref{operator_valued_calculus:eqn:Helffer_Sjoestrand_functional_calculus} or \eqref{operator_valued_calculus:eqn:holomorphic_functional_calculus} to obtain a functional calculus. 
\medskip

\noindent
Speaking of functional calculi, let us briefly review the main ideas. The idea is to pull back the functional calculus from the level of operators 
\begin{align*}
	\varphi \bigl ( \Op^A(f) \bigr ) &= \Op^A \bigl ( \varphi^B(f) \bigr ) 
\end{align*}
to the level of Hörmander symbols. Depending on the class of functions we take $\varphi$ from, we can relax or have to tighten the assumptions on the symbols. One general result ensures the existence of a \emph{holomorphic} functional calculus on $S^0_{\rho,0} \bigl ( \mathcal{B}(\Hil) \bigr )$. 
\begin{theorem}[Holomorphic functional calculus on $S^0_{\rho,0} \bigl ( \mathcal{B}(\Hil) \bigr )$]\label{operator_valued_calculus:thm:holomorphic_functional_calculus}
	Suppose the magnetic field is bounded in the sense of Assumption~\ref{operator_valued_calculus:assumption:bounded_magnetic_fields}, $\rho \in [0,1]$. 
	
	Then there exists a holomorphic functional calculus on $S^0_{\rho,0} \bigl ( \mathcal{B}(\Hil) \bigr )$, \ie for any function $\varphi$ that is holomorphic on some neighborhood of the spectrum $\sigma \bigl ( \Op^A(f) \bigr )$, equation~\eqref{operator_valued_calculus:eqn:holomorphic_functional_calculus} defines a map 
	\begin{align*}
		\varphi \mapsto \varphi^B(f) \in S^0_{\rho,0} \bigl ( \mathcal{B}(\Hil) \bigr ) 
	\end{align*}
	into the Hörmander symbols of order $0$. 
\end{theorem}
Importantly, $f$ need \emph{not} be such that it gives rise to a selfadjoint operator. While one could prove this by hand, we will instead rely on abstract results, following the arguments in \cite[Section~3]{Lein_Mantoiu_Richard:anisotropic_mag_pseudo:2009}. The relevant notion is that of a $\Psi^*$-algebra. Suppose we are given a unital $C^*$-algebra $\mathfrak{A}$ and a $\ast$-subalgebra $\mathfrak{B} \subseteq \mathfrak{A}$ with unit. We call $\mathfrak{B}$ spectrally invariant if and only if $\mathfrak{B} \cap \mathfrak{A}^{-1} = \mathfrak{B}^{-1}$, that is, if the inverse of any invertible element of $\mathfrak{B}$ also lies in the smaller $\ast$-subalgebra $\mathfrak{B}$. Furthermore, we call $\mathfrak{B}$ a $\Psi^*$-algebra if it is spectrally invariant and is endowed with a Fréchet topology so that the inclusion $\mathfrak{B} \hookrightarrow \mathfrak{A}$ is continuous. 

Let us check all ingredients in turn. We have assumed $\Hil = \Hil'$ in the above so that $\mathfrak{A}^B \bigl ( \mathcal{B}(\Hil) \bigr )$ inherits the $C^*$-algebraic nature of $\mathcal{B} \bigl ( L^2(\R^d,\Hil) \bigr )$. 

The Calderón-Vaillancourt Theorem~\ref{operator_valued_calculus:thm:Calderon_Vaillancourt} states we can continuously embed $S^0_{\rho,0} \bigl ( \mathcal{B}(\Hil) \bigr )$ into $\mathfrak{A}^B \bigl ( \mathcal{B}(\Hil) \bigr )$ (\cf also Corollary~\ref{operator_valued_calculus:cor:Calderon_Vaillancourt}). Furthermore, the product of two Hörmander symbols of order $0$ gives another Hörmander symbol of order $0$ (Theorem~\ref{operator_valued_calculus:thm:Weyl_product_Hoermander_symbols}). The adjoint $f \mapsto f^*$ also defines a continuous maps on $S^0_{\rho,0} \bigl ( \mathcal{B}(\Hil) \bigr )$. Consequently, $S^0_{\rho,0} \bigl ( \mathcal{B}(\Hil) \bigr ) \subset \mathfrak{A}^B \bigl ( \mathcal{B}(\Hil) \bigr )$ forms a Fréchet $\ast$-subalgebra. 

We have also proven spectral invariance in Theorem~\ref{operator_valued_calculus:thm:invertiblity}, itself a consequence of Beals' Commutator Criterion~\ref{operator_valued_calculus:thm:Beals_commutator_criterion}. Overall, this proves the following: 
\begin{corollary}
	Suppose the magnetic field is bounded in the sense of Assumption~\ref{operator_valued_calculus:assumption:bounded_magnetic_fields}, $\rho \in [0,1]$. 
	Then $S^0_{\rho,0} \bigl ( \mathcal{B}(\Hil) \bigr ) \subset \mathfrak{A}^B \bigl ( \mathcal{B}(\Hil) \bigr )$ defines a $\Psi^*$-algebra. As such, it is stable under holomorphic functional calculus, the set of invertible elements 
	\begin{align*}
		\Bigl ( S^0_{\rho,0} \bigl ( \mathcal{B}(\Hil) \bigr ) \Bigr )^{-1} \subset \mathfrak{A}^B \bigl ( \mathcal{B}(\Hil) \bigr ) 
	\end{align*}
	is open and the map 
	\begin{align*}
		f \mapsto f^{(-1)_{\Weyl}} \in S^0_{\rho,0} \bigl ( \mathcal{B}(\Hil) \bigr ) 
	\end{align*}
	is continuous. 
\end{corollary}
The $\Psi^*$-algebraic nature of $S^0_{\rho,0} \bigl ( \mathcal{B}(\Hil) \bigr )$ is extremely useful in many applications. For example, any closed $\ast$-subalgebra of a $\Psi^*$-algebra containing the unit element is automatically a $\Psi^*$-algebra. Put another way, we are no longer obliged to prove an analog of Theorem~\ref{operator_valued_calculus:thm:invertiblity} for that smaller Fréchet $\ast$-subalgebra. That has been exploited in \cite{Lein_Mantoiu_Richard:anisotropic_mag_pseudo:2009,Belmonte_Lein_Mantoiu:mag_twisted_actions:2010}, where the notion of anisotropic Hörmander classes was introduced; these are Hörmander classes where the behavior in $x$ is characterized by a suitable subalgebra of $\Cont^{\infty}_{\mathrm{b}}(\R^d)$ (\eg symbols are $\Z^d$-periodic in $x$). 

Another approach is to impose more assumptions on the symbols and relax assumptions on the functions $\varphi$. The Helffer-Sjöstrand formula allows us to make sense of $\varphi^B(h)$ when $h$ defines a selfadjoint magnetic pseudodifferential operator. 
\begin{theorem}[Existence of a functional calculus]\label{operator_valued_calculus:thm:functional_calculus_Helffer_Sjoestrand}
	Suppose the magnetic field is bounded in the sense of Assumption~\ref{operator_valued_calculus:assumption:bounded_magnetic_fields}, $\rho \in [0,1]$ and that $\Hil \hookrightarrow \Hil'$ can be continuously and densely injected into $\Hil'$. 
	Assume we are given a symbol $h \in S^m_{\rho,0} \bigl ( \mathcal{B}(\Hil,\Hil') \bigr )$ of order $m \geq 0$ that is elliptic in the sense of Definition~\ref{operator_valued_calculus:defn:elliptic_symbols} when $m > 0$ and takes values in the selfadjoint operators. 
	Then this defines a functional calculus for $\varphi \in \Cont^{\infty}_{\mathrm{c}}(\R,\C)$ via the Hellfer-Sjöstrand formula~\eqref{operator_valued_calculus:eqn:Helffer_Sjoestrand_functional_calculus}, and for each smooth function with compact support 
	\begin{align*}
		\varphi^B(h) \in S^{-m}_{\rho,0} \bigl ( \mathcal{B}(\Hil',\Hil) \bigr ) 
	\end{align*}
	defines a Hörmander symbol of order $-m$. 
\end{theorem}
\begin{remark}
	The above results extends in a straightforward manner to many functions $\varphi \not\in \Cont^{\infty}_{\mathrm{c}}(\R,\C)$ that do not have compact support. 
\end{remark}
\begin{proof}[Sketch]
	The proof consists of combining the ideas from the proof of \cite[Theorem~6.34]{Iftimie_Mantoiu_Purice:commutator_criteria:2008} with the strategy employed in our proof of Theorem~\ref{operator_valued_calculus:thm:invertiblity}. (Iftimie et al.\ use Bony's Commutator Criterion, which allows for a more straightforward proof. Instead, we will have to repeat the three-step argument from before.) The essential idea is to pull through the derivations into the integral, 
	\begin{align*}
		\partial_{(x,\xi)}^{(a,\alpha)} \bigl ( \varphi^B(h) \bigr ) &= \frac{1}{\pi} \int_{\C} \dd z \, \partial_{\bar{z}} \widetilde{\varphi}(z) \, \partial_{(x,\xi)}^{(a,\alpha)} \bigl ( (h - z)^{(-1)_{\Weyl}} \bigr )
		, 
	\end{align*}
	and then estimate the corresponding expressions for the Moyal resolvents. 
\end{proof}
\begin{figure}[t]
	  \centering
	  \tikzset{every picture/.style={line width=0.75pt}} 
	  \begin{tikzpicture}[x=0.75pt,y=0.75pt,yscale=-1,xscale=1]

		\draw [color={rgb, 255:red, 190; green, 190; blue, 190 }  ,draw opacity=1 ]   (279.03,136.34) -- (279.28,170.84) ;
		\draw [color={rgb, 255:red, 190; green, 190; blue, 190 }  ,draw opacity=1 ]   (233.78,121.09) -- (234.03,155.59) ;
		\draw  [draw opacity=0][fill={rgb, 255:red, 245; green, 166; blue, 35 }  ,fill opacity=0.1 ] (258.03,136.08) .. controls (258.03,136.08) and (258.03,136.08) .. (258.03,136.08) -- (320.03,136.08) .. controls (320.03,136.08) and (320.03,136.08) .. (320.03,136.08) -- (320.03,171.33) .. controls (320.03,171.33) and (320.03,171.33) .. (320.03,171.33) -- (258.03,171.33) .. controls (258.03,171.33) and (258.03,171.33) .. (258.03,171.33) -- cycle ;
		\draw  [draw opacity=0][fill={rgb, 255:red, 74; green, 144; blue, 226 }  ,fill opacity=0.5 ][line width=0.75]  (19.01,104.35) .. controls (46.61,115.95) and (178.61,72.75) .. (209.41,73.15) .. controls (236.61,73.55) and (352.61,109.55) .. (419.01,80.35) .. controls (418.75,76.21) and (419.41,58.09) .. (419.41,49.95) .. controls (410.83,50.35) and (20.65,50.2) .. (19.09,49.8) .. controls (18.7,50.6) and (19.34,91.69) .. (19.01,104.35) -- cycle ;
    \draw  [draw opacity=0][fill={rgb, 255:red, 74; green, 144; blue, 226 }  ,fill opacity=0.5 ][line width=0.75]  (419.72,195.5) .. controls (340.5,173) and (340.5,215) .. (267.5,223) .. controls (200.5,224) and (154.5,173) .. (101.5,157) .. controls (62.5,147) and (53.5,183) .. (19.41,195.55) .. controls (19.67,199.68) and (19.38,242.31) .. (19.4,250.46) .. controls (27.97,250.05) and (418.16,249.65) .. (419.72,250.05) .. controls (420.1,249.25) and (419.41,208.15) .. (419.72,195.5) -- cycle ;

		\draw  [draw opacity=0][fill={rgb, 255:red, 208; green, 2; blue, 27 }  ,fill opacity=0.4 ] (19.88,122.95) .. controls (102.91,98.15) and (280.4,163.09) .. (419.41,137.35) .. controls (419.41,158.15) and (419.8,165.75) .. (419.41,180.95) .. controls (334.44,134.55) and (102.13,127.35) .. (19.88,156.55) .. controls (19.89,144.29) and (20.27,127.35) .. (19.88,122.95) -- cycle ;
		\draw [color={rgb, 255:red, 208; green, 2; blue, 27 }  ,draw opacity=1 ]   (419.41,180.95) .. controls (334.44,134.55) and (102.13,127.35) .. (19.88,156.55) ;
		\draw [color={rgb, 255:red, 208; green, 2; blue, 27 }  ,draw opacity=1 ]   (19.88,122.95) .. controls (102.91,98.15) and (280.4,163.09) .. (419.41,137.35) ;

		\draw  (5.33,150.56) -- (433.73,150.56)(220.38,30.65) -- (220.38,269.85) (426.73,145.56) -- (433.73,150.56) -- (426.73,155.56) (215.38,37.65) -- (220.38,30.65) -- (225.38,37.65)  ;

		\draw  [draw opacity=0][fill={rgb, 255:red, 189; green, 16; blue, 224 }  ,fill opacity=1 ] (232.54,150.65) .. controls (232.54,149.85) and (233.19,149.2) .. (233.99,149.2) .. controls (234.79,149.2) and (235.44,149.85) .. (235.44,150.65) .. controls (235.44,151.45) and (234.79,152.1) .. (233.99,152.1) .. controls (233.19,152.1) and (232.54,151.45) .. (232.54,150.65) -- cycle ;
		\draw  [draw opacity=0][fill={rgb, 255:red, 245; green, 166; blue, 35 }  ,fill opacity=1 ] (277.79,150.65) .. controls (277.79,149.85) and (278.44,149.2) .. (279.24,149.2) .. controls (280.04,149.2) and (280.69,149.85) .. (280.69,150.65) .. controls (280.69,151.45) and (280.04,152.1) .. (279.24,152.1) .. controls (278.44,152.1) and (277.79,151.45) .. (277.79,150.65) -- cycle ;
		\draw  [draw opacity=0][fill={rgb, 255:red, 189; green, 16; blue, 224 }  ,fill opacity=0.1 ] (212.03,121) .. controls (212.03,121) and (212.03,121) .. (212.03,121) -- (285.53,121) .. controls (285.53,121) and (285.53,121) .. (285.53,121) -- (285.53,155.83) .. controls (285.53,155.83) and (285.53,155.83) .. (285.53,155.83) -- (212.03,155.83) .. controls (212.03,155.83) and (212.03,155.83) .. (212.03,155.83) -- cycle ;
		\draw [color={rgb, 255:red, 189; green, 16; blue, 224 }  ,draw opacity=0.4 ][line width=0.75]    (175.03,163.33) .. controls (214.63,133.63) and (192.84,182.01) .. (231.21,153.88) ;
		\draw [shift={(232.4,153)}, rotate = 143.13] [color={rgb, 255:red, 189; green, 16; blue, 224 }  ,draw opacity=0.4 ][line width=0.75]    (10.93,-3.29) .. controls (6.95,-1.4) and (3.31,-0.3) .. (0,0) .. controls (3.31,0.3) and (6.95,1.4) .. (10.93,3.29)   ;
		\draw [color={rgb, 255:red, 245; green, 166; blue, 35 }  ,draw opacity=0.4 ]   (278.78,199.08) .. controls (318.38,169.38) and (238.52,182.48) .. (275.73,153.64) ;
		\draw [shift={(276.9,152.75)}, rotate = 143.13] [color={rgb, 255:red, 245; green, 166; blue, 35 }  ,draw opacity=0.4 ][line width=0.75]    (10.93,-3.29) .. controls (6.95,-1.4) and (3.31,-0.3) .. (0,0) .. controls (3.31,0.3) and (6.95,1.4) .. (10.93,3.29)   ;
		\draw [color={rgb, 255:red, 34; green, 19; blue, 254 }  ,draw opacity=1 ]   (176,105) .. controls (210.8,107.63) and (172.52,122.22) .. (208.1,127.06) ;
		\draw [shift={(209.77,127.27)}, rotate = 186.83] [color={rgb, 255:red, 34; green, 19; blue, 254 }  ,draw opacity=1 ][line width=0.75]    (10.93,-3.29) .. controls (6.95,-1.4) and (3.31,-0.3) .. (0,0) .. controls (3.31,0.3) and (6.95,1.4) .. (10.93,3.29)   ;
		\draw [color={rgb, 255:red, 34; green, 19; blue, 254 }  ,draw opacity=1 ]   (211.7,120.67) -- (211.7,155.49) ;
		\draw [color={rgb, 255:red, 34; green, 19; blue, 254 }  ,draw opacity=1 ]   (285.2,120.67) -- (285.2,155.49) ;
		\draw [color={rgb, 255:red, 189; green, 16; blue, 224 }  ,draw opacity=1 ]   (212.03,121) -- (285.53,121) ;
		\draw [color={rgb, 255:red, 189; green, 16; blue, 224 }  ,draw opacity=1 ]   (212.03,155.83) -- (285.53,155.83) ;
		\draw [color={rgb, 255:red, 245; green, 166; blue, 35 }  ,draw opacity=1 ]   (258.03,171.33) -- (320.03,171.33) ;
		\draw [color={rgb, 255:red, 245; green, 166; blue, 35 }  ,draw opacity=1 ]   (258.03,136.08) -- (320.03,136.08) ;
		\draw [color={rgb, 255:red, 200; green, 190; blue, 20 }  ,draw opacity=1 ]   (258.03,136.08) -- (258.03,171.33) ;
		\draw [color={rgb, 255:red, 200; green, 190; blue, 20 }  ,draw opacity=1 ]   (320.03,136.08) -- (320.03,171.33) ;
		\draw [color={rgb, 255:red, 200; green, 190; blue, 20}  ,draw opacity=1 ]   (367.67,121.33) .. controls (322.79,125.75) and (372.64,150.4) .. (325.78,158.37) ;
		\draw [shift={(324.33,158.61)}, rotate = 351.12] [color={rgb, 255:red, 200; green, 190; blue, 20 }  ,draw opacity=1 ][line width=0.75]    (10.93,-3.29) .. controls (6.95,-1.4) and (3.31,-0.3) .. (0,0) .. controls (3.31,0.3) and (6.95,1.4) .. (10.93,3.29)   ;

		\draw (213.53,14.53) node [anchor=north west][inner sep=0.75pt]    {$E$};
		\draw (440.2,143.07) node [anchor=north west][inner sep=0.75pt]    {$( x,\xi )$};
		\draw (257.73,199.48) node [anchor=north west][inner sep=0.75pt]  [font=\small,color={rgb, 255:red, 245; green, 166; blue, 35 }  ,opacity=1 ]  {$( x_{2} ,\xi _{2})$};
		\draw (155.13,163.71) node [anchor=north west][inner sep=0.75pt]  [font=\small,color={rgb, 255:red, 189; green, 16; blue, 224 }  ,opacity=1 ]  {$( x_{1} ,\xi _{1})$};
		\draw (238.36,159.32) node [anchor=north west][inner sep=0.75pt]  [color={rgb, 255:red, 189; green, 16; blue, 224 }  ,opacity=1 ]  {$U_{1}$};
		\draw (301.98,175.1) node [anchor=north west][inner sep=0.75pt]  [color={rgb, 255:red, 245; green, 166; blue, 35 }  ,opacity=1 ]  {$U_{2}$};
		\draw (125.67,99.07) node [anchor=north west][inner sep=0.75pt]  [font=\small,color={rgb, 255:red, 34; green, 19; blue, 254 }  ,opacity=1 ]  {$\mathrm{supp}( \varphi _{1})$};
		\draw (370.33,112.4) node [anchor=north west][inner sep=0.75pt]  [font=\small,color={rgb, 255:red, 200; green, 190; blue, 20 }  ,opacity=1 ]  {$\mathrm{supp}( \varphi _{2})$};
	\end{tikzpicture}

	\caption{This illustrates the micorolocal construction of the (almost) projection $\pi(x,\xi)$ in a family of neighborhoods. The relevant part of the “pointwise” spectrum $\sigma \bigl ( h(x,\xi) \bigr )$ we wish to isolate is colored in red. It is separated from the remainder of the spectrum (in blue) by a local gap. }
	\label{operator_valued_calculus:figure:microlocal_construction_operator}
\end{figure}
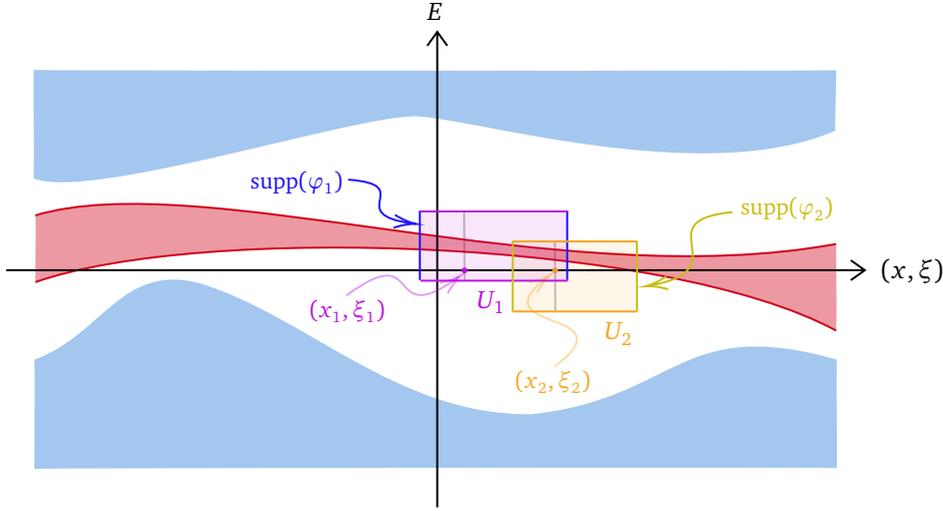
Lastly, let us finish our comparison between Moyal resolvents and parametrices. Indeed, it comes as no surprise that one can construct a functional calculus based on parametrices, and ascribe meaning to 
\begin{align*}
	\pi(x,\xi) := \frac{1}{\pi} \int_{\C} \dd z \, \partial_{\bar{z}} \widetilde{\varphi}(z) \, (h - z)^{(-1)_{\eps}}(x,\xi) + \order(\eps^{\infty})
	, 
\end{align*}
provided the parametrices exist for all relevant values $z$. As the notation suggests, we may construct $\pi$ and thus, the operator $\Op^A(\pi)$ \emph{microlocally}. The price we have to pay is that these operators are only well-defined up to $\order(\eps^{\infty})$. 

A common use case is to construct approximate spectral projections, something which has been implemented many times in the literature (\eg \cite{PST:sapt:2002,PST:effective_dynamics_Bloch:2003,PST:Born-Oppenheimer:2007,DeNittis_Lein:Bloch_electron:2009,DeNittis_Lein:sapt_photonic_crystals:2013,Fuerst_Lein:scaling_limits_Dirac:2008}). For the sake of discussion, we will dispense with mathematical rigor and focus only on the ideas behind the construction. Let us consider the expansion in the adiabatic/semiclassical parameter $\eps$. Here, the terms in the expansion of $\Weyl$ involve pointwise products of $B_{jk}$, the two functions and their derivatives; specifically, the $0$th-order term is just the pointwise product (of operators) (\cf \cite[Theorems~1.1 and 2.12]{Lein:two_parameter_asymptotics:2008}). In this circumstance, the invertibility condition~(c) in Corollary~\ref{operator_valued_calculus:cor:existence_parametrix}, which is necessary to construct the parametrix, simplifies to $z \not\in \sigma \bigl ( h(x,\xi) \bigr )$. 

A typical scenario is sketched in Figure~\ref{operator_valued_calculus:figure:microlocal_construction_operator}: the relevant part of the spectrum $\sigma_{\mathrm{rel}}(x,\xi)$ is the central band (colored in red), the remainder $\sigma \bigl ( h(x,\xi) \bigr ) \setminus \sigma_{\mathrm{rel}}(x,\xi)$ is the union of the upper and lower regions (colored in blue). The two spectral regions are separated by \emph{local} spectral gaps. 

To construct the (approximate) spectral projection $\Op^A(\pi)$ associated to $\sigma_{\mathrm{rel}}$, we need to find a function $\varphi$, which picks out $\sigma_{\mathrm{rel}}$. However, as the gap in Figure~\ref{operator_valued_calculus:figure:microlocal_construction_operator} is only local, but \emph{not global}, no single function $\varphi$ exists that covers only the relevant red part of the spectrum for \emph{all} $(x,\xi) \in T^* \R^d$ — we always get a contribution by the unwanted blue part of the spectrum or miss part of the relevant spectrum for some values of $(x,\xi)$. 

Instead, we use the functional calculus for the parametrix $(h - z)^{(-1)_{\eps}}(x,\xi)$ and define $\pi(x,\xi)$ (micro)locally. We begin by choosing a suitable countable covering $\{ U_j \}_{j \in \mathcal{I}}$ of $T^* \R^d$. Suitable means that on each $U_j$ we may choose a “smoothened characteristic function” $\varphi_j \in \Cont^{\infty}_{\mathrm{c}}(\R,\C)$ that satisfies 
\begin{align*}
	\varphi_j \vert_{\sigma_{\mathrm{rel}}(x,\xi)} &= 1 
	\qquad \forall (x,\xi) \in U_j
	, 
	\\
	\varphi_j \vert_{\sigma (h(x,\xi)) \setminus \sigma_{\mathrm{rel}}(x,\xi)} &= 0 
	\qquad \forall (x,\xi) \in U_j
	. 
\end{align*}
The values $\varphi_j$ takes inside the local gap on $U_j$ are irrelevant. Hence, we may use $\varphi_1$ or $\varphi_2$ from Figure~\ref{operator_valued_calculus:figure:microlocal_construction_operator} to define $\pi(x,\xi)$ on the intersection $U_1 \cap U_2$. Consequently, we are able to patch the local definitions together for they are consistent, and we obtain a globally defined, smooth function. To ensure $\pi$ lies in some Hörmander class, we need additional technical assumptions on $h$ and the spectral gap. 

Technically speaking, $\Op^A(\pi) = \Op^A(\pi)^2 + \order_{\norm{\cdot}}(\eps^{\infty})$ is only an \emph{almost} projection as it does not \emph{exactly} square to itself but only up to $\order(\eps^{\infty})$ in operator norm. Nevertheless, very often in applications we are only interested in finite-order expansions and the difference between true and almost projections is not relevant. What \emph{is} important, though, is that this “approximate spectral projection” exists at all even when the true spectral projection does not. Indeed, the inability of parametrices to reliably detect the spectrum of $\Op^A(h)$ can turn from a liability into an asset. 

Equivalently, we could have used the holomorphic functional calculus with the help of locally constant contours: the spectrum $\sigma \bigl ( h(x,\xi) \bigr )$ decomposes into two disjoint, closed subsets, and on each $U_j$ we may pick a contour $\Gamma_j$ that only encloses $\sigma_{\mathrm{rel}}(x,\xi)$. On the overlap of two neighborhoods $U_j \cap U_k \neq \emptyset$ the two symbols $\pi$ agree up to $\order(\eps^{\infty})$: the complex integral only depends on the spectrum the contour encloses, and $\sigma_{\mathrm{rel}}(x,\xi)$ is the set where the parametrix $(h - z)^{(-1)_{\eps}}(x,\xi)$ has its singularities. 

\subsubsection{Trace-class criterion for operator-valued magnetic $\Psi$DOs} 
\label{operator_valued_calculus:Hoermander_symbols:trace_class_criterion}
The Calderón-Vaillancourt Theorem~\ref{operator_valued_calculus:thm:Calderon_Vaillancourt} gives criteria when a pseudodifferential operator defines a bounded operator between $L^2$-spaces. However, this is not the only relevant Banach space of operators one encounters in applications. The other standard class of operators are the Banach space of trace class operators and its generalization, the $p$-Schatten classes (\cf \cite[Chapter~2]{Simon:trace_ideals_applications:2005}). And one may ask under what conditions a magnetic pseudodifferential operator defines a trace-class or $p$-Schatten class operator. 

This line of inquiry has been addressed in the literature. The first work by  \cite{Rondeaux:schatten_classes_pseudodifferential:1984} gave criteria for non-magnetic pseudodifferential operators for scalar-valued symbols; the analogous results for \emph{magnetic} $\Psi$DOs is due to Athmouni and Purice \cite{Athmouni_Purice:schatten_class_magnetic_Weyl:2018}. Stiepan and Teufel \cite{Stiepan_Teufel:semiclassics_op_valued_symbols:2012} have extended the non-magnetic result to operator-valued symbols. The next results follow from combining \cite{Athmouni_Purice:schatten_class_magnetic_Weyl:2018} with \cite{Stiepan_Teufel:semiclassics_op_valued_symbols:2012}. 

To state the result, we need to introduce the Sobolev spaces of order $(m,\mu)$, $m , \mu \in \N_0$, 
\begin{align*}
	W^{(m,\mu),p}(T^* \R^d , \mathcal{X}) := \Bigl \{ f \in L^p(T^* \R^d , \mathcal{X}) \; \; \big \vert \; \; 
	&(x,\xi) \mapsto \bnorm{\partial_x^a \partial_{\xi}^{\alpha} f(x,\xi)}_{\mathcal{X}} \in L^p(T^* \R^d) 
	\Bigr . \notag \\ 
    &\Bigl . 
    \forall \sabs{a} \leq m , \, \forall \sabs{\alpha} \leq \mu \Bigr \} 
\end{align*}
that take values in a Banach space $\mathcal{X}$. For out intents and purposes, the reader may think of $\mathcal{X} = \mathcal{B}(\Hil,\Hil')$ or $\mathcal{X} = \mathcal{L}^p(\Hil,\Hil')$, $1 \leq p < \infty$ (\cf Definition~\ref{appendix:extension_by_duality_operator_valued_calculus:defn:Schatten_class_operators}). When $\mathcal{X} = \mathcal{B}(\Hil,\Hil')$, we can equivalently characterize $f \in W^{(m,\mu),p} \bigl ( T^* \R^d , \mathcal{B}(\Hil,\Hil') \bigr )$ via the condition 
\begin{align*}
	(x,\xi) \mapsto \bscpro{\psi' \, }{ \, \partial_x^a \partial_{\xi}^{\alpha} f(x,\xi) \, \varphi}_{\Hil'} \in L^p(T^* \R^d) 
	&&
	\forall \varphi \in \Hil 
	, \; 
	\psi' \in \Hil' 
	, \; 
	\sabs{a} \leq m
	, \; 
	\sabs{\alpha} \leq \mu
	. 
\end{align*}
Note that elements of $W^{(m,\mu),p}(T^* \R^d , \mathcal{X})$ are tempered distributions, which is why we denote its elements with capital letters. Moreover, 

Since we may regard any element $f \in W^{(m,\mu),p}(T^* \R^d , \mathcal{X})$ as a tempered distribution when $\mathcal{X} = \mathcal{B}(\Hil,\Hil')$ or $\mathcal{X} = \mathcal{L}^1(\Hil,\Hil')$, we may use the extension of $\Op^A$ from Proposition~\ref{operator_valued_calculus:prop:extension_OpA_tempered_distributions} to tempered distributions in order to make sense of $\Op^A(f)$ as a magnetic pseudodifferential operator. The integrability of $f$ in both variables and all its derivatives up to orders $(m,\mu)$ not only ensure that $\Fourier_{\sigma} f$ exists, but when $m > d$ and $\mu > d$ are chosen large enough (as in the statement below), the (symplectic) Fourier transform is again integrable in both variables. Hence, we may interpret equation~\eqref{operator_valued_calculus:eqn:definition_Op_A} as a Bochner integral and bound the norm of the magnetic pseudodifferential operator $\Op^A(f)$ by 
\begin{align*}
	\bnorm{\Op^A(f)}_{\mathcal{B}(L^2(\R^d,\Hil) , L^2(\R^d,\Hil'))} \leq C \, \snorm{f}_{W^{(m,\mu),1}(T^* \R^d , \mathcal{B}(\Hil,\Hil'))} 
	. 
\end{align*}
With the preamble out of the way, let us state the generalization of \cite[Theorem~1.2]{Athmouni_Purice:schatten_class_magnetic_Weyl:2018}. 
\begin{theorem}\label{operator_valued_calculus:thm:magnetic_PsiDO_trace_class}
	Suppose the magnetic field $B$ is bounded in the sense of Assumption~\ref{operator_valued_calculus:assumption:bounded_magnetic_fields} and that $f \in W^{(m,\mu),1} \bigl ( T^* \R^d , \mathcal{L}^1(\Hil,\Hil') \bigr )$ holds for $m \geq 2 [\nicefrac{d}{2}] + 2$ and $\mu \geq d + [\nicefrac{d}{2}] + 1$. 
	\begin{enumerate}[(1)]
		\item The operator $\Op^A(f) \in \mathcal{L}^1 \bigl ( L^2(\R^d,\Hil) \, , \, L^2(\R^d,\Hil') \bigr )$ is trace class. 
		\item When $\Hil = \Hil'$ we may express the trace of the operator as a phase space integral, 
		\begin{align*}
			\trace_{L^2(\R^d,\Hil)} \bigl ( \Op^A(f) \bigr ) = \frac{1}{(2\pi)^d} \int_{T^* \R^d} \dd x \, \dd \xi \, \trace_{\Hil} \bigl ( f(\eps x,\xi) \bigr ) 
			. 
		\end{align*}
	\end{enumerate}
\end{theorem}
\begin{proof}
	We will prove (2) in the process of proving (1). 
	
	Before we begin with the proof proper, we will justify a few simplifications. First of all, we may assume we are dealing with the more difficult case where $\dim \Hil = \infty = \dim \Hil'$. Otherwise, one of the spaces is finite-dimensional, $\dim \Hil < \infty$ or $\dim \Hil' < \infty$, and the trace over $\Hil$ or $\Hil'$ involves a finite sum. Then also all the relevant sums below will be finite and questions of convergence do not arise. 
	
	To see this, we exploit that the trace norm of $\Op^A(f)$ equals the trace norm of its adjoint: Hilbert spaces are reflexive, and applying \cite[Proposition~47.5]{Treves:topological_vector_spaces:1967} twice yields the equality of both norms, 
	\begin{align*}
		\bnorm{\Op^A(f)^*}_{\mathcal{L}^1(L^2(\R^d,\Hil'),L^2(\R^d,\Hil))} &\leq \bnorm{\Op^A(f)}_{\mathcal{L}^1(L^2(\R^d,\Hil),L^2(\R^d,\Hil'))}
		\\
		&= \bnorm{\Op^A(f)^{\ast \ast}}_{\mathcal{L}^1(L^2(\R^d,\Hil),L^2(\R^d,\Hil'))} 
		\\
		&\leq \bnorm{\Op^A(f)^*}_{\mathcal{L}^1(L^2(\R^d,\Hil'),L^2(\R^d,\Hil))}
		. 
	\end{align*}
	So we may use either $\Op^A(f)$ or $\Op^A(f)^*$ in our estimates. 
	
	By its very Definition~\ref{appendix:extension_by_duality_operator_valued_calculus:defn:trace_class_operators}, the operator $\Op^A(f) \in \mathcal{L}^1 \bigl ( L^2(\R^d,\Hil) \, , L^2(\R^d,\Hil') \bigr )$ is trace class whenever its modulus $\babs{\Op^A(f)} \in \mathcal{L}^1 \bigl ( L^2(\R^d,\Hil) \bigr )$ is. Equivalently, we could have used the absolute value $\babs{\Op^A(f)^*} \in \mathcal{L}^1 \bigl ( L^2(\R^d,\Hil') \bigr )$ of the operator adjoint. When $\Hil$ or $\Hil'$ is finite-dimensional, we pick either $\Op^A(f)$ or $\Op^A(f)^*$ and exploit that the relevant sums are finite. 
	
	Consequently, we will proceed under the assumption $\dim \Hil = \infty = \dim \Hil'$. Even in this case, the above arguments tell us we may assume $\Hil = \Hil'$ and $\Op^A(f) = \babs{\Op^A(f)} \geq 0$ in the remainder of the proof. 
	\medskip
	
	\noindent
	The first step is to make a connection to \cite[Theorem~1.2]{Athmouni_Purice:schatten_class_magnetic_Weyl:2018}, the analogous result for the scalar-valued case. For $\Hil$ we pick some orthonormal basis $\{ \varphi_k \}_{k \in \N}$. Then for each pair $k , j \in \N$ we therefore obtain a scalar-valued function 
	\begin{align*}
		(x,\xi) \mapsto f_{k j}(x,\xi) := \bscpro{\varphi_k}{f(x,\xi) \varphi_j}_{\Hil} \in W^{(m,\mu),1}(T^* \R^d , \C)
	\end{align*}
	that lies in the appropriate Sobolev space.
	
	Each of the $f_{k j}$ satisfies the assumptions of \cite[Theorem~1.2]{Athmouni_Purice:schatten_class_magnetic_Weyl:2018} and hence, $\Op^A(f_{k j})$ is trace-class for all $k , j \in \N$. Consequently, we may compute the trace as the integral of the diagonal of the kernel~\eqref{operator_valued_calculus:eqn:operator_kernel} (\cf \cite[Theorem~3.1]{Brislawn:kernels_trace_class_operators:1988}), 
	\begin{align}
		\trace_{L^2(\R^d)} \bigl ( \Op^A(f_{k j}) \bigr ) &= \int_{\R^d} \dd x \, K^A_{f_{k j}}(x,x)
		\notag \\
		&= \frac{1}{(2\pi)^d} \int_{\R^d} \dd x \int_{\R^d} \dd \xi \, f_{k j}(\eps x,\xi) 
		\notag \\
		&= \frac{1}{(2\pi \eps)^d} \int_{\R^d} \dd x \int_{\R^d} \dd \xi \, f_{k j}(x,\xi) 
		. 
		\label{operator_valued_calculus:eqn:trace_matrix_element_equals_phase_space_integral}
	\end{align}
	For the special case $j = k$ this computation shows that $\Op^A(f_{k k}) \geq 0$ is non-negative exactly when $f_{k k} \geq 0$ is non-negative almost everywhere. As our choice of orthonormal basis was arbitrary, $f_{k k} \geq 0$ for all $k \in \N$ implies $f \geq 0$ on a set of full measure. 
	
	Our assumptions on $f \geq 0$ guarantee that the $L^1 \bigl ( T^* \R^d , \mathcal{L}^1(\Hil) \bigr )$ norm of $\Op^A(f)$ is finite, which implies the finiteness of each of the sums in 
	\begin{align*}
		(2\pi \eps)^{-d} \, \snorm{f}_{L^1(T^* \R^d , \mathcal{L}^1(\Hil))} 
		&= \frac{1}{(2\pi \eps)^d} \int_{T^* \R^d} \dd X \, \bnorm{f(X)}_{\mathcal{L}^1(\Hil)} 
		\\
		&= \frac{1}{(2\pi \eps)^d} \int_{T^* \R^d} \dd X \, \sum_{k \in \N} \bscpro{\varphi_k}{f(X) \varphi_k}_{\Hil} 
		\\
		&= \sum_{k \in \N} \frac{1}{(2\pi \eps)^d} \int_{T^* \R^d} \dd X \, f_{k k}(X) 
		= \sum_{k \in \N} \trace_{L^2(\R^d)} \bigl ( \Op^A(f_{k k}) \bigr ) 
		< \infty 
		.  
	\end{align*}
	Next, let us pick some arbitrary orthonormal basis $\{ \chi_n \}_{n \in \N}$ of $L^2(\R^d)$. Then the (Hilbert space) tensor product decomposition of $L^2(\R^d,\Hil) \cong L^2(\R^d) \otimes \Hil$ means $\bigl \{ \chi_n \otimes \varphi_k \bigr \}_{n , k \in \N}$ is an orthonormal basis of $L^2(\R^d,\Hil)$. Consequently, we can recover the trace on $L^2(\R^d,\Hil)$, 
	\begin{align*}
		(2\pi \eps)^{-d} \, \snorm{f}_{L^1(T^* \R^d , \mathcal{L}^1(\Hil))} &= \sum_{k \in \N} \sum_{n \in \N} \bscpro{\chi_n}{\Op^A(f_{k k}) \chi_n}_{L^2(\R^d)} 
		\\
		&= \sum_{n , k \in \N} \bscpro{\chi_n \otimes \varphi_k \, }{ \, \Op^A(f) \chi_n \otimes \varphi_k}_{L^2(\R^d,\Hil)}
		. 
	\end{align*}
	The value of the trace is independent of our choice of orthonormal basis, and hence, we have established the equality from item~(2) for $f \geq 0$.
	
	These arguments extend from non-negative functions to arbitrary operator-valued functions $f \in W^{(m,\mu),1} \bigl ( T^* \R^d , \mathcal{L}^1(\Hil) \bigr )$ by polarization, we just have to define real and imaginary part with respect to the Hilbert space adjoint of $f$, \eg $f_{\Re} := \tfrac{1}{2} (f + f^*) = f_{\Re,+} - f_{\Re,-}$ and $f_{\Re,+} = \tfrac{1}{2} \bigl ( f_{\Re} + \sabs{f_{\Re}} \bigr )$. 
\end{proof}
This generalizes immediately to $p$-Schatten classes by interpolation: Theorem~\ref{operator_valued_calculus:thm:magnetic_PsiDO_trace_class} covers the case $p = 1$ and the Calderón-Vaillancourt Theorem~\ref{operator_valued_calculus:thm:Calderon_Vaillancourt} the case $p = \infty$. For details, we refer to Remark~1.4 and Corollary~1.5 in \cite{Athmouni_Purice:schatten_class_magnetic_Weyl:2018}. 
\begin{corollary}\label{operator_valued_calculus:cor:magnetic_PsiDO_p_Schatten_class}
	Suppose the magnetic field $B$ is bounded in the sense of Assumption~\ref{operator_valued_calculus:assumption:bounded_magnetic_fields} and that $f \in W^{(m,\mu),p} \bigl ( T^* \R^d , \mathcal{L}^p(\Hil,\Hil') \bigr )$ holds. Then there exist to constants $b_{\min}(d)$ and $\beta_{\min}(d)$ that only depend on the dimension $d$ of $T^* \R^d$ so that whenever $m \geq b_{\min}(d)$ and $\mu \geq \beta_{\min}(d)$ the magnetic pseudodifferential operator $\Op^A(F) \in \mathcal{L}^p \bigl ( L^2(\R^d,\Hil) , L^2(\R^d,\Hil') \bigr )$ lies in the $p$-Schatten class. 
\end{corollary}
Of course, we are primarily interested in the case of magnetic pseudodifferential operators associated to some operator-valued Hörmander symbol. Suppose $f \in S^m_{\rho,\delta} \bigl ( \mathcal{L}^1(\Hil,\Hil') \bigr )$ belongs to some Hörmander class of order $m$ and type $(\rho,\delta)$, which takes values in the trace-class operators; these are defined analogously to Definition~\ref{operator_valued_calculus:defn:Hoermander_symbols} after replacing operator norms with trace norms. The integrability condition in $\xi$ implies a simple condition on $m$: 
\begin{corollary}\label{operator_valued_calculus:cor:magnetic_PsiDO_trace_class_necessary_conditions_order_Hoermander_class}
	Suppose the magnetic field $B$ is bounded in the sense of Assumption~\ref{operator_valued_calculus:assumption:bounded_magnetic_fields} and $f \in S^m_{\rho,\delta} \bigl ( \mathcal{L}^1(\Hil,\Hil') \bigr )$ is a Hörmander symbol taking values in $\mathcal{L}^1(\Hil,\Hil')$. A necessary condition for $\Op^A(f) \in \mathcal{L}^1 \bigl ( L^2(\R^d,\Hil) , L^2(\R^d,\Hil') \bigr )$ being trace class is 
	\begin{align*}
		m + \bigl ( 2 [\nicefrac{d}{2}] + 2 \bigr ) \, \delta < -d 
		. 
	\end{align*}
\end{corollary}
\begin{proof}
	Since the operator norm estimate in Theorem~\ref{operator_valued_calculus:thm:magnetic_PsiDO_trace_class} requires us to control $x$-derivatives of order up to and including $2 [\nicefrac{d}{2}] + 2$, we know that the maximal value of the order of 
	\begin{align*}
		\partial_x^a \partial_{\xi}^{\alpha} f \in S^{m + ( 2 [\nicefrac{d}{2}] + 2 ) \delta}_{\rho,\delta} \bigl ( \mathcal{L}^1(\Hil,\Hil') \bigr ) 
	\end{align*}
	is $m + \bigl ( 2 [\nicefrac{d}{2}] + 2 \bigr ) \, \delta$. To ensure integrability in $\xi$, we therefore need to impose the condition given in the statement of Corollary~\ref{operator_valued_calculus:cor:magnetic_PsiDO_trace_class_necessary_conditions_order_Hoermander_class}. 
\end{proof}
In applications, we will often encounter products of operators where \eg trace class property in $L^2(\R^d)$ and $\Hil$ are “contributed by different operators”. An example would be the situation covered in the following 
\begin{corollary}\label{operator_valued_calculus:cor:product_W_1_Hoermander_class_trace_class}
	Suppose the magnetic field $B$ is bounded in the sense of Assumption~\ref{operator_valued_calculus:assumption:bounded_magnetic_fields} and we are given two functions $f$ and $g$ that satisfy the following assumptions: 
	\begin{enumerate}[(a)]
		\item $f \in W^{(\alpha,\beta),1} \bigl ( T^* \R^d , \mathcal{B}(\Hil',\Hil) \bigr )$
		\item $g \in S^0_{0,0} \bigl ( \mathcal{B}(\Hil,\Hil') \bigr )$
		\item $f \, g : T^* \R^d \longrightarrow \mathcal{L}^1(\Hil)$
		\item $(x,\xi) \mapsto \trace_{\Hil} \bigl ( f(x,\xi) \, g(x,\xi) \bigr ) \in L^1(T^* \R^d)$ 
	\end{enumerate}
	Then the operator $\Op^A(f) \, \Op^A(g) \in \mathcal{L}^1 \bigl ( L^2(\R^d,\Hil) \bigr )$ is trace class and its trace is given by 
	\begin{align*}
		\trace_{L^2(\R^d,\Hil)} \bigl ( \Op^A(f) \, \Op^A(g) \bigr ) = \frac{1}{(2 \pi \eps)^d} \int_{T^* \R^d} \dd x \, \dd \xi \, \trace_{\Hil} \bigl ( f(x,\xi) \, g(x,\xi) \bigr ) 
		. 
	\end{align*}
\end{corollary}
\begin{proof}
	First, we note that $\int_{T^* \R^d} \dd X \, (f \Weyl g)(X) = \int_{T^* \R^d} \dd X \, f(X) \, g(X)$ can be extended from Schwartz functions (\cf Lemma~\ref{operator_valued_calculus:lem:int_Weyl_product_equals_int_product}) to $f$ and $g$ satisfying assumptions~(a) and (b). One way to achieve this is to plug in the integral kernel map~\eqref{operator_valued_calculus:eqn:operator_kernel} and confirm 
	\begin{align}
		\int_{\R^d} \dd x \, \int_{\R^d} \dd y \, K^A_f(x,y) \, K^A_g(y,x) 
		&= \frac{1}{(2 \pi \eps)^d} \int_{T^* \R^d} \dd X \, f(X) \, g(X) 
		\label{operator_valued_calculus:eqn:integral_over_product_of_integral_kernels_equals_product_of_functions}
	\end{align}
	by interpreting the integral on the left as an oscillatory integral. The properties of the integrand on the right can be read off from our assumptions: the pointwise product $f \, g \in W^{(m,\mu),1} \bigl ( T^* \R^d , \mathcal{L}^1(\Hil) \bigr )$ satisfies the assumptions of Theorem~\ref{operator_valued_calculus:thm:magnetic_PsiDO_trace_class}, and is therefore trace class. 
\end{proof}
Our motivation to cover this seemingly trivial example is to be able to introduce the notion of $\mathcal{L}^p_{\mathrm{loc}}$: in view of Theorem~\ref{operator_valued_calculus:thm:magnetic_PsiDO_trace_class} we would like to define it as the set of pseudodifferential operators $\Op^A(f)$ for which the phase space integral 
\begin{align*}
	\frac{1}{(2\pi \eps)^d} \int_{\Lambda} \dd X \, \bnorm{f(X)}_{\mathcal{L}^1(\Hil)} < \infty
\end{align*}
over \emph{any} bounded Borel set $\Lambda \subset T^* \R^d$ is finite. Naïvely, one would like to introduce the operator “$1_{\Lambda}(Q,P)$” associated to the characteristic function. Unfortunately, $Q$ and $P$ do not commute, and this operator is not well-defined. 

Fortunately, we can replace the characteristic function by a smoothened characteristic function $\chi \in \Cont^{\infty}_{\mathrm{c}}(T^* \R^d)$, and its support $\supp \chi$ plays the role of the Borel set $\Lambda \subset T^* \R^d$. Hence, we arrive at the following 
\begin{definition}[$\mathcal{L}^p_{\mathrm{loc},A} \bigl ( L^2(\R^d,\Hil) \, , \, L^2(\R^d,\Hil') \bigr )$]
	We call an operator 
	\begin{align*}
		T \in \mathcal{B} \bigl ( L^2(\R^d,\Hil) \, , L^2(\R^d,\Hil') \bigr ) 
	\end{align*}
	locally $p$-Schatten class if and only if for all $\chi \in \Cont^{\infty}_{\mathrm{c}}(T^* \R^d)$ the product 
	\begin{align*}
		\Op^A(\chi) \, \sabs{T}^p \in \mathcal{L}^p \bigl ( L^2(\R^d,\Hil) \bigr ) 
	\end{align*}
	belongs to the (ordinary) $p$-Schatten class. 
\end{definition}
\begin{corollary}
	Suppose the magnetic field $B$ is bounded in the sense of Assumption~\ref{operator_valued_calculus:assumption:bounded_magnetic_fields}. 
	\begin{enumerate}[(1)]
		\item Suppose $m$ and $\mu$ satisfy the assumptions from Theorem~\ref{operator_valued_calculus:thm:magnetic_PsiDO_trace_class}. Then any magnetic pseudodifferential operator $\Op^A(f) \in \mathcal{L}^1_{\mathrm{loc},A} \bigl ( L^2(\R^d,\Hil) \, , \, L^2(\R^d,\Hil') \bigr )$ defined from $f \in W^{(m,\mu),1}_{\mathrm{loc}} \bigl ( T^* \R^d , \mathcal{L}^1(\Hil,\Hil') \bigr )$ is locally trace class. 
		\item Hörmander symbols of any order $m \in \R$ define locally trace class operators, \ie for any $m \in \R$ and $f \in S^m_{\rho,\delta} \bigl ( \mathcal{L}^1(\Hil,\Hil') \bigr )$ the operator $\Op^A(f) \in \mathcal{L}^1_{\mathrm{loc},A} \bigl ( L^2(\R^d,\Hil) \, , L^2(\R^d,\Hil') \bigr )$. 
	\end{enumerate}
\end{corollary}
\begin{proof}
	The proofs in both cases is identical: we first justify \eqref{operator_valued_calculus:eqn:integral_over_product_of_integral_kernels_equals_product_of_functions} using standard arguments from the theory of oscillatory integrals, and then verify that for any cutoff function $\chi \in \Cont^{\infty}_{\mathrm{c}}(T^* \R^d)$ the product $\chi \, f \in W^{(m,\mu),1} \bigl ( T^* \R^d , \mathcal{L}^1(\Hil) \bigr )$ belongs to the Sobolev space whose order satisfies the assumptions of Theorem~\ref{operator_valued_calculus:thm:magnetic_PsiDO_trace_class}. 
\end{proof}
%

\section{Equivariant magnetic pseudodifferential calculus} 
\label{equivariant_calculus}
When the operator-valued functions involved all satisfy equivariance conditions like \eqref{setting:eqn:equivariance:symbols}, we can expect that the pseudodifferential operator $\Op^A(f)$ restricts to a densely defined operator between equivariant $L^2$ or magnetic Sobolev spaces (\cf equation~\eqref{setting:eqn:equivariant_L2_space}). Indeed, this is exactly what happens and all of the proofs below essentially amount to invoking the corresponding result from Section~\ref{operator_valued_calculus} and checking that equivariance is preserved. 

Throughout this section, \textbf{we shall always tacitly assume the following assumptions are satisfied unless specifically mentioned otherwise:} 
\begin{assumption}[Equivariant $\Psi$DOs]\label{equivariant_calculus:assumption:setting}
	\begin{enumerate}[(a)]
		\item The magnetic field $B$ is bounded in the sense of Definition~\ref{operator_valued_calculus:assumption:bounded_magnetic_fields}. 
		\item All Hilbert spaces $\Hil$, $\Hil'$, $\ldots$ are separable. 
		\item All Hilbert spaces are endowed with group actions 
		\begin{align*}
			\tau &: \Gamma^* \longrightarrow \mathcal{GL}(\Hil) 
			, 
			\\
			\tau' &: \Gamma^* \longrightarrow \mathcal{GL}(\Hil')
			, 
		\end{align*}
		etc.\ of a discrete lattice $\Gamma^* \cong \Z^d$ with tempered growth of \emph{non-negative} orders $q, q' \geq 0$, etc.\ (\cf Definition~\ref{setting:defn:order_tau}).
	\end{enumerate}
\end{assumption}
Conditions~(a) and (b) are taken verbatim from Section~\ref{operator_valued_calculus:Hoermander_symbols}, where we have constructed a pseudodifferential calculus for operator-valued Hörmander symbols; specifically, the boundedness condition on the components of $B$ is necessary to preserve Hörmander symbols. Only condition~(c) is new. For the definition of tempered growth and other essentials like $L^2_{\eq}(\R^d,\Hil)$ and $\Schwartz^*_{\eq}(\R^d,\Hil)$ we point the reader back to Section~\ref{setting:abstract}. An important example of a Hilbert space endowed with a group action that has tempered growth are the magnetic Sobolev spaces $\Hil = H^q_{A_0}(\T^d)$ on the torus (\cf Lemma~\ref{setting:lem:magnetic_Sobolev_norm_operator_tau}).

\subsection{Defining equivariant magnetic pseudodifferential operators} 
\label{equivariant_calculus:construction}
The relevant Hilbert spaces for equivariant operators have first been introduced in Section~\ref{setting:motivating_example:Zak_transform}, which led to the definition~\eqref{setting:eqn:equivariant_L2_space} for the abstract case, 
\begin{align*}
	L^2_{\eq}(\R^d,\Hil) := \Bigl \{ \psi \in L^2_{\mathrm{loc}}(\R^d,\Hil) \; \; \big \vert \; \; &\psi(k - \gamma^*) = \tau(\gamma^*) \, \psi(k) 
	\Bigr . \notag \\
	&\Bigl . 
	\mbox{ for almost all $k \in \R^d$} , \; 
	\forall \gamma^* \in \Gamma^* \Bigr \} 
	. 
\end{align*}
In principle, equivariance — just like periodicity — means it suffices to characterize the behavior of vectors and operators over “one period”, or, more correctly, over the fundamental cell $\BZ$. Concretely, we may identify the $L^2$ space of equivariant $\Hil$-valued vectors with 
\begin{align*}
	L^2(\BZ,\Hil) \cong L^2_{\eq}(\R^d,\Hil) 
	, 
\end{align*}
where $\BZ$ is the Brillouin torus, thought of as the fundamental cell located at the origin when splitting up $\R^d \cong \Gamma^* \times \BZ$ with respect to the dual lattice $\Gamma^*$. Together with the scalar product~\eqref{setting:eqn:equivariant_scalar_product}, this is indeed a Hilbert space that is unitarily equivalent to $L^2_{\eq}(\R^d,\Hil)$. Hence, any operator on $L^2_{\eq}(\R^d,\Hil)$ can be uniquely represented on $L^2(\BZ,\Hil)$ — and vice versa. 

What is more, the embedding of the Brillouin zone $\BZ \subseteq \R^d$ leads to an inclusion 
\begin{align*}
	\imath_0 : L^2(\BZ,\Hil) \longrightarrow L^2(\R^d,\Hil)
	, \quad 
	\psi_0 \mapsto \Psi := \psi_0 \, 1_{\BZ} 
	, 
\end{align*}
into the \emph{ordinary} $L^2$-space by extending $\imath_0(\psi_0)$ with $0$ outside of $\BZ$. Clearly, we can also restrict any $\Psi \in L^2(\R^d,\Hil)$ to $\BZ$, 
\begin{align*}
	\pi_0 : L^2(\R^d,\Hil) \longrightarrow L^2(\BZ,\Hil)
	, \quad 
	\Psi \mapsto \Psi \vert_{\BZ} 
	. 
\end{align*}
Consequently, suitable operators $\widehat{F} : \domain(\widehat{F}) \subseteq L^2(\R^d,\Hil) \longrightarrow L^2(\R^d,\Hil')$ define an operator 
\begin{align*}
	F_0 := \pi_0 \, \widehat{F} \, \imath_0 : \domain(F_0) \subseteq L^2(\BZ,\Hil) \longrightarrow L^2(\BZ,\Hil') 
\end{align*}
on the $L^2$-space over the Brillouin torus and hence, an operator 
\begin{align*}
	F : \domain(F) \subseteq L^2_{\eq}(\R^d,\Hil) \longrightarrow L^2_{\eq}(\R^d,\Hil') 
	. 
\end{align*}
Overall, this gives us a \emph{third} way to view equivariant operators. Not surprisingly, these three views are completely equivalent as we shall state below. 

Before we do, though, we need to introduce the appropriate Sobolev spaces of order $q \geq 0$, 
\begin{align*}
	H^q_{\Fourier}(\R^d,\Hil) := \Fourier H^q(\R^d,\Hil) 
	, 
\end{align*}
endowed with the scalar product 
\begin{align}
	\scpro{\varphi}{\psi}_{H^q_{\Fourier}(\R^d,\Hil)} := \bscpro{\sexpval{\hat{k}}^q \varphi}{\sexpval{\hat{k}}^q \psi}_{L^2(\R^d,\Hil)} 
	. 
	\label{equivariant_calculus:eqn:scalar_product_Sobolev_space}
\end{align}
Here, the weight operator is defined with respect to the multiplication operator $\hat{k} = \Fourier \, (- \ii \nabla_r) \, \Fourier^{-1}$. As the notation suggests, $H^q_{\Fourier}(\R^d,\Hil)$ is just the Fourier transform of the ordinary Sobolev space of order $q$. The reason we need to introduce Sobolev spaces is to cope with the tempered growth of the norms  $\bnorm{\tau^{(\prime)}(\gamma^*)}_{\mathcal{B}(\Hil^{(\prime)})}$. 

We are now in a position to make the equivalence of these three points of views mathematically precise: 
\begin{proposition}\label{equivariant_calculus:prop:characterization_equivariant_operators}
	There is a one-to-one correspondence between 
	\begin{enumerate}[(a)]
		\item a densely defined operator $F : \domain(F) \subseteq L^2_{\eq}(\R^d , \Hil) \longrightarrow L^2_{\eq}(\R^d , \Hil')$, 
		\item an operator $F_0 : \domain(F_0) \subseteq L^2(\BZ,\Hil') \longrightarrow L^2(\BZ,\Hil)$ and two group actions $\tau^{(\prime)} : \Gamma^* \longrightarrow \mathcal{GL}(\Hil^{(\prime)})$, and 
		\item a densely defined operator $\widehat{F} : \widehat{\domain}(F) \subseteq L^2(\R^d , \Hil) \longrightarrow L^2(\R^d , \Hil')$ subject to the equivariance condition 
		\begin{align}
			\widehat{T}'_{\gamma^*} \, \widehat{F} \, \widehat{T}_{\gamma^*}^{-1} = \bigl ( \id_{L^2(\R^d)} \otimes \tau'(\gamma^*) \bigr ) \, \widehat{F} \, \bigl ( \id_{L^2(\R^d)} \otimes \tau(\gamma^*)^{-1} \bigr ) 
			&&
			\forall \gamma^* \in \Gamma^* 
			, 
			\label{equivariant_calculus:eqn:equivariance_condition_L2_Rd}
		\end{align}
		where $(\widehat{T}^{(\prime)}_{\gamma^*} \Psi^{(\prime)})(k) := \Psi(k^{(\prime)} - \gamma^*) \in \Hil^{(\prime)}$ denote the translations by the reciprocal lattice vector $\gamma^* \in \Gamma^*$ on the Hilbert spaces $L^2(\R^d,\Hil^{(\prime)})$. 
	\end{enumerate}
\end{proposition}
The proof, which we have relegated to Appendix~\ref{appendix:equivariant_operators:equivalent_characterizations_operators}, is just an exercise in book-keeping, and consists of making the unitaries and partial isometries that connect the different points of views explicit. 

The reason we bother with this exercise in tedium is that it allows us to \emph{define} equivariant pseudodifferential operators through characterization~(c): we merely have to verify the equivariance condition~\eqref{equivariant_calculus:eqn:equivariance_condition_L2_Rd} and that the operators we are dealing with are well-defined via the calculus from Section~\ref{operator_valued_calculus}. 

Before we continue, we would like to give a characterization of \emph{bounded} operators. While Proposition~\ref{equivariant_calculus:prop:characterization_equivariant_operators} applies to bounded and densely defined, unbounded operators alike, it is useful to have an explicit criterion for the boundedness of norms. Once more, we refer the interested reader to Appendix~\ref{appendix:equivariant_operators:equivalent_characterizations_operators} for a proof. 
\begin{proposition}\label{equivariant_calculus:prop:equivalence_boundedness_equivariant_operators}
	Suppose we are in the setting of Proposition~\ref{equivariant_calculus:prop:characterization_equivariant_operators}. 
	\begin{enumerate}[(1)]
		\item The boundedness of each of the operators from characterizations~(a) and (b) implies the boundedness of the others. Specifically, the norms are related by 
		\begin{align*}
			\snorm{F}_{\mathcal{B}(L^2_{\eq}(\R^d,\Hil),L^2_{\eq}(\R^d,\Hil'))} = \bnorm{F_0}_{\mathcal{B}(L^2(\BZ,\Hil) , \mathcal{B}(L^2(\BZ,\Hil'))} 
			. 
		\end{align*}
		\item The operators from characterizations~(a) and (b) are bounded if and only if 
		\begin{align*}
			\widehat{F} : H^{q + q'}_{\Fourier}(\R^d,\Hil) \longrightarrow L^2(\R^d,\Hil') 
		\end{align*}
		from characterization~(c) defines a bounded operator, where $q$ and $q'$ are the orders of growth for the group actions $\tau$ and $\tau'$. 
	\end{enumerate}
\end{proposition}

\subsubsection{The equivariant building block operators: necessary book-keeping} 
\label{equivariant_calculus:construction:equivariant_building_block_operators}
The starting point for a pseudodifferential calculus are the building block operators position and momentum. They are defined in terms of the building block operators 
\begin{subequations}\label{equivariant_calculus:eqn:building_block_operators}
	\begin{align}
		\Reps &= \ii \eps \nabla_k \otimes \id_{\Hil} \big \vert_{\eq} 
		, 
		\\
		\KA &= \hat{k} \otimes \id_{\Hil} - \lambda A(\Reps) 
		\equiv \hat{k} - \lambda A(\Reps)
		, 
	\end{align}
\end{subequations}
where $(\cdots) \vert_{\eq}$ indicates that we impose equivariant boundary conditions and $\hat{k}$ denotes multiplication with $k$. The associated equivariant Weyl system 
\begin{align*}
	\WeylSys_{\eq}^A(r,k) := 
	\e^{- \ii (k \cdot \Reps - r \cdot \KA)}
\end{align*}
then leads to a \emph{formal} definition for the \emph{equivariant} magnetic pseudodifferential operator 
\begin{align}
	\Opeq^A (h) := \frac{1}{(2\pi)^d} \int_{\R^d} \dd r \int_{\R^d} \dd k \, (\Fs h)(r,k) \, \WeylSys_{\eq}^A(r,k) 
	. 
	\label{equivariant_calculus:magnetic_weyl_calculus:eqn:mag_wq_torus}
\end{align}
The purpose of Section~\ref{equivariant_calculus:construction} is to make sense of this definition. 

In principle, we have to distinguish between 
\begin{align*}
	(\Reps,\KA) = \bigl ( \ii \eps \nabla_k \otimes \id_{\Hil} \big \vert_{\eq} \, , \, \hat{k} \otimes \id_{\Hil} - \lambda A(\Reps) \bigr )
\end{align*}
defined on $L^2_{\eq}(\R^d,\Hil)$ and the analogous operators~\eqref{equivariant_calculus:eqn:building_block_operators} on $L^2_{\eq}(\R^d,\Hil')$ when $\Hil \neq \Hil'$; note that the restriction to equivariant $L^2$ functions intertwines the two factors in the tensor product. For instance, suppose we are given an operator-valued function $f \in \Cont^{\infty} \bigl ( T^* \R^d \, , \, \mathcal{B}(\Hil,\Hil') \bigr )$ that satisfies the equivariance relation
\begin{align}
	f(r,k-\gamma^*) &= \tau'(\gamma^*) \, f(r,k) \, \tau(\gamma^*)^{-1}
	&&
	\forall (r,k) \in T^* \R^d
	, \; 
	\gamma^* \in \Gamma^* 
	. 
	\label{equivariant_calculus:eqn:equivariance_function}
\end{align}
Then the equivariance of the operator-valued function $f$ implies the equivariance of 
\begin{align*}
	f(r,k) \, \e^{- \ii \sigma((r,k),(\Reps,\KA))} = \e^{- \ii \sigma((r,k),(\Reps',\mathsf{K}^{\prime \, A}))} \, f(r,k) 
\end{align*}
in the sense of \eqref{equivariant_calculus:eqn:equivariance_function}, where $\Reps'$ and $\mathsf{K}^{\prime \, A}$ are defined via \eqref{equivariant_calculus:eqn:building_block_operators} but act on $L^2_{\eq}(\R^d,\Hil')$. 
Keeping track of the subtle difference between $(\Reps,\KA)$ and $(\Reps' , \mathsf{K}^{\prime \, A})$ would result in an exercise of tedium without any clear gain. \emph{Therefore, we will unburden the notation and not distinguish between building block operators defined on different equivariant, Hilbert space-valued $L^2$ functions.}

In the context of magnetic systems, the various operators are frequently defined between \emph{magnetic} Sobolev spaces. And given that we had three points of views, we will need three types of magnetic Sobolev spaces: the first one 
\begin{align*}
	H^q_{\Fourier,A}(\R^d,\Hil) := \Fourier H^q_A(\R^d,\Hil)
\end{align*}
is just the Fourier transform of the ordinary magnetic Sobolev space, obtained upon replacing $\sexpval{\hat{k}}^q$ with $\bexpval{\hat{k} - \lambda \, A(\ii \eps \nabla_k)}^q$ in the scalar product~\eqref{equivariant_calculus:eqn:scalar_product_Sobolev_space}. Our assumptions on the magnetic vector potential ensure that $\Hil$-valued Schwartz functions lie densely in magnetic and non-magnetic Sobolev spaces of any order, 
%
\begin{align*}
	\Schwartz(\R^d,\Hil) \subset H^q_{\Fourier,A}(\R^d,\Hil) , \, 
	H^q_{\Fourier}(\R^d,\Hil)
	. 
\end{align*}
Whether magnetic Sobolev spaces nest into non-magnetic ones in the two relevant cases (\cf Assumptions~\ref{operator_valued_calculus:assumption:polynomially_bounded_B} and \ref{operator_valued_calculus:assumption:bounded_magnetic_fields}) is to our knowledge an open question. In fact, this is still an area of active research (see \eg \cite{Nguyen_Pinamonti_Squassina_Vecchi:charactizations_magnetic_Sobolev_spaces:2017,Nguyen_Pinamonti_Squassina_Vecchi:charactizations_magnetic_Sobolev_spaces:2020}).

The second type of magnetic Sobolev space is a subset of the \emph{equivariant} $L^2$-space, namely 
\begin{align*}
	H^q_{\eq,\Fourier,A}(\R^d,\Hil) := \Bigl \{ \psi \in L^2_{\eq}(\R^d,\Hil) \; \; \big \vert \; \; \sexpval{\KA}^q \psi \in  L^2_{\eq}(\R^d,\Hil) \Bigr \} 
\end{align*}
endowed with the scalar product 
\begin{align*}
	\scpro{\varphi}{\psi}_{H^q_{\eq,\Fourier,A}(\R^d,\Hil)} := \bscpro{\sexpval{\KA}^q \varphi}{\sexpval{\KA}^q \psi}_{L^2_{\eq}(\R^d,\Hil)}
	. 
\end{align*}
Finally, its restriction to the fundamental domain $\BZ$ is denoted with 
\begin{align*}
	H^q_{\eq,\Fourier,A}(\BZ,\Hil) := \Bigl \{ \psi \vert_{\BZ} \; \; \big \vert \; \; \psi \in H^q_{\eq,\Fourier,A}(\R^d,\Hil) \Bigr \} 
	. 
\end{align*}
One is tempted to drop “$\eq$” and define the magnetic Sobolev space on $\BZ$ directly. However, the equivariance condition \emph{still} enters the definition of $\KA$ as boundary conditions relating the behavior of elements of $H^q_{\eq,\Fourier,A}(\BZ,\Hil)$ at opposing faces of the Brillouin torus $\BZ$. 
\begin{remark}[Magnetic Sobolev spaces over the Brillouin torus]
	Even though the operators $\hat{k}_j$ are bounded on $L^2(\BZ,\Hil)$, the components of the magnetic vector potential $A_j(\Reps)$ need not be. Consequently, magnetic Sobolev spaces $H^q_{\eq,\Fourier,A}(\BZ,\Hil)$ are different Banach spaces for distinct values of $q \geq 0$. 
	
  However, when all the $A_j \in L^{\infty}(\R^d)$ are bounded, then the $H^q_{\eq,\Fourier,A}(\BZ,\Hil)$ coincide with $L^2(\BZ,\Hil)$ as Banach spaces for all $q \geq 0$.
\end{remark}
\begin{corollary}\label{equivariant_calculus:cor:equivalence_magnetic_equivariant_operators}
	Suppose we are in the setting of Proposition~\ref{equivariant_calculus:prop:characterization_equivariant_operators}. 
	Further, assume we are given a bounded operator 
	\begin{align*}
		\widehat{F}^A : H^{q + q'}_{\Fourier,A}(\R^d,\Hil) \longrightarrow L^2(\R^d,\Hil') 
	\end{align*}
	from the \emph{magnetic} Sobolev space of order $q + q'$, and that $\widehat{F}$ satisfies the covariance condition~\eqref{equivariant_calculus:eqn:equivariance_condition_L2_Rd} from characterization~(c). 
	\begin{enumerate}[(1)]
		\item Then $\widehat{F}^A$ defines the bounded operators 
		\begin{align*}
			F^A &: H^{q + q'}_{\eq,\Fourier,A}(\R^d,\Hil) \longrightarrow L^2_{\eq}(\R^d,\Hil') 
			, 
			\\
			F^A_0 &: H^{q + q'}_{\eq,\Fourier,A}(\BZ,\Hil) \longrightarrow L^2(\BZ,\Hil')
			. 
		\end{align*}
		\item When $q + q' = 0$ or all the components of $A$ are bounded, then the operators $F^A_0$ and $F^A$ define bounded operators between the corresponding $L^2$ spaces, 
		\begin{align*}
			F^A &\in \mathcal{B} \bigl ( L^2_{\eq}(\R^d,\Hil) \, , \, L^2_{\eq}(\R^d,\Hil') \bigr ) 
			, 
			\\
			F^A_0 &\in \mathcal{B} \bigl ( L^2(\BZ,\Hil) \, , \, L^2(\BZ,\Hil') \bigr ) 
			. 
		\end{align*}
		\item These operators and their relations are gauge-covariant in the following sense: given any $\vartheta \in \Cont^{\infty}_{\mathrm{pol}}(\R^d,\R)$, then the operators in the new gauge $A' = A + \eps \, \dd \vartheta$ are related to those in the gauge $A$ via $\e^{+ \ii \lambda \vartheta(Q)}$ and $\e^{+ \ii \lambda \vartheta(\Reps)}$, respectively, 
		\begin{align*}
			\widehat{F}^{A'} = \e^{+ \ii \lambda \vartheta(Q)} \, \widehat{F}^A \, \e^{- \ii \lambda \vartheta(Q)} &: H^{q + q'}_{\eq,\Fourier,A'}(\R^d,\Hil) \longrightarrow L^2_{\eq}(\R^d,\Hil') 
			, 
			\\
			F^{A'} = \e^{+ \ii \lambda \vartheta(\Reps)} \, F^A \, \e^{- \ii \lambda \vartheta(\Reps)} &: H^{q + q'}_{\eq,\Fourier,A'}(\R^d,\Hil) \longrightarrow L^2_{\eq}(\R^d,\Hil') 
			, 
			\\
			F_0^{A'} = \e^{+ \ii \lambda \vartheta(\Reps)} \vert_{\BZ} \, F_0^A \, \e^{- \ii \lambda \vartheta(\Reps)} \vert_{\BZ} &: H^{q + q'}_{\eq,\Fourier,A'}(\BZ,\Hil) \longrightarrow L^2(\BZ,\Hil')
			. 
		\end{align*}
		Here, $\e^{+ \ii \lambda \vartheta(\Reps)} \vert_{\BZ}$ denotes the unique operator obtained from the identification $U^{(\prime)}_0 : L^2_{\eq}(\R^d,\Hil^{(\prime)}) \longrightarrow L^2(\BZ,\Hil^{(\prime)})$ (\cf Lemma~\ref{appendix:equivariant_operators:lem:weighted_Hilbert_spaces_translated_Brillouin_zones}~(3)). 
		Importantly, the magnetic Sobolev spaces — the domains of the operators — use the new gauge $A'$. 
	\end{enumerate}
\end{corollary}
Once more, we refer the interested reader to Appendix~\ref{appendix:equivariant_operators:operators_magnetic_Sobolev_spaces} for a proof.

\subsubsection{Connection to the pseudodifferential calculus from Section~\ref{operator_valued_calculus}} 
\label{equivariant_calculus:construction:connection_operator_valued_calculus}
The building block operators~\eqref{equivariant_calculus:eqn:building_block_operators} prior to restricting to equivariant Hilbert spaces are related to the standard ones~\eqref{operator_valued_calculus:eqn:building_block_operators} by the continuous Fourier transform $\Fourier$. Consequently, the link between the results here and those in Section~\ref{operator_valued_calculus} is to think of an equivariant pseudodifferential operator as the restriction 
\begin{align*}
	\Opeq^A(f) = \Fourier \, \Op^A(f) \, \Fourier^{-1} \big \vert_{\eq} 
\end{align*}
to (suitable subsets of) equivariant distributions. We will make this procedure precise in the course of this section. As a first step, we wish to convince the reader that despite the presence of the continuous Fourier transform $\Fourier$ we may still re-use all of the results from Section~\ref{operator_valued_calculus}. 

Fortunately, the continuous Fourier transform defines topological vector space isomorphisms on $\Schwartz(\R^d,\Hil)$, $\Schwartz^*(\R^d,\Hil)$ and $L^2(\R^d,\Hil)$, and the arguments from Section~\ref{operator_valued_calculus:building_block_operators} tell us that this change of representation does not impact the construction of the magnetic pseudodifferential calculus or the symbol calculus from Section~\ref{operator_valued_calculus}. Moreover, even though the explicit expression for $\Op^A$ becomes more implicit as the magnetic phase now turns into a convolution operator rather than a multiplication operator, thanks to \eqref{operator_valued_calculus:eqn:covariance_Weyl_system} the entire calculus on the level of symbols or distributions remains unchanged. Importantly, this includes the magnetic Weyl product. 

\emph{To streamline the presentation, we will abbreviate $\Fourier \, \Op^A(f) \, \Fourier^{-1}$ with $\Op_{\Fourier}^A(f)$. } 

\subsubsection{Characterization of equivariant pseudodifferential operators} 
\label{equivariant_calculus:construction:characterization_equivariant_operators}
For a consistent extension, we will need to extend the equivariance condition~\eqref{equivariant_calculus:eqn:equivariance_function} from smooth functions to tempered distributions: 
\begin{definition}[Equivariant tempered distributions]\label{equivariant_calculus:defn:equivariance_distribution}
	\begin{enumerate}[(1)]
		\item We call a tempered distribution $\Psi \in \Schwartz^*(\R^d,\Hil)$ equivariant if and only if 
		\begin{align*}
			T_{\gamma^*} \Psi = \tau(\gamma^*) \Psi 
		\end{align*}
		holds for all $\gamma^* \in \Gamma^*$, where the two operators are extended by duality as 
		\begin{align*}
			\scpro{T_{\gamma^*} \Psi}{\varphi}_{\Schwartz(\R^d,\Hil)} :& \negmedspace= \bscpro{\Psi}{T_{-\gamma^*} \varphi}_{\Schwartz(\R^d,\Hil)} 
			, 
			\\
			\bscpro{\tau(\gamma^*) \Psi}{\varphi}_{\Schwartz(\R^d,\Hil)} :& \negmedspace= \bscpro{\Psi}{\tau(\gamma^*)^* \varphi}_{\Schwartz(\R^d,\Hil)} 
			, 
		\end{align*}
		for all $\varphi \in \Schwartz(\R^d,\Hil)$. The subspace of equivariant tempered distributions is denoted with $\Schwartz^*_{\eq}(\R^d,\Hil)$. 
		\item We call a tempered distribution $F \in \Schwartz^* \bigl ( T^* \R^d \, , \, \mathcal{B}(\Hil,\Hil') \bigr )$ equivariant if and only if 
		\begin{align*}
			\scpro{\psi' \, }{ \, T_{\gamma^*} F \varphi}_{\Schwartz(\R^d,\Hil')} &= \Bscpro{\psi' \, }{ \, \bigl ( \id_{\Schwartz'(\R^d)} \otimes \tau'(\gamma^*) \bigr ) \, F \, \bigl ( \id_{\Schwartz'(\R^d)} \otimes \tau(\gamma^*)^{-1} \bigr ) \, \varphi}_{\Schwartz(\R^d,\Hil')}
			\\
			&= \Bscpro{\bigl ( \id_{\Schwartz'(\R^d)} \otimes \tau'(\gamma^*)^* \bigr ) \, \psi' \, }{ \, F \, \bigl ( \id_{\Schwartz'(\R^d)} \otimes \tau(\gamma^*)^{-1} \bigr ) \, \varphi}_{\Schwartz(\R^d,\Hil')}
		\end{align*}
		holds for all $\varphi \in \Schwartz(\R^d,\Hil)$, $\psi' \in \Schwartz(\R^d,\Hil')$ and $\gamma^* \in \Gamma^*$. The subspace of all equivariant, operator-valued tempered distributions is denoted with $\Schwartz^*_{\eq} \bigl ( T^* \R^d \, , \, \mathcal{B}(\Hil,\Hil') \bigr )$. 
	\end{enumerate}
\end{definition}
We may also consider subspaces of distributions that are also equivariant, \eg the equivariant operator-valued Moyal space
\begin{align}
	\MoyalSpace_{\eq} \bigl ( \mathcal{B}(\Hil,\Hil') \bigr ) := \MoyalSpace \bigl ( \mathcal{B}(\Hil,\Hil') \bigr ) \cap \Schwartz^*_{\eq} \bigl ( T^* \R^d \, , \, \mathcal{B}(\Hil,\Hil') \bigr )
	\label{equivariant_calculus:eqn:equivariant_Moyal_space_intersection_Moyal_space_cap_equivariant_distributions}
\end{align}
is the intersection of the operator-valued Moyal space and the subspace of the equivariant distributions. Importantly, distributions from this subspace define pseudodifferential operators that preserve equivariance. 
\begin{proposition}\label{equivariant_calculus:prop:equivariant_distributions_define_equivariant_PsiDOs}
	Suppose $f \in  \MoyalSpace \bigl ( \mathcal{B}(\Hil,\Hil') \bigr )$ is an element from the Moyal space. Then $f \in \MoyalSpace_{\eq} \bigl ( \mathcal{B}(\Hil,\Hil') \bigr )$ lies in the \emph{equivariant} Moyal space if and only if the restriction 
	\begin{align}
		\Opeq^A(f) := \Op_{\Fourier}^A(f) \big \vert_{\Schwartz_{\eq}^*(\R^d,\Hil)} : \Schwartz^*_{\eq}(\R^d,\Hil) \longrightarrow \Schwartz^*_{\eq}(\R^d,\Hil') 
		\label{equivariant_calculus:eqn:definition_OpeqA}
	\end{align}
	defines a continuous operator between \emph{equivariant} tempered distributions. Moreover, $\Opeq^A$ is gauge-covariant, \ie for any $\vartheta \in \Cont^{\infty}_{\mathrm{pol}}(\R^d,\R)$ the gauge operator $\e^{+ \ii \lambda \vartheta(\Reps)}$ relates $\Opeq^A(f)$ and 
	\begin{align*}
		\Opeq^{A + \eps \dd \vartheta}(f) = \e^{+ \ii \lambda \vartheta(\Reps)} \, \Opeq^A(f) \, \e^{- \ii \lambda \vartheta(\Reps)} 
		. 
	\end{align*}
\end{proposition}
\begin{proof}
	$\Longrightarrow$: Suppose $f \in \MoyalSpace_{\eq} \bigl ( \mathcal{B}(\Hil,\Hil') \bigr )$ is an element of the equivariant Moyal space. We need to prove that it restricts to a continuous map between \emph{equivariant} distributions. Continuity of $\Op_{\Fourier}^A(f)$ then follows directly as elements of $\MoyalSpace \bigl ( \mathcal{B}(\Hil,\Hil') \bigr ) \supseteq \MoyalSpace_{\eq} \bigl ( \mathcal{B}(\Hil,\Hil') \bigr )$ define continuous maps between not necessarily equivariant, Hilbert space-valued distributions $\Schwartz^*(\R^d,\Hil^{(\prime)}) \supseteq \Schwartz^*_{\eq}(\R^d,\Hil^{(\prime)})$. 
	
	To aid readability, we will use notation that suggests $f$ and $\Psi$ are functions rather than tempered distributions even though they are not. So instead of $\bigl ( \Opeq^A(f) \Psi \, , \, \varphi' \bigr )_{\Schwartz(\R^d,\Hil')}$ where $\varphi' \in \Schwartz(\R^d,\Hil')$, we will write the action of $\Opeq^A(f)$ on $\Psi \in \Schwartz_{\eq}^{\ast}(\R^d,\Hil)$ as the formal integral 
	\begin{align}
		\bigl ( \Op_{\Fourier}^A(f) \Psi \bigr )(k) &= \frac{1}{(2\pi)^{2d}} \int_{\R^d} \dd r' \int_{\R^d} \dd k' \int_{\R^d} \dd r'' \int_{\R^d} \dd k'' \, \e^{+ \ii \sigma((r',k'),(r'',k''))} \, f(r'',k'') \, 
		\cdot \notag \\
		&\qquad \qquad \qquad \qquad \qquad \qquad \qquad \qquad \qquad \cdot 
		\bigl ( \e^{- \ii \sigma((r',k'),(\Reps,\KA))} \Psi \bigr )(k) 
		\label{equivariant_calculus:eqn:formal_expression_OpeqA_evaluated_k_minus_gamma_ast}
	\end{align}
	with the tacit understanding that we need to apply this to a $\Hil'$-valued Schwartz function and extend all these operations by duality. We will also make repeated use of the fact that elements of $\MoyalSpace \bigl ( \mathcal{B}(\Hil,\Hil') \bigr )$ define continuous maps between (not necessarily equivariant) vector-valued distributions.
	
	That being said, the condition we need to check is 
	\begin{align*}
		T_{\gamma^*} \, \Op_{\Fourier}^A(f) \Psi &= \tau'(\gamma^*) \, \Op_{\Fourier}^A(f) \Psi 
		&&
		\forall \gamma^* \in \Gamma^*
		, \; 
		\Psi \in \Schwartz^*_{\eq}(\R^d,\Hil) 
		, 
	\end{align*}
	and almost all $k \in \R^d$, where the translations by $\gamma^*$ are generated by $\Reps$, 
	\begin{align*}
		T_{\gamma^*} = \e^{+ \frac{\ii}{\eps} \gamma^* \cdot \Reps} 
		. 
	\end{align*}
	Importantly, the $T_{\gamma^*}$ commute with all operators that are functions of $\Reps$. 
	
	Glancing at equation~\eqref{equivariant_calculus:eqn:formal_expression_OpeqA_evaluated_k_minus_gamma_ast}, we see that we need to consider the factor involving the Weyl system first. Writing out the action of the Weyl system, exploiting equivariance of $\Psi$, and commuting $T_{\gamma^*}$ with all operators that are functions of $\Reps$, we get 
	\begin{align*}
		T_{\gamma^*} \, \e^{- \ii \sigma((r',k'),(\Reps,\KA))} \Psi &= \e^{- \ii \frac{\eps}{2} r' \cdot k'} \, \e^{+ \frac{\ii}{\eps} \gamma^* \cdot \Reps} \, \e^{- \ii k' \cdot R} \, \e^{- \frac{\ii}{\eps} \int_{[\Reps,\Reps + \eps r']} A} \, \e^{- \ii r' \cdot \hat{k}} \Psi 
		\\
		&= \e^{- \ii \frac{\eps}{2} r' \cdot k'} \, \e^{- \ii k' \cdot R} \, \e^{- \frac{\ii}{\eps} \int_{[\Reps,\Reps + \eps r']} A} \, \e^{+ \frac{\ii}{\eps} \gamma^* \cdot \Reps} \, \e^{- \ii r' \cdot \hat{k}} \e^{- \frac{\ii}{\eps} \gamma^* \cdot \Reps} \, \e^{+ \frac{\ii}{\eps} \gamma^* \cdot \Reps} \, \Psi 
		. 
	\end{align*}
	Adjoining the translation in position by translations in momentum amounts to replacing $\hat{k}$ by $\hat{k} + \gamma^*$. Moreover, due to the assumed equivariance the translation of $\Psi$ yields the operator $\tau(\gamma^*)$. Lastly, the group action commutes with the Weyl system as it is an operator on the equivariant $L^2$ space. Overall, this gives us two extra factors in the expression: 
	\begin{align*}
		\ldots 
		&= \e^{- \ii \frac{\eps}{2} r' \cdot k'} \, \e^{- \ii k' \cdot R} \, \e^{- \frac{\ii}{\eps} \int_{[\Reps,\Reps + \eps r']} A} \, \e^{- \ii r' \cdot (\hat{k} + \gamma^*)} \, \tau(\gamma^*) \, \Psi 
		\\
		&= \e^{- \ii r' \cdot \gamma^*} \, \e^{- \ii \sigma((r',k'),(\Reps,\KA))} \, \tau(\gamma^*) \, \Psi 
		\\
		&= \e^{- \ii r' \cdot \gamma^*} \, \tau(\gamma^*) \; \e^{- \ii \sigma((r',k'),(\Reps,\KA))} \, \Psi 
	\end{align*}
	The phase factor will ultimately lead to a shift in the $k$ space coordinate where we compute the Fourier transform of the symbol $f$ in $\Op_{\Fourier}^A$: 
	\begin{align*}
		T_{\gamma^*} \, \Op_{\Fourier}^A(f) \Psi &= \frac{1}{(2\pi)^{2d}} \int_{\R^d} \dd r' \int_{\R^d} \dd k' \int_{\R^d} \dd r'' \int_{\R^d} \dd k'' \, \e^{+ \ii (k' \cdot r'' - r' \cdot k'')} \, f(r'',k'') \, 
		\cdot \notag \\
		&\qquad \qquad \qquad \qquad \qquad \qquad \qquad \qquad \qquad \cdot 
		\tau(\gamma^*) \, \e^{- \ii r' \cdot \gamma^*} \, \e^{- \ii \sigma((r',k'),(\Reps,\KA))} \Psi 
		\\
		&= \frac{1}{(2\pi)^{2d}} \int_{\R^d} \dd r' \int_{\R^d} \dd k' \int_{\R^d} \dd r'' \int_{\R^d} \dd k'' \, \e^{+ \ii (k' \cdot r'' - r' \cdot (k'' + \gamma^*))} \, f(r'',k'') \, \tau(\gamma^*) \, 
		\cdot \notag \\
		&\qquad \qquad \qquad \qquad \qquad \qquad \qquad \qquad \qquad \cdot 
		\e^{- \ii \sigma((r',k'),(\Reps,\KA))} \Psi 
	\end{align*}
	A simple change in variables and the covariance condition for $f$ then yields the claim, 
	\begin{align*}
		\ldots 
		&= \frac{1}{(2\pi)^{2d}} \int_{\R^d} \dd r' \int_{\R^d} \dd k' \int_{\R^d} \dd r'' \int_{\R^d} \dd \tilde{k}'' \, \e^{+ \ii (k' \cdot r'' - r' \cdot \tilde{k}'')} \, f(r'',\tilde{k}'' - \gamma^*) \, \tau(\gamma^*) \, 
		\cdot \notag \\
		&\qquad \qquad \qquad \qquad \qquad \qquad \qquad \qquad \qquad \cdot 
		\e^{- \ii \sigma((r',k'),(\Reps,\KA))} \Psi 
		\\
		&= \frac{1}{(2\pi)^{2d}} \int_{\R^d} \dd r' \int_{\R^d} \dd k' \int_{\R^d} \dd r'' \int_{\R^d} \dd \tilde{k}'' \, \e^{+ \ii (k' \cdot r'' - r' \cdot \tilde{k}'')} \, \tau'(\gamma^*) \, f(r'',\tilde{k}'') \, 
		\cdot \notag \\
		&\qquad \qquad \qquad \qquad \qquad \qquad \qquad \qquad \qquad \cdot 
		\e^{- \ii \sigma((r',k'),(\Reps,\KA))} \Psi 
		\\
		&= \tau'(\gamma^*) \, \Op_{\Fourier}^A(f) \Psi 
		. 
	\end{align*}
	$\Longleftarrow$: Suppose that the Weyl quantization $\Op_{\Fourier}^A(f)$ of some $f \in  \MoyalSpace \bigl ( \mathcal{B}(\Hil,\Hil') \bigr )$, which defines a continuous map between Hilbert space-valued distribution spaces, restricts to a continuous operator between equivariant distribution spaces. In view of the above computation, imposing the condition 
	\begin{align*}
		T_{\gamma^*} \, \Op_{\Fourier}^A(f) \Psi \overset{!}{=} \tau'(\gamma^*) \, \Op_{\Fourier}^A(f) \Psi
	\end{align*}
	is equivalent to the equivariance condition 
	\begin{align*}
		f(r,k - \gamma^*) = \tau'(\gamma^*) \, f(r,k) \, \tau(\gamma^*)^{-1} 
		. 
	\end{align*}
	That is, $f \in \MoyalSpace_{\eq} \bigl ( \mathcal{B}(\Hil,\Hil') \bigr )$ lies in the appropriate \emph{equivariant} Moyal space. 
\end{proof}
This proposition is the centerpiece of this section, for it guarantees that $\Opeq^A$ is well-defined. To prevent that we hide this important definition in a proposition, let us spell it out one more time: 
\begin{definition}[Equivariant magnetic pseudodifferential operator]\label{equivariant_calculus:defn:Opeq_A}
	For any equivariant operator-valued distribution $f \in \MoyalSpace_{\eq} \bigl ( \mathcal{B}(\Hil,\Hil') \bigr )$ we define the equivariant magnetic Weyl quantization as the restriction 
	\begin{align*}
		\Opeq^A(f) := \Op_{\Fourier}^A(f) \big \vert_{\Schwartz_{\eq}^*(\R^d,\Hil)} : \Schwartz^*_{\eq}(\R^d,\Hil) \longrightarrow \Schwartz^*_{\eq}(\R^d,\Hil') 
	\end{align*}
	of the map from Proposition~\ref{operator_valued_calculus:prop:extension_OpA_tempered_distributions}~(2) to \emph{equivariant} tempered distributions. 
\end{definition}
Because $\Opeq^A(f)$ is defined as the restriction of $\Op_{\Fourier}^A(f)$ to (subsets of) equivariant distributions, the operator kernel of $\Opeq^A(f)$ is identical to that of $\Op_{\Fourier}^A(f)$. However, as we will need the inverse of the kernel map later on, we will include its explicit form. 
\begin{lemma}\label{equivariant_calculus:lem:kernel_map}
	\begin{enumerate}[(1)]
		\item For any $f \in \MoyalSpace_{\eq} \bigl ( \mathcal{B}(\Hil,\Hil') \bigr )$ the kernel of $\Opeq^A(f)$ is given by 
		\begin{align}
			K_{\eq,f}^A(k,k'') = \frac{1}{(2\pi)^{2d}} \int_{\R^d} \dd r' \int_{\R^d} \dd r'' \int_{\R^d} \dd k' &\e^{- \ii r' \cdot (k - k')} \, \e^{- \ii r'' \cdot (k' - k'')} 
			\cdot \notag \\
			&\cdot 
			\e^{- \ii \frac{\lambda}{\eps} \int_{[\eps r' , \eps r'']} A} \, f \bigl ( \tfrac{\eps}{2} (r' + r'') , k' \bigr ) 
			. 
			\label{equivariant_calculus:eqn:kernel_map}
		\end{align}
		\item The kernel map 
		\begin{align*}
			f \mapsto K_{\eq,f}^A : \MoyalSpace_{\eq} \bigl ( \mathcal{B}(\Hil,\Hil') \bigr ) \longrightarrow \Schwartz^*_{\eq} \bigl ( \R^d \times \R^d , \mathcal{B}(\Hil,\Hil') \bigr ) 
		\end{align*}
		is a continuous injection into the space of equivariant distributions on $\R^d \times \R^d$, \ie distributions which satisfy the equivariance condition 
		\begin{align*}
			K_{\eq,f}^A(k - \gamma^* , k'' - \gamma^*) &= \tau'(\gamma^*) \, K_{\eq,f}^A(k,k'') \, \tau(\gamma^*)^{-1} 
		\end{align*}
		for all $\gamma^* \in \Gamma^*$ in the distributional sense. 
	\end{enumerate}
\end{lemma}
\begin{proof}
	\begin{enumerate}[(1)]
		\item This follows from writing out the kernel for $\Op^A_{\Fourier}(f)$ and collecting the terms of the phase factors to emphasize that we obtain convolution integrals, 
		\begin{align*}
			\e^{- \ii r' \cdot k} \, \e^{- \ii (r'' - r') \cdot k'} \, \e^{+ \ii r'' \cdot k''} = \e^{- \ii r' \cdot (k - k')} \, \e^{- \ii r'' \cdot (k' - k'')} 
			. 
		\end{align*}
		\item Ignoring equivariance for the moment, Corollary~\ref{operator_valued_calculus:lem:Wigner_transform_Schwartz_class} tells us it is a topological vector space isomorphism between tempered distribution spaces. Consequently, it restricts to a continuous linear injection when we regard it as a map 
		\begin{align*}
			f \mapsto K_{\eq,f}^A : \MoyalSpace_{\eq} \bigl ( \mathcal{B}(\Hil,\Hil') \bigr ) \subset \Schwartz^* \bigl ( T^*\R^d , \mathcal{B}(\Hil,\Hil') \bigr )  \longrightarrow \Schwartz^* \bigl ( \R^d \times \R^d , \mathcal{B}(\Hil,\Hil') \bigr ) 
			. 
		\end{align*}
		It remains to show the kernel map takes values in the equivariant distributions on $\R^d \times \R^d$. But that is implied directly by the equivariance of $f$ and a simple change of coordinates $\tilde{k}' := k' + \gamma^*$: 
		\begin{align*}
			K_{\eq,f}^A(k - \gamma^*,k'' - \gamma^*) &= \frac{1}{(2\pi)^{2d}} \int_{\R^d} \dd r' \int_{\R^d} \dd r'' \int_{\R^d} \dd k' \e^{- \ii r' \cdot (k - \gamma^* - k')} \, \e^{- \ii r'' \cdot (k' - k'' + \gamma^*)} 
			\cdot \notag \\
			&\qquad \qquad \qquad \qquad \cdot 
			\e^{- \ii \frac{\lambda}{\eps} \int_{[\eps r' , \eps r'']} A} \, f \bigl ( \tfrac{\eps}{2} (r' + r'') , k' \bigr ) 
			\\
			&= \frac{1}{(2\pi)^{2d}} \int_{\R^d} \dd r' \int_{\R^d} \dd r'' \int_{\R^d} \dd \tilde{k}' \e^{- \ii r' \cdot (k - \tilde{k}')} \, \e^{- \ii r'' \cdot (\tilde{k}' - k'')} 
			\cdot \notag \\
			&\qquad \qquad \qquad \qquad \cdot 
			\e^{- \ii \frac{\lambda}{\eps} \int_{[\eps r' , \eps r'']} A} \, f \bigl ( \tfrac{\eps}{2} (r' + r'') , \tilde{k}' - \gamma^* \bigr )
			\\
			&= \frac{1}{(2\pi)^{2d}} \int_{\R^d} \dd r' \int_{\R^d} \dd r'' \int_{\R^d} \dd \tilde{k}' \e^{- \ii r' \cdot (k - \tilde{k}')} \, \e^{- \ii r'' \cdot (\tilde{k}' - k'')} 
			\cdot \notag \\
			&\qquad \qquad \qquad \qquad \cdot 
			\e^{- \ii \frac{\lambda}{\eps} \int_{[\eps r' , \eps r'']} A} \, \tau'(\gamma^*) \, f \bigl ( \tfrac{\eps}{2} (r' + r'') , \tilde{k}' \bigr ) \, \tau(\gamma^*)^{-1} 
			\\
			&= \tau'(\gamma^*) \, K_{\eq,f}^A(k,k'') \, \tau(\gamma^*)^{-1} 
		\end{align*}
		Hence, the kernel map takes values in the \emph{equivariant} distributions on $\R^d \times \R^d$. 
	\end{enumerate}
\end{proof}
%

\subsubsection{The Meta Theorem for extending results to equivariant magnetic $\Psi$DOs} 
\label{equivariant_calculus:construction:meta_theorem}
Crucially, Proposition~\ref{equivariant_calculus:prop:equivariant_distributions_define_equivariant_PsiDOs} allows one to systematically extend results from the operator-valued to the equivariant context. For example, we can consider a restriction of the Moyal product to equivariant Moyal spaces. 
\begin{proposition}\label{equivariant_calculus:prop:intertwining_operator_product_Weyl_product_Opeq}
	The magnetic Weyl product $\Weyl$ emulates the operator product of two equivariant magnetic $\Psi$DOs in the sense that 
	\begin{align}
		\Opeq^A(f \Weyl g) = \Opeq^A(f) \, \Opeq^A(g) 
		\label{equivariant_calculus:eqn:equivariance_Weyl_product}
	\end{align}
	holds for all $f \in \MoyalSpace_{\eq} \bigl ( \mathcal{B}(\Hil',\Hil'') \bigr )$ and $g \in \MoyalSpace_{\eq} \bigl ( \mathcal{B}(\Hil,\Hil') \bigr )$. Specifically, the magnetic Weyl product $\Weyl$ preserves equivariance.
\end{proposition}
\begin{proof}
	Given that $\MoyalSpace_{\eq} \bigl ( \mathcal{B}(\Hil,\Hil') \bigr ) \subseteq \MoyalSpace \bigl ( \mathcal{B}(\Hil,\Hil') \bigr )$ is a subspace of the larger Moyal space that also includes non-equivariant distributions, the equation 
	\begin{align}
    \Op_{\Fourier}^A(f \Weyl g) = \Op_{\Fourier}^A(f) \, \Op_{\Fourier}^A(g) 
	\end{align}
	makes sense as continuous maps $\Schwartz(\R^d,\Hil) \longrightarrow \Schwartz(\R^d,\Hil'')$ that possesses a continuous extension to $\Schwartz^*(\R^d,\Hil) \longrightarrow \Schwartz^*(\R^d,\Hil'')$. If we restrict ourselves to the smaller space of equivariant distributions, then 
	\begin{align*}
		\Op_{\Fourier}^A(f \Weyl g) \big \vert_{\Schwartz^*_{\eq}(\R^d,\Hil)} : \Schwartz^*_{\eq}(\R^d,\Hil) \subseteq \Schwartz^*(\R^d,\Hil) \longrightarrow \Schwartz^*(\R^d,\Hil'')
	\end{align*}
	is still continuous. 
	
	We merely need to verify equivariance is preserved and invoke Proposition~\ref{equivariant_calculus:prop:equivariant_distributions_define_equivariant_PsiDOs}. But the left-hand side $\Opeq^A(f \Weyl g) = \Op_{\Fourier}^A(f \Weyl g) \big \vert_{\Schwartz^*_{\eq}(\R^d,\Hil)}$ is an equivariant pseudodifferential operator whenever the right-hand side is. And Proposition~\ref{equivariant_calculus:prop:equivariant_distributions_define_equivariant_PsiDOs} tells us that $\Opeq^A(f) = \Op_{\Fourier}^A(f) \vert_{\Schwartz^*_{\eq}(\R^d,\Hil)}$ takes values in the equivariant, $\Hil'$-valued tempered distributions. Consequently, we may replace $\Op_{\Fourier}^A(g)$ with its restriction $\Opeq^A(g)$ and also the composition 
	\begin{align*}
		\Opeq^A(f) \, \Opeq^A(g) 
	\end{align*}
	is equivariant. Thus, also the left-hand side 
	\begin{align*}
		\Opeq^A(f \Weyl g) = \Op_{\Fourier}^A(f \Weyl g) \big \vert_{\Schwartz^*_{\eq}(\R^d,\Hil)} : \Schwartz^*_{\eq}(\R^d,\Hil) \longrightarrow \Schwartz^*_{\eq}(\R^d,\Hil'')
	\end{align*}
	takes values in the \emph{equivariant} $\Hil''$-valued distributions. 
\end{proof}
The first part of the proof just re-states what we know for the operator-valued calculus from Section~\ref{operator_valued_calculus}, the second one elaborates how equivariance is verified and then invokes Proposition~\ref{equivariant_calculus:prop:equivariant_distributions_define_equivariant_PsiDOs}. 

This provides a blueprint to systematically extend properties of the operator-valued pseudodifferential calculus to the equivariant context. Giving us the benefit of vagueness, this can be summarized in the form of a meta theorem: 
\begin{metatheorem}\label{equivariant_calculus:meta_theorem:extension_operator_valued_to_equivariant}
	Suppose a statement for $\Op_{\Fourier}^A$ defined for operator-valued distributions holds true and that equivariance is preserved. Then the statement extends to $\Opeq^A$. 
\end{metatheorem}
The proofs are very similar to the one we have just presented: one needs to combine the statement for operator-valued symbols and invoke Proposition~\ref{equivariant_calculus:prop:equivariant_distributions_define_equivariant_PsiDOs}. Our discussion of an equivariant symbol calculus will make repeated use of this. 

\subsection{A pseudodifferential calculus for equivariant Hörmander symbols} 
\label{equivariant_calculus:calculus}
Equivariant pseudodifferential operators admit a symbol calculus. Almost all of the results in this section are instantiations of Meta Theorem~\ref{equivariant_calculus:meta_theorem:extension_operator_valued_to_equivariant}, and we will only spell out the details for some of the results. We begin by defining equivariant symbol classes and giving one very important property that relates the growth of the group actions $\tau$ and $\tau'$ to the order of the symbol class.

\subsubsection{Definition of equivariant symbol classes and some fundamental properties} 
\label{equivariant_calculus:construction:equivariant_symbol_classes}
\emph{Equivariant} operator-valued Hörmander symbol classes form subspaces of the corresponding operator-valued Hörmander symbol spaces. 
\begin{definition}[Equivariant symbol classes $\Hoereq{m} \bigl ( \mathcal{B}(\Hil' , \Hil) \bigr )$]\label{equivariant_calculus:defn:tau_equivariant_symbol}
	Suppose $m \in \R$. 
	\begin{enumerate}[(1)]
		\item We call symbols $f \in \Hoer{m}_{0,0} \bigl ( \mathcal{B}(\Hil,\Hil') \bigr )$ equivariant if and only if 
		\begin{align}
			f(r, k - \gamma^*) = \tau'(\gamma^*) \, f(r, k) \, \tau(\gamma^*)^{-1} 
			\label{equivariant_calculus:eqn:equivariance_condition_symbols}
		\end{align}
		holds for all $k \in \R^d$, $r \in \R^d$ and $\gamma^* \in \Gamma^*$. The space of equivariant Hörmander symbols of order $m$ will be denoted with $\Hoereq{m} \bigl ( \mathcal{B}(\Hil,\Hil') \bigr )$. 
		\item We call symbols $f \in \Hoer{m}_{0,0} \bigl ( \mathcal{B}(\Hil,\Hil') \bigr )$ periodic if they satisfy \eqref{equivariant_calculus:eqn:equivariance_condition_symbols} for the trivial group actions $\tau^{(\prime)} : \gamma^* \mapsto \id_{\Hil}$. The space of equivariant Hörmander symbols of order $m$ will be denoted with $\Hoer{m}_{\per} \bigl ( \mathcal{B}(\Hil,\Hil') \bigr )$. 
		\item The spaces of equivariant \emph{asymptotic} Hörmander symbols and formal sums of order $m$ (\cf Definition~\ref{operator_valued_calculus:defn:asymptotic_Hoermander_class}) are likewise denoted with $\SemiHoereq{m} \bigl ( \mathcal{B}(\Hil,\Hil') \bigr )$ and $\Sigma \Hoereq{m} \bigl ( \mathcal{B}(\Hil,\Hil') \bigr )$. 
	\end{enumerate}
\end{definition}
In fact, one may think of equivariant symbols as the intersections of the ordinary symbol classes with the space of equivariant distributions or equivariant Moyal spaces, 
\begin{align}
	\Hoereq{m} \bigl ( \mathcal{B}(\Hil,\Hil') \bigr ) &= \Hoer{m}_{0,0} \bigl ( \mathcal{B}(\Hil,\Hil') \bigr ) \cap \Schwartz^*_{\eq} \bigl ( T^* \R^d , \mathcal{B}(\Hil,\Hil') \bigr )
	, 
	\notag \\
	\Hoereq{m} \bigl ( \mathcal{B}(\Hil,\Hil') \bigr ) &= \Hoer{m}_{0,0} \bigl ( \mathcal{B}(\Hil,\Hil') \bigr ) \cap \MoyalSpace_{\eq} \bigl ( \mathcal{B}(\Hil,\Hil') \bigr )
	. 
	\label{equivariant_calculus:eqn:equivariant_Hoermander_class_intersection_usual_Hoermander_class_equivariant_Moyal_space}
\end{align}
Note that in applications it does not suffice to consider $m = 0$ and $\rho = 0$. In fact, the orders of $\tau^{(\prime)}$ determine the order of the Hörmander symbol class the symbol belongs to. 
\begin{lemma}\label{equivariant_calculus:magnetic_PsiDOs:lem:bound_Hoermander_order_by_tau_orders}
	Any Hörmander symbol $f \in \Hoer{m}_{0,0} \bigl ( \mathcal{B}(\Hil , \Hil') \bigr )$ of order $m \in \R$ that satisfies the equivariance condition~\eqref{equivariant_calculus:eqn:equivariance_condition_symbols} is an equivariant Hörmander symbol of order $q + q'$, 
	\begin{align*}
		f \in \Hoer{m}_{0,0} \bigl ( \mathcal{B}(\Hil , \Hil') \bigr )
		\; \wedge \; 
		\mbox{\eqref{equivariant_calculus:eqn:equivariance_condition_symbols}}
		\; \Longrightarrow \; 
		f \in \Hoereq{q+q'} \bigl ( \mathcal{B}(\Hil , \Hil') \bigr ) 
		. 
	\end{align*}
	That is, the order of the symbol class $m \geq q + q'$ has to be greater or equal to the sum of the orders of the actions $\tau$ and $\tau'$. 
\end{lemma}
This innocent-looking and easy-to-prove lemma has some very important consequences. For example, rather than being nested, \emph{equivariant} Hörmander symbol become “constant” as long as $m \geq q + q'$, 
\begin{align*}
	S^m_{\eq} \bigl ( \mathcal{B}(\Hil,\Hil') \bigr ) = S^{q + q'}_{\eq} \bigl ( \mathcal{B}(\Hil,\Hil') \bigr )
	&&
	\forall m \geq q + q'
	. 
\end{align*}
Furthermore, when we take the magnetic Weyl product of two equivariant Hörmander symbols, the order of the product is determined by the order of growth of the group actions on initial and target spaces. 
\begin{proof}
	To prove that $f \in \Hoereq{q+q'} \bigl ( \mathcal{B}(\Hil,\Hil') \bigr )$ belongs to the Hörmander symbol class of order $q+q'$, we split $k = [k] + \gamma^*$ into a reciprocal lattice vector $\gamma^*$ and a vector $[k] \in \BZ$ that is located inside the fundamental cell centered at $0$ and exploit the equivariance~\eqref{equivariant_calculus:eqn:equivariance_condition_symbols} of the symbol, 
	\begin{align*}
		\bnorm{f(r,k)}_{\mathcal{B}(\Hil,\Hil')} &= \bnorm{\tau'(\gamma) \, f(r,[k]) \, \tau(\gamma^*)^{-1}}_{\mathcal{B}(\Hil)}
		\\
		&\leq \bnorm{\tau(\gamma)}_{\mathcal{B}(\Hil)} \, \bnorm{\tau'(\gamma)}_{\mathcal{B}(\Hil')} \, \sup_{\substack{r \in \R^d \\ [k] \in \BZ}} \bnorm{f(r,[k])}_{\mathcal{B}(\Hil,\Hil')}
		. 
	\end{align*}
	Because the symbol $f$ is a continuous function that is $\Cont^{\infty}_{\mathrm{b}}$ in the position variable and we take the supremum over the compact set $\BZ$ in momentum, the supremum 
	\begin{align*}
		C_1(f) := \sup_{\substack{r \in \R^d \\ [k] \in \BZ}} \bnorm{f(r,[k])}_{\mathcal{B}(\Hil,\Hil')} 
		< \infty
	\end{align*}
	is finite. 

	Furthermore, because $\tau^{(\prime)}$ are of order $q^{(\prime)}$, we can estimate the other two norms in a straightforward fashion by 
	\begin{align*}
		\bnorm{\tau^{(\prime)}(\gamma^*)}_{\Hil^{(\prime)}} \leq C_2^{(\prime)} \, \sexpval{\gamma^*}^{q^{(\prime)}} 
		. 
	\end{align*}
	This together with the trivial estimate $\sabs{\gamma^*} \leq 2 \sabs{k}$ lets us control the growth on the right-hand side, 
	\begin{align*}
		\snorm{f(r,k)}_{\mathcal{B}(\Hil,\Hil')} \leq 4 C_1(f) \, C_2 \, C'_2 \, \sexpval{k}^{q+q'} 
		. 
	\end{align*}
	Since also all partial derivatives of $f$ are equivariant, we may in fact replace $f$ by one of its derivatives in the above arguments. Consequently, we arrive at the claim, namely 
	\begin{align*}
		f \in \Hoereq{q+q'} \bigl ( \mathcal{B}(\Hil,\Hil') \bigr ) 
		. 
	\end{align*}
\end{proof}
Sometimes we will need to consider products of periodic \emph{scalar-valued} symbols with equivariant symbols. Fortunately, we may view those as equivariant or periodic operator-valued Hörmander symbols. 
\begin{lemma}\label{equivariant_calculus:magnetic_PsiDOs:lem:periodic_scalar_valued_symbols_embedding_periodic_equivariant_symbol_classes}
	Suppose $f \in \Hoer{m}_{0,0}(\C)$ is periodic in the sense that 
	\begin{align*}
		f(r,k - \gamma^*) = f(r,k) 
		&&
		\forall (r,k) \in T^* \R^d 
		, \; 
		\gamma^* \in \Gamma^* 
	\end{align*}
	holds. Then we may view $f$ either an equivariant (with respect to any group action $\tau$ with tempered growth) or a periodic Hörmander symbol, 
	\begin{align*}
		f \in \Hoer{m}_{0,0}(\C) \mbox{ periodic}
		\; \Longrightarrow \; 
		f \in \Hoereq{2q} \bigl ( \mathcal{B}(\Hil) \bigr ) 
		, 
		\\
		f \in \Hoer{m}_{0,0}(\C) \mbox{ periodic}
		\; \Longrightarrow \; 
		f \in \Hoer{0}_{\per} \bigl ( \mathcal{B}(\Hil) \bigr ) 
		. 
	\end{align*}
\end{lemma}
\begin{proof}
	Any scalar-valued Hörmander symbol $f \in \Hoer{m}_{0,0}(\C)$ defines an element of $\Hoer{m}_{0,0} \bigl ( \mathcal{B}(\Hil) \bigr )$ by tensoring on the identity, 
	\begin{align*}
		f \cong f \otimes \id_{\Hil} \in \Hoer{m}_{0,0} \bigl ( \mathcal{B}(\Hil) \bigr ) 
		. 
	\end{align*}
	Since the identity commutes with any of the $\tau(\gamma^*)$, we may either regard them as equivariant or periodic (in the special case $\tau : \gamma^* \mapsto \id_{\Hil}$). The order of the Hörmander classes follow from Lemma~\ref{equivariant_calculus:magnetic_PsiDOs:lem:bound_Hoermander_order_by_tau_orders}, \ie either $q + q = 2q$ or $0$.
\end{proof}
However, there is one subtle difference in how we treat asymptotic expansions in equivariant versus ordinary symbol classes, where typically $\rho > \delta$. Given that derivatives of equivariant Hörmander symbols lie in the equivariant symbol class of the \emph{same} order, that is, $\rho = 0 = \delta$, we need to deal with asymptotic Hörmander symbols and formal sums differently. As mentioned in Remark~\ref{operator_valued_calculus:rem:lack_resummation_results_rho_equal_delta}, we are not aware of an existence proof of a resummation in case $\rho = \delta$. Hence, we are at present unable to prove an analog of Lemma~\ref{operator_valued_calculus:lem:existence_resummation_formal_sum} for equivariant symbols. For applications this means we need to distinguish between formals sums obtained from \eg perturbation expansions and asymptotic symbols. 

Still, once we adapt our notation somewhat, we can still reformulate results involving asymptotic expansions (\eg of the Weyl product $\Weyl$ or for a parametrix). 
\begin{definition}[Weaker definition of $\order(\epsilon^{\infty})$]\label{equivariant_calculus:defn:weaker_definition_O_epsilon_infty}
	\begin{enumerate}[(1)]
		\item We abbreviate the formal sum $\sum_{n = 0}^{\infty} \epsilon^n \, f_n \in \Sigma \Hoereq{m} \bigl ( \mathcal{B}(\Hil,\Hil') \bigr )$ with $f_{\epsilon}$ and also write $f_{\epsilon} \asymp \sum_{n = 0}^{\infty} \epsilon^n \, f_n$ as 
		\begin{align*}
			f_{\epsilon} = \sum_{n = 0}^{\infty} \epsilon^n \, f_n + \order(\epsilon^{\infty}) 
			. 
		\end{align*}
		\item When $f_{\epsilon}$ and $g_{\epsilon}$ are two formal sums, an asymptotic symbol and a formal sum, or two asymptotic sums, then 
		\begin{align*}
			f_{\epsilon} = g_{\epsilon} + \order(\epsilon^{\infty}) 
		\end{align*}
		is defined to mean that the coefficients $f_n = g_n$ of the asymptotic expansions of $f_{\epsilon} = \sum_{n = 0}^{\infty} \epsilon^n \, f_n + \order(\epsilon^{\infty})$ and $g_{\epsilon} = \sum_{n = 0}^{\infty} \epsilon^n \, g_n + \order(\epsilon^{\infty})$ agree for all $n \in \N_0$. 
	\end{enumerate}
\end{definition}
Item~(1) erases the notational distinction between asymptotic symbols and formal sums. The advantage is that it shortens expressions like the almost-projection property 
\begin{align*}
	\pi_{\epsilon} \Weyl \pi_{\epsilon} = \pi_{\epsilon} + \order(\epsilon^{\infty})
\end{align*}
to a single, easily comprehensible line rather than a rather cumbersome expression for the formal sums. 

Indeed, in many applications it suffices to compute $f_{\epsilon}$ to some finite, but \emph{arbitary} order $N$ and consider the \emph{finite} resummation $f^{(N)}_{\epsilon} = \sum_{n = 0}^N \epsilon^n \, f_n + \order(\epsilon^{N+1})$ of $f_{\epsilon} = f^{(N)}_{\epsilon} + \order(\epsilon^{N+1})$. Seeing as we are only dealing with finite sums here, questions of convergence (\eg in the Fréchet topology of $\Hoereq{m} \bigl ( \mathcal{B}(\Hil,\Hil') \bigr )$) do not arise. And since $N \in \N_0$ can be chosen arbitrarily large, this amounts to working with formal sums. 

We caution the reader that there \emph{are} circumstances where we actually need to work with asymptotic symbols obtained by resumming a formal series. 

\emph{Throughout the remainder of this section, we shall always use the weaker definition of $\order(\epsilon^{\infty})$.}

\subsubsection{Equivariant magnetic pseudodifferential operators on equivariant $L^2$ and Sobolev spaces} 
\label{equivariant_calculus:calculus:boundedness_results}
So far we have regarded equivariant pseudodifferential operators as continuous maps between equivariant distribution spaces. For applications we want to restrict them further, in the simplest case we want to know whether the restriction of $\Opeq^A(f)$ to \eg 
\begin{align*}
	L^2_{\eq}(\R^d,\Hil) \subseteq \Schwartz^*_{\eq}(\R^d,\Hil) 
\end{align*}
defines a continuous, that is, bounded operator 
\begin{align*}
	\Opeq^A(f) : L^2_{\eq}(\R^d,\Hil) \longrightarrow L^2_{\eq}(\R^d,\Hil')
\end{align*}
between equivariant $L^2$ spaces. Like before, such a boundedness result will lead to a boundedness result between equivariant magnetic Sobolev spaces. 
\begin{theorem}\label{equivariant_calculus:thm:Calderon_Vaillancourt}
	Suppose $f \in \SemiHoereq{q + q'} \bigl ( \mathcal{B}(\Hil,\Hil') \bigr )$ is an equivariant Hörmander symbol. 
	\begin{enumerate}[(1)]
		\item Then $\Opeq^A(f) : H^{q + q'}_{\eq,\Fourier,A}(\R^d,\Hil) \longrightarrow L^2_{\eq}(\R^d,\Hil)$ defines a bounded operator. 
		%
		\item When $q + q' = 0$ or the components of the magnetic vector potential $A_j \in L^{\infty}(\R^d,\R)$, $j = 1 , \ldots , d$, are all bounded, then 
		\begin{align*}
			\Opeq^A(f) : L^2_{\eq}(\R^d,\Hil) \longrightarrow L^2_{\eq}(\R^d,\Hil')
		\end{align*}
		defines a bounded operator between equivariant $L^2$-spaces. 
	\end{enumerate}
\end{theorem}
\begin{proof}
	Corollary~\ref{operator_valued_calculus:cor:boundedness_magnetic_Sobolev_spaces} to the Calderón-Vaillancourt Theorem~\ref{operator_valued_calculus:thm:Calderon_Vaillancourt} for operator-valued symbols ensures that the operator $\Op_{\Fourier}^A(f) : H^{q + q'}_{\Fourier,A}(\R^d,\Hil) \longrightarrow L^2(\R^d,\Hil)$ is bounded. Hence, Proposition~\ref{equivariant_calculus:cor:equivalence_magnetic_equivariant_operators} applies for $F^A = \Opeq^A(f)$, and gives us claim (1) and the first half of claim~(2). 
	
	When the components of the vector potential $A_j \in L^{\infty}(\R^d,\R)$, $j = 1 , \ldots , d$, are bounded, then the $\KA_j = \hat{k}_j - A_j(\R)$ all define bounded operators. Consequently, the weights that enter the Sobolev norms are just bounded operators with bounded inverses, which means $H^{q + q'}_{\eq,\Fourier,A}(\R^d,\Hil^{(\prime)})$ and $L^2_{\eq}(\R^d,\Hil^{(\prime)})$ agree as Banach spaces, regardless of the value of $q + q'$. 
\end{proof}
Like in the case of magnetic pseudo\-differential operators defined from scalar- or operator-valued symbols, this result implies norm bounds if we consider $\Opeq^A(f)$ as an operator between equivariant Sobolev spaces. 
\begin{corollary}\label{equivariant_calculus:cor:Calderon_Vaillancourt}
	Any $f \in \Hoereq{q+q'} \bigl ( \mathcal{B}(\Hil , \Hil') \bigr )$ defines a bounded equivariant magnetic pseudo\-differential operator between magnetic equivariant Sobolev spaces of any order $s \geq q+q' \geq 0$, 
	\begin{align*}
		\Opeq^A(f) : H^s_{\eq,\Fourier,A}(\R^d,\Hil) \longrightarrow H^{s-q-q'}_{\eq,\Fourier,A}(\R^d,\Hil')
		.
	\end{align*}
\end{corollary}
The proof is a blend of the proofs of Theorem~\ref{equivariant_calculus:thm:Calderon_Vaillancourt} and Corollary~\ref{operator_valued_calculus:cor:boundedness_magnetic_Sobolev_spaces}: we can express the appropriate operator norm of $\Opeq^A(f)$ similarly to  equation~\eqref{operator_valued_calculus:eqn:proof_corollary_Calderon_Vaillancourt_Sobolev_norm_L2_norm} and then invoke Theorem~\ref{equivariant_calculus:thm:Calderon_Vaillancourt}. Strictly speaking, it involves a product of equivariant Hörmander symbols, but that part of the argument is taken care of in Theorem~\ref{equivariant_calculus:thm:magnetic_Weyl_product_equivariant_Hoermander_symbols} below. 
\medskip

\noindent
Our exposition so far suggests that all results transfer directly from the operator-valued case covered in Section~\ref{operator_valued_calculus} to the equivariant context. Unfortunately, that is not always true. One source of trouble is the common assumption $\rho > \delta$ in the scalar-valued or non-equivariant, operator-valued cases; this condition is necessarily violated for equivariant magnetic $\Psi$DOs as the quasi-periodicity condition~\ref{equivariant_calculus:eqn:equivariance_condition_symbols} implies $\rho = \delta = 0$. We will showcase this with the adaptation of Theorem~\ref{operator_valued_calculus:thm:selfadjointness_elliptic_symbols}: 
\begin{theorem}\label{equivariant_calculus:thm:selfadjointness_elliptic_symbols_weak_result}
	Assume the two Hilbert spaces $\Hil = \Hil'$ agree. 
	Further, we impose \emph{at least one} of the following two conditions: 
	\begin{enumerate}[(a)]
		\item The group actions $\tau^{(\prime)}$ remain uniformly bounded as $\sabs{\gamma^*} \rightarrow \infty$ (\ie $q = 0$). 
		\item The components of the magnetic vector potential $A \in L^{\infty}(\R^d,\R^d)$ are bounded. 
	\end{enumerate}
	Suppose $h \in \Hoereq{2q} \bigl ( \mathcal{B}(\Hil) \bigr )$ is an equivariant Hörmander symbol of order $2q$ that takes values in the selfadjoint operators,
	\begin{align*}
		h(r,k)^* = h(r,k) 
		&&
		\forall (r,k) \in T^* \BZ
		. 
	\end{align*}
	In case $q > 0$ we assume in addition that $h$ is elliptic of order $2q$ in the sense of Definition~\ref{operator_valued_calculus:defn:elliptic_symbols}. 
	
	Then $\Opeq^A(h) = \Opeq^A(h)^*$ defines a \emph{bounded} selfadjoint operator on $L^2_{\eq}(\R^d,\Hil)$. 
\end{theorem}
Compared with Theorem~\ref{operator_valued_calculus:thm:selfadjointness_elliptic_symbols} we have had to add the assumptions $\Hil = \Hil'$, and $q = 0 = q'$ or that the components of $A$ are bounded. Both restrictions are rather strong and are designed to achieve the same thing, namely ensuring that the relevant equivariant magnetic Sobolev space
\begin{align*}
	H^{2q}_{\eq,\Fourier,A}(\R^d,\Hil) = L^2_{\eq}(\R^d,\Hil)
\end{align*}
agrees with the equivariant $L^2$-space as Banach spaces (\cf the proof of Theorem~\ref{equivariant_calculus:thm:Calderon_Vaillancourt}~(2)). 
\begin{proof}
	The proof of Theorem~\ref{equivariant_calculus:thm:selfadjointness_elliptic_symbols_weak_result} amounts to showing that $\Opeq^A(h)$ is a bounded symmetric operator. Boundedness follows from Theorem~\ref{equivariant_calculus:thm:Calderon_Vaillancourt}~(2). 
	
	To prove symmetry, we leverage the symmetry of $\Op^A(h)$ (defined with the \emph{non-equivariant} quantization map from Section~\ref{operator_valued_calculus}) on the dense set 
	\begin{align*}
		\Schwartz(\R^d,\Hil) \subset L^2(\R^d,\Hil) 
		. 
	\end{align*}
	The essential ingredients are $\Op^A(h)^* = \Op^A(h^*)$ and the fact that $h = h^*$ takes values in the selfadjoint operators. 
	
	We then define the set of smooth, equivariant $\Hil$-valued functions, which vanish in a vicinity of the boundary $\partial \BZ$ of the fundamental cell, 
	\begin{align*}
		\Cont^{\infty}_{\eq,0}(\R^d,\Hil) := \Bigl \{ \psi \in \Cont^{\infty}(\R^d,\Hil) \; \; \big \vert \; \; &\psi(k - \gamma^*) = \tau(\gamma^*) \, \psi(k) \; \forall k \in \R^d , \, \gamma^* \in \Gamma^* 
		\Bigr . \\
		&\Bigl. 
		\mbox{$\psi$ vanishes in a neighborhood of $\partial \BZ$}
		\Bigr \}
		.  
	\end{align*}
	Evidently, $\Cont^{\infty}_{\eq,0}(\R^d,\Hil) \subset L^2(\R^d,\Hil)$ lies densely in the equivariant $L^2$-space. Moreover, the map $\imath_{\chi}$ from Lemma~\ref{appendix:equivariant_operators:lem:continuity_inclusion_map_L2_BZ_gamma_ast_L2_Rd} embeds $\Cont^{\infty}_{\eq,0}(\R^d,\Hil) \subseteq \Schwartz(\R^d,\Hil)$ into the Schwartz functions. Thus, the symmetry of $\Op^A(h) \vert_{\Schwartz(\R^d,\Hil)}$ implies the symmetry of $\Opeq^A(h) \vert_{\Cont^{\infty}_{\eq,0}(\R^d,\Hil)}$. 
	
	As a bounded, symmetric operator, its extension $\Opeq^A(h)^* = \Opeq^A(h)$ to $L^2_{\eq}(\R^d,\Hil)$ is selfadjoint. 
\end{proof}
In practice, condition~(a) limits us to \eg $\Hil = L^2(\T^d,\C^n)$. The group actions on magnetic Sobolev spaces $H^m_A(\T^d)$ are of order $m$ (\cf Lemma~\ref{setting:lem:magnetic_Sobolev_norm_operator_tau}), which means $q = 0$ implies $m = 0$. 

Condition~(b), $A \in L^{\infty}(\R^d,\R^d) \cap \Cont^{\infty}_{\mathrm{pol}}(\R^d,\R^d)$, excludes constant magnetic fields $B = \mathrm{const.}$, which is of particular interest in the study of magnetic systems. 

Conditions (a) and (b) allow us to skip one crucial step in the proof of Theorem~\ref{operator_valued_calculus:thm:selfadjointness_elliptic_symbols} and \cite[Theorem~5.1]{Iftimie_Mantoiu_Purice:magnetic_psido:2006}. In order to identify the domain as the relevant magnetic Sobolev space means we needed to prove that the graph norm of $\Op^A(h)$ is equivalent to the magnetic Sobolev norm of order $q + q'$. The upper bound follows from Corollary~\ref{operator_valued_calculus:thm:Calderon_Vaillancourt} — which extends to Corollary~\ref{equivariant_calculus:cor:Calderon_Vaillancourt}. The issue is the lower bound. To obtain that, we needed to obtain a parametrix for $\Op^A(h)$. Initially, the parametrix was constructed as a formal sum (\cf \eg Corollary~\ref{operator_valued_calculus:cor:existence_parametrix}). The assumption $\rho > \delta$ entered the proof in several ways: first of all, we had to choose a resummation for this formal sum. And this resummation is only guaranteed to exist when subsequent terms of the parametrix expansion belong to Hörmander classes of ever smaller order $m_n \rightarrow -\infty$. The latter implies $\rho \gneq \delta$. Secondly, we needed this condition to ensure that the remainder 
\begin{align*}
	R^A := \Op^A \bigl ( h^{(-1)_{\epsilon}} \bigr ) \, \Op^A(h) - \id_{\Hil} \in S^{-\infty} \bigl ( \mathcal{B}(\Hil) \bigr )
\end{align*}
is the magnetic quantization of a smoothing symbol, and thus, not only a bounded operator on $L^2(\R^d,\Hil) \longrightarrow L^2(\R^d,\Hil)$, but small with respect to $\Op^A(h)$. 

In the context of equivariant operators, the order of the symbol is completely determined by the order of growth of $\tau$ and $\tau'$ (\cf Lemma~\ref{equivariant_calculus:magnetic_PsiDOs:lem:bound_Hoermander_order_by_tau_orders}). So the remainder for equivariant symbols cannot be not small with respect to $\Op^A(h)$, it must be of exactly the same order. Hence, \emph{even if} we could construct a resummation for the parametrix's formal sum, the proof would still fail. 

One may try to look for Beals' Commutator Criterion for help, specifically an equivariant analog of Theorem~\ref{operator_valued_calculus:thm:existence_Moyal_resolvent}. The idea is to replace the parametrix with the (true) resolvent defined from the symbol $(h - z)^{(-1)_{\Weyl}}$ for some $z \not\in \sigma \bigl ( \Op^A(h) \bigr )$. However, that would only work if we could localize the spectrum in the complex plane, \ie show $\sigma \bigl ( \Op^A(h) \bigr ) \subseteq \R$. The latter \emph{requires} selfadjointness. Indeed, the conditions in Theorem~\ref{operator_valued_calculus:thm:existence_Moyal_resolvent} guarantee $\Op^A(h) = \Op^A(h)^*$. To summarize, going this route means entering a self-referential loop: in order to prove selfadjointness, we need selfadjointness. 

At present, the only alternative known to us is to prove the selfadjointness of $\Opeq^A(h)$ directly with functional analytic methods. 

\subsubsection{Composition of equivariant magnetic pseudodifferential operators} 
\label{equivariant_calculus:calculus:Weyl_product}
The extension of the magnetic Weyl product to the equivariant context is a straightforward invokation of Meta Theorem~\ref{equivariant_calculus:meta_theorem:extension_operator_valued_to_equivariant}. We have already covered in our motivating example, Proposition~\ref{equivariant_calculus:prop:intertwining_operator_product_Weyl_product_Opeq}, that also in the equivariant context the magnetic Weyl product intertwines $\Opeq^A$ and the operator product, 
\begin{align*}
	\Opeq^A(f \Weyl g) = \Opeq^A(f) \, \Opeq^A(g) 
	. 
\end{align*}
It remains to show that $\Weyl$ preserves \emph{equivariant} symbol classes. 
\begin{theorem}\label{equivariant_calculus:thm:magnetic_Weyl_product_equivariant_Hoermander_symbols}
	\begin{enumerate}[(1)]
		\item The magnetic Weyl product defines a continuous bilinear map 
		\begin{align*}
			\Weyl : \Hoereq{m_1} \bigl ( \mathcal{B}(\Hil',\Hil'') \bigr ) \times \Hoereq{m_2} \bigl ( \mathcal{B}(\Hil,\Hil') \bigr ) \longrightarrow \Hoereq{m_1 + m_2} \bigl ( \mathcal{B}(\Hil,\Hil'') \bigr )
		\end{align*}
		between equivariant Hörmander spaces. 
		\item In fact, if $\tau$, $\tau'$ and $\tau''$ have tempered growth in the sense of Definition~\ref{setting:defn:order_tau} of orders $q$, $q'$ and $q''$, then $m_1 \geq q' + q''$ and $m_2 \geq q + q'$ has to hold and we can sharpen the statement from (1) to 
		\begin{align*}
			\Weyl : \Hoereq{q' + q''} \bigl ( \mathcal{B}(\Hil',\Hil'') \bigr ) \times \Hoereq{q + q'} \bigl ( \mathcal{B}(\Hil,\Hil') \bigr ) \longrightarrow \Hoereq{q + q''} \bigl ( \mathcal{B}(\Hil,\Hil'') \bigr )
		\end{align*}
		being a bilinear continuous map. Specifically, the order of the target space only depends on the order of growth of $\tau$ and $\tau''$. 
		\item The Weyl product between periodic, scalar-valued and equivariant, operator-valued symbols is well-defined and gives rise to the bilinear continuous maps 
		\begin{align*}
			&\Weyl : \Hoer{0}_{\mathrm{per}}(\C) \times \Hoereq{q + q'} \bigl ( \mathcal{B}(\Hil,\Hil') \bigr ) \longrightarrow \Hoereq{q + q'} \bigl ( \mathcal{B}(\Hil,\Hil') \bigr )
			, 
			\\
			&\Weyl : \Hoereq{q + q'} \bigl ( \mathcal{B}(\Hil,\Hil') \bigr ) \times \Hoer{0}_{\per}(\C) \longrightarrow \Hoereq{q + q'} \bigl ( \mathcal{B}(\Hil,\Hil') \bigr )
			. 
		\end{align*}
		\item The asymptotic expansions from Section~\ref{operator_valued_calculus:Hoermander_symbols:Weyl_product} also preserve equivariance, and all formulas for the explicit terms hold verbatim, just seen as maps between the appropriate \emph{equivariant} Hörmander symbol classes of formal sums. Specifically, the asymptotic expansions of Weyl define are continuous bilinear maps 
		\begin{align*}
			\Weyl : \Sigma \Hoereq{q' + q''} \bigl ( \mathcal{B}(\Hil',\Hil'') \bigr ) \times \Sigma \Hoereq{q + q'} \bigl ( \mathcal{B}(\Hil,\Hil') \bigr ) \longrightarrow \Sigma \Hoereq{q + q''} \bigl ( \mathcal{B}(\Hil,\Hil'') \bigr )
			. 
		\end{align*}
	\end{enumerate}
\end{theorem}
\begin{proof}
	As stated before, (1) and (4) are straightforward invocations of Meta Theorem~\ref{equivariant_calculus:meta_theorem:extension_operator_valued_to_equivariant} because equivariant Hörmander classes can be understood as the intersection of the operator-valued Hörmander classes with the equivariant Moyal space (\cf equation~\eqref{equivariant_calculus:eqn:equivariant_Hoermander_class_intersection_usual_Hoermander_class_equivariant_Moyal_space}). 
	
	Item~(3) is a direct consequence of item~(2) combined with Lemma~\ref{equivariant_calculus:magnetic_PsiDOs:lem:periodic_scalar_valued_symbols_embedding_periodic_equivariant_symbol_classes}, which allows us to view any $f \in \Hoer{0}_{\per}(\C) \subset \Hoereq{2q^{(\prime)}} \bigl ( \mathcal{B}(\Hil^{(\prime)}) \bigr )$ as an equivariant, operator-valued symbol. 
	
	Only item~(2) remains to be proven. Thanks to Lemma~\ref{equivariant_calculus:magnetic_PsiDOs:lem:bound_Hoermander_order_by_tau_orders}, we know that $f$ and $g$ are equivariant Hörmander symbols of order $q + q'$ and $q' + q''$, respectively. Thus, we may assume $m_1 = q + q'$ and $m_2 = q' + q''$ without loss of generality. 
	
	Ignoring equivariance, we infer from (1) that the product 
	\begin{align*}
		f \Weyl g \in S^{m_1 + m_2}_{\eq} \bigl ( \mathcal{B}(\Hil,\Hil') \bigr )
	\end{align*}
	is a Hörmander symbol of order $m_1 + m_2 \geq q + 2 q' + q''$. However, applying Lemma~\ref{equivariant_calculus:magnetic_PsiDOs:lem:bound_Hoermander_order_by_tau_orders} a second time, we see that the order of the symbol is determined by the sum of the orders of $\tau$ and $\tau''$, which is $q + q''$. This finishes the proof. 
\end{proof}

\subsection{Extension of some more advanced results} 
\label{equivariant_calculus:advanced_results}
Rather than trying to compile an exhaustive library of more advanced results, the purpose of this section is to showcase how to put Meta Theorem~\ref{equivariant_calculus:meta_theorem:extension_operator_valued_to_equivariant} into practice. Specifically, we will establish an equivariant version of Beals' Commutator Criterion and a functional calculus for equivariant magnetic $\Psi$DOs.

\subsubsection{Beals' Commutator Criterion to identify equivariant magnetic $\Psi$DOs} 
\label{equivariant_calculus:advanced_results:commutator_criteria}
Proving an analog of Beals' Commutator Criterion~\ref{operator_valued_calculus:thm:Beals_commutator_criterion} is straightforward once we have introduced the right notation. However, we will directly skip ahead and furnish a proof of the more general Corollary~\ref{operator_valued_calculus:cor:Beals_commutator_criterion} instead. 

First of all, equivariant Hörmander symbols are necessarily of type $(0,0)$: the order of a symbol and all its derivatives is determined by the growth of the group actions $\tau$ and $\tau'$ on initial and target spaces $\Hil$ and $\Hil'$, respectively. 

Secondly, we need to define the equivariant analog 
\begin{align}
	\mathfrak{A}^B_{\eq} \bigl ( \mathcal{B}(\Hil,\Hil') \bigr ) := \bigl ( \Opeq^A \bigr )^{-1} \Bigl ( \mathcal{B} \bigl ( L^2_{\eq}(\R^d,\Hil) , L^2_{\eq}(\R^d,\Hil') \bigr ) \Bigr ) 
\end{align}
of the Banach space $\mathfrak{A}^B \bigl ( \mathcal{B}(\Hil,\Hil') \bigr )$; as before, we may view 
\begin{align*}
	\mathfrak{A}^B_{\eq} \bigl ( \mathcal{B}(\Hil,\Hil') \bigr ) := \bigl ( \Opeq^A \bigr )^{-1} \Bigl ( \mathcal{B} \bigl ( L^2_{\eq}(\R^d,\Hil) , L^2(\R^d,\Hil') \bigr ) \Bigr ) \subseteq \Schwartz^*_{\eq} \bigl ( \mathcal{B}(\Hil,\Hil') \bigr ) 
\end{align*}
as being composed of equivariant tempered distributions, courtesy of the Schwartz Kernel Theorem~\ref{operator_valued_calculus:lem:Schwartz_kernel_theorem}: we could either choose a direct approach and sandwich 
\begin{align*}
	\Cont^{\infty}_{\eq,0}(\R^d,\Hil) \subset L^2_{\eq}(\R^d,\Hil) \subset \Schwartz^*_{\eq}(\R^d,\Hil)
\end{align*}
in between smooth equivariant functions that vanish near the boundaries of the fundamental cells and the equivariant tempered distributions. Instead, we will exploit that the equivariant $\Psi$DO $\Opeq^A(f) = \Op_{\Fourier}^A(f) \vert_{\Schwartz^*_{\eq}(\R^d,\Hil)}$ can be seen as the restriction of the ordinary magnetic pseudodifferential operator to equivariant tempered distributions (Proposition~\ref{equivariant_calculus:prop:equivariant_distributions_define_equivariant_PsiDOs}). What is more, know that bounded \emph{equivariant} operators define bounded operators between non-equivariant $L^2$-spaces (Proposition~\ref{equivariant_calculus:prop:equivalence_boundedness_equivariant_operators} for $q + q' = 0$). Hence, we may think of 
\begin{align}
	\mathfrak{A}^B_{\eq} \bigl ( \mathcal{B}(\Hil,\Hil') \bigr ) := \bigl ( \Opeq^A \bigr )^{-1} \Bigl ( \mathcal{B} \bigl ( L^2_{\eq}(\R^d,\Hil) , L^2(\R^d,\Hil') \bigr ) \Bigr ) \subseteq \mathfrak{A}^B \bigl ( \mathcal{B}(\Hil,\Hil') \bigr ) 
	\label{equivariant_calculus:eqn:equivariant_Banach_space_A_B_eq_embedded_in_non_equivariant_Banach_space_A_B}
\end{align}
as a closed subspace.

For two composable such Banach spaces we can pull back a product $\Weyl$ and an adjoint ${}^*$ to the level of equivariant tempered distributions as before; clearly, completeness is inherited again from the parent space $\mathcal{B} \bigl ( L^2_{\eq}(\R^d,\Hil) , L^2(\R^d,\Hil') \bigr )$. 

Lastly, non-equivariant and equivariant position and momentum operators are the quantizations of the exact same symbols, \eg 
\begin{align*}
	Q_j &= \Op^A(x_j) 
	\quad \mbox{versus} \quad 
	\Reps_j = \Opeq^A(x_j)
	. 
\end{align*}
Consequently, the derivations 
\begin{align*}
  \ad_{\Reps}^a \, \ad_{\KA}^{\alpha}(F) := \ad_{\KA_1}^{a_1} \circ \cdots \circ \ad_{\KA_d}^{a_d} \circ \ad_{\Reps_1}^{\alpha^1} \circ \cdots \circ \ad_{\Reps_d}^{\alpha_d}(F) 
\end{align*}
on the level of operators pull back to the exact same derivations $\partial_{(r,k)}^{(a,\alpha)} f$ from \eqref{operator_valued_calculus:eqn:Moyal_derivations_for_commutator_criteria}. Only the notation differs, we have replaced $x$ by $r$ and $\xi$ by $k$, \eg 
\begin{align*}
	\ad_{k_j}(f) = [k_j , f]_{\Weyl} = k_j \Weyl f - f \Weyl k_j 
	 . 
\end{align*}
That means that thanks to the inclusion we have proven 
\begin{theorem}[Beals' Commutator Criterion for equivariant $\Psi$DOs]\label{equivariant_calculus:thm:Beals_commutator_criterion}~\\
	An equivariant bounded operator 
	\begin{align*}
		F = \Opeq^A(f) \in \mathcal{B} \bigl ( H^{q + q'}_{\eq,\Fourier,A}(\R^d,\Hil) , L^2_{\eq}(\R^d,\Hil') \bigr ) 
	\end{align*}
	is a magnetic pseudodifferential operator associated to a Hörmander symbol 
	\begin{align*}
		f \in S^{q + q'}_{\eq} \bigl ( \mathcal{B}(\Hil,\Hil') \bigr ) 
	\end{align*}
	if and only if  
	\begin{align}
		w_{-(q + q')} \Weyl \partial^{(a,\alpha)}_{(r,k)} f \in \mathfrak{A}^B_{\eq} \bigl ( \mathcal{B}(\Hil,\Hil') \bigr ) 
		&&
		\forall a , \alpha \in \N_0^d 
		. 
		\label{equivariant_calculus:eqn:Beals_commutator_criterion}
	\end{align}
\end{theorem}
\begin{proof}
	The forward direction where we start with the assumption that $f \in \Hoereq{q + q'} \bigl ( \mathcal{B}(\Hil,\Hil') \bigr )$ follows directly from the boundedness result, Theorem~\ref{equivariant_calculus:thm:Calderon_Vaillancourt}. 
	
	Let us prove non-trivial direction: suppose equation~\eqref{equivariant_calculus:eqn:Beals_commutator_criterion} holds true. Thanks to definition \eqref{equivariant_calculus:eqn:equivariant_Banach_space_A_B_eq_embedded_in_non_equivariant_Banach_space_A_B} and the comments preceding the statement of Theorem~\ref{equivariant_calculus:thm:Beals_commutator_criterion}, we may view $f$ as an element of the non-equivariant Banach space $\mathfrak{A}^B \bigl ( \mathcal{B}(\Hil,\Hil') \bigr )$ from Section~\ref{operator_valued_calculus:Hoermander_symbols:commutator_criteria}. Put another way, the assumptions in Corollary~\ref{operator_valued_calculus:cor:Beals_commutator_criterion} are satisfied and we deduce $f \in S^{q + q'}_{0,0} \bigl ( \mathcal{B}(\Hil,\Hil') \bigr )$ is a Hörmander symbol of order $q + q'$ and type $(0,0)$. 
	
	What is more, $f$ and all its derivations — and thus, all its derivatives — are equivariant as elements of 
	\begin{align*}
		\partial_{(r,k)}^{(a,\alpha)} f \in \mathfrak{A}^B_{\eq} \bigl ( \mathcal{B}(\Hil,\Hil') \bigr ) \subset \Schwartz^*_{\eq} \bigl ( \mathcal{B}(\Hil,\Hil') \bigr ) 
		&&
		\forall a , \alpha \in \N_0^d
		. 
	\end{align*}
	Characterization~\eqref{equivariant_calculus:eqn:equivariant_Hoermander_class_intersection_usual_Hoermander_class_equivariant_Moyal_space} of equivariant Hörmander symbols therefore tells us that $f \in \Hoereq{q + q'} \bigl ( \mathcal{B}(\Hil,\Hil') \bigr )$ is in fact an \emph{equivariant} Hörmander symbol. 
\end{proof}
%

\subsubsection{Results on inversion and functional calculus} 
\label{equivariant_calculus:advanced_results:inversion_and_functional_calculus}
With the Commutator Criterion~\ref{equivariant_calculus:thm:Beals_commutator_criterion} in hand, we proceed as in Section~\ref{operator_valued_calculus:Hoermander_symbols:inversion_and_functional_calculus}. To avoid being repetitive we will omit proofs either largely or in their entirety. 

The Moyal resolvent is again the preimage of the resolvent operator under the quantization, 
\begin{align*}
	\bigl ( \Opeq^A(h) - z \bigr )^{-1} = \Opeq^A \bigl ( (h - z)^{(-1)_{\Weyl}} \bigr ) 
	. 
\end{align*}
For suitable functions we can then define functional calculi, \eg via the Helffer-Sjöstrand formula~\eqref{operator_valued_calculus:eqn:Helffer_Sjoestrand_functional_calculus} or the complex integral \eqref{operator_valued_calculus:eqn:holomorphic_functional_calculus}. Clearly, the central point is invertibility. 
\begin{theorem}[Invertibility]\label{equivariant_calculus:thm:invertiblity}
	Suppose the Hilbert spaces $\Hil$ and $\Hil'$ are isomorphic. 
	Suppose $f \in \Hoereq{q + q'} \bigl ( \mathcal{B}(\Hil,\Hil') \bigr )$ is invertible either as an element of $\mathfrak{A}^B \bigl ( \mathcal{B}(\Hil,\Hil') \bigr )$ ($q + q' = 0$) or $\mathcal{M}^B \bigl ( \mathcal{B}(\Hil,\Hil') \bigr )$ ($q + q' \geq 0$). 
	\begin{enumerate}[(1)]
		\item Then also its inverse $f^{(-1)_{\Weyl}} \in \Hoereq{q + q'} \bigl ( \mathcal{B}(\Hil',\Hil) \bigr )$ is a Hörmander symbol of the \emph{same} order $q + q' \geq 0$. 
		\item The products of $f$ and $f^{(-1)_{\Weyl}}$ belong to the equivariant Hörmander spaces 
		\begin{align*}
			f \Weyl f^{(-1)_{\Weyl}} &= \id_{\Hil'} \in \Hoereq{2q'} \bigl ( \mathcal{B}(\Hil') \bigr ) 
			, 
			\\
			f^{(-1)_{\Weyl}} \Weyl f &= \id_{\Hil} \in \Hoereq{2q} \bigl ( \mathcal{B}(\Hil) \bigr ) 
			. 
		\end{align*}
	\end{enumerate}
\end{theorem}
\begin{proof}
	\begin{enumerate}[(1)]
		\item The fact that $f^{(-1)_{\Weyl}} \bigl ( \mathcal{B}(\Hil',\Hil) \bigr )$ exists as a Hörmander symbol follows by applying Theorem~\ref{operator_valued_calculus:thm:invertiblity} and checking that equivariance is indeed preserved. 
		
		The only question is as to the order of the symbol $f^{(-1)_{\Weyl}}$, and that is covered by Lemma~\ref{equivariant_calculus:magnetic_PsiDOs:lem:bound_Hoermander_order_by_tau_orders}. 
		\item Theorem~\ref{equivariant_calculus:thm:magnetic_Weyl_product_equivariant_Hoermander_symbols} tells us these products exist in $S^{2(q + q')}_{0,0} \bigl ( \mathcal{B}(\Hil') \bigr )$ and $S^{2(q + q')}_{0,0} \bigl ( \mathcal{B}(\Hil) \bigr )$, respectively. However, Lemma~\ref{equivariant_calculus:magnetic_PsiDOs:lem:bound_Hoermander_order_by_tau_orders} informs the actual order is lower, $q'$ and $q$, respectively. 
	\end{enumerate}
\end{proof}
A direct corollary to Theorem~\ref{equivariant_calculus:thm:invertiblity} is the existence of Moyal resolvents. 
\begin{theorem}[Existence of Moyal resolvent]\label{equivariant_calculus:thm:existence_Moyal_resolvent}
	Suppose $\Hil \hookrightarrow \Hil'$ can be continuously and densely injected into $\Hil'$. 
	
	Assume we are given a symbol $f \in \Hoereq{q + q'} \bigl ( \mathcal{B}(\Hil,\Hil') \bigr )$ and $f - z$ is an invertible element of $\mathcal{M}^B_{\eq} \bigl ( \mathcal{B}(\Hil,\Hil') \bigr )$ for some $z \not\in \sigma \bigl ( \Op^A(f) \bigr )$. 
	\begin{enumerate}[(1)]
		\item Then the Moyal resolvent $(f - z)^{(-1)_{\Weyl}} \in \Hoereq{q + q'} \bigl ( \mathcal{B}(\Hil',\Hil) \bigr )$ exists as a Hörmander symbol of order $q + q' \geq 0$. 
		\item We may view their products 
		\begin{align*}
			(f - z)^{(-1)_{\Weyl}} \Weyl (f - z) &= \id_{\Hil} \in \Hoereq{2q} \bigl ( \mathcal{B}(\Hil) \bigr ) \cap \Hoereq{2q'} \bigl ( \mathcal{B}(\Hil') \bigr ) 
			, 
			\\
			(f - z) \Weyl (f - z)^{(-1)_{\Weyl}} &= \id_{\Hil'} \in \Hoereq{2q} \bigl ( \mathcal{B}(\Hil) \bigr ) \cap \Hoereq{2q'} \bigl ( \mathcal{B}(\Hil') \bigr ) 
			. 
		\end{align*}
		as elements of two different equivariant Hörmander classes. 
	\end{enumerate}
\end{theorem}
\begin{proof}
	We only need to comment on the Hörmander classes of the products: because $\Hil \hookrightarrow \Hil'$ can be viewed as a dense subspace, we may view $\id_{\Hil'}$ as a continuous extension of $\id_{\Hil}$. Equivalently, $\id_{\Hil}$ is the restriction of $\id_{\Hil'}$ to the dense subspace $\Hil$. 
	
	The order of the equivariant Hörmander classes are then determined by the growth of the group actions $\tau$ and $\tau'$ (Lemma~\ref{equivariant_calculus:magnetic_PsiDOs:lem:bound_Hoermander_order_by_tau_orders}). 
\end{proof}
One case where we can ensure the invertibility is when $h$ is an elliptic Hörmander symbol. 
\begin{corollary}
	Suppose $\Hil \hookrightarrow \Hil'$ can be continuously and densely injected into $\Hil'$. 
	
	Assume we are given a symbol $h \in \Hoereq{q + q'} \bigl ( \mathcal{B}(\Hil,\Hil') \bigr )$ that is elliptic in the sense of Definition~\ref{operator_valued_calculus:defn:elliptic_symbols} when $q + q' > 0$ and takes values in the selfadjoint operators. Then for any $z \not\in \sigma \bigl ( \Op^A(h) \bigr )$ the Moyal resolvent 
	\begin{align*}
		(h - z)^{(-1)_{\Weyl}} \in \Hoereq{q + q'} \bigl ( \mathcal{B}(\Hil,\Hil') \bigr ) 
	\end{align*}
	exists as a Hörmander symbol.
\end{corollary}
With the Moyal resolvents in hand, we can proceed to define functional calculi. The first one is for situations where we can use the Helffer-Sjöstrand formula: 
\begin{theorem}[Existence of a functional calculus]
	Suppose $\Hil \hookrightarrow \Hil'$ can be continuously and densely injected into $\Hil'$. 
	Assume we are given a symbol $h \in \Hoereq{q + q'} \bigl ( \mathcal{B}(\Hil,\Hil') \bigr )$ that is elliptic in the sense of Definition~\ref{operator_valued_calculus:defn:elliptic_symbols} when $q + q' > 0$ and takes values in the selfadjoint operators. 
	Then this defines a functional calculus for $\varphi \in \Cont^{\infty}_{\mathrm{c}}(\R,\C)$ via the Hellfer-Sjöstrand formula~\eqref{operator_valued_calculus:eqn:Helffer_Sjoestrand_functional_calculus}, and for each smooth function with compact support 
	\begin{align*}
		\varphi^B(h) \in \Hoereq{q + q'} \bigl ( \mathcal{B}(\Hil',\Hil) \bigr ) 
	\end{align*}
	defines a Hörmander symbol of order $q + q'$. 
\end{theorem}
\begin{proof}
	$\varphi^B(h)$ exists as an operator-valued Hörmander symbol by Theorem~\ref{operator_valued_calculus:thm:functional_calculus_Helffer_Sjoestrand}. Evidently, $\varphi^B(h)$ inherits the equivariance of the Moyal resolvent (Theorem~\ref{equivariant_calculus:thm:existence_Moyal_resolvent}). Therefore, Lemma~\ref{equivariant_calculus:magnetic_PsiDOs:lem:bound_Hoermander_order_by_tau_orders} tells us the order of the \emph{equivariant} symbol equals the sum of the orders of growth of the group actions $\tau$ and $\tau'$. 
\end{proof}
In principle, we can also define a holomorphic functional calculus, and a proof of the following statement is identical to the one above, we merely have to set $\Hil = \Hil'$ and point to Theorem~\ref{operator_valued_calculus:thm:holomorphic_functional_calculus} in the first step instead. 
\begin{theorem}[Holomorphic functional calculus on $\Hoereq{2q} \bigl ( \mathcal{B}(\Hil) \bigr )$]\label{operator_valued_calculus:thm:holomorphic_functional_calculus}
	There exists a holomorphic functional calculus on $\Hoereq{2q} \bigl ( \mathcal{B}(\Hil) \bigr )$, \ie for any function $\varphi$ that is holomoprhic on some neighborhood of the spectrum $\sigma \bigl ( \Op^A(f) \bigr )$, equation~\eqref{operator_valued_calculus:eqn:holomorphic_functional_calculus} defines a map 
	\begin{align*}
		\varphi \mapsto \varphi^B(f) \in \Hoereq{2q} \bigl ( \mathcal{B}(\Hil) \bigr ) 
	\end{align*}
	into the Hörmander symbols of order $2q$. 
\end{theorem}
Just like in the operator-valued case, we can often construct a parametrix for equivariant magnetic $\Psi$DOs. The symbol agrees with the Moyal resolvent up to $\order(\epsilon^{\infty})$ — if the latter exists. Of course, the crucial assumption is that a $0$th-order approximation exists and is equivariant. Then the equivariance of the calculus ensures that all higher-order terms are equivariant as well. Compared to the analogous results from Section~\ref{operator_valued_calculus:Hoermander_symbols:Weyl_product}, the parametrix exists only as a formal sum and we do not know whether a resummation exists (\cf Remark~\ref{operator_valued_calculus:rem:lack_resummation_results_rho_equal_delta}). 
\begin{theorem}[Existence of a parametrix]\label{equivariant_calculus:cor:existence_parametrix}
	Suppose the following assumptions are satisfied: 
	\begin{enumerate}[(a)]
		\item $\Hil \hookrightarrow \Hil'$ can be continuously and densely injected into $\Hil'$. 
		\item We are given a symbol $f \in \Hoereq{q + q'} \bigl ( \mathcal{B}(\Hil,\Hil') \bigr )$. 
		\item We are given an asymptotic expansion of the magnetic Weyl product (\cf Definition~\ref{operator_valued_calculus:defn:asymptotic_Hoermander_class}) 
		\begin{align*}
			f \Weyl g \asymp \sum_{n = 0}^{\infty} \epsilon^n \, (f \Weyl g)_{(n)} 
			, 
		\end{align*}
		of $f \in S^{q' + q''}_{\rho,\delta} \bigl ( \mathcal{B}(\Hil',\Hil'') \bigr )$ and $g \in S^{q + q'}_{\rho,\delta} \bigl ( \mathcal{B}(\Hil,\Hil') \bigr )$, where $\epsilon \ll 1$ is a small parameter and the terms of the expansion satisfy 
		\begin{align*}
			(f \Weyl g)_{(n)} \in \Hoereq{q + q''} \bigl ( \mathcal{B}(\Hil,\Hil'') \bigr ) 
			. 
		\end{align*}
		\item There exists a symbol $g_0 \in \Hoereq{q + q'} \bigl ( \mathcal{B}(\Hil',\Hil) \bigr )$ that satisfies 
		\begin{align*}
			g_0 \Weyl f - \id_{\Hil} &= \order(\epsilon) \in \Hoereq{2q} \bigl ( \mathcal{B}(\Hil) \bigr ) \cap \Hoereq{2q'} \bigl ( \mathcal{B}(\Hil') \bigr )
			, 
			\\
			f \Weyl g_0 - \id_{\Hil'} &= \order(\epsilon) \in \Hoereq{2q} \bigl ( \mathcal{B}(\Hil) \bigr ) \cap \Hoereq{2q'} \bigl ( \mathcal{B}(\Hil') \bigr )
			. 
		\end{align*}
	\end{enumerate}
	Then there exists a parametrix $f^{(-1)_{\epsilon}} \in \Sigma \Hoereq{q + q'} \bigl ( \mathcal{B}(\Hil',\Hil) \bigr )$, which satisfies 
	\begin{align*}
			f^{(-1)_{\epsilon}} \Weyl f &= \id_{\Hil} + \order(\epsilon^{\infty}) \in \Sigma \Hoereq{2q} \bigl ( \mathcal{B}(\Hil) \bigr ) \cap \Sigma \Hoereq{2q'} \bigl ( \mathcal{B}(\Hil') \bigr ) 
			, 
			\\
			f \Weyl f^{(-1)_{\epsilon}} &= \id_{\Hil'} + \order(\epsilon^{\infty}) \in \Sigma \Hoereq{2q} \bigl ( \mathcal{B}(\Hil) \bigr ) \cap \Sigma \Hoereq{2q'} \bigl ( \mathcal{B}(\Hil') \bigr ) 
			. 
	\end{align*}
\end{theorem}
Applying the parametrix construction to $h - z$ yields an “approximate Moyal resolvent” $(h - z)^{(-1)_{\epsilon}}$, which is then the basis for a functional calculus. 
\begin{remark}[Equivariant functional calculi using parametrices]
	We could have repeated the discussion from Section~\ref{operator_valued_calculus:Hoermander_symbols:inversion_and_functional_calculus} and set up functional calculi with respect to parametrices rather than Moyal resolvents. The basis for those is Theorem~\ref{equivariant_calculus:cor:existence_parametrix}. Our arguments hold verbatim after adding the keyword equivariant in the right places. Indeed, this has been exploited in the literature on perturbed periodic systems (see \eg \cite{PST:effective_dynamics_Bloch:2003,DeNittis_Lein:Bloch_electron:2009,DeNittis_Lein:sapt_photonic_crystals:2013,Freund_Teufel:non_trivial_Bloch_sapt:2013}). 
\end{remark}
The fact that finite resummations of $(h - z)^{(-1)_{\epsilon}}$ in general do not converge may actually be an advantage in many situations: when $\Op^A(h)$ is selfadjoint and $z \not\in \sigma \bigl ( \Opeq^A(h) \bigr )$, Beal's Commutator Criterion~\ref{equivariant_calculus:thm:Beals_commutator_criterion} ensures the existence of the exact Moyal resolvent $(h - z)^{(-1)_{\Weyl}} \in \Hoereq{q + q'} \bigl ( \mathcal{B}(\Hil',\Hil) \bigr )$. When in addition $\bigl ( (h - z)^{(-1)_{\epsilon}} \bigr )_0$ the principal symbol in the $\epsilon$ expansion of $\Weyl$ exists (condition~(d) above), then we can construct the formal sum $(h - z)^{(-1)_{\epsilon}}$ as well. The two are related by 
\begin{align*}
	(h - z)^{(-1)_{\Weyl}} = (h - z)^{(-1)_{\epsilon}} + \order(\epsilon^{\infty})
	, 
\end{align*}
\ie the formal sum is the asymptotic expansion of the true Moyal resolvent. Results like \cite[Theorem~1.4]{Mantoiu_Purice:continuity_spectra:2009}, which guarantee the continuity of spectra for parameter-dependent magnetic pseudodifferential operators with scalar-valued symbols then suggest that we should expect $z \not\in \sigma \bigl ( \Opeq^A(h) \bigr )$ for $\epsilon = \epsilon_0$ extends to a neighborhood of $\epsilon_0$. In that case, $(h - z)^{(-1)_{\Weyl}} \in \SemiHoereq{q + q'} \bigl ( \mathcal{B}(\Hil',\Hil) \bigr )$ defines an asymptotic equivariant symbol. 

These relations between approximate and true Moyal resolvents extend to symbols obtained via the two functional calculi from the same function. 
%
%
%
%
\section{Conclusion \& outlook} 
\label{outlook}
The aim of this work was providing a solid foundation for magnetic pseudodifferential calculi for operator-valued and equivariant operator-valued functions and distributions. Up until now the mathematical theory of which lived in the appendices of some works (see \eg \cite{PST:sapt:2002,PST:effective_dynamics_Bloch:2003,Teufel:adiabatic_perturbation_theory:2003}) or in short explanations that merely sketched these ideas (as in \eg \cite{DeNittis_Lein:Bloch_electron:2009,Fuerst_Lein:scaling_limits_Dirac:2008,DeNittis_Lein:sapt_photonic_crystals:2013,DeNittis_Lein:ray_optics_photonic_crystals:2014}). The downside was that this had to be — partly or entirely — repeated for each publication, and more advanced results were usually out of reach. What is more, there are some technical subtleties one may not think of when outlining the construction in broad strokes. 

We hope that this work has filled this gap. Apart from providing fully featured pseudodifferential calculi, which importantly include \emph{functional calculi,} our discussion clarifies how to overcome the technical hurdles when extending analytical results to operator-valued and equivariant magnetic $\Psi$DOs. A common theme was that \emph{morally} we can often think of some operator-valued function space 
\begin{align*}
	\mathcal{A} \bigl ( \mathcal{B}(\Hil,\Hil') \bigr ) \neq \mathcal{A}(\C) \otimes \mathcal{B}(\Hil,\Hil')
\end{align*}
as the tensor product of the corresponding scalar-valued function space $\mathcal{A}(\C)$ with the Banach space of bounded operators. Unfortunately, unless $\mathcal{A}(\C)$ is nuclear or $\mathcal{B}(\Hil,\Hil')$ is finite-dimensional, this is false as there are several inequivalent topologies with respect to which we can complete the algebraic tensor product. Extra care needs to be taken to make the mathematical arguments rigorous. The crucial ingredient is that the relevant spaces of Schwartz functions and tempered distributions \emph{are} nuclear, and on this level the tensor product decomposition can be made rigorous and exploited systematically. 

Furthermore, we have conceptually clarified the relation between the operator-valued and the equivariant calculus. Big picture, there are two important ingredients: 
\begin{enumerate}[(1)]
	\item On the level of distributions equivariant magnetic $\Psi$DOs are the restriction of operator-valued magnetic $\Psi$DOs between (not necessarily equivariant) tempered distributions (\cf Proposition~\ref{equivariant_calculus:prop:equivariant_distributions_define_equivariant_PsiDOs}). 
	\item Bounded operators between ordinary, non-equivariant $L^2$ spaces that satisfy an equivariance condition define bounded operators between equivariant magnetic Sobolev and $L^2$ spaces (\cf \eg Propositions~\ref{equivariant_calculus:prop:characterization_equivariant_operators} and \ref{equivariant_calculus:prop:equivalence_boundedness_equivariant_operators}). 
\end{enumerate}
Exploiting this has a number of advantages. For example, rather than introducing suitable weights as in \cite[Appendix~B]{PST:effective_dynamics_Bloch:2003}, our proof of the Calderón-Vaillancourt Theorem for equivariant $\Psi$DOs (Theorem~\ref{equivariant_calculus:thm:Calderon_Vaillancourt}) rests on the fact that any bounded operator between (non-equivariant) magnetic Sobolev spaces that satisfies the equivariance condition~\eqref{equivariant_calculus:eqn:equivariance_condition_L2_Rd} uniquely defines a bounded \emph{equivariant} operator between \emph{equivariant} magnetic Sobolev spaces. This is one incarnation of Meta~Theorem~\ref{equivariant_calculus:meta_theorem:extension_operator_valued_to_equivariant}, which provides a general blueprint for extending other results. 

We believe that this review forms a good basis for more advanced results and generalizations. We have purposefully included the small magnetic field limit $\lambda \rightarrow 0$, which has been investigated by \eg Cornean, Iftimie and Purice \cite{Cornean_Iftimie_Purice:Peierls_substitution_magnetic_PsiDOs:2019} and by \cite{Fuerst_Lein:scaling_limits_Dirac:2008} for the non-relativistic limit of the Dirac dynamics. 

Our initial motivation for this paper came from our work on the semiclassical limit for metallic crystalline solids \cite{DeNittis_Lein_Seri:semiclassics_Bloch_electron_Fermi_surface:2021}. Having split off the technical aspects related to the pseudodifferential calculus for equivariant $\Psi$DOs allows us to focus on the ideas and physical interpretation. Moreover, should we or someone else decide to generalize the setting of \cite{DeNittis_Lein_Seri:semiclassics_Bloch_electron_Fermi_surface:2021}, then they can make use of a robust and fully developed pseudodifferential calculus. And should they need more advanced results that have not been covered here, then they only need to prove those. 

Further natural developments for the calculus could be the extension to quasicrystals or to more general $2$-cocycles that need not be the exponential of a magnetic flux. Or one could adapt the algebraic point of view developed in \cite{Mantoiu_Purice_Richard:twisted_X_products:2004,Mantoiu_Purice_Richard:Cstar_algebraic_framework:2007,Lein_Mantoiu_Richard:anisotropic_mag_pseudo:2009,Belmonte_Lein_Mantoiu:mag_twisted_actions:2010}. The latter could be rewarding not just from the point of view of applications, going beyond \emph{abelian} “coefficient $C^*$-algebras” (to use the terminology from \eg \cite{Mantoiu_Purice_Richard:twisted_X_products:2004,Lein_Mantoiu_Richard:anisotropic_mag_pseudo:2009}) introduces interesting mathematical obstacles. 
\appendix
\section{Tempered growth of group action on $H^m_A(\T^d)$} 
\label{appendix:tempered_growth_magnetic_Sobolev_spaces_torus}
An important example of a Hilbert space that is endowed with a group action $\tau$ of order $m$ is the magnetic Sobolev space $\Hil = H^m_{A_0}(\T^d)$ of the same order $m$ over the torus. 
\begin{proof}[Lemma~\ref{setting:lem:magnetic_Sobolev_norm_operator_tau}]
	The magnetic Sobolev norm we use is defined by 
	\begin{align*}
		\snorm{\psi}_{H^m_{A_0}(\T^d)} := \bnorm{\bexpval{- \ii \nabla_y - A_0(\hat{y})} \, \psi}_{L^2(\T^d)} 
		, 
	\end{align*}
	where the weight is defined through functional calculus. 
	
	For $\psi \in H^m_{A_0}(\T^d)$ we can express the Sobolev norm in terms of the usual $L^2$ norm, 
	\begin{align*}
		\bnorm{\e^{+ \ii \gamma^* \cdot \hat{y}} \psi}_{H^m_{A_0}(\T^d)} &= \Bnorm{\bexpval{- \ii \nabla_y - A_0(\hat{y})}^m \; \e^{+ \ii \gamma^* \cdot \hat{y}} \psi}_{L^2(\T^d)}
		\\
		&\leq \Bnorm{\e^{+ \ii \gamma^* \cdot \hat{y}} \, \bexpval{\gamma^* - A_0(\hat{y})}^m \, \Psi}_{L^2(\T^d)} + \Bnorm{\e^{+ \ii \gamma^* \cdot \hat{y}} \, \bexpval{- \ii \nabla_y - A_0(\hat{y})}^m \, \psi}_{L^2(\T^d)}
		\\
		&= \Bnorm{\bexpval{\gamma^* - A_0(\hat{y})}^m \, \Psi}_{L^2(\T^d)} + \Bnorm{\bexpval{- \ii \nabla_y - A_0(\hat{y})}^m \, \psi}_{L^2(\T^d)}
		\\
		&\leq \Bnorm{\bexpval{\gamma^* - A_0(\hat{y})}^m}_{\mathcal{B}(L^2(\T^d))} \; \snorm{\psi}_{L^2(\T^d)} + \snorm{\psi}_{H^m_{A_0}(\T^d)} 
		\\
		&\leq \left ( 1 + \Bnorm{\bexpval{\gamma^* - A_0(\hat{y})}^m}_{\mathcal{B}(L^2(\T^d))} \right ) \; \snorm{\psi}_{H^m_{A_0}(\T^d)} 
		. 
	\end{align*}
	After we take the supremum over $\psi \in H^m_{A_0}(\T^d)$ that are normalized with respect to the $m$th magnetic Sobolev norm, we obtain the claimed estimate: because the multiplication operators $A_{0,j}(\hat{y})$ are bounded on $L^2(\T^d)$ (the polynomial growth at $\infty$ does not matter as $\T^d$ is compact), we can estimate the first factor by 
	\begin{align*}
		1 + \Bnorm{\bexpval{\gamma^* - A_0(\hat{y})}^m}_{\mathcal{B}(L^2(\T^d))} \leq C \, \sexpval{\gamma^*}^m  
		. 
	\end{align*}
	This finishes the proof. 
\end{proof}
%

\section{Auxiliary results to extend $\Op^A$ from Schwartz functions to tempered distributions} 
\label{appendix:identification_operator_valued_Schwartz_spaces_distributions}
The most significant technical complication of Section~\ref{operator_valued_calculus} enters the extension by duality. To correctly adapt the strategy of \cite{Mantoiu_Purice:magnetic_Weyl_calculus:2004} from the scalar-valued to the operator-valued case, we have to carefully select the right Fréchet spaces and pick the right topologies. 

To make it easier for the reader to compare our results with our main resource, Treves' excellent book on topological vector spaces \cite{Treves:topological_vector_spaces:1967}, we will revert back to more standard notation. For example, we will label Hilbert spaces $\Hil_1 = \Hil$, $\Hil_2 = \Hil'$ and so forth to avoid confusing the prime with the strong or the Banach space dual $'$ later on.

\subsection{Space of trace class operators $\mathcal{L}^1(\Hil_1,\Hil_2)$ between two Hilbert spaces} 
\label{appendix:extension_by_duality_operator_valued_calculus:trace_class_operators}
In applications, we \emph{have} to consider cases where our functions and distributions take values in bounded operators between two \emph{different} Hilbert spaces. Frequently, our magnetic $\Psi$DO is defined by an \emph{unbounded}-operator-valued function; in that case $\Hil_1 = \mathcal{D} \subset \Hil_2$ is the pointwise operator domain that lies densely in the target space $\Hil_2$. 

When $\Hil_1 = \Hil_2$, our construction relies on well-known properties of $p$-Schatten classes (\cf \eg \cite[Section~VI.6]{Reed_Simon:M_cap_Phi_1:1972}). For example, the dual $\mathcal{L}^1(\Hil_1)' = \mathcal{B}(\Hil_1)$ of the trace class operators are the bounded operators and $\mathcal{L}^1(\Hil_1) \subseteq \mathcal{B}(\Hil_1)$ is a two-sided ideal with respect to multiplication. Moreover, the trace is cyclical, $\trace_{\Hil_1} (A \, B) = \trace_{\Hil_1} (B \, A)$. 

The purpose of this subsection is to furnish a definition of $\mathcal{L}^1(\Hil_1,\Hil_2)$ and convince the reader that this space acts in very much the same way and possesses the same properties as $\mathcal{L}^1(\Hil_1)$. All of the material we need can be found in \cite{Treves:topological_vector_spaces:1967}, where the topic is developed in much greater generality. When $\Hil_1$ and $\Hil_2$ are replaced by two Fréchet spaces $\mathcal{X}_1$ and $\mathcal{X}_2$, then $\mathcal{L}^1(\mathcal{X}_1,\mathcal{X}_2)$ is called the space of \emph{nuclear operators} mapping between $\mathcal{X}_1$ and $\mathcal{X}_2$. However, in our setting it makes sense to call elements of $\mathcal{L}^1(\Hil_1,\Hil_2)$ trace class: not only does $\mathcal{L}^1(\Hil_1,\Hil_1)$ coincide with the usual space of trace class operators, it inherits all essential properties of $\mathcal{L}^1(\Hil_1)$ even when $\Hil_1 \neq \Hil_2$. 

Let us offer up a definition that exploits our setting: 
\begin{definition}[$\mathcal{L}^1(\Hil_1,\Hil_2)$]\label{appendix:extension_by_duality_operator_valued_calculus:defn:trace_class_operators}
	We endow the vector space of \emph{trace class operators} 
  \begin{align*}
		\mathcal{L}^1(\Hil_1,\Hil_2) := \Bigl \{ T \in \mathcal{B}(\Hil_1,\Hil_2) \; \; \big \vert \; \; \norm{T}_{\mathcal{L}^1(\Hil_1,\Hil_2)} < \infty \Bigr \} 
	\end{align*}
	with the norm	%
	\begin{align*}
		\norm{T}_{\mathcal{L}^1(\Hil_1,\Hil_2)} := \bnorm{\sabs{T}}_{\mathcal{L}^1(\Hil_1)}
	\end{align*}
	where $\sabs{T} := \sqrt{T^* T}$ is the absolute value of $T$. 
\end{definition}
This is not the first-principles definition, because starting with that would necessitate a long — and for our purposes \emph{completely unnecessary} — preamble. Fortunately, an operator $T$ between Hilbert spaces is nuclear if and only if its modulus $\sabs{T} \in \mathcal{L}^1(\Hil_1)$ is trace class in the usual sense (\cf \cite[Theorem~48.2]{Treves:topological_vector_spaces:1967}). Hence, we have taken the liberty to replace the first-principles definition, \cite[Definition~47.2]{Treves:topological_vector_spaces:1967}, with Definition~\ref{appendix:extension_by_duality_operator_valued_calculus:defn:trace_class_operators} that only requires knowledge of standard results from functional analysis. 

A second characterization of trace class operators closely mimics the well-known case $\mathcal{L}^1(\Hil_1)$: an operator $T \in \mathcal{L}^1(\Hil_1)$ is trace class in the usual sense if and only if $T$ is the sum of rank-$1$ operators and the eigenvalues of $\sabs{T}$ are absolutely summable. The Corollary on p.~494 of \cite{Treves:topological_vector_spaces:1967} states that this is still true for elements of $\mathcal{L}^1(\Hil_1,\Hil_2)$: an operator $T$ between Hilbert spaces is nuclear if and only if there exist two (potentially finite) orthonormal sequences $\{ \varphi_{1 \, j} \}_{j \in \mathcal{I}} \subset \Hil_1$, $\mathcal{I} \subseteq \N$, and $\{ \varphi_{2 \, j} \}_{j \in \mathcal{I}} \subset \Hil_2$ as well as a sequence of complex numbers $\{ \mu_j \}_{j \in \mathcal{I}} \subset \C$ so that the action of the operator 
\begin{align*}
	T \psi_1 &= \sum_{j \in \mathcal{I}} \mu_j \, \scpro{\varphi_{1 \, j}}{\psi_1} \, \varphi_{2 \, j} 
	\\
	&\equiv \sum_{j \in \mathcal{I}} \mu_j \, \sket{\varphi_{2 \, j}}_{\Hil_2} \sbra{\varphi_{1 \, j}}_{\Hil_1} \, \psi_1 
\end{align*}
can be expressed as the linear combination of rank-$1$ operators and the coefficients 
\begin{align*}
	\sum_{j \in \mathcal{I}} \sabs{\mu_j} < \infty 
\end{align*}
are absolutely summable. This is the reason why we can use our abbreviated Definition~\ref{appendix:extension_by_duality_operator_valued_calculus:defn:trace_class_operators}. Furthermore, it explains why requiring $\sqrt{T^* \, T} \in \mathcal{L}^1(\Hil_1)$ is equivalent to $\sqrt{T \, T^*} \in \mathcal{L}^1(\Hil_2)$ — both yield the same summability condition $\{ \mu_j \}_{j \in \mathcal{I}} \in \ell^1(\mathcal{I})$ on the coefficients. 

Let us summarize some relevant properties. Essentially, the next Proposition merely states that $\mathcal{L}^1(\Hil_1,\Hil_2)$ has all the (suitably generalized) properties of $\mathcal{L}^1(\Hil_1)$. 
\begin{proposition}[Properties of $\mathcal{L}^1(\Hil_1,\Hil_2)$]\label{appendix:extension_by_duality_operator_valued_calculus:prop:properties_generalized_L1_spaces}
	\begin{enumerate}[(1)]
		\item $\mathcal{L}^1(\Hil_1,\Hil_2)$ is a Banach space. 
		\item The Hilbert space adjoint ${}^* : \mathcal{L}^1(\Hil_1,\Hil_2) \longrightarrow \mathcal{L}^1(\Hil_2,\Hil_1)$ is an isometry between Banach spaces. 
		\item $\mathcal{L}^1(\Hil_1,\Hil_2) \subseteq \mathcal{K}(\Hil_1,\Hil_2)$ is a subset of the compact operators. 
		\item The trace class property is stable under multiplication with bounded operators from the left and the right, \ie for all $A \in \mathcal{B}(\Hil_3,\Hil_4)$, $T \in \mathcal{L}^1(\Hil_2,\Hil_3)$ and $C \in \mathcal{B}(\Hil_1,\Hil_2)$ we have 
		\begin{align}
			\bnorm{A \, T C}_{\mathcal{L}^1(\Hil_1,\Hil_4)} &\leq \snorm{A}_{\mathcal{B}(\Hil_3,\Hil_4)} \, \snorm{T}_{\mathcal{L}^1(\Hil_2,\Hil_3)} \, \snorm{C}_{\mathcal{B}(\Hil_1,\Hil_2)} 
			, 
			\label{appendix:extension_by_duality_operator_valued_calculus:eqn:trace_norm_estimate_product}
		\end{align}
		which in particular implies $A \, T C \in \mathcal{L}^1(\Hil_1,\Hil_4)$. 
		\item The trace is cyclical in the sense that 
		\begin{align}
			\trace_{\Hil_2} \bigl ( A \, B \bigr ) &= \trace_{\Hil_1} \bigl ( B \, A \bigr ) 
			\label{appendix:extension_by_duality_operator_valued_calculus:eqn:cyclicity_trace}
		\end{align}
		holds whenever \emph{at least one} of the operators $A \in \mathcal{B}(\Hil_1,\Hil_2)$ and $B \in \mathcal{B}(\Hil_2,\Hil_1)$ is trace class. 
		\item The $\eps$-tensor product (or \emph{injective} tensor product) of two Hilbert spaces 
		\begin{align*}
			\Hil_1 \otimes_{\eps} \Hil_2 &= \mathcal{K}(\Hil_1,\Hil_2) 
		\end{align*}
		equals the compact operators $\mathcal{K}(\Hil_1,\Hil_2)$ mapping between them. 
		\item The $\pi$-tensor product (or \emph{projective} tensor product) of two Hilbert spaces 
		\begin{align*}
			\Hil_1 \otimes_{\pi} \Hil_2 &= \mathcal{L}^1(\Hil_1,\Hil_2) 
		\end{align*}
		equals the trace class operators $\mathcal{L}^1(\Hil_1,\Hil_2)$ mapping between them. 
		\item The dual of the trace class operators are the bounded operators, 
		\begin{align*}
			\mathcal{L}^1(\Hil_1,\Hil_2)^* = \mathcal{B}(\Hil_1,\Hil_2)
			, 
		\end{align*}
		where the ${}^{\ast}$ indicates that we have used the sesquilinear duality bracket  
		\begin{align*}
			\scpro{A}{T}_{\mathcal{L}^1(\Hil_1,\Hil_2)} &= \trace_{\Hil_1} \bigl ( A^* \, T \bigr )
			\\
			&= \overline{\scpro{T}{A}_{\mathcal{L}^1(\Hil_2,\Hil_1)}}
			&&
			\forall A \in \mathcal{B}(\Hil_1,\Hil_2) 
			, \; 
			T \in \mathcal{L}^1(\Hil_1,\Hil_2) 
			. 
		\end{align*}
	\end{enumerate}
\end{proposition}
Item~(2) can be succinctly understood as the natural generalization of the two-sided ideal property of $\mathcal{L}^1(\Hil_1)$. 

Moreover, the two tensor products can be viewed as completions of the algebraic tensor product $\Hil_1 \odot \Hil_2$ with respect to different topologies, \ie \emph{finite} linear combinations of rank-$1$ operators 
\begin{align*}
	\Hil_1 \times \Hil_2 \ni (\psi_1,\psi_2) \mapsto \sket{\psi_2}_{\Hil_2} \sbra{\psi_1}_{\Hil_1} 
	. 
\end{align*}
\begin{proof}
	\begin{enumerate}[(1)]
		\item The first principles definition, \cite[Definition~47.2]{Treves:topological_vector_spaces:1967}, declares $\mathcal{L}^1(\Hil_1,\Hil_2) \cong \Hil_1 \otimes_{\eps} \Hil_2 / N$ as the quotient of the Banach space $\Hil_1 \otimes_{\eps} \Hil_2$ with a closed subspace $N$ (\cf the first paragraph on \cite[p.~479]{Treves:topological_vector_spaces:1967}). Thus, it is a Banach space. 
		\item For the Banach space adjoint the relevant statement can be found as \cite[Proposition~47.5]{Treves:topological_vector_spaces:1967} and its Corollary on p.~484. Given that Hilbert spaces are reflexive, these facts immediately extend to the Hilbert space adjoint. 
		\item Any nuclear operator is compact (\cf \cite[Proposition~47.3]{Treves:topological_vector_spaces:1967}). 
		\item This is the content of \cite[Proposition~47.1]{Treves:topological_vector_spaces:1967}. 
		\item Here, we use a trick by identifying $A \in \mathcal{B}(\Hil_1,\Hil_2)$ and $B \in \mathcal{B}(\Hil_2,\Hil_1)$ with the block operators 
		\begin{align*}
			\tilde{A} := \left (
			\begin{matrix}
				0 & 0 \\
				A & 0 \\
			\end{matrix}
			\right )
			, 
			\quad
			\tilde{B} := \left (
			\begin{matrix}
				0 & B \\
				0 & 0 \\
			\end{matrix}
			\right )
			, 
		\end{align*}
		on the direct sum $\Hil_1 \oplus \Hil_2$. We can then use the cyclicity of the trace on $\Hil_1 \oplus \Hil_2$ to deduce 
		\begin{align*}
			\trace_{\Hil_2} \bigl ( A \, B \bigr ) &= \trace_{\Hil_1 \oplus \Hil_2} \bigl ( \tilde{A} \, \tilde{B} \bigr ) 
			= \trace_{\Hil_1 \oplus \Hil_2} \bigl ( \tilde{B} \, \tilde{A} \bigr ) 
			\\
			&= \trace_{\Hil_1} \bigl ( B \, A \bigr ) 
			. 
		\end{align*}
		\item We refer to \cite[Theorem~48.3]{Treves:topological_vector_spaces:1967}. 
		\item This follows from \cite[Theorem~48.4]{Treves:topological_vector_spaces:1967} and the arguments in the paragraph preceding Theorem~48.4 in the reference. 
    \item The reader may find the claim as \cite[Theorem~48.5']{Treves:topological_vector_spaces:1967}. 
	\end{enumerate}
\end{proof}
For the sake of completeness, we give the straight-forward extension of Definition~\ref{appendix:extension_by_duality_operator_valued_calculus:defn:trace_class_operators} to $p$-Schatten class operators. 
\begin{definition}[$\mathcal{L}^p(\Hil_1,\Hil_2)$]\label{appendix:extension_by_duality_operator_valued_calculus:defn:Schatten_class_operators}
	For any $1 \leq p < \infty$ we endow the vector space of \emph{$p$-Schatten class operators} 
	\begin{align*}
		\mathcal{L}^p(\Hil_1,\Hil_2) := \Bigl \{ T \in \mathcal{B}(\Hil_1,\Hil_2) \; \; \big \vert \; \; \norm{T}_{\mathcal{L}^p(\Hil_1,\Hil_2)} < \infty \Bigr \} 
	\end{align*}
	with the norm	%
	\begin{align*}
		\norm{T}_{\mathcal{L}^p(\Hil_1,\Hil_2)} := \bnorm{\sabs{T}}_{\mathcal{L}^p(\Hil_1)}
		= \bigl ( \trace_{\Hil_1} \sabs{T}^p \bigr )^{\nicefrac{1}{p}} 
	\end{align*}
	where $\sabs{T} := \sqrt{T^* T}$ is the absolute value of $T$. 
\end{definition}
%

\subsection{Proof of Proposition~\ref{operator_valued_calculus:prop:extension_OpA_tempered_distributions}: extension of $\Op^A$ by duality} 
\label{appendix:extension_by_duality_operator_valued_calculus:proof_central_proposition_extension}
To allow us to refer to the relevant results from \cite{Treves:topological_vector_spaces:1967} more directly, in this subsection we will work with the \emph{strong} or the \emph{Banach space dual} ${}'$ rather than the dual ${}^*$ obtained from a sesquilinear duality bracket. We leave it to the reader to restore the notation from Section~\ref{operator_valued_calculus} at the end. 

Many of the technical results we will need are summarized in Proposition~\ref{appendix:extension_by_duality_operator_valued_calculus:prop:properties_generalized_L1_spaces}. Item~(8) justifies why looking at trace class operators between $\Hil_1$ and $\Hil_2$ is the right Banach space to start with — its dual are the bounded operators. Moreover, item~(7) connects $\mathcal{L}^1(\Hil_1,\Hil_2) \cong \Hil_1 \otimes_{\pi} \Hil_2$ to \emph{a} tensor product. 

As we have emphasized several times in the main body of the text, dealing with the subtleties of tensor products is one of the key technical challenges. For example, even though for \emph{most} of the proof, our choice of tensor product does not matter — \emph{most} of the spaces are nuclear and for nuclear spaces injective and projective tensor products coincide (\cf \cite[Theorem~50.1]{Treves:topological_vector_spaces:1967}). In one key step, though, it \emph{does:} the tensor products of $\Hil_1$ and $\Hil_2$ in Proposition~\ref{appendix:extension_by_duality_operator_valued_calculus:prop:properties_generalized_L1_spaces}~(6) and (7) yield different spaces of operators. Sussing out exactly \emph{where} we need to make a choice of tensor product in our arguments and why is \emph{very} subtle. Ultimately, it can be traced to the fact that the $\pi$-tensor product is always \emph{associative} whereas the $\eps$-tensor product is \emph{not} (\cf \eg the paragraph below Definition~2.7 in \cite{Dabrowski_Kerjean:models_linear_logic_Schwartz_eps_product:2019}). 
\begin{proof}[Proposition~\ref{operator_valued_calculus:prop:extension_OpA_tempered_distributions}]
	\begin{enumerate}[(1)]
		\item First of all, all topological vector spaces involved are derived from Fréchet spaces or their duals. Moreover, the vector spaces $\Schwartz(\R^d)$ and $\Schwartz'(\R^d)$ are nuclear (\cf \cite[p.~530, Corollary]{Treves:topological_vector_spaces:1967}), and therefore the completions $\Schwartz(\R^d) \otimes \Hil_j$, $j = 1 , 2$, of the corresponding algebraic tensor products are uniquely defined and we may identify them with $\Schwartz(\R^d) \otimes \Hil_j \cong \Schwartz(\R^d,\Hil_j)$ (\cf \cite[p.~533]{Treves:topological_vector_spaces:1967}). Owing to the nuclearity of these spaces we have the freedom of interpreting $\otimes$ as the projective tensor product $\otimes_{\pi}$ and to use its algebraic properties in the proof. 
		
      Moreover, the canonical identification of $\Hil_j$ and its Banach space dual $\Hil_j' \cong \Hil_j$ via Riesz Lemma leads to an identification of the dual space $\Schwartz(\R^d,\Hil_j)' \cong \Schwartz'(\R^d,\Hil_j)$.
		
		We will now start with the target space in Proposition~\ref{operator_valued_calculus:prop:extension_OpA_tempered_distributions}~(1) and transform it to the initial space on the left-hand side. The first identification comes courtesy of \cite[equation~(50.17)]{Treves:topological_vector_spaces:1967}, namely 
		\begin{align*}
			\mathcal{B} \bigl ( \Schwartz'(\R^d,\Hil_1) , \Schwartz(\R^d,\Hil_2) \bigr ) &\cong \mathcal{B} \bigl ( \Schwartz(\R^d,\Hil_1)' , \Schwartz(\R^d,\Hil_2) \bigr ) 
			\\
			&\cong \Schwartz(\R^d,\Hil_1) \otimes \Schwartz(\R^d,\Hil_2) 
			. 
		\end{align*}
		In order to group the factors of the tensor product, let us insert $\pi$ back for emphasis so that we can use the the associativity and the commutativity of the projective tensor product, 
		\begin{align*}
			\bigl ( \Schwartz(\R^d) \otimes \Hil_1 \bigr ) \otimes \bigl ( \Schwartz(\R^d) \otimes \Hil_2 \bigr ) &= \bigl ( \Schwartz(\R^d) \otimes_{\pi} \Hil_1 \bigr ) \otimes_{\pi} \bigl ( \Schwartz(\R^d) \otimes_{\pi} \Hil_2 \bigr ) 
			\\
			&\cong \Schwartz(\R^d) \otimes_{\pi} \Bigl ( \bigl ( \Hil_1 \otimes_{\pi} \Schwartz(\R^d) \bigr ) \otimes_{\pi} \Hil_2 \Bigr ) 
			\\
			&\cong \bigl ( \Schwartz( \R^d) \otimes_{\pi} \Schwartz( \R^d) \bigr ) \otimes_{\pi} \bigl ( \Hil_1 \otimes_{\pi} \Hil_2 \bigr ) 
			. 
		\end{align*}
		After invoking \cite[Theorem~51.6]{Treves:topological_vector_spaces:1967} and identifying $\R^{2d} \cong T^* \R^d$ with the cotangent bundle, the first factor reduces to $\Schwartz(T^* \R^d)$. For the second, we make use of Proposition~\ref{appendix:extension_by_duality_operator_valued_calculus:prop:properties_generalized_L1_spaces}~(7), which leads to the first identification in 
		\begin{align*}
			\Schwartz(T^* \R^d) \otimes_{\pi} \bigl ( \Hil_1 \otimes_{\pi} \Hil_2 \bigr ) &\cong  \Schwartz(T^* \R^d) \otimes \mathcal{L}^1(\Hil_1,\Hil_2) 
			\\
			&\cong  \Schwartz \bigl ( T^* \R^d, \mathcal{L}^1(\Hil_1,\Hil_2) \bigr ) 
			. 
		\end{align*}
		The last isomorphism that pulls in $\mathcal{L}^1(\Hil_1,\Hil_2)$ is again taken from \cite[p.~533]{Treves:topological_vector_spaces:1967}. This then proves item~(1). 
		\item We will build on some of the arguments from the proof of (1). The central identification is the last equality on \cite[p.~534]{Treves:topological_vector_spaces:1967}, 
		\begin{align*}
			\Schwartz' \bigl ( T^* \R^d , \mathcal{B}(\Hil_1,\Hil_2) \bigr ) &\cong \Schwartz' \bigl ( T^* \R^d , \mathcal{L}^1(\Hil_1,\Hil_2)' \bigr ) 
			\\
			&\cong \Bigl ( \Schwartz \bigl ( T^* \R^d , \mathcal{L}^1(\Hil_1,\Hil_2) \bigr ) \Bigr )' 
			.
		\end{align*}
		Note that the position of the dual ${}'$ is essential, and the additional pair of brackets in the last line are merely for emphasis. 
		In the proof of item~(1) we have shown that we may identify 
		\begin{align*}
			\Schwartz \bigl ( T^* \R^d , \mathcal{L}^1(\Hil_1,\Hil_2) \bigr ) \cong \Schwartz(\R^d,\Hil_1) \otimes \Schwartz(\R^d,\Hil_2) 
		\end{align*}
		with \emph{the} tensor product of two Hilbert space-valued Schwartz functions; we may write \emph{the} tensor product since Schwartz functions are nuclear and all tensor products are one and the same. 
		
		But the dual of the tensor product is nothing but the tensor product of the duals (\cf \cite[equation~(50.19)]{Treves:topological_vector_spaces:1967}), 
		\begin{align*}
			\Bigl ( \Schwartz \bigl ( T^* \R^d , \mathcal{L}^1(\Hil_1,\Hil_2) \bigr ) \Bigr )' &\cong \Schwartz(\R^d,\Hil_1)' \otimes \Schwartz(\R^d,\Hil_2)' 
			\\
			&\cong \Schwartz'(\R^d,\Hil_1) \otimes \Schwartz'(\R^d,\Hil_2) 
			. 
		\end{align*}
		According to the first displayed equation on \cite[p.~525]{Treves:topological_vector_spaces:1967}, the right-hand side can be seen as 
		\begin{align*}
			\Schwartz'(\R^d,\Hil_1) \otimes \Schwartz'(\R^d,\Hil_2) &\cong \mathcal{B} \bigl ( \Schwartz(\R^d,\Hil_1) \, , \, \Schwartz'(\R^d,\Hil_2) \bigr ) 
			. 
		\end{align*}
		Thus, we have proven item~(2). 
		\item We need to prove gauge-covariance. Assumption~\ref{operator_valued_calculus:assumption:polynomially_bounded_B} ensures that the gauge function $x \mapsto \e^{+ \ii \lambda \vartheta(\eps x)} \in \Cont^{\infty}_{\mathrm{pol}}(\R^d,\C)$ is a smooth, bounded function with polynomially bounded derivatives. 
		
		Consequently, multiplication with $\e^{+ \ii \lambda \vartheta(\eps x)}$ defines a continuous map 
		\begin{align*}
			x \mapsto \e^{+ \ii \lambda \vartheta(\eps x)} : \Schwartz(\R^d,\Hil_j) \longrightarrow \Schwartz(\R^d,\Hil_j) 
			, 
			&&
			j = 1 , 2 
			, 
		\end{align*}
		between Schwartz spaces, which extends by duality to a continuous map on the tempered distributions. 
		
		Thus, no matter if we view $f \mapsto \Op^A(f)$ as maps on $\Schwartz \bigl ( T^* \R^d , \mathcal{L}^1(\Hil_1,\Hil_2) \bigr )$ or on $\Schwartz' \bigl ( T^* \R^d , \mathcal{B}(\Hil_1,\Hil_2) \bigr )$, the right-hand side of equation~\eqref{operator_valued_calculus:eqn:covariance_Op_A} is the concatenation of three continuous linear maps. Hence, the right-hand side is continuous. 
		
		Equality of left- and right-hand side follows directly for Schwartz functions and extends by duality to tempered distributions. Therefore, the maps from item~(1) and (2) inherit gauge-covariance~\eqref{operator_valued_calculus:eqn:covariance_Op_A}. This finishes the proof. 
	\end{enumerate}
\end{proof}
%

\section{Characterizing equivariant operators} 
\label{appendix:equivariant_operators}
The central piece to extending the magnetic pseudodifferential calculus to \emph{equivariant} operator-valued symbols is Proposition~\ref{equivariant_calculus:prop:characterization_equivariant_operators}. The main point of this appendix is to prove it rigorously. However, it turns out that in addition to the three equivalent points of views, we need a fourth one, which corresponds to another identification of $L^2_{\eq}(\R^d,\Hil)$ with a Hilbert space over the Brillouin zone $\BZ_{\gamma^*} := \BZ - \gamma^*$ that has been shifted by $\gamma^* \in \Gamma^*$. For consistency, we omit the index when the fundamental cell $\BZ = \BZ_0$ is centered at the origin $0 \in \Gamma^*$.

\subsection{Relevant Hilbert spaces and identifications between them} 
\label{appendix:equivariant_operators:hilbert_spaces}
The first three relevant spaces have already been introduced, $L^2_{\eq}(\R^d,\Hil)$, $L^2(\BZ,\Hil)$ and $L^2(\R^d,\Hil)$. The fourth is the Hilbert space $\mathfrak{h}_{\gamma^*}(\Hil)$ defined for any $\gamma^* \in \Gamma^*$ as the \emph{Banach} space $L^2(\BZ_{\gamma^*},\Hil)$ endowed with the weighted scalar product 
\begin{align*}
	\scpro{\varphi_{\gamma^*}}{\psi_{\gamma^*}}_{\mathfrak{h}_{\gamma^*}(\Hil)} := \Bscpro{\bigl ( \id_{L^2(\BZ_{\gamma^*})} \otimes \tau(\gamma^*)^{-1} \bigr ) \, \varphi_{\gamma^*} \, }{ \, \bigl ( \id_{L^2(\BZ_{\gamma^*})} \otimes \tau(\gamma^*)^{-1} \bigr ) \, \psi_{\gamma^*}}_{L^2(\BZ_{\gamma^*},\Hil)} 
	. 
\end{align*}
Since $\tau(0) = \id_{\Hil}$, the \emph{Hilbert} spaces $\mathfrak{h}_0(\Hil) = L^2(\BZ_0,\Hil) \equiv L^2(\BZ,\Hil)$ agree. The presence of the weight ensures $L^2_{\eq}(\R^d,\Hil)$ is \emph{unitarily} equivalent to $\mathfrak{h}_{\gamma^*}(\Hil)$ for \emph{any} $\gamma^* \in \Gamma^*$. 

Let us list and prove the relations between $\mathfrak{h}_{\gamma^*}(\Hil)$, $L^2(\BZ_{\gamma^*},\Hil)$ and $L^2_{\eq}(\R^d,\Hil)$. We emphasize that we will continue to tacitly assume that Assumption~\ref{equivariant_calculus:assumption:setting} is satisfied, \ie all Hilbert spaces are separable and all group actions have at most polynomial growth. 
\begin{lemma}\label{appendix:equivariant_operators:lem:weighted_Hilbert_spaces_translated_Brillouin_zones}
	Let $\gamma^* , \tilde{\gamma}^* \in \Gamma^*$ be two reciprocal lattice vectors. 
	\begin{enumerate}[(1)]
		\item For $\gamma^* = 0$ the \emph{Hilbert} spaces $\mathfrak{h}_0(\Hil) = L^2(\BZ_0,\Hil) \equiv L^2(\BZ,\Hil)$ all coincide. 
		\item $\mathfrak{h}_{\gamma^*}(\Hil)$ and $L^2(\BZ_{\gamma^*},\Hil)$ agree as \emph{Banach} spaces. 
		\item The linear map 
		\begin{align*}
			U_{\gamma^*} : \, &L^2_{\eq}(\R^d,\Hil) \longrightarrow \mathfrak{h}_{\gamma^*}(\Hil)
			, 
			\\
			&\psi \mapsto \psi \vert_{\BZ_{\gamma^*}} 
			, 
		\end{align*}
		is a unitary between Hilbert spaces. Furthermore, for almost all $k \in \BZ_{\gamma^*}$ we have 
		\begin{align*}
			(U_{\gamma^*} \psi)(k) = \psi(k) = \tau(\gamma^*) \, \psi(k + \gamma^*) 
			. 
		\end{align*}
		\item The composition 
		\begin{align*}
			U_{\gamma^*,\tilde{\gamma}^*} := U_{\gamma^*} \, U_{\tilde{\gamma}^*}^{-1} : \mathfrak{h}_{\tilde{\gamma}^*}(\Hil) \longrightarrow \mathfrak{h}_{\gamma^*}(\Hil) 
		\end{align*}
		is a unitary that acts on $\psi_{\tilde{\gamma}^*} \in \mathfrak{h}_{\tilde{\gamma}^*}(\Hil)$ as 
		\begin{align*}
			\bigl ( U_{\gamma^*,\tilde{\gamma}^*} \psi_{\tilde{\gamma}^*} \bigr )(k) &= \tau(\gamma^*) \, \tau(\tilde{\gamma}^*)^{-1} \, \psi_{\tilde{\gamma}^*}(k + \gamma^* - \tilde{\gamma}^*)
			\\
			&= \tau(\tilde{\gamma}^* - \gamma^*)^{-1} \, \psi_{\tilde{\gamma}^*}(k + \gamma^* - \tilde{\gamma}^*)
		\end{align*}
		for almost all $k \in \BZ_{\gamma^*}$. 
	\end{enumerate}
\end{lemma}
\begin{proof}
	\begin{enumerate}[(1)]
		\item This follows from $\tau(0) = \id_{\Hil}$ and $\BZ_0 = \BZ$. 
		\item The norms of $\mathfrak{h}_{\gamma^*}(\Hil)$ and $L^2(\BZ_{\gamma^*},\Hil)$ are equivalent, 
		\begin{align}
			\bnorm{\tau(\gamma^*)}_{\mathcal{B}(\Hil)}^{-1} \, \bnorm{\psi_{\gamma^*}}_{L^2(\BZ_{\gamma^*},\Hil)} \leq \bnorm{\psi_{\gamma^*}}_{\mathfrak{h}_{\gamma^*}(\Hil)} 
			\leq \bnorm{\tau(\gamma^*)^{-1}}_{\mathcal{B}(\Hil)} \, \bnorm{\psi_{\gamma^*}}_{L^2(\BZ_{\gamma^*},\Hil)} 
			, 
			\label{appendix:equivariant_operators:eqn:equivalence:h_gammaStar_L2_BZ_gammaStar_norms}
		\end{align}
		because both, $\tau(\gamma^*)$ and its inverse $\tau(\gamma^*)^{-1} = \tau(-\gamma^*)$ are bounded operators on $\Hil$ by Assumption~\ref{equivariant_calculus:assumption:setting}~(c). 
		\item First of all, we can verify by direct computation that $U_{\gamma^*}$ preserves the scalar product: for any $\varphi , \psi \in L^2_{\eq}(\R^d,\Hil)$ we deduce 
		\begin{align*}
			\bscpro{U_{\gamma^*} \varphi}{U_{\gamma^*} \psi}_{\mathfrak{h}_{\gamma^*}(\Hil)} &= \int_{\BZ_{\gamma^*}} \dd k \, \bscpro{\tau(\gamma^*)^{-1} \varphi(k)}{\tau(\gamma^*)^{-1} \psi(k)}_{\Hil} 
			\\
			&= \int_{\BZ} \dd k' \, \bscpro{\tau(\gamma^*)^{-1} \varphi(k' + \gamma^*)}{\tau(\gamma^*)^{-1} \psi(k + \gamma^*)}_{\Hil} 
			\\
			&= \int_{\BZ} \dd k' \, \scpro{\varphi(k')}{\psi(k)}_{\Hil} 
			= \scpro{\varphi}{\psi}_{L^2_{\eq}(\R^d,\Hil)} 
			. 
		\end{align*}
		Furthermore, $U_{\gamma^*}$ is surjective as we can use the equivariance condition to construct a preimage to any $\psi_{\gamma^*} \in \mathfrak{h}_{\gamma^*}(\Hil)$. Therefore, $U_{\gamma^*}$ is unitary. 
		
		The explicit expression for $(U_{\gamma^*} \psi)(k)$ follows directly from the equivariance condition that characterizes elements of $L^2_{\eq}(\R^d,\Hil)$. 
		\item As the composition of unitaries, $U_{\gamma^*,\tilde{\gamma}^*}$ is unitary. The explicit formula for the point evaluation of $U_{\gamma^*,\tilde{\gamma}^*} \psi_{\tilde{\gamma}^*}$ follows from (3). 
	\end{enumerate}
\end{proof}
The second lemma collects facts on how to relate the \emph{collection} $\mathfrak{h}_{\gamma^*}(\Hil)$ to $L^2(\R^d,\Hil)$. Clearly, we can identify the Hilbert spaces 
\begin{align*}
	L^2(\R^d,\Hil) \cong \bigoplus_{\gamma^* \in \Gamma^*} L^2(\BZ_{\gamma^*},\Hil)
\end{align*}
by carving up $\R^d \cong \bigcup_{\gamma^* \in \Gamma^*} \BZ_{\gamma^*}$ into fundamental cells. If we take the direct sum of the $\mathfrak{h}_{\gamma^*}(\Hil)$ instead, in general we will obtain a dense subspace of $L^2(\R^d,\Hil)$ due to the presence of the weights. 
\begin{lemma}\label{appendix:equivariant_operators:lem:direct_sum_weighted_Hilbert_spaces_translated_Brillouin_zones}
	Suppose the order of the group action $\tau : \Gamma^* \longrightarrow \mathcal{GL}(\Hil)$ on $\Hil$ is $q \geq 0$. Define the Hilbert spaces 
	\begin{align*}
		\mathfrak{h}_{\Gamma^*}(\Hil) := \bigoplus_{\gamma^* \in \Gamma^*} \mathfrak{h}_{\gamma^*}(\Hil) 
	\end{align*}
	and $\mathfrak{h}_{\Gamma^*}^q(\Hil)$ that consists of those elements of $\bigoplus_{\gamma^* \in \Gamma^*} L^2(\BZ_{\gamma^*},\Hil)$ for which the norm 
	\begin{align*}
		\bnorm{(\psi_{\gamma^*})_{\gamma^* \in \Gamma^*}}_{\mathfrak{h}_{\Gamma^*}^q(\Hil)}^2 := \sum_{\gamma^* \in \Gamma^*} \sexpval{\gamma^*}^q \, \snorm{\psi_{\gamma^*}}_{L^2(\BZ_{\gamma^*},\Hil)}^2 
	\end{align*}
	is finite. 
	\begin{enumerate}[(1)]
		\item The \emph{Banach} spaces $H^q_{\Fourier}(\R^d,\Hil) := \Fourier H^q(\R^d,\Hil)$ and $\mathfrak{h}_{\Gamma^*}^q(\Hil)$ are isomorphic under the map $U_{\Gamma^*}^q : \Psi \mapsto \bigl ( \Psi \vert_{\BZ_{\gamma^*}} \bigr )_{\gamma^* \in \Gamma^*}$. 
		\item The embedding $\imath : \mathfrak{h}_{\Gamma^*}^q(\Hil) \longrightarrow \mathfrak{h}_{\Gamma^*}(\Hil)$ is continuous. 
		\item The set of finite sequences 
		\begin{align*}
			c_{0,\Gamma^*}(\Hil) := \biggl \{ (\psi_{\gamma^*})_{\gamma^* \in \Gamma^*} \in \bigoplus_{\gamma^* \in \Gamma^*} L^2(\BZ_{\gamma^*},\Hil) \; \; \big \vert \; \; &\exists R > 0 : \; \psi_{\gamma^*} = 0 \quad 
			\biggr. \\
			&\biggl . 
			\forall \gamma^* \in \Gamma^* \mbox{ with } \sabs{\gamma^*} \geq R \biggr \} 
		\end{align*}
		lies dense in $\mathfrak{h}_{\Gamma^*}(\Hil)$, $\mathfrak{h}_{\Gamma^*}^q(\Hil)$ and $\bigoplus_{\gamma^* \in 
		\Gamma^*} L^2(\BZ_{\gamma^*},\Hil)$. Therefore, $\mathfrak{h}_{\Gamma^*}^q(\Hil) \subseteq \mathfrak{h}_{\Gamma^*}(\Hil)$ and $\mathfrak{h}_{\Gamma^*}(\Hil)$, $\mathfrak{h}_{\Gamma^*}^q(\Hil) \subseteq \bigoplus_{\gamma^* \in 
		\Gamma^*} L^2(\BZ_{\gamma^*},\Hil)$ are dense as well. 
		\item Suppose that $\tau$ is elliptic of order $q$ in the sense that there exist two constants $C > c > 0$ so that 
		\begin{align}
			c \, \sexpval{\gamma^*}^q \, \bnorm{\psi_{\gamma^*}}_{L^2(\BZ_{\gamma^*},\Hil)} \leq \bnorm{\psi_{\gamma^*}}_{\mathfrak{h}_{\gamma^*}(\Hil)} 
			\leq C \, \sexpval{\gamma^*}^q \, \bnorm{\psi_{\gamma^*}}_{L^2(\BZ_{\gamma^*},\Hil)}
			\label{appendix:equivariant_operators:eqn:direct_sum_weighted_Hilbert_spaces_ellipticity_norms_plaquettes}
		\end{align}
		holds for all $(\psi_{\gamma^*})_{\gamma^* \in \Gamma^*} \in \mathfrak{h}_{\Gamma^*}^q(\Hil)$. Then $\mathfrak{h}_{\Gamma^*}^q(\Hil) = \mathfrak{h}_{\Gamma^*}(\Hil)$ agree as Banach spaces. 
	\end{enumerate}
\end{lemma}
\begin{proof}
	\begin{enumerate}[(1)]
		\item Clearly, we can view $H^q_{\Fourier}(\R^d,\Hil) \ni \Psi \cong (\psi_{\gamma^*})_{\gamma^* \in \Gamma^*}$ as a collection of vectors $\psi_{\gamma^*} := \Psi \vert_{\BZ_{\gamma^*}} \in L^2(\BZ_{\gamma^*},\Hil)$. That gives an injection 
		\begin{align*}
			\imath_q : H^q_{\Fourier}(\R^d,\Hil) \hookrightarrow \bigoplus_{\gamma^* \in \Gamma^*} L^2(\BZ_{\gamma^*},\Hil) 
		\end{align*}
		which is continuous since $\sexpval{K}^q \geq 1$ holds and therefore 
		\begin{align*}
			\bnorm{\Psi}_{H^q_{\Fourier}(\R^d,\Hil)}^2 &= \bnorm{\sexpval{K}^q \, \Psi}_{L^2(\R^d,\Hil)}^2 
			= \sum_{\gamma^* \in \Gamma^*} \bnorm{\sexpval{K}^q \, \Psi \vert_{\BZ_{\gamma^*}}}_{L^2(\BZ,\Hil)}^2 
			\\
			&\geq \sum_{\gamma^* \in \Gamma^*} \bnorm{\Psi \vert_{\BZ_{\gamma^*}}}_{L^2(\BZ,\Hil)}^2 
			. 
		\end{align*}
		So the question that remains is whether the two norms are equivalent. That amounts to showing there exist three constants $\tilde{C} > \tilde{c} > 0$ and $R > 0$ for which 
		\begin{align}
			\tilde{c} \, \sexpval{\gamma^*}^q \, \bnorm{\Psi \vert_{\BZ_{\gamma^*}}}_{L^2(\BZ_{\gamma^*},\Hil)} 
			\leq \bnorm{\sexpval{K}^q \, \Psi \vert_{\BZ_{\gamma^*}}}_{L^2(\BZ,\Hil)}
			\leq \tilde{C} \, \sexpval{\gamma^*}^q \, \bnorm{\Psi \vert_{\BZ_{\gamma^*}}}_{L^2(\BZ_{\gamma^*},\Hil)}
			\label{appendix:equivariant_operators:eqn:equivalence_Hq_hq_norms_estimate}
		\end{align}
		holds for all $\bigl ( \Psi \vert_{\BZ_{\gamma^*}} \bigr )_{\gamma^* \in \Gamma^*}$ and $\gamma^* \in \Gamma^*$ provided $\sabs{\gamma^*} \geq R$. Since $\sexpval{K}^q$ is a multiplication operator upper and lower bound follow from the elementary estimate 
		\begin{align*}
			\tfrac{1}{2} \, \sabs{\gamma^*} \leq \sabs{k} \leq 2 \, \sabs{\gamma^*} 
			, 
		\end{align*}
		which is true for all $k \in \BZ_{\gamma^*}$ and $\gamma^* \in \Gamma^*$ as long as $\sabs{\gamma^*} \geq R$ is large enough. 
		
		Thus, the norms are in fact equivalent and the spaces $H^q_{\Fourier}(\R^d,\Hil)$ and $\mathfrak{h}_{\Gamma^*}^q(\Hil)$ are isomorphic as Banach spaces. 
		\item The density follows immediately from the inclusions 
		\begin{align*}
			c_{0,\Gamma^*}(\Hil) \subset \mathfrak{h}_{\Gamma^*}^q(\Hil) 
			\subseteq \mathfrak{h}_{\Gamma^*}(\Hil) 
			\subseteq \bigoplus_{\gamma^* \in \Gamma^*} L^2(\BZ_{\gamma^*},\Hil)
		\end{align*}
		and the density of $c_{0,\Gamma^*}(\Hil) \subset \bigoplus_{\gamma^* \in \Gamma^*} L^2(\BZ_{\gamma^*},\Hil)$. 
		\item Assumption~\ref{equivariant_calculus:assumption:setting}~(c), namely that $\tau$ is of order $q$, directly implies that there exists a constant $C_{\tau} > 0$ for which 
		\begin{align*}
			\bnorm{\psi_{\gamma^*}}_{\mathfrak{h}_{\gamma^*}(\Hil)} &= \bnorm{\tau(\gamma^*)^{-1} \psi_{\gamma^*}}_{L^2(\BZ_{\gamma^*},\Hil)} 
			\leq C_{\tau} \, \sexpval{\gamma^*}^q \, \bnorm{\psi_{\gamma^*}}_{L^2(\BZ_{\gamma^*},\Hil)} 
		\end{align*}
		is satisfied for all $\gamma^* \in \Gamma^*$ and $\psi_{\gamma^*} \in L^2(\BZ_{\gamma^*},\Hil)$. Writing out the $\mathfrak{h}_{\Gamma^*}^q$ norm and inserting this estimate yields 
		\begin{align*}
			\bnorm{\imath \bigl ( (\psi_{\gamma^*})_{\gamma^* \in \Gamma^*} \bigr )}_{\mathfrak{h}_{\Gamma^*}(\Hil)} \leq C_{\tau} \, \bnorm{(\psi_{\gamma^*})_{\gamma^* \in \Gamma^*}}_{\mathfrak{h}_{\Gamma^*}^q(\Hil)} 
			. 
		\end{align*}
		\item We have to show that the norms of $\mathfrak{h}_{\Gamma^*}(\Hil)$ and $\mathfrak{h}_{\Gamma^*}^q(\Hil)$ are equivalent. The upper bound has already been shown in item~(2). The lower bound can be proven analogously to the upper bound, exploiting the lower bound~\eqref{appendix:equivariant_operators:eqn:direct_sum_weighted_Hilbert_spaces_ellipticity_norms_plaquettes} we have assumed to be true. 
	\end{enumerate}
\end{proof}
Characterization~(c) of Proposition~\ref{equivariant_calculus:prop:characterization_equivariant_operators}, which corresponds to characterization~(d) in the extended Proposition~\ref{appendix:equivariant_operators:prop:characterization_equivariant_operators} below, involves (the Fourier transform of) non-magnetic Sobolev spaces. This is to deal with the polynomial growth of $\bnorm{\tau(\gamma^*)}_{\mathcal{B}(\Hil)}$ as $\sabs{\gamma^*} \rightarrow \infty$. 
\begin{corollary}\label{appendix:equivariant_operators:cor:direct_sum_weighted_Hilbert_spaces_translated_Brillouin_zones}
	Suppose we are in the setting of Lemma~\ref{appendix:equivariant_operators:lem:direct_sum_weighted_Hilbert_spaces_translated_Brillouin_zones}. Then the following holds true: 
	\begin{enumerate}[(1)]
		\item The map 
		\begin{align*}
			U_{\Gamma^*} : \, &H^q_{\Fourier}(\R^d,\Hil) \longrightarrow \mathfrak{h}_{\Gamma^*}(\Hil) 
			, 
			\\
			&\Psi \mapsto \bigl ( \Psi \vert_{\BZ_{\gamma^*}} \bigr )_{\gamma^* \in \Gamma^*}
		\end{align*}
		is a continuous injection with dense range. 
		\item In the above claim, we can replace $q$ by a larger $q + q'$, $q' > 0$, without altering the conclusion. 
		%
		\item When the group action is elliptic in the sense that it satisfies equation~\eqref{appendix:equivariant_operators:eqn:direct_sum_weighted_Hilbert_spaces_ellipticity_norms_plaquettes}, then the map $U_{\Gamma^*}$ is a bijection with bounded inverse. 
		\item In the special case where the group action $\tau : \Gamma^* \longrightarrow \mathcal{U}(\Hil)$ takes values in the \emph{unitary} operators, the map $U_{\Gamma^*}$ is a unitary.
	\end{enumerate}
\end{corollary}
\begin{proof}
	\begin{enumerate}[(1)]
		\item The map $U_{\Gamma^*}$ is the composition of the map $\Psi \mapsto \bigl ( \Psi \vert_{\BZ_{\gamma^*}} \bigr )_{\gamma^* \in \Gamma^*}$ that facilitates the identification of $H^q_{\Fourier}(\R^d,\Hil)$ and $\mathfrak{h}_{\Gamma^*}^q$, and the injection $\mathfrak{h}_{\Gamma^*}^q(\Hil) \hookrightarrow \mathfrak{h}_{\Gamma^*}(\Hil)$ from Lemma~\ref{appendix:equivariant_operators:lem:direct_sum_weighted_Hilbert_spaces_translated_Brillouin_zones}~(1) and (2), respectively. According to the Lemma, both maps are continuous and thus, their composition is as well. 
		
		Injectivity is evident from the definition: $U_{\Gamma^*} \Psi = 0$ holds if and only if all the restrictions to the plaquettes $\BZ_{\gamma^*}$ vanish on a set of full measure. Thus, $\Psi = 0$ is the only solution and $U_{\Gamma^*}$ is injective. 
		\item This follows directly from $H^{q+q'}_{\Fourier}(\R^d,\Hil) \subseteq H^q_{\Fourier}(\R^d,\Hil)$ and 
		\begin{align*}
			\snorm{\Psi}_{H^{q+q'}_{\Fourier}(\R^d,\Hil)} \leq \snorm{\Psi}_{H^q_{\Fourier}(\R^d,\Hil)} 
			&&
			\forall \Psi \in H^{q+q'}_{\Fourier}(\R^d,\Hil)
			, 
		\end{align*}
		which is a consequence of $\sexpval{K}^q \leq \sexpval{K}^{q + q'}$. 
		\item Lemma~\ref{appendix:equivariant_operators:lem:direct_sum_weighted_Hilbert_spaces_translated_Brillouin_zones}~(4) tells us that if $\tau$ is elliptic in the sense that equation~\eqref{appendix:equivariant_operators:lem:direct_sum_weighted_Hilbert_spaces_translated_Brillouin_zones} holds, then the Banach spaces $H^q_{\Fourier}(\R^d,\Hil)$ and $\mathfrak{h}_{\Gamma^*}(\Hil)$ are isomorphic Banach spaces. In fact, $U_{\Gamma^*}$ is the isomorphism we are looking for. 
		
		Lower and upper bounds in the estimate 
		\begin{align*}
			c \, \snorm{\Psi}_{H^q_{\Fourier}(\R^d,\Hil)} \leq \bnorm{U_{\Gamma^*} \Psi}_{\mathfrak{h}_{\Gamma^*}^q(\Hil)} 
			= \bnorm{\bigl( \Psi \vert_{\BZ_{\gamma^*}} \bigr )_{\gamma^* \in \Gamma^*}}_{\mathfrak{h}_{\Gamma^*}(\Hil)} 
			\leq C \, \snorm{\Psi}_{H^q_{\Fourier}(\R^d,\Hil)} 
		\end{align*}
		that shows equivalence of the norms give us continuity of $U_{\Gamma^*}^{-1}$ and $U_{\Gamma^*}$. 
		\item This follows directly from (3): when the group action $\tau$ is unitary, then we may choose $q = 0$. Moreover, the scalar products — and thus, the induced norms — of $\mathfrak{h}_{\Gamma^*}^q(\Hil)$, $\mathfrak{h}_{\Gamma^*}(\Hil)$ and $H^0_{\Fourier}(\R^d,\Hil) = L^2(\R^d,\Hil)$ all coincide. Consequently, $U_{\Gamma^*}$ is a surjective norm-preserving map between Hilbert spaces, and as such unitary. 
	\end{enumerate}
\end{proof}
%

\subsection{Equivalent characterizations of operators} 
\label{appendix:equivariant_operators:equivalent_characterizations_operators}
We will now generalize and extend Proposition~\ref{equivariant_calculus:prop:characterization_equivariant_operators} in two ways: first of all, rather than just look at $F_0$ as a linear map between  
\begin{align*}
	\domain(F_0) \subseteq L^2(\BZ,\Hil) = \mathfrak{h}_0(\Hil) \longrightarrow L^2(\BZ,\Hil') = \mathfrak{h}_0(\Hil')
	, 
\end{align*}
we will instead allow one to pick \emph{any} fundamental cell $\BZ_{\gamma^*}$. Due to the presence of $\tau(\gamma^*)$ in the equivariance condition that defines $L^2_{\eq}(\R^d,\Hil)$, it is more suitable to work with the weighted Hilbert spaces $\mathfrak{h}_{\gamma^*}(\Hil^{(\prime)})$ rather than $L^2(\BZ_{\gamma^*},\Hil^{(\prime)})$. 

The second modification is that we will look at the collection $( F_{\gamma^*} )_{\gamma^* \in \Gamma^*}$ as a whole for two reasons: first of all, it is easier to make the equivariance condition explicit through the unitary $U_{\gamma^*,\tilde{\gamma}^*}$ from Lemma~\ref{appendix:equivariant_operators:lem:weighted_Hilbert_spaces_translated_Brillouin_zones}~(4). And secondly, upon identifying $\mathfrak{h}_{\Gamma^*}(\Hil)$ with a dense subspace of $L^2(\R^d,\Hil)$, we can more easily relate $\widehat{F} : \domain(\widehat{F}) \subseteq L^2(\R^d,\Hil) \longrightarrow L^2(\R^d,\Hil')$ with a collection $(F_{\gamma^*})_{\gamma^* \in \Gamma^*}$ of operators. 

The third and final modification is that we have reordered the items to follow the natural order in the proof. 
\begin{proposition}\label{appendix:equivariant_operators:prop:characterization_equivariant_operators}
	There is a one-to-one correspondence between 
	\begin{enumerate}[(a)]
		\item a single operator $F_{\gamma^*} : \domain(F_{\gamma^*}) \subseteq \mathfrak{h}_{\gamma^*}(\Hil) \longrightarrow \mathfrak{h}_{\gamma^*}(\Hil')$ for some $\gamma^* \in \Gamma^*$ and two group actions $\tau^{(\prime)} : \Gamma^* \longrightarrow \mathcal{GL}(\Hil^{(\prime)})$, 
		\item a collection of densely defined operators $F_{\gamma^*} : \domain(F_{\gamma^*}) = \tau(\gamma^*) \domain(F_0) \subseteq L^2(\BZ,\Hil) \longrightarrow L^2(\BZ,\Hil')$ related by 
		\begin{align}
			F_{\gamma^*} &= U'_{\gamma^*,\tilde{\gamma}^*} \, F_{\tilde{\gamma}^*} \, U_{\gamma^*,\tilde{\gamma}^*}^{-1}
			&&
			\forall \gamma^* , \tilde{\gamma}^* \in \Gamma^* 
			, 
			\label{equivariant_calculus:eqn:equivariance_condition_plaquettes}
		\end{align}
		\item a densely defined operator $F : \domain(F) \subseteq L^2_{\eq}(\R^d , \Hil) \longrightarrow L^2_{\eq}(\R^d , \Hil')$, and 
		\item a densely defined operator $\widehat{F} : \widehat{\domain}(F) \subseteq L^2(\R^d , \Hil) \longrightarrow L^2(\R^d , \Hil')$ subject to the equivariance condition 
		\begin{align}
			\widehat{T}'_{\gamma^*} \, \widehat{F} \, \widehat{T}_{\gamma^*}^{-1} = \bigl ( \id_{L^2(\R^d)} \otimes \tau'(\gamma^*) \bigr ) \, \widehat{F} \, \bigl ( \id_{L^2(\R^d)} \otimes \tau(\gamma^*)^{-1} \bigr ) 
			&&
			\forall \gamma^* \in \Gamma^* 
			, 
			\label{appendix:equivariant_operators:eqn:equivariance_condition_L2_Rd}
		\end{align}
		and the domain 
		\begin{align}
			\widehat{T}_{\gamma^*} \bigl ( \domain(\widehat{F}) \bigr ) &= \bigl (\id_{L^2(\BZ_{\gamma^*})} \otimes \tau(\gamma^*) \bigr ) \Bigl ( \domain(\widehat{F}) \Bigr )
			, 
			\label{appendix:equivariant_operators:eqn:equivariance_condition_domain_L2_Rd}
		\end{align}
		where $\bigl ( \widehat{T}^{(\prime)}_{\gamma^*} \Psi^{(\prime)} \bigr )(k) := \Psi(k^{(\prime)} - \gamma^*) \in \Hil^{(\prime)}$ denotes the translation by $\gamma^* \in \Gamma^*$ on $L^2(\R^d,\Hil^{(\prime)})$. 
	\end{enumerate}
\end{proposition}
\begin{proof}
	The structure of the proof is as follows: 
	\begin{align*}
		\bfig
			\node a(-400,400)[\mbox{(a)}]
			\node b(0,200)[\mbox{(b)}]
			\node c(-400,0)[\mbox{(c)}]
			\node d(400,200)[\mbox{(d)}]
			\arrow/=>/[a`b;]
			\arrow/=>/[b`c;]
			\arrow/=>/[c`a;]
			\arrow/<=>/[b`d;]
		\efig
	\end{align*}
  Since the Hilbert space structure is unimportant, \textbf{we shall be working in the category of Banach spaces}. Consequently, we can identify Hilbert spaces whose norms are equivalent such as $\mathfrak{h}_{\gamma^*}(\Hil) = L^2(\BZ_{\gamma^*},\Hil)$ even though these are different \emph{Hilbert} spaces. It is not important that operators such as $U_{\tilde{\gamma}^*,\gamma^*} \in \mathcal{GL} \bigl ( L^2(\BZ_{\gamma^*},\Hil) \, , \, L^2(\BZ_{\tilde{\gamma}^*},\Hil) \bigr )$ are \emph{unitary} with respect to one scalar product. All that matters is that they are bounded operators with bounded inverses, a fact that stays true in any equivalent norm. 
	
  Lastly, in the preceding Lemmas and Corollary, we have introduced several maps on $\Hil$-valued Banach spaces. Whenever we deal with $\Hil'$-valued Banach spaces, we will systematically add a prime so that \eg $U'_{\tilde{\gamma}^*,\gamma^*} \in \mathcal{GL} \bigl ( L^2(\BZ_{\gamma^*},\Hil') \, , \, L^2(\BZ_{\tilde{\gamma}^*},\Hil') \bigr )$. 
	\medskip
	
	\noindent
	(a) $\Longrightarrow$ (b): Suppose we are given a single densely defined operator $F_{\gamma^*} : \domain(F_{\gamma^*}) \subseteq L^2(\BZ_{\gamma^*},\Hil)$ and two group actions $\tau$ and $\tau'$. To construct the operator at $\tilde{\gamma}^* \in \Gamma^*$ we set 
	\begin{align*}
		F_{\tilde{\gamma}^*} := U'_{\tilde{\gamma}^*,\gamma^*} \, F_{\gamma^*} \, U_{\tilde{\gamma}^*,\gamma^*}^{-1}
	\end{align*}
	and endow it with the domain 
	\begin{align*}
		\domain(F_{\tilde{\gamma}^*}) := U_{\tilde{\gamma}^*,\gamma^*} \bigl ( \domain(F_{\gamma^*}) \bigr ) 
		. 
	\end{align*}
	As the image of a dense set under a $\mathcal{GL} \bigl ( L^2(\BZ_{\gamma^*},\Hil) \, , \, L^2(\BZ_{\tilde{\gamma}^*},\Hil) \bigr )$ map, the domain $\domain(F_{\tilde{\gamma}^*}) \subseteq L^2(\BZ_{\tilde{\gamma}^*},\Hil)$ is again dense. And evidently, the covariance condition~\eqref{equivariant_calculus:eqn:equivariance_condition_plaquettes} — which implies also the covariance of the domains — is baked right into the definition of $(F_{\gamma^*})_{\gamma^* \in \Gamma^*}$. 
	\medskip
	
	\noindent
	(b) $\Longrightarrow$ (c): Suppose we are given a collection of operators $(F_{\gamma^*})_{\gamma^* \in \Gamma^*}$ that satisfies the equivariance condition~\eqref{equivariant_calculus:eqn:equivariance_condition_plaquettes}. Pick any $\gamma^* \in \Gamma^*$. With the help of 
	\begin{align*}
		U_{\gamma^*} \in \mathcal{GL} \bigl ( L^2_{\eq}(\R^d,\Hil) \, , \, L^2(\BZ_{\gamma^*},\Hil) \bigr ) 
	\end{align*}
	from Lemma~\ref{appendix:equivariant_operators:lem:weighted_Hilbert_spaces_translated_Brillouin_zones}, we define the operator 
	\begin{align}
		F := U^{\prime \, -1}_{\gamma^*} \, F_{\gamma^*} \, U_{\gamma^*} 
		\label{appendix:equivariant_operators:eqn:relation_F_L2_eq_F_gamma_ast}
	\end{align}
	endowed with the domain $\domain(F) := U_{\gamma^*}^{-1} \bigl ( \domain(F_{\gamma^*}) \bigr )$. Clearly, this gives us a well-defined operator between $L^2_{\eq}(\R^d,\Hil)$ and $L^2_{\eq}(\R^d,\Hil')$. As the range of a dense set under a $\mathcal{GL}$ map, the domain $\domain(F) \subseteq L^2_{\eq}(\R^d,\Hil)$ is again dense. 
	
	Our remaining task is to show that $F$ is independent of our choice of $\gamma^* \in \Gamma^*$. So let $0 \in \Gamma^*$ be the zero element of our group. Then analogously to above we can define the operator 
	\begin{align*}
		\tilde{F} := U^{\prime \, -1}_0 \, F_0 \, U_0 
	\end{align*}
	equipped with the domain $\domain(\tilde{F}) := U_0^{-1} \bigl ( \domain(F_0) \bigr )$. Covariance~\eqref{equivariant_calculus:eqn:equivariance_condition_plaquettes} relates 
	\begin{align*}
		F_{\gamma^*} &= U'_{\gamma^*,0} \, F_0 \, U_{\gamma^*,0}^{-1} 
		\\
		&= U'_{\gamma^*} \, U^{\prime \, -1}_0 \, F_0 \, U_0 \, U_{\gamma^*}^{-1} 
	\end{align*}
	to $F_0$, which is equivalent to 
	\begin{align*}
		F &= U^{\prime \, -1}_{\gamma^*} \, F_{\gamma^*} \, U_{\gamma^*} 
		\\
		&= U^{\prime \, -1}_{\gamma^*} \, U'_{\gamma^*} \, U^{\prime \, -1}_0 \, F_0 \, U_0 \, U_{\gamma^*}^{-1} \, U_{\gamma^*} 
		\\
		&= U^{\prime \, -1}_0 \, F_0 \, U_0
		= \tilde{F} 
		. 
	\end{align*}
	Likewise, the domains coincide by covariance, 
	\begin{align*}
		\domain(F_0) = U_0 \, U_{\gamma^*}^{-1} \bigl ( \domain(F_{\gamma^*}) \bigr ) 
		. 
	\end{align*}
	This shows that the operator $F$ is independent of our choice of $\gamma^* \in \Gamma^*$.
	\medskip
	
	\noindent
	(c) $\Longrightarrow$ (a):
	Suppose we are given an operator $F : \domain(F) \subseteq L^2_{\eq}(\R^d,\Hil) \longrightarrow L^2_{\eq}(\R^d,\Hil')$. Pick a lattice vector $\gamma^* \in \Gamma^*$. Then we define the operator 
	\begin{align*}
		F_{\gamma^*} := U'_{\gamma^*} \, F \, U_{\gamma^*}^{-1} 
	\end{align*}
	by concatenating $F$ with the $\mathcal{GL}$ maps $U_{\gamma^*}^{-1}$ and $U'_{\gamma^*}$ from Lemma~\ref{appendix:equivariant_operators:lem:weighted_Hilbert_spaces_translated_Brillouin_zones}~(3), and equipping it with the domain $\domain(F_{\gamma^*}) := U_{\gamma^*} \bigl ( \domain(F) \bigr )$. As the range of a dense set under a $\mathcal{GL}$ map, the domain $\domain(F_{\gamma^*}) \subseteq L^2(\BZ_{\gamma^*},\Hil)$ is again dense. 
	\medskip
	
	\noindent
	(b) $\Longrightarrow$ (d): Suppose we are given a collection of covariant operators $(F_{\gamma^*})_{\gamma^* \in \Gamma^*}$. For any lattice vector $\gamma^* \in \Gamma^*$ we define the surjection 
	\begin{align*}
		\pi_{\gamma^*} : L^2(\R^d,\Hil) \longrightarrow L^2(\BZ_{\gamma^*},\Hil)
		, \; 
		\Psi \mapsto \Psi \vert_{\BZ_{\gamma^*}} 
		, 
	\end{align*}
	that restricts vectors to the fundamental cell located at $\gamma^*$. Furthermore, we define the Banach space isomorphisms   
	\begin{align*}
		V_{\Gamma^*}^{(\prime)} : \, &L^2(\R^d,\Hil^{(\prime)}) \longrightarrow \bigoplus_{\gamma^* \in \Gamma^*} L^2(\BZ_{\gamma^*},\Hil^{(\prime)}) 
		, 
		\\
		&\Psi^{(\prime)} \mapsto \bigl ( \pi_{\gamma^*}(\Psi^{(\prime)}) \bigr )_{\gamma^* \in \Gamma^*}
		. 
	\end{align*}
	Then considering the covariant collection $(F_{\gamma^*})_{\gamma^* \in \Gamma^*}$ as a densely defined blockdiagonal operator on the Hilbert space $\bigoplus_{\gamma^* \in \Gamma^*} L^2(\BZ_{\gamma^*},\Hil^{(\prime)})$, we define $\widehat{F}$ as the operator 
	\begin{align*}
		\widehat{F} := V^{\prime \, -1}_{\Gamma^*} \, (F_{\gamma^*})_{\gamma^* \in \Gamma^*} \, V_{\Gamma^*} 
	\end{align*}
	endowed with the domain 
	\begin{align*}
		\domain(\widehat{F}) := V_{\Gamma^*}^{-1} \Bigl ( \bigl ( \domain(F_{\gamma^*}) \bigr )_{\gamma^* \in \Gamma^*} \Bigr ) 
		\subseteq L^2(\R^d,\Hil)
		. 
	\end{align*}
	The fact that $V_{\Gamma^*} \in \mathcal{GL} \Bigl ( L^2(\R^d,\Hil^{(\prime)}) \, , \, \bigoplus_{\gamma^* \in \Gamma^*} L^2(\BZ_{\gamma^*},\Hil^{(\prime)})  \Bigr )$ and the density of all the $\domain(F_{\gamma^*}) \subseteq L^2(\BZ_{\gamma^*},\Hil)$ imply that the domain of $\widehat{F}$ lies dense. 
	
	To verify the covariance condition~\eqref{appendix:equivariant_operators:eqn:equivariance_condition_L2_Rd}, we note that for any $\tilde{\gamma}^* \in \Gamma^*$ and $k \in \BZ_{\gamma^*}$ the action of the lattice translation $\widehat{T}_{\tilde{\gamma}^*}$ simplifies to 
	\begin{align}
		\bigl ( V_{\Gamma^*} \, \widehat{T}_{\tilde{\gamma}^*} \, V_{\Gamma^*}^{-1} (\psi_{\gamma^*})_{\gamma^* \in \Gamma^*} \bigr )_{\gamma^*}(k) &= \psi_{\gamma^* + \tilde{\gamma}^*}(k - \tilde{\gamma}^*) 
		\in L^2(\BZ_{\gamma^*},\Hil)
		. 
		\label{appendix:equivariant_operators:eqn:translations_in_direct_sum_plaquettes_representation}
	\end{align}
	Hence, the covariance condition~\eqref{equivariant_calculus:eqn:equivariance_condition_plaquettes} for the $(F_{\gamma^*})_{\gamma^* \in \Gamma^*}$ and their domains implies \eqref{appendix:equivariant_operators:eqn:equivariance_condition_L2_Rd}. Moreover, the equivariance of the domains of the $F_{\gamma^*}$ translates to \eqref{appendix:equivariant_operators:eqn:equivariance_condition_domain_L2_Rd}. 
	\medskip
	
	\noindent
	(d) $\Longrightarrow$ (b): Suppose we are given a densely defined operator $\widehat{F} : \domain(\widehat{F}) \subseteq L^2(\R^d,\Hil) \longrightarrow L^2(\R^d,\Hil')$. For the most part, we just need to read the proof for “(b) $\Longrightarrow$ (d)” in reverse: the covariance condition~\eqref{appendix:equivariant_operators:eqn:equivariance_condition_L2_Rd} guarantees that the operator 
	\begin{align}
		V_{\Gamma^*} \, \widehat{F} \, V_{\Gamma^*}^{-1} = (F_{\gamma^*})_{\gamma^* \in \Gamma^*}
		\label{appendix:equivariant_operators:eqn:relation_widehat_F_collection_F_gamma_ast}
	\end{align}
	is blockdiagonal with respect to the direct sum decomposition. The domain 
	\begin{align*}
		\domain \bigl ( (F_{\gamma^*})_{\gamma^* \in \Gamma^*} \bigr ) := V_{\Gamma^*} \bigl ( \domain(\widehat{F}) \bigr ) 
		\subseteq \bigoplus_{\gamma^* \in \Gamma^*} L^2(\BZ_{\gamma^*},\Hil)
	\end{align*}
	is necessarily dense as the image of a dense set under a $\mathcal{GL}$ map. Moreover, equation~\eqref{appendix:equivariant_operators:eqn:translations_in_direct_sum_plaquettes_representation} guarantees the domain satisfies \eqref{appendix:equivariant_operators:eqn:equivariance_condition_domain_L2_Rd}. 
\end{proof}
The first three characterizations~(a)–(c) lead to clear criteria for when these operators are bounded: since $L^2_{\eq}(\R^d,\Hil)$ and $\mathfrak{h}_{\gamma^*}(\Hil)$ are \emph{unitarily} equivalent, the norms of all these operators are the same. In contrast, the identification between $\widehat{F}$ and $F_{\gamma^*}$ involves an injection and a surjection as well as different weights. Consequently, there is no simple relation between the operator norm of $\widehat{F}$ and those of $F$ or $F_{\gamma^*}$. 
\begin{proposition}\label{appendix:equivariant_operators:prop:equivalence_boundedness_equivariant_operators}
	Suppose we are in the setting of Proposition~\ref{appendix:equivariant_operators:prop:characterization_equivariant_operators}. 
	\begin{enumerate}[(1)]
		\item The boundedness of each of the operators from characterizations~(a)–(c) implies the boundedness of the others. Specifically, the norms are related by 
		\begin{align*}
			\bnorm{F_{\gamma^*}}_{\mathcal{B}(\mathfrak{h}_{\gamma^*}(\Hil) , \mathfrak{h}_{\gamma^*}(\Hil'))} &= \bnorm{F_0}_{\mathcal{B}(L^2(\BZ,\Hil) , L^2(\BZ,\Hil'))} 
			\\
			&= \snorm{F}_{\mathcal{B}(L^2_{\eq}(\R^d,\Hil),L^2_{\eq}(\R^d,\Hil'))}
			. 
		\end{align*}
		\item The operators from characterizations~(a)–(c) are bounded if and only if 
		\begin{align*}
			\widehat{F} : H^{q + q'}_{\Fourier}(\R^d,\Hil) \longrightarrow L^2(\R^d,\Hil') 
		\end{align*}
		from characterization~(d) defines a bounded operator, where $q$ and $q'$ are the orders of growth for the group actions $\tau$ and $\tau'$. 
	\end{enumerate}
\end{proposition}
\begin{proof}
	For this proof, we need to distinguish between $\mathfrak{h}_{\gamma^*}(\Hil)$ and $L^2(\BZ_{\gamma^*},\Hil)$ since we claim equality of specific operator norms. Thus, \textbf{we will work in the category of Hilbert spaces}. 
	\begin{enumerate}[(1)]
		\item Suppose we are given a single operator $F_{\gamma^*}$ and two group actions $\tau^{(\prime)}$ as in characterization~(a), where $\gamma^*$ indicates the location of the unit cell. Then we may use the equivariance condition~\eqref{equivariant_calculus:eqn:equivariance_condition_plaquettes} to translate it back to the origin, and obtain the norm estimate 
		\begin{align*}
			\bnorm{F_{\gamma^*}}_{\mathcal{B}(\mathfrak{h}_{\gamma^*}(\Hil) , \mathfrak{h}_{\gamma^*}(\Hil'))} &= \sup_{\substack{\psi_{\gamma^*} \in \mathfrak{h}_{\gamma^*}(\Hil) \\ \norm{\psi_{\gamma^*}}_{\mathfrak{h}_{\gamma^*}(\Hil)} = 1}} \bnorm{F_{\gamma^*} \psi_{\gamma^*}}_{\mathfrak{h}_{\gamma^*}(\Hil')}
			\\
			&= \sup_{\substack{\psi_0 \in L^2(\BZ,\Hil) \\ \norm{\psi_0}_{L^2(\BZ,\Hil)} = 1}} \bnorm{U'_{0,\gamma^*} \, F_{\gamma^*} \, U_{0,\gamma^*}^{-1} \psi_0}_{L^2(\BZ,\Hil')}
			\\
			&= \sup_{\substack{\psi_0 \in L^2(\BZ,\Hil) \\ \norm{\psi_0}_{L^2(\BZ,\Hil)} = 1}} \bnorm{F_0 \psi_0}_{L^2(\BZ,\Hil')}
			\\
			&= \bnorm{F_0}_{\mathcal{B}(L^2(\BZ,\Hil) , \mathcal{B}(L^2(\BZ,\Hil'))}
			. 
		\end{align*}
		Note that the Hilbert spaces $\mathfrak{h}_0(\Hil^{(\prime)}) = L^2(\BZ_0,\Hil^{(\prime)}) \equiv L^2(\BZ,\Hil^{(\prime)})$ agree since the weight operator $\tau^{(\prime)}(0) = \id_{\Hil^{(\prime)}}$ reduces to the identity. So clearly, $F_{\gamma^*}$ is bounded if and only if $F_0$ is, and their norms agree. As $\gamma^*$ had been chosen arbitrarily, we deduce the boundedness of each of the operators in the collection $(F_{\gamma^*})_{\gamma^* \in \Gamma^*}$ from characterization~(b). 
		
		Now we move on to characterization~(c): our choice of scalar product on $L^2_{\eq}(\R^d,\Hil^{(\prime)})$ picks the unit cell at the origin $0 \in \Gamma^*$ in its definition, which links it to the operator $F_0$ from (b). The operators $F$ and $F_0$ are related by \eqref{appendix:equivariant_operators:eqn:relation_F_L2_eq_F_gamma_ast} for the special choice $\gamma^* = 0 \in \Gamma^*$. Since the maps that facilitate the change of representation are unitary (Lemma~\ref{appendix:equivariant_operators:lem:weighted_Hilbert_spaces_translated_Brillouin_zones}~(3)), the norms of $F$ and $F_0$ necessarily agree, 
		\begin{align*}
			\bnorm{F_0}_{\mathcal{B}(L^2(\BZ,\Hil),L^2(\BZ,\Hil'))} = \snorm{F}_{\mathcal{B}(L^2_{\eq}(\R^d,\Hil),L^2_{\eq}(\R^d,\Hil'))}
			. 
		\end{align*}
		\item $\Longrightarrow$: Suppose $\widehat{F} : H^{q + q'}_{\Fourier}(\R^d,\Hil) \longrightarrow L^2(\R^d,\Hil')$ is bounded, where $q \geq 0$ and $q' \geq 0$ are the orders of growth of the group actions $\tau$ and $\tau'$. 
		Then for any $(\psi_{\gamma^*})_{\gamma^* \in \Gamma^*} \in \mathfrak{h}_{\Gamma^*}^{q+q'}(\Hil)$ we can estimate 
		\begin{align*}
			\bnorm{F_0 \psi_0}_{L^2(\BZ,\Hil')}^2 &\leq \sum_{\gamma^* \in \Gamma^*} \bnorm{F_{\gamma^*} \psi_{\gamma^*}}_{L^2(\BZ_{\gamma^*},\Hil')}^2 
			\\
			&= \bnorm{(F_{\gamma^*})_{\gamma^* \in \Gamma^*} \, (\psi_{\gamma^*})_{\gamma^* \in \Gamma^*}}_{\mathcal{B}(\mathfrak{h}^{q+q'}_{\Gamma^*}(\Hil),\bigoplus_{\gamma^* \in \Gamma^*} L^2(\BZ_{\gamma^*},\Hil'))}^2 
			. 
		\end{align*}
		The covariant collection $(F_{\gamma^*})_{\gamma^* \in \Gamma^*}$ is related to $\widehat{F}$ via equation~\eqref{appendix:equivariant_operators:eqn:relation_widehat_F_collection_F_gamma_ast} by the unitaries 
		\begin{align*}
			V_{\Gamma^*}^{(\prime)} : L^2(\R^d,\Hil^{(\prime)}) \longrightarrow \bigoplus_{\gamma^* \in \Gamma^*} L^2(\BZ_{\gamma^*},\Hil^{(\prime)})
		\end{align*}
		from the proof of Proposition~\ref{appendix:equivariant_operators:prop:characterization_equivariant_operators}. Moreover, the image 
		\begin{align*}
			V_{\Gamma^*}^{-1} \bigl ( \mathfrak{h}_{\Gamma^*}^{q + q'}(\Hil) \bigr ) = H^{q + q'}_{\Fourier}(\R^d,\Hil) 
		\end{align*}
		is nothing by the Sobolev space of order $q + q'$ and the restriction of the unitary 
		\begin{align*}
			U_{\Gamma^*}^{q + q'} := V_{\Gamma^*} \big \vert_{H^{q + q'}_{\Fourier}(\R^d,\Hil)} :  H^{q + q'}_{\Fourier}(\R^d,\Hil) \longrightarrow \mathfrak{h}_{\Gamma^*}^{q + q'}(\Hil)
		\end{align*}
		is unitary with respect to the weighted scalar products (\cf Corollary~\ref{appendix:equivariant_operators:cor:direct_sum_weighted_Hilbert_spaces_translated_Brillouin_zones}). 
		
		Consequently, the norms 
		\begin{align*}
			\bnorm{(F_{\gamma^*})_{\gamma^* \in \Gamma^*}}_{\mathcal{B}(\mathfrak{h}_{\Gamma^*}^{q + q'}(\Hil),\bigoplus_{\gamma^* \in \Gamma^*} L^2(\BZ_{\gamma^*},\Hil'))} = \norm{\widehat{F}}_{\mathcal{B}(H^{q+q'}_{\Fourier}(\R^d,\Hil),L^2(\R^d,\Hil'))} < \infty 
		\end{align*}
		coincide, which shows that $\norm{F_0}_{\mathcal{B}(L^2(\BZ,\Hil),L^2(\BZ,\Hil'))} < \infty$ is finite. By (1), the norms of the $F_{\gamma^*}$ are also necessarily finite. 
		\medskip
		
		\noindent
		$\Longleftarrow$: Suppose one of the operators from characterizations~(a)–(c) is bounded. By (1) we may therefore assume that the covariant family of operators $(F_{\gamma^*})_{\gamma^* \in \Gamma^*}$ is bounded. 
		
		Then with the help of the unitary $V_{\Gamma^*}$ and equation~\eqref{appendix:equivariant_operators:eqn:relation_widehat_F_collection_F_gamma_ast} we can rewrite the square of the norm as 
		\begin{align*}
			&\snorm{\widehat{F}}_{\mathcal{B}(H^{q+q'}_{\Fourier}(\R^d,\Hil),L^2(\R^d,\Hil'))}^2 
			= \\
			&\qquad 
			= \sup_{\snorm{\Psi}_{H^{q+q'}_{\Fourier}(\R^d,\Hil)} = 1} \bnorm{\widehat{F} \Psi}_{L^2(\R^d,\Hil')}^2 
			\\
			&\qquad 
			= \sup_{\snorm{\Phi}_{L^2(\R^d,\Hil)} = 1} \bnorm{\widehat{F} \, \sexpval{K}^{-(q+q')} \Phi}_{L^2(\R^d,\Hil')}^2 
			\\
			&\qquad 
			= \sum_{\gamma^* \in \Gamma^*} \bnorm{F_{\gamma^*} \, \sexpval{K}^{-(q+q')} \pi_{\gamma^*}(\Phi)}_{L^2(\BZ_{\gamma^*},\Hil')}^2
			\\
			&\qquad 
			\leq \sum_{\gamma^* \in \Gamma^*} \bnorm{F_{\gamma^*}}_{\mathcal{B}(L^2(\BZ_{\gamma^*},\Hil),L^2(\BZ_{\gamma^*},\Hil'))}^2 \, \bnorm{\sexpval{K}^{-(q+q')}}_{\mathcal{B}(L^2(\BZ_{\gamma^*}))}^2 \bnorm{\pi_{\gamma^*}(\Phi)}_{L^2(\BZ_{\gamma^*},\Hil)}^2
			. 
		\end{align*}
		Let us discuss each of the factors in turn. First of all, we note that the maps 
		\begin{align*}
			\id_{L^2(\BZ_{\gamma^*})} \otimes \tau^{(\prime)}(\gamma^*) : \mathfrak{h}_{\gamma^*}(\Hil^{(\prime)}) \longrightarrow L^2(\BZ_{\gamma^*},\Hil^{(\prime)})
		\end{align*}
		are unitary. Therefore, we obtain a uniform bound for the norm of $F_{\gamma^*}$, 
		\begin{align*}
			&\bnorm{F_{\gamma^*}}_{\mathcal{B}(L^2(\BZ_{\gamma^*},\Hil),L^2(\BZ_{\gamma^*},\Hil'))} 
			= \\
			&\qquad 
			= \bnorm{\bigl ( \id_{L^2(\BZ_{\gamma^*})} \otimes \tau'(\gamma^*) \bigr )^{-1} \, F_{\gamma^*} \, \bigl ( \id_{L^2(\BZ_{\gamma^*})} \otimes \tau(\gamma^*) \bigr )}_{\mathcal{B}(\mathfrak{h}_{\gamma^*}(\Hil),\mathfrak{h}_{\gamma^*}(\Hil'))}
			\\
			&\qquad 
			\leq \bnorm{\tau'(\gamma^*)^{-1}}_{\mathcal{B}(\Hil')} \, \bnorm{\tau(\gamma^*)}_{\mathcal{B}(\Hil)} \, \bnorm{F_{\gamma^*}}_{\mathcal{B}(\mathfrak{h}_{\gamma^*}(\Hil),\mathfrak{h}_{\gamma^*}(\Hil'))}
			\\
			&\qquad 
			\leq C_{\tau} \, C_{\tau'} \, \sexpval{\gamma^*}^{q + q'} \, \bnorm{F_0}_{\mathcal{B}(L^2(\BZ,\Hil),L^2(\BZ,\Hil'))} 
			, 
		\end{align*}
		exploiting our assumption that $\tau$ and $\tau'$ are of orders $q$ and $q'$ as well as (1). Note that compared to (1), we use different norms for $\mathfrak{h}_{\gamma^*}(\Hil^{(\prime)})$ and $L^2(\BZ_{\gamma^*},\Hil^{(\prime)})$, the extra factors come precisely from the weights $\tau^{(\prime)}(\gamma^*)^{-1}$ that relate the two norms. 
		
		Moreover, for $\gamma^* \in \Gamma^*$ with $\sabs{\gamma^*} \geq R$ large enough, we can bound 
		\begin{align*}
			\tfrac{1}{2} \, \sabs{\gamma^*} \leq \sup_{k \in \BZ_{\gamma^*}} \sabs{k} 
			= \bnorm{\sabs{K}}_{\mathcal{B}(L^2(\BZ_{\gamma^*}))} 
			\leq 2 \, \sabs{\gamma^*} 
		\end{align*}
		from above and below, which immediately leads to the estimate 
		\begin{align*}
			\tfrac{1}{2^{q + q'}} \, \sexpval{\gamma^*}^{\pm (q + q')} \leq \sup_{k \in \BZ_{\gamma^*}} \sexpval{k}^{\pm (q + q')} = \bnorm{\sexpval{K}^{\pm (q + q')}}_{\mathcal{B}(L^2(\BZ_{\gamma^*}))} \leq 2^{q + q'} \, \sexpval{\gamma^*}^{\pm (q + q')}
		\end{align*}
		for all $\gamma^* \in \Gamma^*$ with $\sabs{\gamma^*} \geq R$ and $k \in \BZ_{\gamma^*}$. The growth of the norm due to the presence of $\tau$ and $\tau'$ is exactly canceled by the presence of $\sexpval{K}^{-(q + q')}$. Putting all the pieces together yields 
		\begin{align*}
			&\bnorm{F_{\gamma^*}}_{\mathcal{B}(L^2(\BZ_{\gamma^*},\Hil),L^2(\BZ_{\gamma^*},\Hil'))} 
			= \\
			&\qquad 
			\leq 4^{q + q'} \, C_R \, C_{\tau}^2 \, C_{\tau'}^2 \, \bnorm{F_0}_{\mathcal{B}(L^2(\BZ,\Hil),L^2(\BZ,\Hil'))}^2 \, \sup_{\snorm{\Phi}_{L^2(\R^d,\Hil)} = 1} \sum_{\gamma^* \in \Gamma^*} \bnorm{\pi_{\gamma^*}(\Phi)}_{L^2(\BZ_{\gamma^*},\Hil)}^2
			\\
			&\qquad 
			= 4^{q + q'} \, C_R \, C_{\tau}^2 \, C_{\tau'}^2 \, \bnorm{F_0}_{\mathcal{B}(L^2(\BZ,\Hil),L^2(\BZ,\Hil'))}^2 
			, 
		\end{align*}
		where $C_R > 0$ is a constant that takes into account that the estimate on the norm of $\sexpval{K}^{-(q + q')}$ only holds for sufficiently large $\sabs{\gamma^*} \geq R$. Thus, the norm of the operator $\widehat{F} : H^{q + q'}_{\Fourier}(\R^d,\Hil) \longrightarrow L^2(\R^d,\Hil')$ is bounded. 
	\end{enumerate}
\end{proof}
%

\subsection{Equivalence for operators between magnetic Sobolev spaces} 
\label{appendix:equivariant_operators:operators_magnetic_Sobolev_spaces}
Given the context, \emph{magnetic} pseudodifferential theory, many of the operators are naturally defined on  \emph{magnetic} Sobolev spaces. These give rise to \emph{bounded} operators when we regard them as maps between the corresponding magnetic Sobolev and $L^2$-spaces. 
\begin{corollary}\label{appendix:equivariant_operators:cor:equivalence_boundedness_equivariant_operators}
	Suppose we are in the setting of Proposition~\ref{appendix:equivariant_operators:prop:characterization_equivariant_operators}. 
	Further, assume we are given a bounded operator 
	\begin{align*}
		\widehat{F}^A : H^{q + q'}_{\Fourier,A}(\R^d,\Hil) \longrightarrow L^2(\R^d,\Hil') 
	\end{align*}
	from the \emph{magnetic} Sobolev space of order $q + q'$ that is covariant in the sense of characterization~(d). 
	\begin{enumerate}[(1)]
		\item Then $\widehat{F}^A$ defines the bounded operators 
		\begin{align*}
			F^A_{\gamma^*} &: H^{q + q'}_{\eq,\Fourier,A}(\BZ_{\gamma^*},\Hil) \longrightarrow L^2(\BZ_{\gamma^*},\Hil')
			, 
			&&
			\gamma^* \in \Gamma^*
			, 
			\\
			F^A &: H^{q + q'}_{\eq,\Fourier,A}(\R^d,\Hil) \longrightarrow L^2_{\eq}(\R^d,\Hil') 
			. 
		\end{align*}
		\item When $q + q' = 0$ or all the components of $A$ are bounded, then the operators $F^A_{\gamma^*}$ and $F^A$ define bounded operators between the corresponding $L^2$ spaces. 
		\item These operators and their relations are gauge-covariant in the following sense: given any $\vartheta \in \Cont^{\infty}_{\mathrm{pol}}(\R^d,\R)$, then the operators in the gauge $A' = A + \eps \, \dd \vartheta$ are related to those in the gauge $A$ via $\e^{+ \ii \lambda \vartheta(Q)}$ and $\e^{+ \ii \lambda \vartheta(\Reps)}$, respectively, 
		\begin{align*}
			\widehat{F}^{A'} = \e^{+ \ii \lambda \vartheta(Q)} \, \widehat{F}^A \, \e^{- \ii \lambda \vartheta(Q)} &: H^{q + q'}_{\eq,\Fourier,A'}(\R^d,\Hil) \longrightarrow L^2_{\eq}(\R^d,\Hil') 
			, 
			\\
			F^{A'} = \e^{+ \ii \lambda \vartheta(\Reps)} \, F^A \, \e^{- \ii \lambda \vartheta(\Reps)} &: H^{q + q'}_{\eq,\Fourier,A'}(\R^d,\Hil) \longrightarrow L^2_{\eq}(\R^d,\Hil') 
			, 
			\\
			F_{\gamma^*}^{A'} = \e^{+ \ii \lambda \vartheta(\Reps)} \big \vert_{\BZ_{\gamma^*}} \, F_{\gamma^*}^A \, \e^{- \ii \lambda \vartheta(\Reps)} \big \vert_{\BZ_{\gamma^*}} &: H^{q + q'}_{\eq,\Fourier,A'}(\BZ_{\gamma^*},\Hil) \longrightarrow \mathfrak{h}_{\gamma^*}(\Hil^{(\prime)})
			. 
		\end{align*}
		Here, $\e^{+ \ii \lambda \vartheta(\Reps)} \big \vert_{\BZ_{\gamma^*}}$ denotes the unique operator obtained from the identification $U^{(\prime)}_{\gamma^*} : L^2_{\eq}(\R^d,\Hil^{(\prime)}) \longrightarrow \mathfrak{h}_{\gamma^*}(\Hil^{(\prime)})$ (\cf Lemma~\ref{appendix:equivariant_operators:lem:weighted_Hilbert_spaces_translated_Brillouin_zones}~(3)). 
	\end{enumerate}
\end{corollary}
In order to connect the norms of the operators $\widehat{F}^A$ with those of $F^A$ and $F^A_{\gamma^*}$, $\gamma^* \in \Gamma^*$, we will need a way to identify suitable elements of the \emph{equivariant} magnetic Sobolev space $H^m_{\eq,\Fourier,A}(\R^d,\Hil)$ with elements of the non-equivariant magnetic Sobolev space $H^m_{\Fourier,A}(\R^d,\Hil)$. This is facilitated through the following injections: 
\begin{definition}\label{appendix:equivariant_operators:defn:inclusion_map_L2_BZ_gamma_ast_L2_Rd}
	Fix $\gamma^* \in \Gamma^*$. Let $\chi \in \Cont^{\infty}_{\mathrm{c}}(\R^d,\R)$ be a cutoff function, namely $\ran \chi = [0,1]$, for the unit cell located at $\BZ_{\gamma^*}$, \ie $\chi \vert_{\BZ_{\gamma^*}} = 1$. 
	Then we define the maps 
	\begin{align*}
		\imath_{\chi,\R^d} &: L^2_{\eq}(\R^d,\Hil) \longrightarrow L^2(\R^d,\Hil) 
		, \quad
		\bigl ( \imath_{\chi,\R^d}(\psi) \bigr )(k) := \chi(k) \, \psi(k) 
		\\
		\imath_{\chi} &: L^2(\BZ_{\gamma^*},\Hil) \longrightarrow L^2(\R^d,\Hil) 
		, \quad
		\imath_{\chi} := \imath_{\chi,\R^d} \circ U_{\gamma^*}^{-1} 
		. 
	\end{align*}
\end{definition}
Of course, we will have to show that these maps are well defined (case $m = 0$ below) and continuous, and that they restrict to continuous maps between magnetic Sobolev spaces. 
\begin{lemma}\label{appendix:equivariant_operators:lem:continuity_inclusion_map_L2_BZ_gamma_ast_L2_Rd}
	Fix any $\gamma^* \in \Gamma^*$ and $m \geq 0$. Then the maps 
	\begin{align*}
		\imath_{\chi} : H^m_{\eq,\Fourier,A}(\BZ_{\gamma^*},\Hil) \longrightarrow H^m_{\Fourier,A}(\R^d,\Hil) 
		, 
		\\
		\imath_{\chi,\R^d} : H^m_{\eq,\Fourier,A}(\R^d,\Hil) \longrightarrow H^m_{\Fourier,A}(\R^d,\Hil) 
		,
	\end{align*}
	between magnetic Sobolev spaces are bounded, \ie continuous. 
\end{lemma}
\begin{proof}
	First of all, it suffices to consider the map $\imath_{\chi,\R^d}$, because $U_{\gamma^*}$ restricts to an isomorphism between magnetic Sobolev spaces. Hence, $\imath_{\chi}$ is bounded whenever $\imath_{\chi,\R^d}$ is. 
	
	Secondly, in principle, we need to ensure $\imath_{\chi}$ and $\imath_{\chi,\R^d}$ are well-defined. But this corresponds to the case $m = 0$, and given that we will make the arguments below for $m \geq 0$, it is not necessary to study the case $m = 0$ separately. 
	
	So let $m \geq 0$. The idea of our proof is to write $\imath_{\chi,\R^d} = M_{\chi} \, \pi_{\Lambda_{\chi}}$ as the product of two operators between suitable magnetic Sobolev spaces, and show that both are bounded, 
	\begin{align}
		&\bnorm{\imath_{\chi,\R^d}}_{\mathcal{B}(H^m_{\eq,\Fourier,A}(\R^d,\Hil) , H^m_{\Fourier,A}(\R^d,\Hil))} 
		\leq \notag \\
		&\qquad 
		\leq \bnorm{M_{\chi}}_{\mathcal{B}(H^m_{\eq,\Fourier,A}(\Lambda_{\chi},\Hil) , H^m_{\Fourier,A}(\R^d,\Hil))} \, \bnorm{\pi_{\Lambda_{\chi}}}_{\mathcal{B}(H^m_{\eq,\Fourier,A}(\R^d,\Hil) , H^m_{\eq,\Fourier,A}(\Lambda_{\chi},\Hil))}
		. 
		\label{appendix:equivariant_operators:eqn:imath_chi_Rd_norm_estimate_product_two_operators}
	\end{align}
	Let us begin with a definition of the second operator $\pi_{\Lambda_{\chi}}$. Here, the set 
	\begin{align*}
  \Lambda_{\chi} := \overline{\bigcup_{\tilde{\gamma}^* \in \mathcal{I}_\chi} \BZ_{\tilde{\gamma}^*}} 
		\subset \R^d 
	\end{align*}
	is the closure of the union of unit cells that has non-empty intersection with $\supp \chi$, defined through the set of reciprocal lattice vectors lying inside $\supp \chi$, 
	\begin{align*}
		\mathcal{I}_\chi := \Bigl \{ \tilde{\gamma}^* \in \Gamma^* \; \; \big \vert \; \; \supp \chi \cap \overline{\BZ_{\tilde{\gamma}^*}} \neq \emptyset \Bigr \} 
		\subset \Gamma^* 
		. 
	\end{align*}
	The cardinality of $\mathcal{I}_\chi$ is merely the number of Brillouin zones that intersect with the support of $\chi$; as the support is compact by assumption, $\sabs{\mathcal{I}_\chi} < \infty$ is necessarily finite. 
	To make sure that $\supp \chi \subset \mathrm{Int}(\Lambda_{\chi})$ lies in the \emph{interior} of $\Lambda_{\chi}$, a fact we will later need, the set $\mathcal{I}_\chi$ is defined with respect to the \emph{closure} of the unit cells. 
	
	We label the restriction of $H^m_{\eq,\Fourier,A}(\R^d,\Hil)$ to $\Lambda_{\chi}$ with 
	\begin{align*}
		\pi_{\Lambda_{\chi}} : H^m_{\eq,\Fourier,A}(\R^d,\Hil) \longrightarrow H^m_{\eq,\Fourier,A}(\Lambda_{\chi},\Hil) 
		, \quad
		\psi \mapsto \psi \vert_{\Lambda_{\chi}} 
	\end{align*}
	We may view the equivariant magnetic Sobolev space 
	\begin{align}
		H^m_{\eq,\Fourier,A}(\Lambda_{\chi},\Hil) \cong \bigoplus_{\tilde{\gamma}^* \in \mathcal{I}_\chi} H^m_{\eq,\Fourier,A}(\BZ_{\tilde{\gamma}^*},\Hil)
		\label{appendix:equivariant_operators:eqn:plaquette_decomposition_equivariant_H_m_A_Lambda_chi}
	\end{align}
	that acts as the target space either as the direct sum of the equivariant magnetic Sobolev spaces on the relevant unit cells or define it as those elements $\psi$ of $L^2(\Lambda_{\chi},\Hil)$ for which the equivariant Sobolev norm 
	\begin{align*}
		\snorm{\psi}_{H^m_{\eq,\Fourier,\Hil}(\Lambda_{\chi},\Hil)} := \sqrt{\Bscpro{\bexpval{\KA \vert_{\Lambda_{\chi}}}^m \psi \, }{ \, \bexpval{\KA \vert_{\Lambda_{\chi}}}^m \psi}_{L^2(\Lambda_{\chi},\Hil)}} < \infty 
	\end{align*}
	is finite. Note that equivariance is still imposed since we use the \emph{equivariant} momentum operator $\KA$ in the norm. 
	
	The operator norm of $\pi_{\Lambda_{\chi}}$ is the square root of the number of unit cells that intersects with $\supp \chi$: 
	\begin{align*}
		\bnorm{\pi_{\Lambda_{\chi}}(\psi)}_{H^m_{\eq,\Fourier,A}(\Lambda_{\chi},\Hil)}^2 &= \bnorm{\psi \vert_{\Lambda_{\chi}}}_{H^m_{\eq,\Fourier,A}(\Lambda_{\chi},\Hil)}^2 
		= \sum_{\tilde{\gamma}^* \in \mathcal{I}_\chi} \bnorm{\psi \vert_{\BZ_{\tilde{\gamma}^*}}}_{H^m_{\eq,\Fourier,A}(\BZ_{\tilde{\gamma}^*},\Hil)}^2 
		\\
		&= \sabs{\mathcal{I}_\chi} \, \bnorm{\psi \vert_{\BZ_{\gamma^*}}}_{H^m_{\eq,\Fourier,A}(\BZ_{\gamma^*},\Hil)}^2
		= \sabs{\mathcal{I}_\chi} \, \snorm{\psi}_{H^m_{\eq,\Fourier,A}(\R^d,\Hil)}^2 
	\end{align*}
	Hence, the second operator $\pi_{\Lambda_{\chi}}$ is bounded. 
	
	Now it is our turn to define the other operator, 
	\begin{align*}
		M_{\chi} : \; \; &H^m_{\eq,\Fourier,A}(\Lambda_{\chi},\Hil) \longrightarrow H^m_{\eq,\Fourier,A}(\R^d,\Hil) 
		\\
		&\bigl ( M_{\chi} \psi_{\Lambda_{\chi}} \bigr )(k) := 
		\begin{cases}
			\chi(k) \, \psi_{\Lambda_{\chi}}(k) & k \in \Lambda_{\chi} \\
			0 & k \not\in \Lambda_{\chi} \\
		\end{cases}
		. 
	\end{align*}
	Its boundedness will be proven in two steps: we will first furnish the estimate 
	\begin{align}
		\bnorm{M_{\chi}}_{\mathcal{B}(H^m_{\eq,\Fourier,A}(\Lambda_{\chi},\Hil) , H^m_{\Fourier,A}(\R^d,\Hil))} \leq \bnorm{\chi(\hat{k})}_{\mathcal{B}(H^m_{\Fourier,A}(\R^d,\Hil))}
		\label{appendix:equivariant_operators:eqn:norm_estimate_M_chi_chi_hat_k}
	\end{align}
	by the operator norm of the multiplication operator 
	\begin{align*}
		\chi(\hat{k}) : H^m_{\Fourier,A}(\R^d,\Hil) \longrightarrow H^m_{\Fourier,A}(\R^d,\Hil) 
		. 
	\end{align*}
	And then we will give a proof that $\chi(\hat{k})$ is bounded as an operator between \emph{non-equivariant} magnetic Sobolev spaces. 
	
	The proof of \eqref{appendix:equivariant_operators:eqn:norm_estimate_M_chi_chi_hat_k} is subtle. We wish to exploit that $\Lambda_{\chi}$ is a subset of $\R^d$ and that we can drop equivariance. However, the order in which we do that matters, and we cannot do both simultaneously since there is \emph{no continuous embedding} of the equivariant magnetic Sobolev space $H^m_{\eq,\Fourier,A}(\Lambda_{\chi},\Hil)$ into the \emph{non}-equivariant magnetic Sobolev space $H^m_{\Fourier,A}(\R^d,\Hil)$. 
	
	For technical reasons, we will have to introduce yet another magnetic Sobolev space, 
	\begin{align*}
		H^m_{\Fourier,A,0}(\Lambda_{\chi},\Hil) := \overline{\bigl \{ \psi \in \Cont^{\infty}(\Lambda_{\chi},\Hil) \; \; \vert \; \; \supp \psi \subset \mathrm{Int}(\Lambda_{\chi}) \bigr \}}^{\norm{ \, \cdot \, }_{H^m_{\Fourier,A,0}(\Lambda_{\chi},\Hil)}}
	\end{align*}
	defined as the completion of the smooth functions whose support lies in the \emph{interior} of $\Lambda_{\chi}$ with respect to the norm induced by the scalar product 
	\begin{align*}
		\scpro{\varphi}{\psi}_{H^m_{\Fourier,A,0}(\Lambda_{\chi},\Hil)} := \Bscpro{\bexpval{P^A_{\Lambda_{\chi},0}}^m \varphi \, }{ \, \bexpval{P^A_{\Lambda_{\chi},0}}^m \psi}_{L^2(\Lambda_{\chi},\Hil)} 
	\end{align*}
	that involve the magnetic Laplacian 
	\begin{align*}
		\bigl ( P^A_{\Lambda_{\chi},0} \bigr )^2 := \bigl ( \hat{k} - \lambda \, A(\ii \eps \nabla_k) \bigr )^2 
	\end{align*}
	with Dirichlet boundary conditions. Importantly, we can continuously embed 
	\begin{align}
		H^m_{\Fourier,A,0}(\Lambda_{\chi},\Hil) \hookrightarrow H^m_{\Fourier,A}(\R^d,\Hil) 
		. 
		\label{appendix:equivariant_operators:eqn:inclusion_H_m_A_0_Lambda_chi_H_m_A_Rd}
	\end{align}
	into the non-equivariant magnetic Sobolev space on $\R^d$. 
	
	The reason we can involve $H^m_{\Fourier,A,0}(\Lambda_{\chi},\Hil)$ in our arguments is that when computing the operator norm of $M_{\chi}$, it suffices to take the supremum over all functions whose essential support is contained in $\supp \chi$, 
	\begin{align*}
		\bnorm{M_{\chi}}_{\mathcal{B}(H^m_{\eq,\Fourier,A}(\Lambda_{\chi},\Hil) , H^m_{\Fourier,A}(\R^d,\Hil))} &= \sup_{\snorm{\psi}_{H^m_{\eq,\Fourier,A}(\Lambda_{\chi},\Hil)} = 1} \bnorm{M_{\chi} \psi}_{H^m_{\Fourier,A}(\R^d,\Hil))} 
		\\
		&= \sup_{\substack{\snorm{\psi}_{H^m_{\eq,\Fourier,A}(\Lambda_{\chi},\Hil)} = 1 \\ \mathrm{ess \, supp} \, \psi \subseteq \supp \chi}} \bnorm{M_{\chi} \psi}_{H^m_{\Fourier,A}(\R^d,\Hil))} 
		. 
	\end{align*}
	By construction, $\supp \chi \subset \mathrm{Int}(\Lambda_{\chi})$ lies in the interior of $\Lambda_{\chi}$ and hence, the boundary conditions on $\partial \Lambda_{\chi}$ do not matter for functions that approximately maximize the norm. Restricted to $\supp \chi$, we may view $H^m_{\Fourier,A,0}(\Lambda_{\chi},\Hil)$ as a superset of 
	\begin{align*}
		H^m_{\eq,\Fourier,A}(\Lambda_{\chi},\Hil) \big \vert_{\supp \chi} \subset H^m_{\Fourier,A,0}(\Lambda_{\chi},\Hil) \big \vert_{\supp \chi} 
		, 
	\end{align*}
	since all we do is drop the equivariance condition. We combine this with the inclusion~\eqref{appendix:equivariant_operators:eqn:inclusion_H_m_A_0_Lambda_chi_H_m_A_Rd} to arrive at 
	\begin{align*}
		\bnorm{M_{\chi}}_{\mathcal{B}(H^m_{\eq,\Fourier,A}(\Lambda_{\chi},\Hil) , H^m_{\Fourier,A}(\R^d,\Hil))} &\leq \sup_{\substack{\snorm{\psi}_{H^m_{\Fourier,A,0}(\Lambda_{\chi},\Hil)} = 1 \\ \mathrm{ess \, supp} \, \psi \subseteq \supp \chi}} \bnorm{\chi(\hat{k}) \psi}_{H^m_{\Fourier,A}(\R^d,\Hil))} 
		\\
		&= \sup_{\snorm{\psi}_{H^m_{\Fourier,A,0}(\Lambda_{\chi},\Hil)} = 1} \bnorm{\chi(\hat{k}) \psi}_{H^m_{\Fourier,A}(\R^d,\Hil))} 
		\\
		&\leq \sup_{\snorm{\psi}_{H^m_{\Fourier,A}(\R^d,\Hil)} = 1} \bnorm{\chi(\hat{k}) \psi}_{H^m_{\Fourier,A}(\R^d,\Hil))} 
		= \bnorm{\chi(\hat{k})}_{\mathcal{B}(H^m_{\Fourier,A}(\R^d,\Hil))} 
		. 
	\end{align*}
	In conclusion, we have shown estimate~\eqref{appendix:equivariant_operators:eqn:norm_estimate_M_chi_chi_hat_k}, and proceed to prove that the operator norm on the right is finite as well. 
	
	Let us introduce some better notation to write out the $H^m_{\Fourier,A}(\R^d,\Hil)$ norm. Specifically, we abbreviate the kinetic momentum operator in momentum representation with 
	\begin{align*}
		P^A_{\Fourier} := \Fourier \, \bigl ( - \ii \nabla_x - \lambda \, A(\eps \hat{x}) \bigr ) \, \Fourier^{-1} = \hat{k} - \lambda \, A(\ii \eps \nabla_k)
		. 
	\end{align*}
	Then the weights that enter the scalar product of $H^m_{\Fourier,A}(\R^d,\Hil)$ are $\sexpval{P^A_{\Fourier}}^m = \Op_{\Fourier}^A(\sexpval{k}^m)$; they can also be viewed as a magnetic pseudodifferential operator. Note that Assumption~\ref{equivariant_calculus:assumption:setting} applies and hence, the components of the magnetic field $B$ are all $\Cont^{\infty}_{\mathrm{b}}$ functions. 
	
	We can replace the weights $\sexpval{P^A_{\Fourier}}^m = \Op_{\Fourier}^A(\sexpval{k}^m)$ with the weights 
	\begin{align*}
		w_m(P^A_{\Fourier}) = \sexpval{P^A_{\Fourier}}^m = \Op_{\Fourier}^A(\sexpval{k}^m) + \lambda(m) 
	\end{align*}
	from equation~\eqref{operator_valued_calculus:eqn:definition_alternate_weights_magnetic_Sobolev_spaces} that lead to an \emph{equivalent} norm on the magnetic Sobolev space; they differ only by a constant $\lambda(m) \geq 0$ chosen so that $w_m(P^A_{\Fourier})$ is an invertible operator with bounded inverse. 
	
	Owing to our assumptions on the magnetic field (Assumption~\ref{equivariant_calculus:assumption:setting}~(a)), we know $w_m(P^A_{\Fourier}) = \Op_{\Fourier}^A(w_m)$ and its inverse 
	\begin{align*}
		w_m(P^A_{\Fourier})^{-1} = \Op_{\Fourier}^A \bigl ( w_m^{(-1)_{\Weyl}} \bigr ) 
		= \Op_{\Fourier}^A(w_{-m})
	\end{align*}
	are magnetic pseudodifferential operators associated to Hörmander symbols, 
	\begin{align*}
		w_{\pm m} \in S^{\pm m}_{1,0}(\C) 
		. 
	\end{align*}
	Corollary~\ref{operator_valued_calculus:cor:Calderon_Vaillancourt} tells us that 
	\begin{align*}
		\Op_{\Fourier}^A(w_{\pm m}) : H^s_{\Fourier,A}(\R^d,\Hil) \longrightarrow H^{s \mp m}_{\Fourier,A}(\R^d,\Hil) 
	\end{align*}
	define continuous operators between magnetic Sobolev spaces for any $s \in \R$. What is more, by construction these operators are invertible with bounded inverse. 
	
	Hence, we can compute the relevant operator norm by inserting 
	\begin{align*}
		\id_{H^s_{\Fourier,A}(\R^d,\Hil)} = w_m(P^A_{\Fourier}) \, w_{-m}(P^A_{\Fourier}) 
	\end{align*}
	and writing out the Sobolev norm, 
	\begin{align*}
		\bnorm{\chi(\hat{k})}_{\mathcal{B}(H^m_{\Fourier,A}(\R^d,\Hil))} &= \sup_{\snorm{\Psi}_{H^m_{\Fourier,A}(\R^d,\Hil)} = 1} \bnorm{\chi(\hat{k}) \Psi}_{H^m_{\Fourier,A}(\R^d,\Hil)} 
		\\
		&= \sup_{\snorm{\Phi}_{L^2(\R^d,\Hil)} = 1} \bnorm{\sexpval{P^A_{\Fourier}}^m \, \chi(\hat{k}) \, w_{-m}(P^A_{\Fourier}) \, w_m(P^A_{\Fourier}) \, \Phi}_{H^m_{\Fourier,A}(\R^d,\Hil)} 
		\\
		&\leq \bnorm{\sexpval{P^A_{\Fourier}}^m}_{\mathcal{B}(L^2(\R^d,\Hil),H^m_{\Fourier,A}(\R^d,\Hil))} \, \bnorm{\chi(\hat{k})}_{\mathcal{B}(L^2(\R^d,\Hil))} 
		\cdot \\
		&\quad \, \cdot 
		\, \bnorm{w_{-m}(P^A_{\Fourier})}_{\mathcal{B}(L^2(\R^d,\Hil))} \, \bnorm{w_m(P^A_{\Fourier})}_{\mathcal{B}(H^m_{\Fourier,A}(\R^d,\Hil) , L^2(\R^d,\Hil))} 
		\\
		&= \bnorm{\id_{H^m_{\Fourier,A}(\R^d,\Hil)}}_{\mathcal{B}(H^m_{\Fourier,A}(\R^d,\Hil))} \, \bnorm{\chi(\hat{k})}_{\mathcal{B}(L^2(\R^d,\Hil))} 
		\cdot \\
		&\quad \, \cdot 
		\, \bnorm{w_{-m}(P^A_{\Fourier})}_{\mathcal{B}(L^2(\R^d,\Hil))} \, \bnorm{\sexpval{P^A_{\Fourier}}^m + \lambda(m)}_{\mathcal{B}(H^m_{\Fourier,A}(\R^d,\Hil) , L^2(\R^d,\Hil))} 
		. 
	\end{align*}
	The operator norm of the first factor, $\bnorm{\id_{H^m_{\Fourier,A}(\R^d,\Hil)}}_{\mathcal{B}(H^m_{\Fourier,A}(\R^d,\Hil))} = 1$, is evidently $1$. 
	
	The rapid decay of $\chi$ implies it can be regarded as an element of $S^{-\infty}(\C)$, and hence, by the Calderón-Vaillancourt Theorem~\ref{operator_valued_calculus:thm:Calderon_Vaillancourt} the operator $\chi(\hat{k}) = \Op_{\Fourier}^{A = 0}(\chi)$ defines a bounded operator on $L^2(\R^d,\Hil)$. 
	
	For the third operator norm, we exploit the inclusion $S^{-m}_{1,0}(\C) \subset S^0_{1,0} \bigl ( \mathcal{B}(\Hil) \bigr )$ of Hörmander classes and deduce once more from the Calderón-Vaillancourt Theorem~\ref{operator_valued_calculus:thm:Calderon_Vaillancourt} that 
	\begin{align*}
		w_{-m}(P^A_{\Fourier}) = \Op_{\Fourier}^A(w_{-m}) \in \mathcal{B} \bigl ( L^2(\R^d,\Hil) \bigr ) 
	\end{align*}
	defines a bounded operator on $L^2(\R^d,\Hil)$. 
	
	The norm of the fourth and final term in the product is bounded as well, courtesy of Corollary~\ref{operator_valued_calculus:cor:Calderon_Vaillancourt}. 
	
	This means $\chi(\hat{k}) \in \mathcal{B} \bigl ( H^m_{\Fourier,A}(\R^d,\Hil) \bigr )$ is a continuous linear operator, which, in view of equation~\eqref{appendix:equivariant_operators:eqn:imath_chi_Rd_norm_estimate_product_two_operators}, also proves the continuity of $\imath_{\chi,\R^d}$. 
\end{proof}
Armed with this knowledge, we can provide a proof of the last remaining Corollary. 
\begin{proof}[Corollary~\ref{appendix:equivariant_operators:cor:equivalence_boundedness_equivariant_operators}]
	\begin{enumerate}[(1)]
		\item Given that the domain $\domain(\widehat{F}^A) = H^{q + q'}_{\Fourier,A}(\R^d,\Hil) \subseteq L^2(\R^d,\Hil)$ is dense, Proposition~\ref{appendix:equivariant_operators:prop:characterization_equivariant_operators} applies and the operators $F_{\gamma^*}^A$ and $F^A$ exist as densely defined operators. 
		
		Let us start with the family of operators from characterization~(b). Pick some $\gamma^* \in \Gamma^*$. The equivariance condition~\eqref{appendix:equivariant_operators:eqn:equivariance_condition_domain_L2_Rd} implies that the domains $\domain(F_{\gamma^*}) = U_{\gamma^*,\tilde{\gamma}^*} \, \bigl ( \domain(F_{\tilde{\gamma}^*}^A) \bigr )$ for different reciprocal lattice vectors $\gamma^* , \tilde{\gamma}^* \in \Gamma^*$ are related by a combination of translations and the group action $\tau(\gamma^* - \tilde{\gamma}^*)$ (\cf Lemma~\ref{appendix:equivariant_operators:lem:weighted_Hilbert_spaces_translated_Brillouin_zones}~(4)). Consequently, imposing the equivariance condition on the kinetic momentum operator and restricting it to functions on $\BZ_{\gamma^*}$, 
		\begin{align*}
			\bigl ( P^A \vert_{\BZ_{\gamma^*}} \bigr ) \big \vert_{\eq} &= \bigl ( \bigl ( \hat{k} + \lambda \, A(\ii \eps \nabla_k) \bigr ) \bigr ) \big \vert_{\BZ_{\gamma^*}} \Big \vert_{\eq} 
			\\
			&= \bigl ( \hat{k} + \lambda \, A(\Reps) \bigr ) \big \vert_{\BZ_{\gamma^*}}
			= \KA \big \vert_{\BZ_{\gamma^*}}
			, 
		\end{align*}
		gives the equivariant kinetic momentum operator restricted to a single fundamental cell. This is precisely the operator that enters into the definition of $H^{q + q'}_{\eq,\Fourier,A}(\BZ_{\gamma^*},\Hil)$ through the weight 
		\begin{align*}
			\bigl ( \sexpval{P^A}^{q + q'} \bigr ) \big \vert_{\BZ_{\gamma^*}} \Big \vert_{\eq} = \sexpval{\KA}^{q + q'} \big \vert_{\BZ_{\gamma^*}} 
			. 
		\end{align*}
		Consequently, 
		the domain $\domain(F_{\gamma^*}^A) = H^{q + q'}_{\eq,\Fourier,A}(\BZ_{\gamma^*},\Hil)$ is just the magnetic Sobolev space of order $q + q'$, where we have imposed equivariant boundary conditions. 
		
		Translated to the operator $F^A$ from characterization~(c), the above arguments show that the domain is $H^{q + q'}_{\eq,\Fourier,A}(\R^d,\Hil)$. 
		
		It remains to prove that the operators $F^A_{\gamma^*}$, $\gamma^* \in \Gamma^*$, and $F^A$ are bounded when $\widehat{F}^A$ is. Our construction of $F^A_{\gamma^*}$ from $\widehat{F}^A$ by splitting up $L^2(\R^d,\Hil)$ into a sum over shifted Brillouin tori and the definition of the map $\imath_{\chi}$ (\cf Definition~\ref{appendix:equivariant_operators:defn:inclusion_map_L2_BZ_gamma_ast_L2_Rd}) tells us that the action of 
		\begin{align*}
			\bigl( \widehat{F}^A \imath_{\chi}(\psi_{\gamma^*}) \bigr ) \big \vert_{\BZ_{\gamma^*}} = F^A_{\gamma^*} \psi_{\gamma^*}
		\end{align*}
		on the relevant unit cell $\BZ_{\gamma^*}$ coincides with the action of $F^A_{\gamma^*}$. 
		
		Consequently, for any $\psi_{\gamma^*} \in H^{q + q'}_{\eq,\Fourier,A}(\BZ_{\gamma^*},\Hil)$ this yields the first estimate of 
		\begin{align*}
			\bnorm{F^A_{\gamma^*} \psi_{\gamma^*}}_{\mathfrak{h}_{\gamma^*}(\Hil')} &\leq \bnorm{\widehat{F}^A}_{\mathcal{B}(H^{q + q'}_{\Fourier,A}(\R^d,\Hil),L^2(\R^d,\Hil'))} \, \bnorm{\imath_{\chi}(\psi_{\gamma^*})}_{H^{q + q'}_{\Fourier,A}(\R^d,\Hil)} 
			\\
			&\leq C \, \bnorm{\widehat{F}^A}_{\mathcal{B}(H^{q + q'}_{\Fourier,A}(\R^d,\Hil),L^2(\R^d,\Hil'))} \, \snorm{\psi_{\gamma^*}}_{H^{q + q'}_{\eq,\Fourier,A}(\BZ_{\gamma^*},\Hil)} 
			, 
		\end{align*}
		and the continuity of the map $\imath_{\chi} : H^{q + q'}_{\eq,\Fourier,A}(\BZ_{\gamma^*},\Hil) \longrightarrow H^{q + q'}_{\Fourier,A}(\R^d,\Hil)$ (Lemma~\ref{appendix:equivariant_operators:lem:continuity_inclusion_map_L2_BZ_gamma_ast_L2_Rd}) the second. Hence, boundedness of $F^A_{\gamma^*}$ follows from the boundedness of $\widehat{F}^A$.
		
		Moreover, we can use the construction of $F^A$ from $F^A_{\gamma^*}$ in step (b)~$\Longrightarrow$~(c) of the proof of Proposition~\ref{appendix:equivariant_operators:prop:characterization_equivariant_operators} to infer the boundedness of $F^A$: regarding the domain, the unitary $U_{\gamma^*} : L^2_{\eq}(\R^d,\Hil) \longrightarrow \mathfrak{h}_{\gamma^*}(\Hil)$ from Lemma~\ref{appendix:equivariant_operators:lem:weighted_Hilbert_spaces_translated_Brillouin_zones}~(3) restricts to an isomorphism 
		\begin{align*}
			U_{\gamma^*} \big \vert_{H^{q + q'}_{\eq,\Fourier,A}(\BZ_{\gamma^*},\Hil)} : H^{q + q'}_{\eq,\Fourier,A}(\BZ_{\gamma^*},\Hil) \longrightarrow H^{q + q'}_{\eq,\Fourier,A}(\R^d,\Hil)
			, \quad
			\psi \mapsto \psi \vert_{\BZ_{\gamma^*}}
		\end{align*}
		between magnetic Sobolev spaces. Combining this with the boundedness of $F^A_{\gamma^*}$ yields the boundedness of $F^A$. 
		\item When $q + q' = 0$, the magnetic Sobolev spaces reduce to the corresponding $L^2$ spaces, \eg $H^0_{\eq,\Fourier,A}(\BZ_{\gamma^*},\Hil) = L^2(\BZ_{\gamma^*},\Hil)$. 
		
		So suppose $q + q' > 0$, but assume all the $A_j \in L^{\infty}(\R^d,\R)$, $j = 1 , \ldots , d$, are bounded. Then as the sum of two bounded operators, the components of the kinetic momentum operator $\KA_j \vert_{\BZ_{\gamma^*}} \in \mathcal{B} \bigl ( L^2(\BZ_{\gamma^*},\Hil) \bigr )$ are all bounded. Therefore, $H^{q + q'}_{\eq,\Fourier,A}(\BZ_{\gamma^*},\Hil)$ agrees with $L^2(\BZ_{\gamma^*},\Hil)$ as \emph{Banach} spaces and $F^A_{\gamma^*}$ is bounded. In view of Proposition~\ref{equivariant_calculus:prop:equivalence_boundedness_equivariant_operators} this means all $F^A_{\gamma^*}$, $\gamma^* \in \Gamma^*$, and $F^A$ are bounded as well. 
		\item The fact that the relations between the operators transform in a gauge-covariant way follows from 
		\begin{align*}
			P^{A + \eps \dd \vartheta}_{\Fourier} := \Fourier \, \bigl ( - \ii \nabla_x - \lambda \, A(\eps \hat{x}) - \lambda \, \eps \, \nabla_x \vartheta(\eps x) \bigr ) 
			= \e^{+ \ii \lambda \vartheta(\ii \eps \nabla_k)} \, P^A_{\Fourier} \, \e^{- \ii \lambda \vartheta(\ii \eps \nabla_k)} 
		\end{align*}
		and similar relations for the other kinetic momentum operators. Those imply that the gauge transformations relate the various magnetic Sobolev spaces as well as the operators for the gauges $A$ and $A + \eps \dd \vartheta$. 
	\end{enumerate}
\end{proof}
%


\printbibliography

@article{Mantoiu_Purice:magnetic_Weyl_calculus:2004,
	author = {Marius Măntoiu and Radu Purice},
	title = "{The Magnetic Weyl Calculus}",
	journal = {J. Math. Phys.},
	volume = {45},
	number = {4},
	pages = {1394–1417},
	doi = {10.1063/1.1668334},
	year = {2004}
}

@article{Mueller:product_rule_gauge_invariant_Weyl_symbols:1999,
	author = {M. Müller},
	journal = {J. Phys. A},
	pages = {1035–1052},
	title = "{Product rule for gauge invariant Weyl symbols and its application to the semiclassical description of guiding centre motion }",
	url = {http://arxiv.org/abs/math-ph/0401043},
	volume = {32},
	year = {1999}
}

@book{Folland:harmonic_analysis_hase_space:1989,
	author = {Gerald B. Folland},
	title = "{Harmonic Analysis on Phase Space}",
	series = {Annals of Mathematics Studies},
	volume = {122},
	publisher = {Princeton University Press},
	isbn = {978-0691085289},
	year = {1989}
}

@inbook{Mantoiu_Purice_Richard:twisted_X_products:2004,
	booktitle = "{Operator Algebras and Mathematical Physics}",
	author = {Marius Măntoiu and Radu Purice and Serge Richard},
	title = "{Twisted Crossed Products and Magnetic Pseudodifferential Operators}",
	pages = {137–172},
	publisher = {Theta},
	url = {http://arxiv.org/abs/math-ph/0403016},
	year = {2005}
}

@article{Mantoiu_Purice_Richard:Cstar_algebraic_framework:2007,
	author = {Marius Măntoiu and Radu Purice and Serge Richard},
	Doi = {DOI: 10.1016/j.jfa.2007.05.020},
	Issn = {0022-1236},
	journal = {Journal of Functional Analysis},
	Number = {1},
	pages = {42 - 67},
	title = "{Spectral and propagation results for magnetic Schrödinger operators; A C*-algebraic framework}",
	url = {http://www.sciencedirect.com/science/article/B6WJJ-4P5RM74-3/2/e871df164b3a194735a5ac539f6b9aa2},
	volume = {250},
	year = {2007}
}

@article{PST:effective_dynamics_Bloch:2003,
	author = {Gianluca Panati and Herbert Spohn and Stefan Teufel},
	Eprint = {0212041},
	journal = {Commun. Math. Phys.},
	Month = {10},
	pages = {547–578},
	title = "{Effective dynamics for Bloch electrons: Peierls substitution and Beyond}",
	volume = {242},
	issue = {3},
	doi = {10.1007/s00220-003-0950-1},
	issn = {1432-0916},
	year = {2003}
}

@article{PST:sapt:2002,
	author = {Gianluca Panati and Herbert Spohn and Stefan Teufel},
	title = "{Space-Adiabatic Perturbation Theory}",
	journal = {Adv. Theor. Math. Phys.},
	pages = {145–204},
	volume = {7},
	number = {1},
	doi = {10.4310/ATMP.2003.v7.n1.a6},
	year = {2003}
}

@article{Iftimie_Mantoiu_Purice:magnetic_psido:2006,
	author = {Viorel Iftimie and Marius Măntoiu and Radu Purice},
	title = "{Magnetic Pseudodifferential Operators}",
	journal = {Publications of the Research Institute for Mathematical Sciences},
	volume = {43},
	issue = {3},
	pages = {585–623},
	doi = {10.2977/prims/1201012035},
	year = {2007}
}

@article{Panati_Sparber_Teufel:polarization:2006,
	author = {Gianluca Panati and Christof Sparber and Stefan Teufel},
	affiliation = {Università di Roma “La Sapienza” Dipartimento di Matematica Piazzale Aldo Moro 2 00185 Roma Italy},
	title = "{Geometric Currents in Piezoelectricity}",
	journal = {Arch. Rational Mech. Anal.},
	publisher = {Springer Berlin/Heidelberg},
	issn = {0003-9527},
	pages = {387-422},
	volume = {191},
	issue = {3},
	doi = {10.1007/s00205-008-0111-y},
	year = {2009}
}

@article{Fuerst_Lein:scaling_limits_Dirac:2008,
	author = {Martin Fürst and Max Lein},
	title = "{Semi- and Non-relativistic Limit of the Dirac Dynamics with External Fields}",
	journal = {Annales Henri Poincaré},
	volume = {14},
	issue = {5},
	pages = {1305–1347},
	doi = {10.1007/s00023-012-0213-9},
	year = {2013}
}

@article{Iftimie_Mantoiu_Purice:commutator_criteria:2008,
	author = {Viorel Iftimie and Marius Măntoiu and Radu Purice},
	journal = {Communications in Partial Differential Equations},
	title = "{Commutator Criteria for Magnetic Pseudodifferential Operators}",
	volume = {35},
	issue = {6},
	pages = {1058–1094},
	doi = {10.1080/03605301003717118},
	year = {2010}
}

@book{Teufel:adiabatic_perturbation_theory:2003,
	author = {Stefan Teufel},
	title = "{Adiabatic Perturbation Theory in Quantum Dynamics}",
	publisher = {Springer-Verlag},
	series = {Lecture Notes in Mathematics},
	volume = {1821},
	isbn = {978-3-540-40723-2},
	year = {2003}
}

@article{Cordes:pseudodifferential_FW_transform:2004,
	author = {H. O. Cordes},
	Doi = {10.1007/s00028-003-0128-5},
	journal = {Journal of Evolution Equations},
	Month = {March},
	Number = {1},
	Numpages = {13},
	pages = {128–138},
	title = "{A precise pseudodifferential Foldy-Wouthuysen transform for the Dirac equation}",
	url = {http://dx.doi.org/10.1007/s00028-003-0128-5},
	volume = {4},
	year = {2004}
}

@article{Cordes:pseudodifferential_FW_transform:1983,
	author = {H. O. Cordes},
	journal = {Communications in Partial Differential Equations},
	Number = {13},
	Numpages = {10},
	pages = {1475–1485},
	title = "{A pseudodifferential Foldy-Wouthuysen transform}",
	volume = {8},
	year = {1983}
}

@article{Davies:functional_calculus:1995,
	author = {E. B. Davies},
	title = "{The Functional Calculus}",
	journal = {J. London Math. Soc.},
	volume = {52},
	issue = {1},
	pages = {166–176},
	doi = {10.1112/jlms/52.1.166},
	year = {1995}
}

@article{Lein_Mantoiu_Richard:anisotropic_mag_pseudo:2009,
	author = {Max Lein and Marius Măntoiu and Serge Richard},
	title = "{Magnetic pseudodifferential operators with coefficients in $C^*$-algebras}",
	journal = {Publ. RIMS Kyoto Univ.}, 
	volume = {46},
	issue = {4},
	pages = {755–788},
	doi = {10.2977/PRIMS/25},
	year = {2010}
}

@book{Treves:topological_vector_spaces:1967,
	author = {François Treves},
	title = "{Topological vector spaces, distributions and kernels}",
	publisher = {Academic Press}, 
	year = {1967}
}

@article{DeNittis_Lein:Bloch_electron:2009,
	author = {Giuseppe {De Nittis} and Max Lein}, 
	title = "{Applications of Magnetic $\Psi$DO Techniques to SAPT – Beyond a simple review}", 
	journal = {Rev. Math. Phys.}, 
	volume = {23},
	issue = {3},
	pages = {233-260},
	doi = {10.1142/S0129055X11004278},
	year = {2011}
}

@phdthesis{Lein:progress_magWQ:2010, 
	author = {Max Lein}, 
	title = "{Semiclassical Dynamics and Magnetic Weyl Calculus}", 
	school = {Technische Universität München, Munich, Germany},
	type = {PhD thesis},
	url = {http://arxiv.org/abs/1202.4668},
	year = {2011}
}

@article{Lein:two_parameter_asymptotics:2008,
	author = {Max Lein}, 
	title = "{Two-parameter Asymptotics in Magnetic Weyl Calculus}", 
	journal = {J. Math. Phys.}, 
	volume = {51},
	issue = {12},
	pages = {123519},
	doi = {10.1063/1.3499660},
	year = {2010}
}

@article{PST:Born-Oppenheimer:2007,
	author = {Panati, Gianluca and Spohn, Herbert and Teufel, Stefan},
	title = "{The time-dependent Born-Oppenheimer approximation}",
	DOI= {10.1051/m2an:2007023},
	url= {http://dx.doi.org/10.1051/m2an:2007023},
	journal = {M2AN},
	year = {2007},
	volume = {41},
	number = {2},
	pages = {297-314},
	month = {March}
}

@book{Reed_Simon:M_cap_Phi_1:1972,
	author = {Michael Reed and Barry Simon},
	title = "{Methods of Mathematical Physics I: Functional Analysis}",
	series = {Methods of Modern Mathematical Physics},
	volume = {1},
	publisher = {Academic Press}, 
	isbn = {978-0125850506},
	year = {1972}
}

@article{Beals:characterization_psido:1977,
	author = {R. Beals},
	title = "{Characterization of Pseudodifferential Operators}",
	journal = {Duke Mathematics Journal},
	volume = {44},
	pages = {45-57},
	doi = {10.1215/S0012-7094-77-04402-7},
	year = {1977}
}

@article{Bony:characterization_psido:1996,
	author = {J. M. Bony},
	title = "{Caractérisacion des opérateurs pseudo-différentiels}",
	journal = {École Polytechnique, Séminaire EDP},
	volume = {XXIII},
	year = {1996}
}

@book{Dimassi_Sjoestrand:spectral_asymptotics:1999,
	author = {M. Dimassi and Johannes Sjöstrand},
	title = "{Spectral Asymtptotics in the Semi-Classical Limit}",
	publisher = {London Mathematical Society},
	series = {Lecture Notes Series},
	volume = {268},
	isbn = {9780521665445},
	year = {1999}
}

@article{Mantoiu_Purice:continuity_spectra:2009,
	author = {Nassim Athmouni and Marius Măntoiu and Radu Purice},
	title = "{On the continuity of spectra for families of magnetic pseudodifferential operators}",
	journal = "{J. Math. Phys.}",
	volume = {51},
	url = {http://arxiv.org/abs/0912.0652},
	pages = {083517},
	doi = {10.1063/1.3470118},
	year = {2010}
}

@article{Belmonte_Lein_Mantoiu:mag_twisted_actions:2010,
	author = {Fabian Belmonte and Max Lein and Marius Măntoiu},
	title = "{Magnetic twisted actions on general abelian $C^*$-algebras}",
	journal = {Journal of Operator Theory},
	volume = {69},
	issue = {1},
	pages = {33–58},
	doi = {10.7900/jot.2010jun30.1896},
	year = {2013}
}

@article{Nenciu:effective_dynamics_Bloch:1991,
	author = {Gheorghe Nenciu},
	title = "{Dynamics of band electrons in electric and magnetic fields: rigorous justification of the effective Hamiltonians}",
	journal = {Rev. Mod. Phys.},
	volume = {63},
	number = {1},
	pages = {91–127},
	doi = {10.1103/RevModPhys.63.91},
	year = {1991}
}

@article{Stiepan_Teufel:semiclassics_op_valued_symbols:2012,
	author = {Hans-Michael Stiepan and Stefan Teufel},
	title = "{Semiclassical approximations for Hamiltonians with operator-valued symbols}",
	journal = {Commun. Math. Phys.},
	volume = {320},
	issue = {3},
	pages = {821–849},
	doi = {10.1007/s00220-012-1650-5},
	year = {2013}
}

@article{DeNittis_Lein:sapt_photonic_crystals:2013,
	author = {Giuseppe {De Nittis} and Max Lein},
	title = "{Effective Light Dynamics in Perturbed Photonic Crystals}",
	journal = {Commun. Math. Phys.},
	volume = {332},
	issue = {1},
	pages = {221–260},
	doi = {10.1007/s00220-014-2083-0},
	year = {2014}
}

@article{DeNittis_Lein:adiabatic_periodic_Maxwell_PsiDO:2013,
	author = {Giuseppe {De Nittis} and Max Lein},
	title = "{The Perturbed Maxwell Operator as Pseudodifferential Operator}",
	journal = {Documenta Mathematica},
	volume = {19},
	pages = {63–101},
	year = {2014}
}

@article{DeNittis_Lein:ray_optics_photonic_crystals:2014,
	author = {Giuseppe {De Nittis} and Max Lein},
	title = "{Derivation of Ray Optics Equations in Photonic Crystals Via a Semiclassical Limit}",
	journal = {Annales Henri Poincaré},
	volume = {18},
	issue = {5},
	pages = {1789–1831},
	doi = {10.1007/s00023-017-0552-7},
	year = {2017}
}

@article{Panati_Spohn_Teufel:sapt_PRL:2002,
	title = "{Space-Adiabatic Perturbation Theory in Quantum Dynamics}",
	author = {Gianluca Panati and Herbert Spohn and Stefan Teufel},
	journal = {Phys. Rev. Lett.},
	volume = {88},
	issue = {25},
	pages = {250405},
	numpages = {4},
	year = {2002},
	month = {6},
	publisher = {American Physical Society},
	doi = {10.1103/PhysRevLett.88.250405},
	url = {http://link.aps.org/doi/10.1103/PhysRevLett.88.250405}
}

@article{Freund_Teufel:non_trivial_Bloch_sapt:2013,
	author = {Silvia Freund and Stefan Teufel},
	title = "{Peierls substitution for magnetic Bloch bands}",
	journal = {},
	volume = {9},
	number = {4},
	pages = {773–811},
	doi = {10.2140/apde.2016.9.773},
	year = {2016}
}

@article{Lampert_Teufel:adiabatic_limit_Schroedinger_operators_fiber_bundles:2014,
	author = {Jonas Lampert and Stefan Teufel},
	title = "{The adiabatic limit of Schrödinger operators on fibre bundles}",
	journal = {arxiv},
	pages = {1402.0382},
	year = {2014}
}

@article{Iftimie_Purice:magnetic_Fourier_integral_operators:2011,
	author = {Viorel Iftimie and Radu Purice},
	title = "{Magnetic Fourier integral operators}",
	journal = {J. Pseudo-Differ. Oper. Appl.},
	volume = {2},
	issue = {2},
	pages = {141–218},
	doi = {10.1007/s11868-011-0028-3},
	year = {2011}
}

@article{Khosravi_Asgari:tensor_product_normalized_tight_frame:2003,
	author = {Amir Khosravi and M. S. Asgari},
	title = "{Frames and Bases in Tensor Product of Hilbert Spaces}",
	journal = {Intern. Math. Journal},
	volume = {4},
	number = {6},
	pages = {527–537},
	year = {2003}
}

@article{Cornean_Helffer_Purice:simple_proof_Beals_criterion_magnetic_PsiDOs:2018,
	author = {Horia Cornean and Bernard Helffer and Radu Purice},
	title = "{A Beals Criterion for magnetic pseudo-differential operators proved with magnetic Gabor frames}",
	journal = {Comm. PDE},
	volume = {43},
	number = {8},
	pages = {1196–1204},
	doi = {10.1080/03605302.2018.1499777},
	year = {2018}
}

@inbook{Witt:weak_topology_symbol_spaces:1997,
	booktitle = "{Differential Equations, Asymptotic Analysis and Mathematical Physics}",
	editors = {Michael Demuth and Bert-Wolfgang Schulze},
	title = "{The Weak Symbol Topology and Continuity of Pseudodifferential Operators}",
	author = {Ingo Witt},
	pages = {412–422},
	publisher = {Wiley-VCH Verlag GmbH},
	isbn = {978-3055017698},
	year = {1997}
}

@article{DeNittis_Lein_Seri:semiclassics_Bloch_electron_Fermi_surface:2021,
	author = {Giuseppe {De Nittis} and Max Lein and Marcello Seri},
	title = "{Semiclassical Dynamics for Bloch Electrons on Fermi Surfaces}",
	journal = {in preparation},
	year = {2021}
}

@article{Mantoiu_Purice_Richard:spectral_propagation_results_magnetic_Schroedinger_operators:2007,
	author = {Marius Măntoiu and Radu Purice and Serge Richard},
	title = "{Spectral and propagation results for magnetic Schrödinger operators; A $C^*$-algebraic framework}",
	journal = {Journal of Functional Analysis},
	volume = {250},
	issue = {1},
	pages = {42–67},
	doi = {10.1016/j.jfa.2007.05.020},
	year = {2007}
}

@article{Cornean_Iftimie_Purice:Peierls_substitution_magnetic_PsiDOs:2019,
	author = {Horia D. Cornean and Viorel Iftimie and Radu Purice},
	title = "{Peierls’ substitution via minimal coupling and magnetic pseudo-differential calculus}",
	journal = {Reviews in Mathematical Physics},
	volume = {31},
	number = {3},
	pages = {1950008},
	doi = {10.1142/S0129055X19500089},
	year = {2019}
}

@book{Teschl:real_analysis:2019,
	author = {Gerald Teschl},
	title = "{Topics in Real Analysis}",
	series = {to appear in Graduate Studies in Mathematics},
	publisher = {American Mathematical Society},
	year = {2019}
}

@article{Dabrowski_Kerjean:models_linear_logic_Schwartz_eps_product:2019,
	author = {Yoann Dabrowski and Marie Kerjean},
	title = "{Models of Linear Logic Based on the Schwartz $\eps$-product}",
	journal = {Theory and Applications of Categories},
	volume = {34},
	number = {45},
	pages = {1440–1525},
	year = {2019}
}

@book{Yosida:Functional_analysis:1995,
	doi = {10.1007/978-3-642-61859-8},
	url = {https://doi.org/10.1007/978-3-642-61859-8},
	year = 1995,
	publisher = {Springer Berlin Heidelberg},
	author = {K{\^{o}}saku Yosida},
	title = {Functional Analysis}
}

@book{Shubin:pseudodifferential:2001,
	doi = {10.1007/978-3-642-56579-3},
	url = {https://doi.org/10.1007/978-3-642-56579-3},
	year = 2001,
	publisher = {Springer Berlin Heidelberg},
	author = {Mikhail A. Shubin},
	title = {Pseudodifferential Operators and Spectral Theory}
}

@incollection{McKeag_Safarov:Pseudodifferential_manifolds:2011,
	doi = {10.1007/978-3-0348-0024-2_6},
	url = {https://doi.org/10.1007/978-3-0348-0024-2_6},
	year = 2011,
	publisher = {Springer Basel},
	pages = {321--341},
	author = {Peter McKeag and Yuri Safarov},
	title = {Pseudodifferential Operators on Manifolds: A Coordinate-free Approach},
	booktitle = {Partial Differential Equations and Spectral Theory}
}

@article{Derezinski_al:Pseudodifferential_pseudoriemannian:2020,
	doi = {10.1007/s00023-020-00890-9},
	url = {https://doi.org/10.1007/s00023-020-00890-9},
	year = 2020,
	month = {feb},
	publisher = {Springer Science and Business Media {LLC}},
	volume = {21},
	number = {5},
	pages = {1595--1635},
	author = {Jan Derezi{\'{n}}ski and Adam Latosi{\'{n}}ski and Daniel Siemssen},
	title = {Pseudodifferential Weyl Calculus on (Pseudo-)Riemannian Manifolds},
	journal = {Annales Henri Poincar{\'{e}}}
}

@article{berezanskii1968expansions,
	title={Expansions in eigenfunctions of selfadjoint operators, Transl. Math},
	author={Berezanskii, Ju M},
	journal={Monographs AMS},
	volume={17},
	number={2},
	year={1968}
}

@book{Ruzhansky_Turunen:pseudodifferential_operators_symmetries:2010,
	author = {Michael Ruzhansky and Ville Turunen},
	title = "{Pseudo-Differential Operators — Theory and Applications Vol. 2}",
	publisher = {Birkhäuser},
	isbn = {978-3-7643-8513-2},
	year = {2010}
}

@article{Athmouni_Purice:schatten_class_magnetic_Weyl:2018,
	author = {Nassim Athmouni and Radu Purice},
	title = {A Schatten{\textendash}von Neumann class criterion for the magnetic Weyl calculus},
	journal = {Communications in Partial Differential Equations},
	publisher = {Informa {UK} Limited},
	volume = {43},
	number = {5},
	pages = {733--749},
	doi = {10.1080/03605302.2018.1475486},
	url = {https://doi.org/10.1080/03605302.2018.1475486},
	year = {2018}
}

@article{Rondeaux:schatten_classes_pseudodifferential:1984,
	doi = {10.24033/asens.1466},
	url = {https://doi.org/10.24033/asens.1466},
	year = 1984,
	publisher = {Societe Mathematique de France},
	volume = {17},
	number = {1},
	pages = {67--81},
	author = {Christiane Rondeaux},
	title = {Classes de Schatten d{\textquotesingle}op{\'{e}}rateurs pseudo-diff{\'{e}}rentiels},
	journal = {Annales scientifiques de l{\textquotesingle}{\'{E}}cole normale sup{\'{e}}rieure}
}

@book{Christensen:frames_Riesz_bases:2016,
	author = {Ole Christensen},
	title = "{An Introduction to Frames and Riesz Bases}",
	publisher = {Birkhäuser},
	series = {Applied and Numerical Harmonic Analysis},
	isbn = {978-3-319-25611-5},
	doi = {10.1007/978-3-319-25613-9},
	year = {2016}
}

@article{Brislawn:kernels_trace_class_operators:1988,
	author = {Chris Brislawn},
	title = {Kernels of Trace Class Operators},
	journal = {Proc. Amer. Math. Soc.},
	volume = {104},
	number = {4},
	pages = {1181–1190},
	year = {1988}
}

@book{Simon:trace_ideals_applications:2005,
	author = {Barry Simon},
	title = "{Trace Ideals and Their Applications}",
	publisher = {American Mathematical Society},
	series = {Mathematical Surveys and Monographs},
	volume = {120},
	isbn = {978-0-8218-3581-4},
	year = {2005}
}

@article{Nguyen_Pinamonti_Squassina_Vecchi:charactizations_magnetic_Sobolev_spaces:2020,
	author = {Hoai-Minh Nguyen and Andrea Pinamonti and Marco Squassina and Eugenio Vecchi}, 
	title = "{Some characterizations of magnetic Sobolev spaces}",
	journal = {Complex Variables and Elliptic Equations},
	volume = {65},
	issue = {7},
	pages = {1104–1114},
	doi = {10.1080/17476933.2018.1520850},
	year = {2020}
}

@article{Nguyen_Pinamonti_Squassina_Vecchi:charactizations_magnetic_Sobolev_spaces:2017,
	author = {Hoai-Minh Nguyen and Andrea Pinamonti and Marco Squassina and Eugenio Vecchi}, 
	title = "{New characterizations of magnetic Sobolev spaces}",
	journal = {Advances in Nonlinear Analysis},
	volume = {7},
	number = {2},
	pages = {227–245},
	doi = {10.1515/anona-2017-0239},
	year = {2017}
}

\end{document}